\documentclass[twoside,11pt]{article}

\usepackage{blindtext}

 \usepackage[abbrvbib, preprint]{jmlr2e}

\usepackage{amsmath}

\newtheorem{assumption}{Assumption}

\newcommand{\p}{\mathcal P}

\newcommand{\I}{{\mathcal I }}

\newcommand{\lt}{ {\mathcal L^2}}

 \usepackage{stmaryrd}
\usepackage{xcolor}

\DeclareMathOperator*{\argmin}{arg\,min}

\usepackage{lastpage}
\jmlrheading{26}{2025}{1-\pageref{LastPage}}{3/24; Revised 6/26}{}{}{Carlos Misael Madrid Padilla, Oscar Hernan Madrid Padilla and Daren Wang}

\ShortHeadings{Spatio-temporal model via Locally Adaptive Regression Splines}{Madrid Padilla, Madrid Padilla and Wang}
\firstpageno{1}

\begin{document}

\title{Spatio-temporal  model via Locally Adaptive Regression Splines}

\author{\name Carlos Misael Madrid Padilla \email carlos.madrid@cimat.mx \\
       \addr Department of Statistics and Data Science\\
       Washington University in St. Louis\\
       1 Brookings Drive, St. Louis, MO 63130, USA
       \AND
       \name Oscar Hernan Madrid Padilla \email oscar.madrid@stat.ucla.edu \\
       \addr Department of Statistics and Data Science\\
       University of California, Los Angeles\\
       8125 C/D Math Sciences Bldg.
        Los Angeles, CA 90095-1554, USA
       \AND
       \name Daren Wang \email darenwangoffcampus@gmail.com \\
       \addr Department of Mathematics\\
       University of California, San Diego\\
       9500 Gilman Drive, La Jolla, CA 92093-0112, USA
       }

\editor{}

\maketitle

\begin{abstract}
This paper focuses on the estimation of a non-parametric regression function in the presence of data with spatio-temporal dependencies. In such a context, we study {\color{black}{Locally Adaptive Regression Splines}}, a nonparametric estimator introduced by \cite{mammen1997locally} and \cite{rudin1992nonlinear}. To the best of our knowledge, this estimator has not previously been examined in a similar context. For univariate settings, the signals we consider are assumed to have a kth weak derivative with bounded total variation, allowing for a general degree of smoothness. In the multivariate setting, we study a variant of the $K$-Nearest Neighbor fused lasso estimator studied by \cite{padilla2018adaptive}. For this case,  the function is required to have bounded variation and satisfy a property that extends a piecewise Lipschitz continuity criterion, or the function is assumed to be piecewise Lipschitz. We develop an ADMM algorithm for practical computation. By aligning with lower bounds, the minimax optimality of our univariate and multivariate estimators is shown. A unique phase transition phenomenon, previously unprecedented in {\color{black}{Locally Adaptive Regression Splines}} studies, emerges through our analysis. Both simulation studies and real data applications underscore the superior performance of our method when compared with established techniques in the existing literature.
\end{abstract}

\begin{keywords}
  Spatio-temporal data, Locally Adaptive Regression Splines, Trend Filtering, $K$-Nearest Neighbor fused lasso, minimax rates, nonparametric regression
\end{keywords}

\section{Introduction}

The recent advancements in technology have transformed data collection and analysis,
leading to increased availability of spatio-temporal data in various fields such as
environmental science \citep[e.g.,][]{cressie2015statistics,wen2019novel},
epidemiology \citep[e.g.,][]{aswi2019bayesian,briz2020spatio},
urban planning \citep[e.g.,][]{liu2020high,tian2021assessing},
agriculture \citep[e.g.,][]{dong20174d,jayasiri2022spatio}, and others.
With the surge in the availability of rich datasets that capture both temporal and
spatial dimensions, there has been a corresponding rise in the development and
application of spatio-temporal models. These models address the unique challenges and
harness the opportunities presented by data that span time and space. For a
comprehensive overview of advancements in spatio-temporal modeling, readers are
directed to \cite{wikle2019spatio}.

In this paper, we consider a spatio-temporal nonparametric setting in which the
observations
\(
\{(x_{i,j}, y_{i,j})\}_{i=1,j=1}^{n,m_i} \subseteq [0,1]^d \times \mathbb{R}
\)
are generated according to
\begin{align}\label{eq:model l2}
    y_{i,j}
    = f^*(x_{i,j}) + \delta_i(x_{i,j}) + \epsilon_{i,j},
    \qquad i=1,\ldots,n,\;\; j=1,\ldots,m_i .
\end{align}
Here, $\{x_{i,j}\}_{i=1,j=1}^{n,m_i}\subseteq[0,1]^d$ denote the random spatial locations
where the (noisy) spatio-temporal observations
$\{y_{i,j}\}_{i=1,j=1}^{n,m_i}\subseteq\mathbb{R}$ are collected.
The function $f^*:[0,1]^d\to\mathbb{R}$ represents the unknown deterministic mean
structure, $\{\delta_i:[0,1]^d\to\mathbb{R}\}_{i=1}^n$ denotes spatial noise that varies
over time, and $\{\epsilon_{i,j}\}_{i=1,j=1}^{n,m_i}$ are measurement errors.
In this context, the index $i$ represents the temporal dimension (e.g., days or
replicates), while the index $j$ corresponds to spatial locations observed at each time
point, with $m_i$ denoting the number of spatial observations at time $i$.
We refer to Assumption~\ref{assume:tv functions} for detailed technical conditions on the
model.

Estimating the unknown regression function $f^*$ has been the subject of extensive
research. Wavelet shrinkage methods have been studied for data with long-range
dependence \citep{wang1996function}, while \cite{rice1991estimating} focused on estimating
mean and covariance structures when curves are observed at different locations.
For regression models with correlated errors and smooth mean functions,
\cite{hall1990nonparametric} proposed a nonparametric approach.
In settings with correlated measurements,
\cite{opsomer2001nonparametric} and \cite{rice2001nonparametric} investigated smoothing
methods, while \cite{johnstone1997wavelet} studied wavelet-based approaches and
\cite{cai2011optimal} focused on kernel methods.
Unlike these works, we allow for dependence of the measurements
$\{(x_{i,j}, y_{i,j})\}_{i=1,j=1}^{n,m_i}$ across the temporal index $i$.

{\color{black}{Over the last decade, there has been substantial interest in a class of
$\ell_1$-regularized nonparametric regression estimators based on total variation
penalties, introduced by \citet{mammen1997locally}.
This approach, known as \emph{locally adaptive regression splines}, was further developed
and studied in, e.g.,
\citet{tibshirani2005sparsity,kim2009ell_1,tibshirani2014adaptive}, through a
closely related estimator referred to as \emph{Trend Filtering}.
These methods aim to perform nonparametric regression, typically under independent
errors, by fitting piecewise polynomial functions.
In one dimension, this class of estimators is known to perform well for estimating
signals whose $k$th weak derivative has bounded total variation
\citep{mammen1997locally,tibshirani2014adaptive}.
Specifically, the locally adaptive regression spline estimator solves the optimization
problem given by}}
\begin{equation}\label{eqn:e1}
    \underset{f \in \mathcal{F}_k}{\min}
    \left\{
    \sum_{i=1}^{n} \sum_{j=1}^{m_i}
    \big(y_{i,j}-f(x_{i,j})\big)^2
    + \lambda J_k(f)
    \right\},
\end{equation}
where $J_k(f)=\mathrm{TV}(f^{(k-1)})$, $f^{(k-1)}$ denotes the $(k-1)$st weak derivative of
$f$, and $\lambda>0$ is a tuning parameter. The function class is defined as
\[
\mathcal{F}_k
:= \Big\{
f:[0,1]\to\mathbb{R}
:\; f \text{ is $(k-1)$ times weakly differentiable and }
\mathrm{TV}(f^{(k-1)})<\infty
\Big\}.
\]
For a function $f:[0,1]\to\mathbb{R}$, its total variation is given by
\begin{equation}\label{eqn:tv_operator}
\mathrm{TV}(f)
=
\sup_{\substack{0\le a_1\le\cdots\le a_N\le 1\\ N\in\mathbb{N}}}
\sum_{l=1}^{N-1} |f(a_l)-f(a_{l+1})|.
\end{equation}
When $f$ is differentiable, $\mathrm{TV}(f)=\int_0^1 |f'(t)|\,dt$, and thus the penalty
controls the overall roughness of the function.

A key property of the variational problem \eqref{eqn:e1} is that its minimizer is a
piecewise polynomial function with adaptively chosen knots.
Indeed, \citet{mammen1997locally} showed that the minimizer is a $k$th degree spline.
Thus, for $k=1$ the estimator is piecewise constant, for $k=2$ piecewise linear, for
$k=3$ piecewise quadratic, and so on.

{\color{black}{Closely related to the variational formulation \eqref{eqn:e1} is the estimator commonly
referred to as Trend Filtering, which replaces the continuous total variation
penalty $\mathrm{TV}(f^{(k-1)})$ by a discretized finite-difference analogue defined on a
grid \citep{tibshirani2014adaptive}.
As discussed in \citet{tibshirani2014adaptive}, for $k=0$ and $k=1$ the resulting estimator
coincides exactly with the corresponding locally adaptive regression spline.
For $k\ge 2$, the two procedures are formally distinct, but they are asymptotically
equivalent and empirically very similar in a wide range of settings.
Throughout this paper, we work with the locally adaptive regression spline formulation
\eqref{eqn:e1}, and view Trend Filtering as a closely related discrete analogue based on
the same total variation regularization principle.}}

It is worth noting that in all the existing literature on {\color{black}{locally adaptive regression spline and}} Trend Filtering, the focus has only been on models where $m_i = 1$ for all $i$ and with independent error terms. While such models can be useful for nonparametric regression, they fall short when dealing with spatio-temporal data, where $m_i$ can be larger than $1$, and there can be dependence between the error terms. In this paper, we overcome these limitations by studying a version of {\color{black}{Locally Adaptive Regression Splines}} for spatio-temporal data that also allows for dependence of the measurements across $i$.

In addressing the challenges outlined, {\color{black}{we work with the continuous variational formulation
of locally adaptive regression splines, following \cite{sadhanala2019additive}.  However, we also acknowledge that the literature often utilizes a discrete version commonly referred to as Trend Filtering, see for example \cite{tibshirani2014adaptive}, \cite{guntuboyina2020adaptive}, among others.}} In practice, it is indeed the discrete version, Trend Filtering, that is typically employed to compute the estimator. Next, we provide some notation before presenting a summary of our results.


\subsection{Notation}  \label{sec:notation}
Throughout, given $s \in \mathbb{N}$, we denote by $[s]$ the set $\{1, \ldots, s\}$.
Given a distribution $Q$ supported on $[0,1]^d, $ and, $\{x_{i,j}\}_{i=1,j=1}^{n,m_i}$ identically distributed according to $Q$. Denote by $Q_{nm}$ the associated empirical distribution where $m=\Big(\frac{1}{n}\sum_{i=1}^{n}\frac{1}{m_i}\Big)^{-1}$.
For any function $f:\ [0,1]^d\to \mathbb R$ and for $1 \leq p < \infty$, define $\|f\|_{p} = (\int_{[0,1]^d} |f(x)|^p d Q (x))^{\frac{1}{p}}$ and for $p = \infty$, define $\|f\|_{\infty} = \sup_{x\in{[0,1]}}\vert f(x)\vert$. Denote by $\mathcal{L}_p(Q) = \{f: \, [0, 1]^d \to \mathbb{R}, \, \|f\|_p<\infty \}$. For $p=2$, define  $\langle f,g\rangle_{2}=\int_{[0,1]}f(x)g(x)dQ(x)$  where $f,g:\ [0,1]^d \to \mathbb R$. Similarly, for $Q_{nm}$ we define  $\langle f,g\rangle_{nm}=\frac{1}{n}\sum_{i=1}^{n}\frac{1}{m_i}\sum_{j=1}^{m_i}f(x)g(x)dx$  with corresponding norm $\vert\vert \cdot\vert\vert_{nm}^2=\langle \cdot,\cdot\rangle_{nm}$. Define,  for $t>0$
\begin{equation} 
\label{eqn:tv_ball}
\begin{aligned}
&B_{\mathrm{TV}}^{(k-1)}(t) := \bigg\{ f:[0,1]\to \mathbb{R}: \mathrm{TV}(f^{(k-1)}) \le t \bigg\}, \\
&\text{and} \quad
\,
B_{\infty}(t) := \bigg\{ f:[0,1]\to \mathbb{R}: \| f \|_\infty \le t \bigg\},
\end{aligned}
\end{equation}
where $\mathrm{TV}(\cdot)$ is the total variation operator defined in (\ref{eqn:tv_operator}).

Given a set of elements $\{\theta_{i,j}\}_{i=1,j=1}^{n,m_i}\subset \mathbb{R}$, we denote it as a vector  $\theta\in{\mathbb{R}^{\sum_{i=1}^{n}m_i}}$ and each entry is of the form $\theta_{\llbracket i,j\rrbracket}=\theta_{i,j}$ where $\llbracket i,j\rrbracket=\sum_{k=0}^{i-1}m_k+j$ with $m_0=0$,  for any $i\in\{1,\ldots,n\},j\in\{1,\ldots m_i\}$.
Similarly we make the abuse of notation for $\theta\in{\mathbb{R}^{\sum_{i=1}^{n}m_i}}$ by seeing it as  $\{\theta_{i,j}\}_{i=1,j=1}^{n,m_i}\subset \mathbb{R}$ where $\theta_{i,j}=\theta_{\llbracket i,j\rrbracket}$.
Let now $\widehat{\theta}\in \mathbb{R}^{\sum_{i=1}^nm_i}$ be an estimator of the set of elements  $\{\theta_{i,j}\}_{i=1,j=1}^{n,m_i}\subset \mathbb{R}$, the Mean Square Error (MSE) of $\widehat{\theta}$, is defined as $\frac{1}{n}\sum_{i=1}^{n}\frac{1}{m_i}\sum_{j=1}^{m_i}(\theta_{i,j}-\widehat{\theta}_{i,j})^2.$ For any $A\subset\mathbb{R}^d$ consider $\mathbf{1}_{A}(x)=\Big\{\begin{array}{cc}
     1 &  \text{if} \ x\in A \\
    0 &  \text{otherwise} 
\end{array}$. For $s\in\mathbb{N}$, we set $\mathbf{1}_s=(1,\ldots,1)^T\in\mathbb{R}^s.$ 
Given integers $a$ and $b$, we use $a \% b$ to denote the remainder of $a$ when divided by $b$ (i.e., the modulo operation). For two positive sequences $\{a_n\}_{n\in \mathbb N^+ }$ and $\{b_n\}_{n\in \mathbb N ^+ }$, we write $a_n = O(b_n)$ or $a_n\lesssim b_n$ if $a_n\le Cb_n$ with some constant $C > 0$ that does not depend on $n$, and $a_n = \Theta(b_n)$ or $a_n\asymp b_n$ if $a_n = O(b_n)$ and $b_n = O(a_n)$. Let $X_n$ be a sequence of random $\mathbb{R}$-valued variables.  For a deterministic or random $\mathbb{R}$-valued sequence $a_n$,  we write  $X_n=O_{\mathbb{P}}\left(a_n\right)$ if $\lim _{M \rightarrow \infty} \lim \sup _n \mathbb{P}\left(\left|X_n\right| \geq M a_n\right)=0$. We also write $X_n=o_{\mathbb{P}}\left(a_n\right)$ if $\limsup _n \mathbb{P}\left(\left|X_n\right| \geq M a_n\right)=0$ for all $M>0$. Let $X$ be a random variable in $\mathbb{R}^d$. The $\sigma$-algebra generated by $X$ is denoted as $\sigma(X)$. For a set of random variables $\left(X_t, t \in I\right)\subset \mathbb{R}^d$, $\sigma\left(X_t, t \in I\right)$ represents the $\sigma$-algebra generated by these variables.
\subsection{Summary of Results}
\label{sec:results}

In this paper, for $k \geq 1 $ and $d=1$, we start by studying \textcolor{black}{a modified version of} the constrained estimator 
 \begin{equation}
 	\label{eqn:e2}
 	 \widehat f _{k,V_n}\,=\,   \underset{f}{ \argmin }  \,\frac{1}{ n} \sum_{i=1}^n \frac{1}{m_i} \sum_{j=1}^{m_i} \big (  y_{i,j} -f(x_{i,j})  \big)^2 \quad  \text{subject to} \quad    \mathrm{TV}(f^{(k-1)}) \leq V_n,
 \end{equation}
for a tuning parameter $V_n>0$. The optimization problem given in Equation (\ref{eqn:e2}) can be viewed as a constrained version of {\color{black}{Locally Adaptive Regression Splines}} for spatio-temporal data.  This can be seen by comparing it to the optimization problem in Equation (\ref{eqn:e1}) and the loss function in \cite{cai2011optimal}. In Equation (\ref{eqn:e2}), $V_n >0$ is a tuning parameter that controls the $(k-1)$th order total variation. Moreover, the estimator defined in Equation (\ref{eqn:e2}) can be considered as a variational version of the estimator studied in \cite{wakayama2021trend}, which addresses some computational aspects of {\color{black}{Locally Adaptive Regression Splines and}} Trend Filtering for functional data. 

Equation (\ref{eqn:e1}) and Equation (\ref{eqn:e2}) are quite different, especially when dealing with issues found in spatio-temporal data such as spatial noise and changing amounts of data over time. Equation (\ref{eqn:e2}) addresses these by adjusting its loss function to inversely weight observations based on the number of observations at each time point, ensuring a balanced impact from each time point. This method works better on spatio-temporal data by making its predictions more accurate despite the noise and fluctuating data amounts.

Our framework supposes that the sequences $\{ (x_{i,1},\ldots,x_{i,m_i}) \}_{i=1}^n $,  $\{ (\epsilon_{i,1},\ldots,\epsilon_{i,m_i}) \}_{i=1}^n$, and $\{\delta_i\}_{i=1}^n \subset \mathcal{L}_2([0,1])$ have elements that are identically distributed and $\beta$-mixing with coefficients having exponential decay, see Assumption  
\ref{assume:tv functions} for the explicit details.  Moreover, we require that $\delta_i(x)$ is subGaussian, with mean zero, at every location $x$ and $i \in \{1,\ldots,n\}$.  Thus our data can have temporal dependence across $i$ and spatial dependence across $j$. Under these conditions, we show that \textcolor{black}{a variant } $\widehat f_{k,V_n}$ of the estimator defined in (\ref{eqn:e2}) satisfies 
\begin{equation}
	\label{eqn:e3}
	\| \widehat f _{k,V_n}- f^*\|_2 ^2 =O_{\mathbb{P}}\Big(   V_n^2   \frac{\log^{\frac{4k}{2k+1}}(n)}{(mn)^{\frac{ 2 k   }{2 k +1 }} }    +  \frac{\log^{\frac{5}{2}} (n)}{  n }\Big),
\end{equation}
provided that  $V_n \geq \mathrm{TV}((f^*)^{(k-1)})$ and $m_i \asymp m$ for all $i$. Hence, ignoring logarithmic factors, (\ref{eqn:e3}) implies that $\widehat f _{k,V_n}$  has a convergence rate of order $(nm)^{\frac{-2k}{2k+1}} + n^{-1}$.

Our result in (\ref{eqn:e3})  extends the previous convergence rate established in Theorem 2.1 of \cite{guntuboyina2020adaptive} for fixed design and $m_i=O(1)$, to the case where the $m_i$'s can diverge, the data are dependent, and spatial noise is present. Moreover, we also allow for temporal dependence across $i$ while \cite{guntuboyina2020adaptive} required independence of the data. Additionally, the rate $(nm)^{\frac{-2k}{2k+1}} +  n^{-1}$ is minimax optimal when estimating signals in the class of functions $B_{\mathrm{TV}}^{(k-1)}(V_n)$ with $V_n = O(1)$.   This follows from Theorem 3.1 in \cite{cai2011optimal} where the authors considered Sobolev function classes. In the existing literature, the condition $V_n=O(1)$ is commonly referred to as the canonical scaling for functions residing in a total variation ball with radius $V_n.$ See for instance \cite{madrid2022risk} and \cite{sadhanala2016total}.

In addition to the constrained estimator, we study \textcolor{black}{a modified version of} the penalized estimator defined as 

\begin{equation}
	\label{eqn:e4}
	\widehat f _{k,\lambda}\,= \,\underset{f}{ \argmin }\,   \left\{  \frac{1}{ n} \sum_{i=1}^n \frac{1}{m_i} \sum_{j=1}^{m_i} \big (  y_{i,j} -f(x_{i,j})  \big)^2    \,+\,   \lambda \mathrm{TV}(f^{(k-1)})  \right\},
\end{equation}
for a  tuning parameter $\lambda >0$.  \textcolor{black}{We show that for an appropriate choice of $\lambda$, the resulting variant, which we also denote as $\widehat f _{k,\lambda}$, satisfies}
\begin{equation}
	\label{eqn:e51}
	\| \widehat f _{k,V_n}- f^*\|_2 ^2 =O_{\mathbb{P}}\Big(  \frac{\log^{\frac{4k}{2k+1}}(n)}{(mn)^{\frac{ 2 k   }{2 k +1 }} }\left(V_n^*\right)^2+\frac{\log^{\frac{5}{2}}(n)}{n}\Big),
\end{equation}
with $V_n^* =   \mathrm{TV}(  (f^*)^{(k-1)}  )\ge1$.  Thus, when  $V_n    =  O(V_n^*)$ both the constrained and penalized estimator attain the same rate of convergence.

This paper also addresses the problem of nonparametric spatio-temporal data estimation for dimension $d>1$. To this end, we propose the spatio-temporal data version of the $K$-Nearest Neighbor fused lasso estimator ($K$-NN-FL) studied in \cite{elmoataz2008nonlocal,ferradans2014regularized,madrid2020adaptive}. The $K$-NN-FL estimator is a two-step method. First, we construct a $K$-nearest neighbor graph $G_K$ with edge set $E_K$, where two pairs  $(i_1,j_1)$ and $(i_2,j_2)$ are connected by an edge if and only if $x_{i_1,j_1}$ is among the $K$-nearest neighbors of $x_{i_2,j_2}$, or vice versa. Second, we run the graph fused lasso estimator in the graph $G_K$. 
However, our version of the $K$-NN-FL differs from that in \cite{madrid2020adaptive} in the loss function that we use given the spatio-temporal data. Specifically, our objective function is structured around comparing the data $y_{i,j}$ with the function evaluation $f(x_{i,j})$ across different time points $i$. Recognizing that these time points have varying numbers of measurements $m_i$, we use a weighting mechanism in the loss function, proportionate to the number of observations $m_i$ at each time point. This adjustment ensures a more balanced contribution from each time period, see \cite{cai2011optimal}. To address the complexities introduced by this weighted loss function, we develop an alternating direction method of multipliers (ADMM) specific to our objective function. For more details on ADMM algorithms, see \cite{boyd2011distributed}.

We show that the resulting $K$-NN-FL estimator attains a convergence rate of order $(nm)^{-1/d} + n^{-1}$, in terms of the square of the $\ell_2$ distance, if we ignore logarithmic factors. This rate holds under the assumption that the true function $f^*$ has bounded variation and satisfies a property which is a generalization of piecewise Lipschitz continuity. Our result effectively extends Theorems 1 and 2 in  \cite{madrid2020adaptive}  to the  spatio-temporal data setup. These findings are particularly relevant for nonparametric spatio-temporal data analysis in multivariate settings, where the proposed ADMM algorithm can be used to efficiently estimate the $K$-NN-FL estimator.

Our study shows that our proof methods are quite adaptable. They can easily be applied to other models exhibiting short-range dependence on spatio-temporal data. Specifically, for the univariate model ($d=1$), we introduce unprecedented concentration inequalities that work even with spatial noise. This is shown by the coupling of empirical and $L_2$ norms, (refer to Lemma \ref{lemma2}). See Appendix \ref{sec-proof-thm2} for concrete details. Notably, our findings go beyond what \cite{tibshirani2014adaptive} discussed. The mentioned study, \cite{tibshirani2014adaptive}, based its findings on data without any spatial noise and assumed that covariates and measurement errors are independent. Furthermore, we have come up with a way to bound the supremum norm for functions that have an empirical mean of zero and a bounded empirical $L_2$ norm, as detailed in Lemma \ref{lemma1} in Appendix \ref{sec-proof-thm2}. This adds to what was established in Lemma 7 by \cite{tibshirani2014adaptive}. For the multivariate framework ($d>1$), we have successfully expanded on everything found in the study \cite{madrid2020adaptive}, adapting them to spatio-temporal data characterized by short-range dependence.

To contextualize the results above, it is crucial to acknowledge their deviation from traditional {\color{black}{Locally Adaptive Regression Splines}} and Trend Filtering theories. These differences are mainly in two ways: the way data interacts and the inclusion of spatial noise. The analysis of the former requires a rigorous examination of beta-mixing block-independent instances. On the other hand, tackling spatial noise calls for sophisticated techniques, such as symmetrization and the Ledoux-Talagrand inequality, among others.

\subsection{Other relevant work}

While there is a vast body of literature on total variation denoising with fixed design, such as \citep{wang2015trend,padilla2017dfs,sadhanala2021multivariate}, the techniques used in the analysis of such settings do not directly apply to the random design setup. In the latter case, some other notable works using total variation include the additive Trend Filtering framework proposed in \cite{sadhanala2019additive}; the quantile $K$-Nearest Neighbor fused lasso estimator studied in \cite{ye2021non}; and the recent work from \cite{hu2022voronoigram}, which has very strong adaptivity properties, but is computationally much more intensive than the $K$-NN-FL estimator studied for spatio-temporal data in this paper.

{\color{black}{
Complementary Bayesian studies have established frequentist guarantees for scalable spatial and spatio-temporal modeling. \citet{JMLR:v23:20-543} derive minimax-optimal posterior convergence rates in $L_2$ for both the varying-coefficient functions and the induced mean regression function within a distributed framework using a multivariate Gaussian-process prior. \citet{10.1214/22-STS868} establish nearly minimax-optimal Bayes $L_2$-risk rates for combined pseudo-posteriors arising from a subclass of distributed spatial Gaussian-process methods and provide upper bounds on the number of subsets. Using basis expansions and random data sketching, \citet{JMLR:v26:23-0505} establish posterior contraction rates for estimating varying coefficients and predicting outcomes at new locations. In contrast, we develop frequentist total-variation estimators for a deterministic mean function under temporal dependence, time-varying spatial noise, and measurement error.
}}

\subsection{Outline}

The rest of the paper is organized as follows. In Section \ref{MainRes}, the underlying assumptions of the spatio-temporal model are delineated, followed by a presentation of our theoretical findings. Section \ref{sec:uni_constrained} delves into the univariate setting, highlighting that a constrained estimator is a minimax optimal. This is further complemented in Section \ref{sec:uni_Pen}, which underscores the minimax optimality of a penalized estimator. In contrast, the multivariate setting is presented in Sections \ref{sec:Mult_Cons} and \ref{sec:Mult_Pen}, introducing minimax optimal constrained and penalized estimators, respectively.  Section \ref{simu-data} evaluates the practical performance of our proposed techniques via various simulations and real data analysis. Outlines of our theoretical proofs can be found in Sections \ref{OutlineT1-2} and \ref{OutlineT3-4}. 
Finally, Section \ref{sec-conclusion} concludes
with a discussion. All proofs can be found in the Appendix.

\section{Main results}\label{MainRes}

\subsection{Assumptions}

In this paper, we focus on a specific statistical framework for nonparametric spatio-temporal data estimation. This framework is described by Assumption \ref{assume:tv functions} below, which outlines the assumptions on the data generating process for the model given in Equation \eqref{eq:model l2}.

\begin{assumption}\label{assume:tv functions}
  The data $\{(x_{i, j}, y_{i, j})\}_{i = 1,j = 1}^{n,m_i} \subseteq [0,1]^d \times \mathbb{R}$ are generated by  model \eqref{eq:model l2}.
\\
{\bf a.}  For each  fixed $i\in \{1,\ldots, n\}$, the random design   $\{ x_{i,j}  \}_{ j = 1}^{ m_i} \subseteq [0, 1]^d$ are  independently   sampled from a common density function $v:\, [0,1]^d \to \mathbb R$. Here $0< v_{\text{min}} \leq v(x) \leq v_{\text{max}}$  for all $x \in [0,1]^d$. Denote $ \sigma _x(i) =\sigma(\{ x_{i,j}  \}_{ j = 1}^{ m_i} )   $. Suppose that  $\{ \sigma_x(i)\}_{i=0}^n$ is $\beta$-mixing with beta coefficients  \footnote{{See Appendix \ref{BetaMsection} for the $\beta$-mixing definition and Section \ref{sec:notation} for the definition of $\sigma(\cdot)$}.} satisfying $\beta_x(l)\le e^{-C_\beta l}$ for some constant $C_\beta>0.$  
\\
{\bf b.} The regression  function $f^* : \mathbb [0,1]^d \to \mathbb R$ satisfies the following:
\vskip 0.1cm
\noindent
 1. For $d=1$,   suppose  that $V_n^* :=  \mathrm{TV}( (f^*)^{(k-1)}  )  < \infty$. Here   $V_n^*$  depends on $n$,  and   $k \in \mathbb Z^+ $.  
\vskip 0.1cm
\noindent
2. For $d>1$, suppose $f^*$ has bounded variation {\color{red}\footnote{See  Appendix \ref{app:BV-def} for the definition of bounded variation in the multidimensional case. }} and it satisfies a more general condition than being piecewise Lipschitz or it 
  is piecewise Lipschitz.\footnote{See  Appendix \ref{plsection} for more details. }
Additionally, suppose  that 
$\| f^*\| _{\infty}\le C_{f^*}$ for some positive constant $C_{f^*}$.\\
{\bf c.} The spatial noise  $\left\{\delta_i\right\}_{i=0}^n \subset \mathcal{L}_2\left([0,1]^d\right)$is identically distributed such that 
$$ \mathbb{E}\left(\delta_i (x) \right) = 0 \ \ \text{and} \ \
\sup _{x \in[0,1]^d} \mathbb{E}\left\{\exp \left(t \delta_i(x)\right)\right\} \leq \exp \left(\varpi_\delta^2 t^2\right)     \text{ for all } t \in \mathbb{R} .
$$
Here  $\varpi_\delta>0$ is an absolute constant. Moreover,  suppose that 
$\{ \sigma(\delta_i)\}_{i=0}^n$ are $\beta$-mixing with beta coefficients  satisfying $\beta_\delta(l)\le e^{-C_{\beta}l}$.
\\
{\bf d.} For each  fixed $i\in \{1,\ldots, n\}$,  the measurement errors $\{ \epsilon_{i,j}\}_{  j=1}^{ m_i}  \subset \mathbb R $ are    identically distributed    such that 
$  \mathbb{E}(\epsilon_{i,j})= 0 $  and $\mathbb{E}\{ \exp ( t \epsilon_{i,j}) \} \le  \exp(\varpi_\epsilon^2 t ^2) $     for all  $ t\in \mathbb R$.  
Here  $\varpi_\epsilon>0$ is an absolute constant. Moreover, let  $ \sigma _\epsilon (i) =\sigma( \{\epsilon_{i,j}\}_{  j=1}^{ m_i} ) $ and suppose that 
$\{ \sigma_\epsilon (i)\}_{i=0}^n$ are $\beta$-mixing with beta coefficients  satisfying  $\beta_\epsilon(l)\le e^{-C_\beta l}$.  
\\
{\bf e.} Let $\sigma_x=\sigma(\{x_{i,j}\}_{i=1,j=1}^{n,m_i}), \sigma_\delta=\sigma(\{\delta_{i}\}_{i=1}^{n})$ and $\sigma_\epsilon=\sigma(\{\epsilon_{i,j}\}_{i=1,j=1}^{n,m_i})$.
Suppose $\sigma_x, \sigma_\delta$ and $\sigma_\epsilon$ are independent.  
\\
{\bf f.}
For any $i\in\{1,...,n\}$, we assume that $cm\le m_i\le Cm$, where $m$ is the harmonic mean of $m_i$ defined by \eqref{Har-mean}. Here, $c$ and $C$ are positive constants and 
\begin{equation}\label{Har-mean}
m= \left(\frac{1}{n}\sum_{i=1}^{n}\frac{1}{m_i}\right)^{-1}.
\end{equation}
Thus, we essentially require that $m_i \asymp m_j$ for all $i$ and $j$.   
  \end{assumption}

Assumption \ref{assume:tv functions}\textbf{a} allows the observed spatio-temporal data to be represented on random design points. Moreover,  the sampling distribution $v$ must be upper and lower bounded on the support $[0, 1]^d$. This restriction is common in nonparametric regression, for instance, see \cite{madrid2020adaptive} and references therein. The intuition behind the $\beta$-mixing condition here is to quantify the extent of dependence between observations as they become more temporally separated. Specifically, $\beta$-mixing ensures that this dependence diminishes, which is crucial for the reliability of statistical inference in the presence of dependent data. Exponential decay of the beta coefficients, $\beta_x(l)\le e^{-C_\beta l}$ for some constant $C_\beta>0$, indicates a rapid decrease in dependency, facilitating the incorporation of short-range dependencies in the spatio-temporal model.



Assumption \ref{assume:tv functions}\textbf{b} pertains to the smoothness properties of the mean function $f^*$ of the observed spatio-temporal data. In the univariate case, we assume that $f^*$ belongs to a $B_{\mathrm{TV}}^{(k-1)}(V_n)$ ball, which accommodates greater fluctuations in the local level of smoothness of $f^*$ and includes more traditional function classes such as H\"{o}lder and Sobolev spaces as particular cases. In the multivariate case, we allow for functions that can have discontinuities, meaning they can be piecewise Lipschitz, or can satisfy a more general notion than being  piecewise Lipschitz along with a bounded variation restriction. We refer the interested reader to  \cite{madrid2020adaptive} for examples and discussion.



Our model takes into account two distinct sources of noise that can affect the observed spatio-temporal data. The first type of noise is spatial noise, which arises from the inherent variability of the underlying spatio-temporal processes. In other words, even if we were to observe the true underlying spatio-temporal process without any measurement error, we would still see some level of variability due to factors such as natural randomness, biological variation, or other sources of variation specific to the application at hand. This type of noise is described in Assumption \ref{assume:tv functions}\textbf{c} of our model.

The second type of noise is measurement error given as $\{ \epsilon_{i,j}\}_{i=1, j=1}^{n,m_i}$, which results from imperfect measurement devices or other sources of external interference. Measurement error can introduce systematic bias or random fluctuations in the observed data that are not reflective of the true underlying process. This type of noise is described in Assumption \ref{assume:tv functions}\textbf{d} of our model.

By including both spatial noise and measurement error in our model, we are accounting for the fact that the observed data may not be a perfect reflection of the true underlying process, and that there may be variability and errors that need to be taken into account when analyzing the data. This allows to consider a more realistic modeling framework. 

In {\color{black}{Locally Adaptive Regression Splines}} and Trend Filtering analysis, one of the often overlooked aspects has been the dependence of covariates. The nature of the dependence among covariates, which can often be observed in datasets, can play a pivotal role in ensuring the stability and clarity of the trends identified. Recognizing this, our approach goes beyond conventional boundaries by allowing the covariates to be dependent, see Assumption \ref{assume:tv functions}{\bf{a}}. Additionally, we acknowledge the presence of dependence within the spatial noise and also within the measurement errors, see Assumption \ref{assume:tv functions}{\bf{c}} and {\bf{d}}, respectively. In the literature, such a comprehensive consideration has been notably absent from {\color{black}{Locally Adaptive Regression Splines or}} Trend Filtering studies. This comprehensive approach brings a new perspective to our methods, enhancing the reliability and effectiveness of {\color{black}{Locally Adaptive Regression Splines and}} Trend Filtering. By taking these factors into account, we lay the groundwork for a more in-depth and complete understanding of the underlying processes.

Assumption \ref{assume:tv functions}{\bf{f}} is standard in the literature with spatial noise. For instance, \cite{wang2022low} focuses on the nonparametric estimation of covariance functions for functional data and employs a similar assumption, as shown on Page 3 of that paper. 

Overall, our model provides a flexible and comprehensive framework for analyzing spatio-temporal data that can accommodate various sources of noise and variability, making it suitable for a wide range of applications in which spatio-temporal data are encountered.

\subsection{Univariate constrained {\color{black}{Locally Adaptive Regression Splines}}}
\label{sec:uni_constrained}

In this subsection, we present our main result for the constrained spatio-temporal data {\color{black}{Locally Adaptive Regression Splines}} estimator.

\begin{theorem}  \label{thm:main tv} Consider $k\ge 1$, and  observations $\{(x_{i,j}, y_{i,j})\}_{i=1,j=1}^{n,m_i}$ satisfying Assumption 
 \ref{assume:tv functions}. Let $\bar{y}$ be given as 
 \[
    \bar{y} \,:=\, \frac{1}{ \sum_{i=1}^n m_i }    \sum_{i=1}^n \sum_{j=1}^{m_i} y_{i,j}.
 \]
 Consider  the estimator, $    \widehat f_{k,V_n} :=   \widehat g_{k,V_n}  \,+\,       \bar{y} $ where 
\begin{equation}
\label{eqn:constrained}
\begin{aligned}
\widehat g_{k,V_n} := \underset{f}{\argmin} \ \frac{1}{ n} \sum_{i=1}^n \frac{1}{m_i} \sum_{j=1}^{m_i} 
\big (  y_{i,j} -f(x_{i,j}) - \bar{y} \big)^2, 
 \\
 \text{subject to} \quad  
J_k(f) \le V_n \quad \text{and} \quad 
\sum_{i=1}^{n}\sum_{j=1}^{m_i}f(x_{i,j}) = 0.
\end{aligned}
\end{equation}
It holds that $$  \| \widehat f_{k,V_n}  - f^*\|_2 ^2 =O_{\mathbb{P}}\Big(  V_n^2  \frac{\log^{\frac{4k}{2k+1}}(n)}{(mn)^{\frac{ 2 k   }{2 k +1 }} }   +  \frac{\log^{\frac{5}{2}} (n)}{  n }\Big),$$
provided that $V_n \geq \max\{V_n^*,1\}$, with $V_n^*$ as in Assumption \ref{assume:tv functions}{\bf{b}}.
\end{theorem}

Some comments are in order. First, the constrained estimator in Theorem \ref{thm:main tv} involves the constraint 
\begin{equation}
    \label{eqn:sum_zero}
    \sum_{i=1}^{n}\sum_{j=1}^{m_i}f(x_{i,j})=0.
\end{equation}
This is standard in {\color{black}{Locally Adaptive Regression Splines and}} Trend Filtering with arbitrary design points, see for instance a similar constraint in Problem 8 in \cite{sadhanala2019additive}. Furthermore, the condition $V_n \geq V_n^*$ is meant to ensure that $f^*$ is feasible for the optimization problem (\ref{eqn:constrained}). A similar condition appeared in Theorem 1 in \cite{madrid2022risk}.   \textcolor{black}{We note that the result in Theorem \ref{thm:main tv} does not correspond exactly to the estimator defined in (\ref{eqn:e2}). This discrepancy arises because our proof relies on existing covering number results for total variation  balls, which are only known to hold under an additional constraint involving the $\|\cdot\|_{\infty}$ norm. Since this constraint is not included in the construction of our estimator, we instead incorporate the condition (\ref{eqn:sum_zero}), following \citet{sadhanala2019additive}. This adjustment enables the use of an orthogonality argument based on projection onto the space of polynomials and its orthogonal complement; see Appendix \ref{sec-proof-thm2} for details. }

As for the upper bound,  Theorem \ref{thm:main tv}  implies that the constrained {\color{black}{Locally Adaptive Regression Splines}} estimator achieves a rate up to logarithmic factors, in terms of $\ell_2$ estimation, of order $(nm)^{\frac{-2k}{2k+1}} +  n^{-1}$ in the canonical setting where  $V_n^*  = O(1)$ and $V_n \asymp V_n^*$.  Notably, an interesting phase transition phenomenon is revealed. Specifically, with a boundary at $m=n^{\frac{1}{2 k}} (\log (n))^{\frac{-2k-5}{4k}}$. When the sampling frequency $m$ satisfies $m\lesssim n^{\frac{1}{ 2 k}} (\log (n))^{ \frac{-  2k-5}{4k} }$, the  rate is of  order $ (\log (n))^{\frac{4k}{2k+1}}  (n m)^{\frac{-2 k}{ 2 k+1}}$ which depends jointly on the values of both $m$ and $n$. In the case of high sampling frequency with $m \gtrsim  n^{\frac{1}{2 k}}(\log( n))^{ \frac{-  2k-5}{4k} }$, the  rate is of order $ (\log^{\frac{5}{2}} (n))  n^{-1}$ and does not depend on $m$. Importantly, the rate  $(nm)^{\frac{-2k}{2k+1}} +  n^{-1}$ is minimax optimal for estimating functions in the class of signals $B_{\mathrm{TV}}^{(k-1)}(V_n)$. This follows from Theorem 3.1 in \cite{cai2011optimal}. In particular, the authors showed that for any estimate $\widetilde{f}$ based on observations $\{(x_{i,j},y_{i,j})\}_{i=1,j=1}^{n,m_i}$, it holds that 
\begin{equation}
    \label{eqn:lower1}
     \underset{ n \rightarrow \infty}{\lim\sup}\,\underset{f^*  \in   \mathcal{G}_k(U)   }{ \sup     }\,\mathbb{P}\left(  \| \widetilde{f}-f^* \|_2^2 \,\geq \,  C\left(    \frac{1}{n }  \,+\, \frac{1}{  (nm)^{\frac{2k}{2k+1}} } \right) \right)\,>0\,,
\end{equation}
where  $U>0$ is a constant, $C >0 $ depends  on $U$, the observations satisfy Assumption \ref{assume:tv functions} but instead of Assumption \ref{assume:tv functions}\textbf{b} they require $f^*\in \mathcal{G}_k(U) $ and assume that the data are independent. Here, 
\[
\mathcal{G}_k(U)   \,:=\, \left\{  f:[0,1]\to \mathbb{R}  \,:   \,  \int_0^1 \vert f^{(k)}(t) \vert^2 dt \le U \right\}.
\]
Hence, if $V_n \geq 1$  with $V_n =  O(1)$, and $U = V_n-1$, for $f \in \mathcal{G}_k(U)$, from the inequality
\[
 \mathrm{TV}(f^{(k-1)})  \,=\, \int_0^1 \vert f^{(k)}(t) \vert dt \,\leq\, \int_0^1 \vert f^{(k)}(t) \vert^2 dt \,+\, 1   \,\leq\,   V_n,
\]
we obtain that $f  \in B_{\mathrm{TV}}^{(k-1)}(V_n) $, with $B_{\mathrm{TV}}^{(k-1)}(\cdot)$ defined in (\ref{eqn:tv_ball}). Therefore, (\ref{eqn:lower1}) implies  
\begin{equation}
    \label{eqn:lower2}
     \underset{ n \rightarrow \infty}{\lim\sup}\,\underset{f^*  \in  B_{\mathrm{TV}}^{(k-1)}(V_n)   }{ \sup     }\,\mathbb{P}\left(  \| \widetilde{f}-f^* \|_2^2 \,\geq \,  C\left(    \frac{1}{n }  \,+\, \frac{1}{  (nm)^{\frac{2k}{2k+1}} } \right) \right)\,>0.
\end{equation}
The latter allows to conclude that $\widehat f_{k,V_n} $ defined in (\ref{eqn:constrained}), except for logarithmic factors, is minimax optimal for estimation in the function class $B_{\mathrm{TV}}^{(k-1)}(V_n)  $ when $V_n \asymp 1$.


Notably, to the best of our knowledge, Theorem \ref{thm:main tv} is the first result for {\color{black}{Locally Adaptive Regression Splines, or}} Trend Filtering, in univariate settings where the data are allowed to be dependent. 

\subsection{Univariate penalized {\color{black}{Locally Adaptive Regression Splines}}}
\label{sec:uni_Pen}

In this subsection, we investigate the statistical performance of penalized {\color{black}{Locally Adaptive Regression Splines}} for univariate spatio-temporal data. While the constrained estimator studied in the previous subsection can be formulated as a quadratic program, this formulation can be computationally expensive when dealing with large sample sizes. To overcome such computational burden, penalized estimators offer a natural alternative. Efficient algorithms for penalized {\color{black}{Locally Adaptive Regression Splines and}} Trend Filtering have been studied in the literature, as discussed in \cite{johnson2013dynamic,barbero2018modular}  and \cite{ramdas2016fast}.


For a tuning parameter $\lambda>0$, the penalized {\color{black}{Locally Adaptive Regression Splines}}  estimator of order $k \geq 1$ for spatio-temporal data is defined as   $\widehat{f}_{k,\lambda}  =   \widehat{g}_{k,\lambda}  + \bar{y}$ where 
$$
\widehat{g}_{k,\lambda}  =  \underset{f}{\arg \min } \frac{1}{n} \sum_{i=1}^n \frac{1}{m_i} \sum_{j=1}^{m_i}\left(y_{i ,j}- \bar{y} -f\left(x_{i, j}\right)\right)^2+\lambda J_k(f), \quad  \text{subject to} \quad \ \sum_{i=1}^{n}\sum_{j=1}^{m_i}f(x_{i,j})=0.
$$
With this notation, we are now ready to state our main result for the penalized estimator. 

\begin{theorem}\label{Penalized} Let $k \geq 1$ and $\{(x_{i,j}, y_{i,j})\}_{i=1,j=1}^{n,m_i}$ be observations  satisfying Assumption \ref{assume:tv functions}.  For  the estimator $\widehat{f}_{k,\lambda} $, it holds that
$$
\left\|\widehat{f}_{k,\lambda}-f^*\right\|_{2}^2=O_{\mathbb{P}}\left(r_{n,m}\right), \ \text{where} \ r_{n,m}:=\frac{\log^{\frac{4k}{2k+1}}(n)}{(mn)^{\frac{ 2 k   }{2 k +1 }} }\left(J_k^2\left(f^*\right)+1\right)+\frac{\log^{\frac{5}{2}}(n)}{n},
$$
for a choice of $\lambda$ satisfying $\lambda \asymp r_{n,m}\left(J_k\left(f^*\right)\right)^{-1}$, and provided that $r_{n,m}=O(1)$.
\end{theorem}

Theorem \ref{Penalized} demonstrates that with the canonical scaling $J_k(f^*) = O(1)$, the penalized {\color{black}{Locally Adaptive Regression Splines}} estimator of order $k \geq 1$ for spatio-temporal data attains the rate $(nm)^{\frac{-2k}{2k+1}} + n^{-1}$, up to logarithmic factors. This rate is the same as that attained by the constrained estimator in Theorem \ref{thm:main tv}. Therefore, as discussed in Section \ref{sec:uni_constrained}, the penalized estimator $\widehat{f}_{k,\lambda}$, like the constrained estimator, is also minimax optimal for estimation in the function class $B_{\mathrm{TV}}^{(k-1)}(V_n)$ with $V_n = O(1)$.

We need to address the minor differences in our upper bound compared to existing in the literature. For instance, when 
$k=1$ and $d=1$, our upper bound depends on 
$V_n^2$. This contrasts with other research, such as the study by \cite{guntuboyina2020adaptive}, which establishes a tighter upper bound depending on 
$V_n^{\frac{2}{3} }$. However, this distinction is not critical in the context of our primary research goal, which is to achieve minimax optimality for our estimator.  This goal is primarily considered under the assumption of canonical scaling, meaning that $V_n = O(1)$, as discussed by \cite{sadhanala2016total} among others. Under this assumption, the specific dependence of the upper bound on $V_n$ does not significantly impact the efficiency of the estimate. 

\subsection{Constrained \texorpdfstring
{$K$-}{K}NN fused lasso for spatio-temporal data}
\label{sec:Mult_Cons}

Now, we focus on the multivariate case ($d>1$). We start by introducing a constrained estimator. Let  $G_K$ be the $K$-NN graph with edge set $E_K$ constructed from $\{ x_{i,j}\}_{i=1, j = 1}^{n,m_i}$ as discussed in Section  \ref{sec:results}. Then, for a signal $\theta \in  \mathbb{R}^{  \sum_{i=1}^n m_i  }$ we define its total variation along the graph $G_K$ as
\[
 \| \nabla_{G_K}\theta  \|_1\,:=\, \sum_{ \{ (i_1,j_1), (i_2,j_2) \}  \in E_K  } \vert \theta_{i_1 ,j_1} - \theta_{ i_2, j_2  } \vert.
\]
Intuitively, if the signal $\theta $ consists of function evaluations, then one would expect that $\| \nabla_{G_K}\theta  \|_1 << nm$ since the $K$-NN graph would capture the geometry of the domain of the $x_{i,j}$'s and most of the quantities $\vert \theta_{i_1, j_1} - \theta_{ i_2 ,j_2  } \vert$ would be small. Such intuition is made precise in Lemma \ref{Lemma-nabla}, where we show that, with high probability,  
\begin{equation}
    \label{eqn:tv_bound}
    \| \nabla_{G_K}\theta^*  \|_1 \,\asymp \, (nm)^{1-\frac{1}{d}},
\end{equation}
with $\theta^*_{i,j} = f^*(x_{i,j})$, and  provided that either $f^*$  has bounded variation, along with an additional condition
that generalizes piecewise Lipschitz continuity; or $f^*$ is piecewise Lipschitz. This result extends  the one shown in \cite{madrid2020adaptive} to the spatio-temporal  scenario considered in Assumption \ref{assume:tv functions}.
Thus, a natural estimator for $\theta^*$ is   $\widehat{\theta}_{V_n}$ given as
\begin{equation}
    \label{eqn:theta_hat_knn}
    \widehat{\theta}_{V_n}\,=\,  \underset{ \theta \in  \mathbb{R}^{  \sum_{i=1}^n m_i  },\,\,\,  \left\|\nabla_{G_K}\theta\right\|_1 \leq V_n  }{\arg \min  }\,\left\{\frac{1}{n}\sum_{i=1}^n \frac{1}{m_i}\sum_{j=1}^{m_i}\left(y_{i, j}-\theta_{i, j}\right)^2 \right\},   
\end{equation}
for a tuning parameter $V_n>0$. Once $\widehat{\theta}_{V_n}$ is constructed, we let $\widehat{f}_{V_n} : [0,1]^d \,\rightarrow \, \mathbb{R} $ be given as
\begin{equation}
    \label{eqn:est1}
     \widehat{f}_{V_n}(x_{i,j})   \,=\,    (\widehat{\theta}_{V_n})_{i,j}\,\,\,\,\,\, \text{for}\,\,\,\,  i = 1,\ldots,n,\,\,\,\, j =1 ,\ldots, m_i,
\end{equation}
and 
\begin{equation}
    \label{eqn:est2}
     \widehat{f}_{V_n}(x)   \,=\,   \frac{1}{ \sum_{i=1}^n \sum_{j=1}^{m_i}  \mathcal{I}(x,x_{i,j})  }  \sum_{i=1}^n \sum_{j=1}^{m_i}  \mathcal{I}(x,x_{i,j})   (\widehat{\theta}_{V_n})_{i,j}, 
\end{equation}
for $x \in  [0,1]^d \backslash \{ x_{i,j}\}_{i=1, j = 1}^{n, m_i}$, where  $\mathcal{N}_K(x)$ is the set of $K$-nearest neighbors of $x$ among $\{ x_{i,j}\}_{i=1, j = 1}^{n,m_i}$, and $\mathcal{I}(x,x_{i,j}) = 1$ if $x_{i,j} \in \mathcal{N}_K(x)$ and $\mathcal{I}(x,x_{i,j}) = 0$ otherwise.    Notice in (\ref{eqn:est2})
we are using the same estimation rule described in \cite{madrid2020adaptive}.

With this notation, we are ready to state our main result regarding the constrained $K$-NN fused lasso for spatio-temporal data   (\ref{eqn:est1})--(\ref{eqn:est2}). 

\begin{theorem}\label{Cest-d>1} Let $\{(x_{i,j},y_{i,j})\}_{i=1,j=1}^{n,m_i}$ generated based on \eqref{eq:model l2} and Assumption \ref{assume:tv functions} with $d>1$. Suppose in addition that
 $K \asymp \log^{2+l}(nm)$ for some $l>0$. The estimator $\widehat{\theta}_{V_n}$ defined in (\ref{eqn:theta_hat_knn}) satisfies that 
\begin{equation}
    \label{eqn:e41}
    \frac{1}{n} \sum_{i=1}^n \frac{1}{m_i}\sum_{j=1}^{m_i}\left((\widehat{\theta}_{V_n})_{i,j}-\theta_{i, j}^*\right)^2 \, = \, O_{\mathbb{P}}\left(r_{n,m}\right),
\end{equation}
where
$$r_{n,m}=\frac{\log^{8+2l}(nm)}{n}+\left\{\frac{\log ^{5+l}(n m)}{n m}\right\}  V_n,$$
and provided that $V_n \geq \left\|\nabla_{G_K} \theta^*\right\|_1 $. In addition,  $\widehat{f}_{V_n}$ defined in (\ref{eqn:est1})--(\ref{eqn:est2})  satisfies that 
\begin{equation}
    \label{eqn:e5}
     \| \widehat{f}_{V_n} -f^*\|_2^2 = O_{\mathbb{P}}(\widetilde{r}_{n,m}),
\end{equation}
for 
\[
\widetilde{r}_{n,m}=\frac{\log^{8+2l}(nm)}{n}+\left\{\frac{\log ^{5+l}(n m)}{n m}\right\}  V_n + \frac{\log^{\frac{2+l}{d}}(nm)}{(nm)^{1/d}}.
\]
\end{theorem}

Note that  (\ref{eqn:tv_bound}) implies that     $V_n  = \|\nabla_{G } \theta^*\|_1 \asymp (nm)^{1-\frac{1}{d}}$ with high probability. Therefore, Theorem \ref{Cest-d>1} implies that the $K$-NN-FL estimator for spatio-temporal data achieves 
\begin{equation}
    \label{eqn:e5 1}
     \| \widehat{f}_{V_n} -f^*\|_2^2 =\mathrm{poly}(\log( nm)) \bigg\{ \frac{ 1  }{  n}  + \frac{  1 }{  (nm)^{ 1/d}}  \bigg\} ,
\end{equation}
  where $\mathrm{poly}(\cdot) $ indicates a polynomial function. Ignoring the logarithmic factors, we see that the rate becomes $n^{-1}$ when $ n^{d-1}\lesssim m$, while in the sparse design regime $m \lesssim n^{d-1}$ the rate is order $(nm)^{\frac{-1}{d}}$. In particular, in the case $m \asymp 1$, we recover the upper bound obtained in Theorem 2 \cite{madrid2020adaptive}.  However, unlike \cite{madrid2020adaptive}, we allow for dependent measurements and spatial noise. 


In Theorem~\ref{Cest-d>1}, we assume that the discrete total variation $\left \|\nabla_{G_K} \theta^*\right\|_1 $ of the unknown regression function is bounded. 
In Lemma~\ref{Lemma-nabla}, we show that piecewise Lipschitz functions possess bounded discrete total variation. Consequently,  the risk bound in Theorem~\ref{Cest-d>1} is also valid for piecewise Lipschitz functions. In the following lemma, we justify that the risk bound established in Theorem~\ref{Cest-d>1} is minimax optimal up to logarithmic factors   by establishing a matching lower bound in  the class of piecewise Lipschitz functions.

\begin{lemma}\label{Mlb-M}
Consider the model described in Assumption 
 \ref{assume:tv functions} with $d\ge2.$ Let $\widetilde{f}$ any estimate based on the observations $\{(x_{i,j},y_{i,j})\}_{i=1,j=1}^{n,m_i}$, it holds that
\[
    \underset{ n \rightarrow \infty}{\lim\sup}\,\underset{f^* \in \mathcal{F}(L)     }{ \sup     }\,\mathbb{P}\left(  \| \widetilde{f}-f^* \|_2^2 \,\geq \,  C_{\text{opt}}\left(    \frac{1}{n }  \,+\, \frac{1}{  (nm)^{\frac{1}{d}} } \right) \right)\,>0.\,
   \]
   where $C_{\text{opt}} > 0$ is constant, and $\mathcal{F}(L)$ is the class of piecewise  Lipschitz functions, with  Lipschitz constant $L >0$.
\end{lemma}

See Appendix \ref{ProofLemma1} for the proof of Lemma \ref{Mlb-M}.

\subsection{Penalized \texorpdfstring{$K$-}{K}NN fused lasso for spatio-temporal data  }\label{sec:Mult_Pen}

Although the constrained estimator $\widehat{f}_{V_n}$ defined in (\ref{eqn:est1})--(\ref{eqn:est2}) is minimax optimal for estimating piecewise Lipschitz functions, its computational cost increases significantly for large  $n$ or $m$, due to the need to solve the quadratic program given in (\ref{eqn:theta_hat_knn}). This computational burden can pose a significant challenge in practice. To address this issue, we propose  the penalized $K$-NN fused lasso for spatio-temporal data defined as follows. First,  we compute 
\begin{equation}
    \label{eqn:penalized_knn}
    \widehat{\theta}_{\lambda}=\underset{\theta \in \mathbb{R}^{\sum_{i=1}^nm_i}}{\arg \min } \frac{1}{n}\sum_{i=1}^n \frac{1}{m_i}\sum_{j=1}^{m_i}\left(y_{i, j}-\theta_{i, j}\right)^2+\lambda\left\|\nabla_{G_K} \theta\right\|_1,
\end{equation}
for a tuning parameter $\lambda >0$. Once $\widehat{\theta}_{\lambda}$ is constructed, we define $\widehat{f}_{\lambda}$  as in (\ref{eqn:est1})--(\ref{eqn:est2}) but replacing $\widehat{\theta}_{V_n}$  with $\widehat{\theta}_{\lambda}$. 

To construct the estimator defined in (\ref{eqn:penalized_knn}), we develop an ADMM algorithm (see e.g. \cite{boyd2011distributed}), exploiting the max-flow algorithm from \cite{chambolle2009total}.
To arrive at our ADMM, we first notice that 
\begin{align*}
    &\underset{\theta \in \mathbb{R}^{\sum_{i=1}^{n}m_i}}{\arg \min } \frac{1}{n}\sum_{i=1}^n \frac{1}{m_i}\sum_{j=1}^{m_i}\left(y_{i,j}-\theta_{i, j}\right)^2+\lambda\left\|\nabla_{G_K} \theta\right\|_1\\
    =&
    \underset{\theta \in \mathbb{R}^{\sum_{i=1}^{n}m_i}}{\arg \min } \frac{1}{n}\sum_{i=1}^n \frac{1}{m_i}\sum_{j=1}^{m_i}\left(y_{i,j}-\theta_{i, j}\right)^2+\lambda\left\|\nabla_{G_K} z\right\|_1 \ \text{subject to} \quad    \theta-z=0. \
\end{align*}
Here, $\lambda>0$ is a tuning parameter, and $z$ is a slack variable. Then, proceeding with ADMM, (\ref{eqn:penalized_knn}) can be solved iteratively with the updates
\begin{align}
\theta^{l+1} & := \underset{\theta \in \mathbb{R}^{\sum_{i=1}^{n}m_i}}{\arg \min }\left(\sum_{i=1}^n \sum_{j=1}^{m_i}\frac{1}{nm_i}\left(y_{i,j}-\theta_{i, j}\right)^2+(\rho / 2)\left\|\theta- z^l+u^l\right\|_2^2\right),\label{eq:theta_min} \\  
z^{l+1} & := \underset{z \in \mathbb{R}^{\sum_{i=1}^{n}m_i}}{\arg \min }\left(\lambda\left\|\nabla_{G_K} z\right\|_1+(\rho / 2)\left\|\theta^{l+1}-z+u^l\right\|_2^2\right), \label{eq:z_min}\\
u^{l+1} & :=u^l+\theta^{l+1}- z^{l+1},
\end{align}
where $u$ is the scaled dual variable, and $\rho$ is the ADMM penalty parameter. To find the minimizer in  (\ref{eq:theta_min}), we use the optimality condition which leads to 
\begin{equation*}
   \theta^{l+1}_{i,j}=\frac{-\rho(y_{i,j}-z^{l}_{i,j}+u^{l}_{i,j})}{\rho
+\frac{2}{nm_i}}+y_{i,j}. 
\end{equation*}
As for the minimizer in (\ref{eq:z_min}), we use the max-flow algorithm from \cite{chambolle2009total}, invoking the Matlab function \texttt{knfl}, with penalty parameter $\frac{\lambda}{\rho}$. We vary $\rho$ using the scheme,
$$
\rho^{l+1}:= \begin{cases}\tau^{\text {incr }} \rho^l & \text { if }\left\|r^l\right\|_2>\mu\left\|s^l\right\|_2, \\ \rho^l / \tau^{\text {decr }} & \text { if }\left\|s^l\right\|_2>\mu\left\|r^l\right\|_2, \\ \rho^l & \text { otherwise },\end{cases}
$$
where $\mu>1, \tau^{\text {incr }}>1$, and $\tau^{\text {decr }}>1$ are parameters and $r^l$ corresponds to the primal residual, $r^l=\theta^l-z^l$, and $s^l$ corresponds to the dual residual, $s^l=\rho^l(z^{l}-z^{l-1})$. In practice, we choose  $\mu=10$ and $\tau^{\text {incr }}=\tau^{\text {decr }}=2$.

Next, with an algorithm in hand for computing the estimator defined in (\ref{eqn:penalized_knn}),  we provide an upper bound on its estimation error.

\begin{theorem}\label{Pest-d>1}
	Let $\{(x_{i,j},y_{i,j})\}_{i=1,j=1}^{n,m_i}$ be generated based on (\ref{eq:model l2}) and  suppose that Assumption  \ref{assume:tv functions} holds. Suppose in addition that
	$K \asymp \log^{1+l}(nm)$ for some $l>0$. It holds that, with an appropriate choice of $\lambda$ satisfying $\lambda \asymp  r_{n,m}\left\|\nabla_G \theta^*\right\|_1^{-1}$, we have 
	\begin{equation}
	    \label{PT-eq1}
     \frac{1}{n} \sum_{i=1}^n \frac{1}{m_i} \sum_{j=1}^{m_i}\left(\widehat{\theta}_{i ,j}-\theta_{i, j}^*\right)^2 =O_{\mathbb{P}} \left(r_{n,m}\right),
	\end{equation}
	where
	$$r_{n,m}=\frac{\log^{\alpha_1}(nm)}{n}+\frac{\log^{\alpha_2}(nm)}{(nm)^{\frac{1}{d}}},
	$$
	for $\alpha_1=8+2l$ and $\alpha_2=7+2l+\frac{2+l}{d}.$ Moreover,
	\begin{equation}
	    \label{PT-eq2}
     \vert\vert f^*- \widehat{f}_{V_n}\vert\vert_2^2=O_{\mathbb{P}} \left(\frac{\log^{  \widetilde{\alpha_1} }(nm)}{n}+\frac{\log^{ \widetilde{\alpha_2} }   (nm)}{(nm)^{\frac{1}{d}}}\right),
	\end{equation}
	where $\widetilde{\alpha_1}=8+2l$ and $\widetilde{\alpha_2}=7+2l+\frac{4+2l}{d}.$
\end{theorem}



Similarly to the constrained estimator, the penalized $K$-NN-FL for spatio-temporal data achieves nearly minimax rate,  $(nm)^{\frac{-1}{d}} +  n^{-1}$  except for logarithm factors, for estimating functions that have bounded variation and meet an additional condition, which extends piecewise Lipschitz continuity or satisfy the piecewise Lipschitz property. The advantage over the constrained estimator is that the penalized estimator can be efficiently found using our proposed ADMM.

\section{Experiments}\label{simu-data}

This section provides numerical evidence to support our theoretical results. We start with empirical experiments to evaluate the performance of the penalized estimators studied in Theorem \ref{Penalized} and Theorem \ref{Pest-d>1}, denoted here as {\color{black}{Locally Adaptive Regression Splines}} and $K$-NN-FL, respectively, for both settings $d=1$ and $d>1$. To establish a comparison, we include various competing methods as benchmarks. Finally, we present a real-world data example to demonstrate the practical applicability of our approach. The code to replicate all of our experiments can
be found at \url{https://github.com/cmadridp/K-NN-FL}.

\subsection{Simulated data}
\label{simu-data-1}

To conduct our simulations, we consider different model setups, which we call scenarios, for both settings $d=1$ and $d>1$. We describe these scenarios in the following subsections. Next, we outline and explain the shared features of our simulations for both the univariate and multivariate settings.

For the data created under each scenario, we perform a $75\%/25\% $ split of the data into a training set
and a test set. All competing models are then fit on the training data, with $5$-fold cross-validation to select tuning parameter
values. The performance of each method is evaluated by the MSE obtained when using the test data set. We then report the average MSE over 100 Monte Carlo simulations.



For each scenario below, we vary $n$ in a range of values. For each fixed value $n$ we then vary a parameter $m_{\text{mult}}$ in $\{0.5,1,1.5,2\}$, to obtain different sets of values for $\{m_i\}_{i=1}^n$. Specifically, we consider:
\begin{itemize}
  \item $m_1=...=m_{n/4}=10(m_{\text{mult}}+2)$,
  \item $m_{n/4+1}=...=m_{n/2}=10m_{\text{mult}}$,
  \item $m_{n/2+1}=...=m_{3n/4}=10(m_{\text{mult}}+3)$,
  \item $m_{3n/4+1}=...=m_{n}=10(m_{\text{mult}}+1)$.
\end{itemize}
The spatial noise is generated as $\delta_i (x)=0.5\delta_{i-1} (x)+\sum_{t=1}^{50} t^{-1} b_{t,i} h_{t}(x)$, where $$\{h_t(x)=\prod_{j=1}^{d}(1/\sqrt{2})\pi\sin(tx^{(j)})\}_{t=1}^{50},$$ are basis functions and  $\{b_{t,i}\}_{t = 1, i=1}^{ 50,n}$ are considered to be i.i.d.~$\mathcal{N}(0,1)$  through different scenarios on this section.  The measurement error is generated as follows. Set $M=10(m_{\text{mult}}+3)$. For each $i\in[n]$, the vector $\widetilde\epsilon_{i}\in\mathbb{R}^M$  is generated recursively as, $\widetilde\epsilon_{i}=0.3\widetilde\epsilon_{i-1}+\xi_i$, where $\{\xi_i\}_{i = 1}^n$ are i.i.d.~$\mathcal{N}(0,0.5\mathbf{I}_{M})$. To form the measurement errors, for each $i\in[n]$, we take the first $m_i$ entries of the vector $\widetilde\epsilon_{i}$ and denote the resulting vector as $\epsilon_i$. We observe the noisy spatio-temporal data $\{y_{i,j}\}_{i=1,j=1}^{n,m_i}$ at design points $\{x_{i,j}\}_{i=1,j=1}^{n,m_i}$  sampled  as follows. First we generate $\{x_{1,j}\}_{j=1}^{m_1}\sim\mathrm{Unif}([0,1]^d)$, then for any $1<i\le n$,
$$x_{i,j}=\begin{cases}
    x_{i-1,j},&\ \text{with probability}\ \phi=0.1.\\
    X_{i,j}\sim\mathrm{Unif}([0,1]^d), & \ \text{otherwise},
\end{cases},\ \forall j=1,...,m_i.$$

Additional numerical experiments are presented in Appendix \ref{ExtraSim}, for some of the scenarios described below. First, in Appendix \ref{ExtraSimD}, we adjust the spatial noise in different ways. We do this by looking at two aspects. To begin, we consider $\{b_{t,i}\}_{t = 1, i=1}^{ 150,n}$ and $\{b_{t,i}\}_{t = 1, i=1}^{ 250,n}$ to be i.i.d.~$\mathcal{N}(0,3)$, and $\mathcal{N}(0,1)$, respectively. Furthermore we modify $\{b_{t,i}\}_{t = 1, i=1}^{ 50,n}$ to be i.i.d.~$\mathcal{N}(0,2)$. In these experiments, we also vary the probability of the mixture $\phi$, to be $0.2$ and $0.05$. We also experiment with the measurement errors in this appendix. We set $\xi_i$ to be $\mathcal{N}(0, \mathbf{I}_{m_i})$ and $\mathcal{N}(0,2 \mathbf{I}_{m_i})$. Moreover, Appendix \ref{sec-DetailedRDataE} provides more details of the real data example presented in this Section.

\subsubsection{Univariate simulation study} 
All experiments were performed in \texttt{R}. To employ the {\color{black}{Locally Adaptive Regression Splines}} estimator, we use the \texttt{trendfilter} function in the \texttt{glmgen} package. We compared our method with additive smoothing splines in a wide range of scenarios, using the \texttt{smooth.spline} function. For all these scenarios, we consider $f^*=\widetilde{f}-\int_{[0,1]}\widetilde{f}(x)dx$ for functions $\widetilde{f}$ defined below. Moreover, we vary $n\in\{200,400,600,800\}$.

{$\bullet$ {\bf{Scenario 1} }} Piece-wise constant. \begin{equation*}
\widetilde{f}(x)=
    \begin{cases}
        5 & \text{if } x\le 0.2,\\
        2 & \text{if } 0.2<x \le 0.4,\\
        8 & \text{if } 0.4<x \le 0.6,\\
        1 & \text{if } 0.6<x. 
    \end{cases}
\end{equation*}

{$\bullet$ {\bf{Scenario 2} }}  Piece-wise linear. 
\begin{equation*}
\widetilde{f}(x)=
    \begin{cases}
        2.5x & \text{if } x\le 0.4,\\
        45x-17 & \text{if } 0.4<x \le 0.6,\\
        -40x+34 & \text{if } 0.6<x \le 0.8,\\
        30x-22 & \text{if } 0.8<x. 
    \end{cases}
\end{equation*}

{$\bullet$ {\bf{Scenario 3} } } Lagrange polynomial.
\begin{equation*}
    \widetilde{f}(x)=-\frac{2125}{6}x^{4}+\frac{2050}{3}x^{3}-\frac{2465}{6}x^{2}+\frac{248}{3}x.
\end{equation*}

{$\bullet$ {\bf{Scenario 4} }}  Sinusoidal function.
\begin{equation*}
\widetilde{f}(x)=2\sin\Big(\frac{2\pi}{(x+0.1)^{\frac{5}{10}}}\Big)+5.
\end{equation*}

Notice that one appealing aspect of our method in the univariate setting is that it can be adapted to different degrees of smoothness. In {\bf{Scenario 1}} to {\bf{Scenario 4}}, the way we generate data covers a wide range of simulation setups, with different levels of smoothness. This comprehensive approach ensures that our findings are robust across a variety of simulation settings, highlighting the versatility and relevance of our method in practical applications.

We apply 5-fold cross-validation to select the {\color{black}{Locally Adaptive Regression Splines}} order \( k \). The selected values align with the known smoothness of the true function \( f^* \): \( k = 1 \) for Scenario 1, \( k = 2 \) for Scenario 2, and \( k = 3 \) for Scenarios 3 and 4. This is consistent with standard practice (e.g., \citet{tibshirani2014adaptive}), where \( k \) is chosen using either cross-validation or structural assumptions.

The performance of each method is showcased in Figure \ref{fig:S1-d=1} and \ref{fig:S3-d=1}. From these visuals, one can discern that our proposed method consistently outperforms the competitors across all scenarios. This superiority in performance aligns with our expectations. As previously established in our theoretical findings, the {\color{black}{Locally Adaptive Regression Splines}} estimator for spatio-temporal data achieves nearly minimax rates in very general classes of functions.


\begin{figure}
    \centering
    \includegraphics[width=0.95\textwidth]{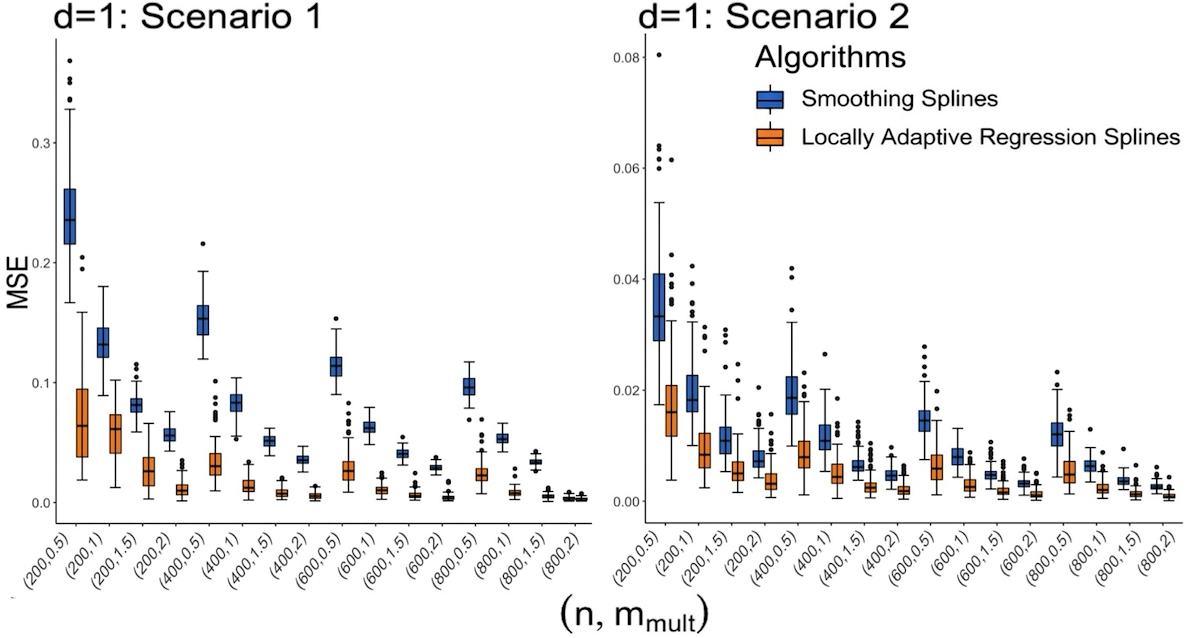}
    
    \caption{Box plots summarizing the results for {\bf{Scenarios 1}}  and {\bf{2}}.} 
    \label{fig:S1-d=1}

\end{figure}

\begin{figure}
    \centering
    \includegraphics[width=0.95\textwidth]{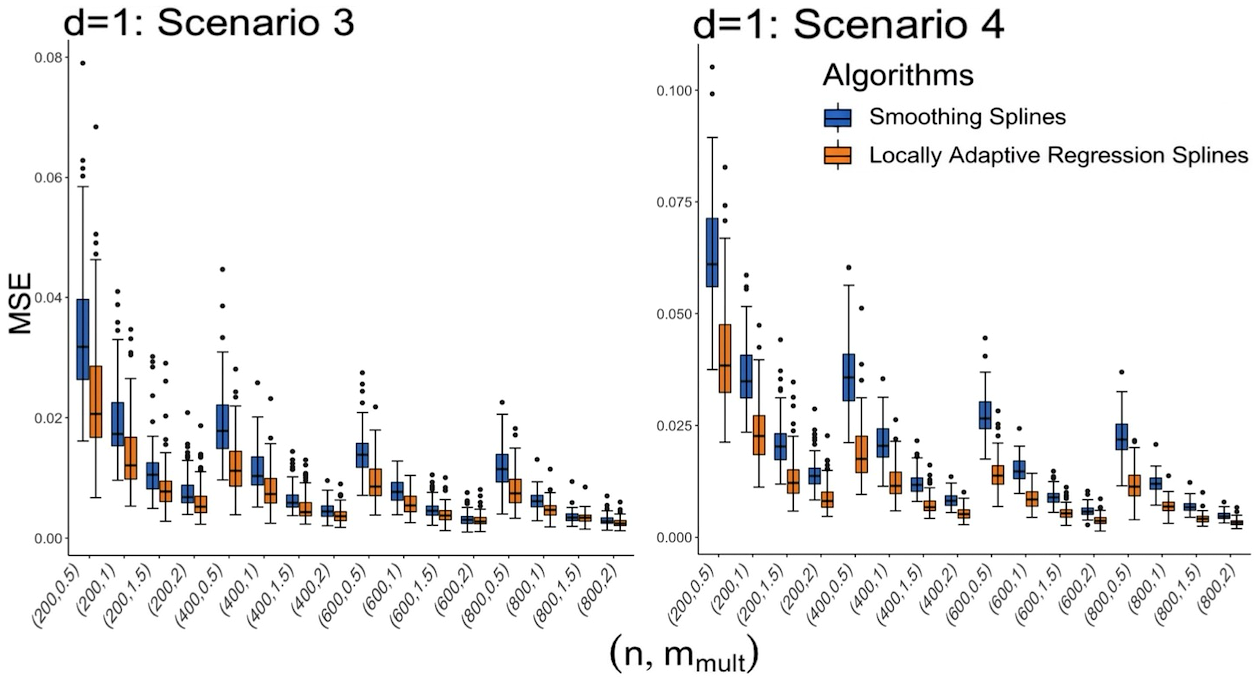}
    
    \caption{Box plots summarizing the results for {\bf{Scenarios 3}}  and {\bf{4}}.} 
    \label{fig:S3-d=1}

\end{figure}

\subsubsection{Multivariate simulation study}\label{Mult-Sim}

Following \cite{madrid2020adaptive}, we let  the number of neighbors in $K$-NN-FL be $K=5$. \textcolor{black}{This choice follows \cite{madrid2020adaptive} and has been found to work well in practice; see Appendix \ref{app:sensitivity-K} for a sensitivity analysis with respect to $K$.}  We compare our estimator with six competitors, 
\begin{itemize}
    \item CART (\cite{breiman1984richard}), implemented in the \texttt{R} package \texttt{rpart}, with the complexity parameter treated as a tuning parameter. CART is a decision tree learning technique that continually splits the data into subsets to aid in predictions. An essential feature of CART is the complexity parameter, which controls the size and complexity of the tree. 
    \item MARS (\cite{friedman1991multivariate}), implemented in the \texttt{R} package \texttt{earth}, with the penalty parameter treated as a tuning parameter. MARS is a non-linear regression approach that uses spline functions to model data. A penalty parameter in MARS decides the flexibility of the splines. 
    \item Random forests (\cite{breiman2001random}), implemented in the \texttt{R} package \texttt{randomForest}, with the number of trees fixed at 500, and with the minimum size of each terminal node treated as a tuning parameter. When deploying Random Forests, the number of trees is a crucial parameter, often set at 500 in many applications. Another tunable parameter is the minimum size of each terminal node, which ensures sufficient data in each decision-making leaf of the trees.
     \item KNN Regression (\cite{stone1977consistent}), implemented in  \texttt{Matlab}, with the number of neighbors treated as a tuning parameter.
\end{itemize}
Additionally, for the case $d=2,$ we also compare our proposed method with,
\begin{itemize}

      \item Local Linear Estimator (LLE). Originally implemented in the univariate case in \cite{li2010uniform}. A slight modification for the $2$-dimensional case can be achieved easily. Let $x\in[0,1]^2$, then a local-linear estimator of $f^*(x)$, is given by $\widehat{a}_0$, where
      $$\left(\widehat{a}_0, \widehat{a}_1\right)=\underset{a_0\in{\mathbb{R}}, a_1\in\mathbb{R}^2}{\arg \min } \frac{1}{n} \sum_{i=1}^n \frac{1}{m_i} \sum_{j=1}^{m_i}\left\{y_{i,j}-a_0-a_1\left(x_{i,j}-x\right)\right\}^2 \mathcal{K}_{h}\left(x_{i,j}-x\right)$$
      with $\mathcal{K}$ a given kernel with bandwidth $h$. Here $h$ is the tuning parameter, and the chosen kernel is Gaussian.
      \item Voronoigram. Proposed in \cite{hu2022voronoigram}, it is considered with the penalty parameter $\lambda$ and the weighted penalty operator treated as tuning parameters. The code kindly provided by the authors for this algorithm is only available for the two-dimensional setting. 
\end{itemize}
We consider scenarios of spatio-temporal data inspired by \cite{madrid2020adaptive}. Specifically, we investigate the subsequent scenarios: for {\bf{Scenario 5}} and {\bf{Scenario 6}}, we set $d=2$ and consider $n \in \{1000,2000\}$;  while for {\bf{Scenario 7}} and {\bf{Scenario 8}}, we keep $n$ fixed at 1000 and vary $d$ between $3$ and $10$.

{$\bullet$ {\bf{Scenario 5} }}
The function $f^*:[0,1]^2 \rightarrow \mathbb{R}$ is piecewise constant,
$$
f^*(x)=\mathbf{1}_{\left\{t \in \mathbb{R}^2:\left\|t-\frac{3}{4}(1,1)^T\right\|_2<\left\|t-\frac{1}{2}(1,1)^T\right\|_2\right\}}(x) .
$$

{$\bullet$ {\bf{Scenario 6} }}
The function $f^*:[0,1]^2 \rightarrow \mathbb{R}$ is piecewise constant,
$$
f^*(x)=\mathbf{1}_{\left\{\left\|x-\frac{1}{2}(1,1)^T\right\|_2^2 \leq \frac{2}{1000}\right\}}(x).
$$

{$\bullet$ {\bf{Scenario 7} }}
The function $f^*:[0,1]^d \rightarrow \mathbb{R}$ is defined as,
$$
f^*(x)= \begin{cases}1 & \text { if }\left\|x-\frac{1}{4} \mathbf{1}_d\right\|_2<\left\|x-\frac{3}{4} \mathbf{1}_d\right\|_2, \\ -1 & \text { otherwise. }\end{cases}
$$

{$\bullet$ {\bf{Scenario 8} }} The function $f^*:[0,1]^d \rightarrow \mathbb{R}$ is defined as
$$
f^*(x)= \begin{cases}2 & \text { if }\left\|x-q_1\right\|_2<\min \left\{\left\|x-q_2\right\|_2,\left\|x-q_3\right\|_2,\left\|x-q_4\right\|_2\right\}, \\ 1 & \text { if }\left\|x-q_2\right\|_2<\min, \left\{\left\|x-q_1\right\|_2,\left\|x-q_3\right\|_2,\left\|x-q_4\right\|_2\right\}, \\ 0 & \text { if }\left\|x-q_3\right\|_2<\min \left\{\left\|x-q_1\right\|_2,\left\|x-q_2\right\|_2,\left\|x-q_4\right\|_2\right\}, \\ -1 & \text { otherwise, }\end{cases}
$$
where $q_1=\left(\frac{1}{4} \mathbf{1}_{\lfloor d / 2\rfloor}^T, \frac{1}{2} \mathbf{1}_{d-\lfloor d / 2\rfloor}^T\right)^T, q_2=\left(\frac{1}{2} \mathbf{1}_{\lfloor d / 2\rfloor}^T, \frac{1}{4} \mathbf{1}_{d-\lfloor d / 2\rfloor}^T\right)^T, q_3=\left(\frac{3}{4} \mathbf{1}_{\lfloor d / 2\rfloor}^T, \frac{1}{2} \mathbf{1}_{d-\lfloor d / 2\rfloor}^T\right)^T$ and $q_4=\left(\frac{1}{2} \mathbf{1}_{\lfloor d / 2\rfloor}^T, \frac{3}{4} \mathbf{1}_{d-\lfloor d / 2\rfloor}^T\right)^T$.

\begin{figure}[h]
    \centering
     \includegraphics[width=1.\textwidth]{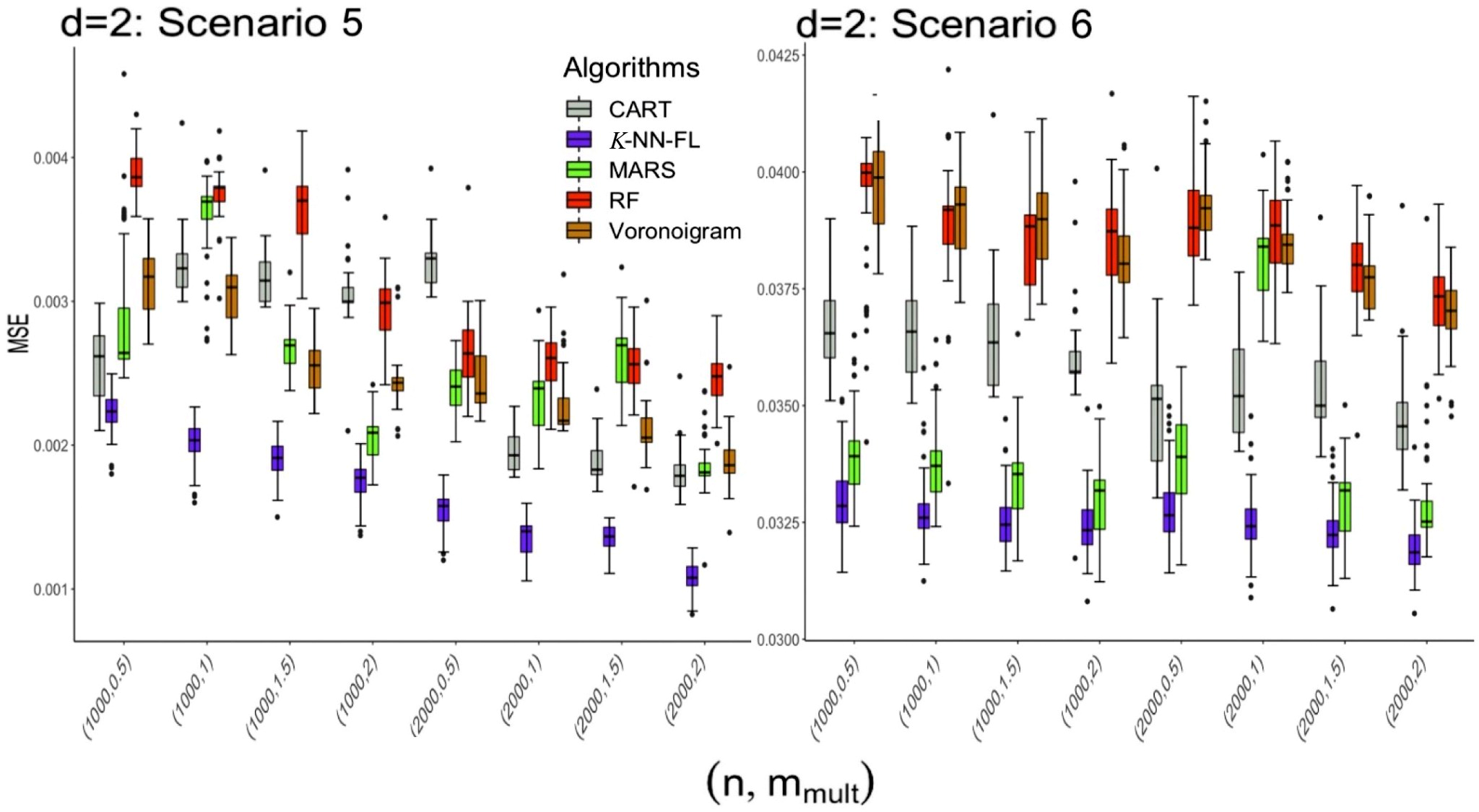}
    
    \caption{
    From left to right, we present a box plot of the performance of the top 5 methods under {\bf{Scenario 5}} and {\bf{Scenario 6}}, over $100$  repetitions, respectively. In each figure, the performance of each competitor is measured by the MSE, averaged over the $100$ repetitions.}
    \label{fig:S5S6}

\end{figure}

\begin{figure}
    \centering
    \includegraphics[width=0.9\textwidth]{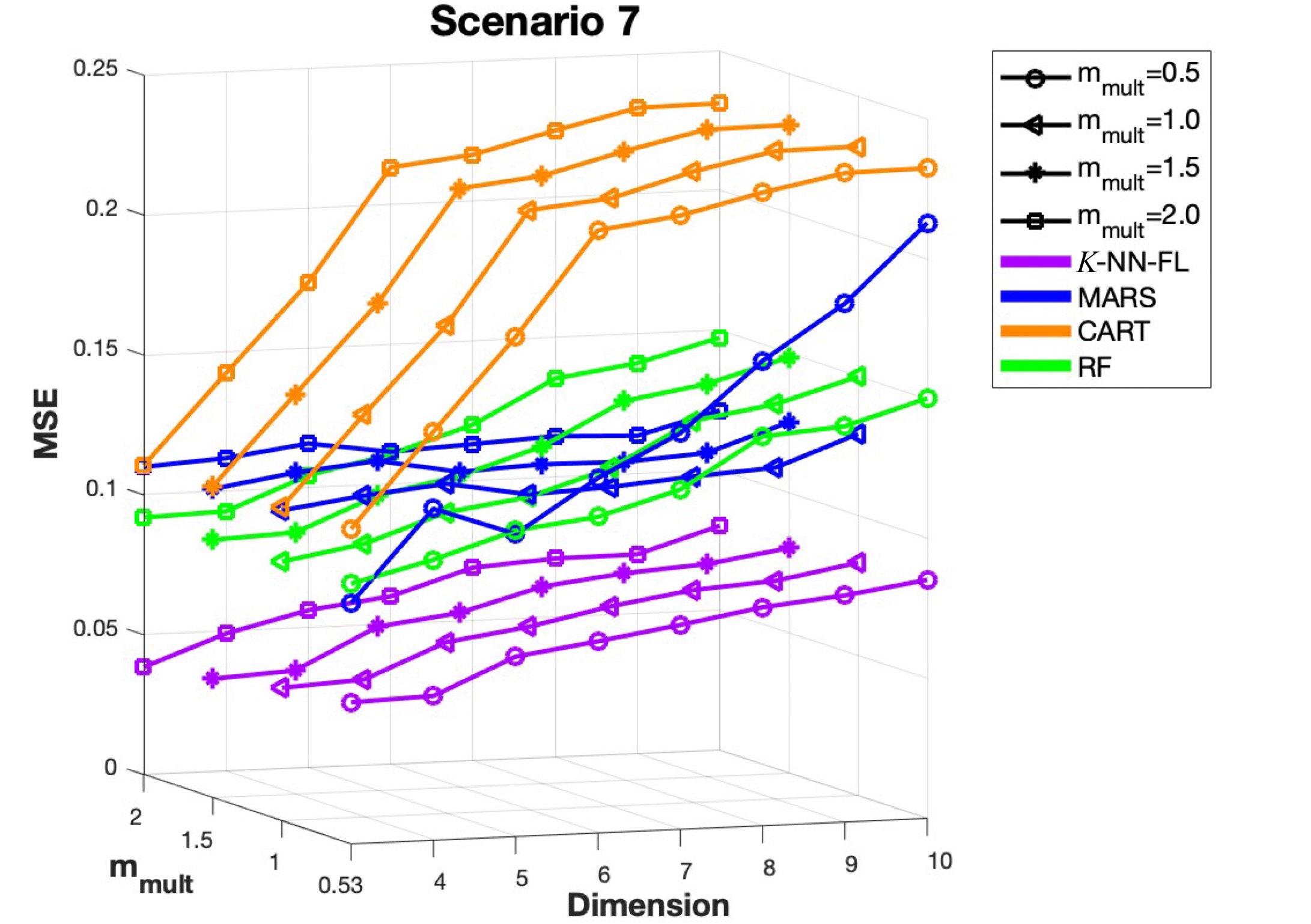}
    
    \caption{MSE, averaged over 100 Monte Carlo simulations, for {\bf{Scenario 7}}.}
    \label{fig:S7}

\end{figure} 

\begin{figure}
    \centering
    \includegraphics[width=0.9\textwidth]{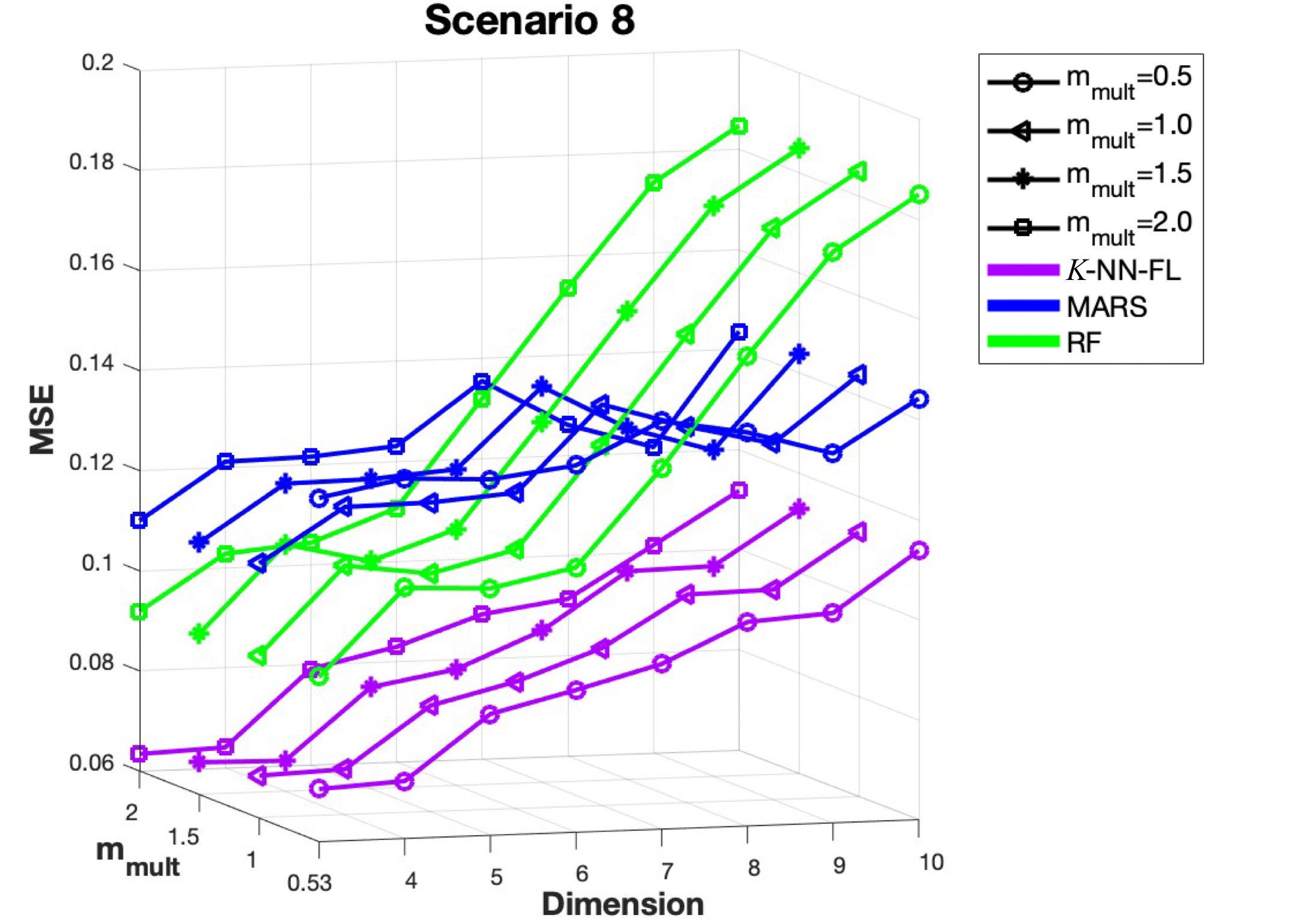}
    
    \caption{MSE, averaged over 100 Monte Carlo simulations, for {\bf{Scenario 8}}.}
    \label{fig:S4-d>2-1}

\end{figure}

 Figures \ref{fig:S5S6}--\ref{fig:S4-d>2-1} provide a visualization of the performance best performing methods in each scenario, while more comprehensive results can be found in Section \ref{ExtraSim}. Remarkably, $K$-NN-FL  outperforms all other competitors. It is worth noting that across all scenarios, the MSE for each competitor method decreases as the sample size increases. However, in Section \ref{ExtraSimD}, where spatial noise, measurement error, and probability of the mixture are varied, $K$-NN-FL is the only method that maintains this behavior.

\begin{table}[h]
	\caption{Summary of the result of the Real Data Example.  In this table, we present the averaged Prediction Error (and standard deviation), both multiplied by $100$.  In each setting, we highlight the best result in \textbf{bold} and
        the second best result in \textbf{\textit{bold and italic}}}
	\label{RdTable}
	\begin{center}
		\begin{small}
			\begin{sc}
				\begin{tabular}{lccr}
					\hline
					
					Method   & June & July & August  \\
					\hline
					\\
					& \multicolumn{3}{c}{Region $\mathbb{I}$}\\
					\multicolumn{4}{c}{average (standard deviation) of Prediction Error} \\
					$K$-NN-FL & \textbf{0.740 (0.287)}& \textbf{0.660 (0.257)} &\textbf{ 0.590 (0.251)} \\
					MARS & \textbf{\textit{0.745 (0.296)}} & \textbf{\textit{0.671 (0.258)}} &\textbf{\textit{ 0.595 (0.251) }}\\
					CART& 4.454 (0.985) & 4.556 (0.994) &  4.048 (0.901) \\
					RF & 0.914 (0.522) & 0.888 (0.472) & 0.788 (0.321) \\
					KNN-R & 63.121 (10.731) & 53.818 (9.010) & 49.831 (8.678) \\
					LLE & 314.162 (53.941) & 298.510 (50.918) & 201.399 (78.510) \\
     Voronoigram  & 5.478 (0.311) & 7.301 (0.292) & 12.590 (0.563) \\
					\\
					& \multicolumn{3}{c}{Region $\mathbb{II}$}\\
					\multicolumn{4}{c}{average (standard deviation) of Prediction Error}  \\
					$K$-NN-FL & \textbf{0.610 (0.227)}& \textbf{0.810 (0.280)} & \textbf{0.880 (0.309)} \\
					MARS & \textbf{\textit{0.622 (0.231)}} & \textbf{\textit{0.839 (0.292)}} & \textbf{\textit{0.938 (0.330)}} \\
					CART& 2.383 (0.472) & 2.534 (0.509) &  2.506 (0.489) \\
					RF & 1.224 (3.010) & 1.412 (0.490) & 0.788 (0.321) \\
					KNN-R & 34.219 (0.728) & 39.810 (0.791) & 70.291 (10.033) \\
					LLE & 130.511 (39.610) & 149.190 (39.999) & 159.211 (40.192) \\
     Voronoigram  & 12.594 (1.160) & 27.742 (1.213) & 12.285 (1.090) \\
					\\
					& \multicolumn{3}{c}{Region $\mathbb{III}$}\\
					\multicolumn{4}{c}{average (standard deviation) of Prediction Error}  \\
					$K$-NN-FL & \textbf{0.200 (0.181)}& \textbf{0.180} \textbf{\textit{(0.182)}} & \textbf{0.160 (0.160) }\\
					MARS & \textbf{0.200} \textbf{\textit{(0.182)}} & \textbf{\textit{0.189}} \textbf{(0.170)} &\textbf{\textit{ 0.167 (0.162)}} \\
					CART& 0.218 (0.280) & 0.196 (0.200) &  0.178 (0.168) \\
					RF & 0.274 (0.310) & 0.259 (0.215) & 0.240 (0.190) \\
					KNN-R & 2.349 (0.689) & 2.197 (0.612) & 2.003 (0.519) \\
					LLE & 58.300 (1.630) & 46.109 (0.930) & 42.851 (0.901) \\
     Voronoigram  & 1.740 (0.419) & 2.317 (0.520) & 2.624 (0.529) \\
					\\
					\hline
				\end{tabular}
			\end{sc}
		\end{small}
	\end{center}
\end{table}

\begin{figure}
	\centering
	\includegraphics[width=0.79\textwidth]{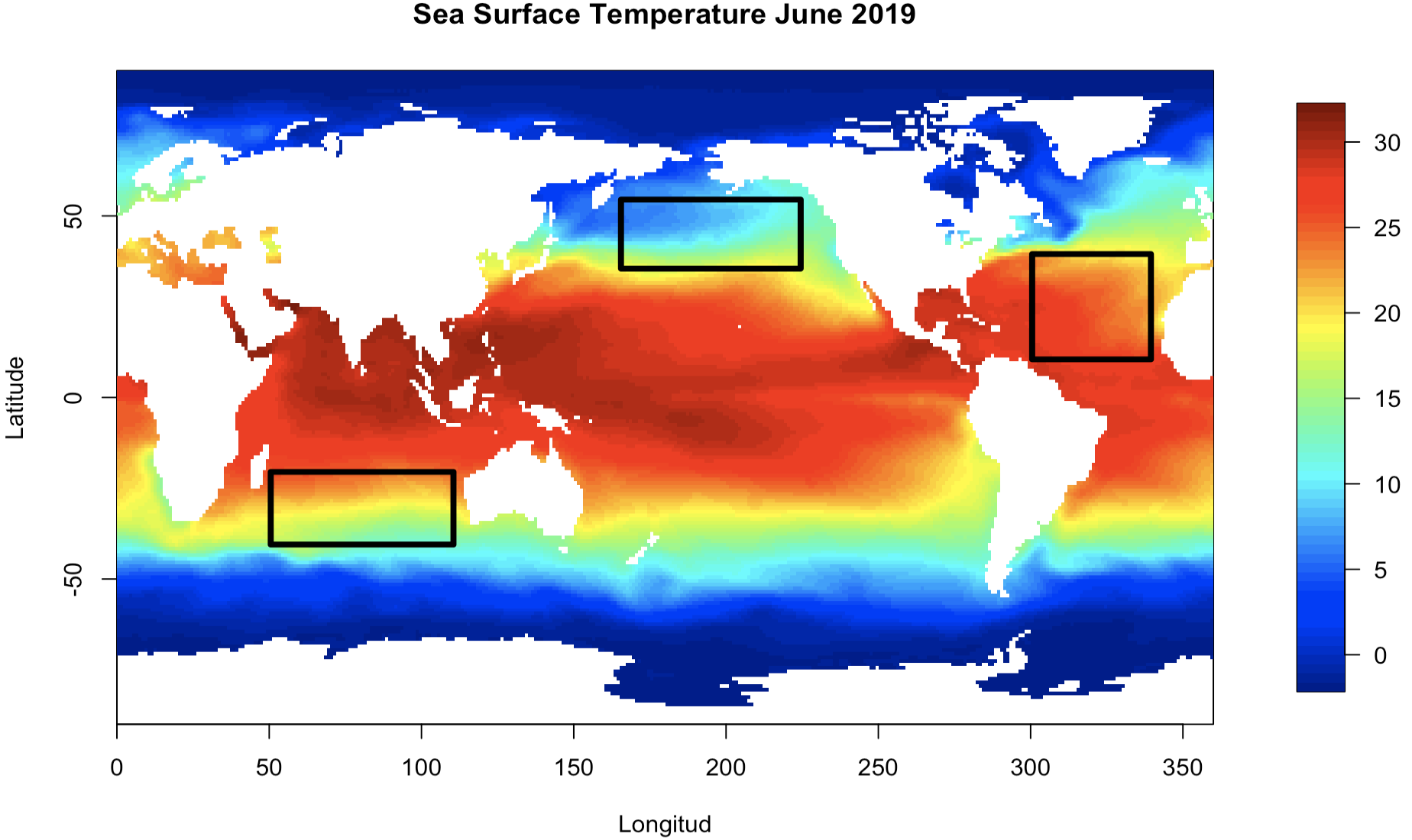}
	
	\caption{SST data captured for the specific time period of June 2019. Highlighted in the graphic are three black rectangles, each symbolizing one of the distinct designs under consideration in our study.} 
	\label{Regions}
	
\end{figure}

\begin{figure}
	\centering
	\includegraphics[width=0.79\textwidth]{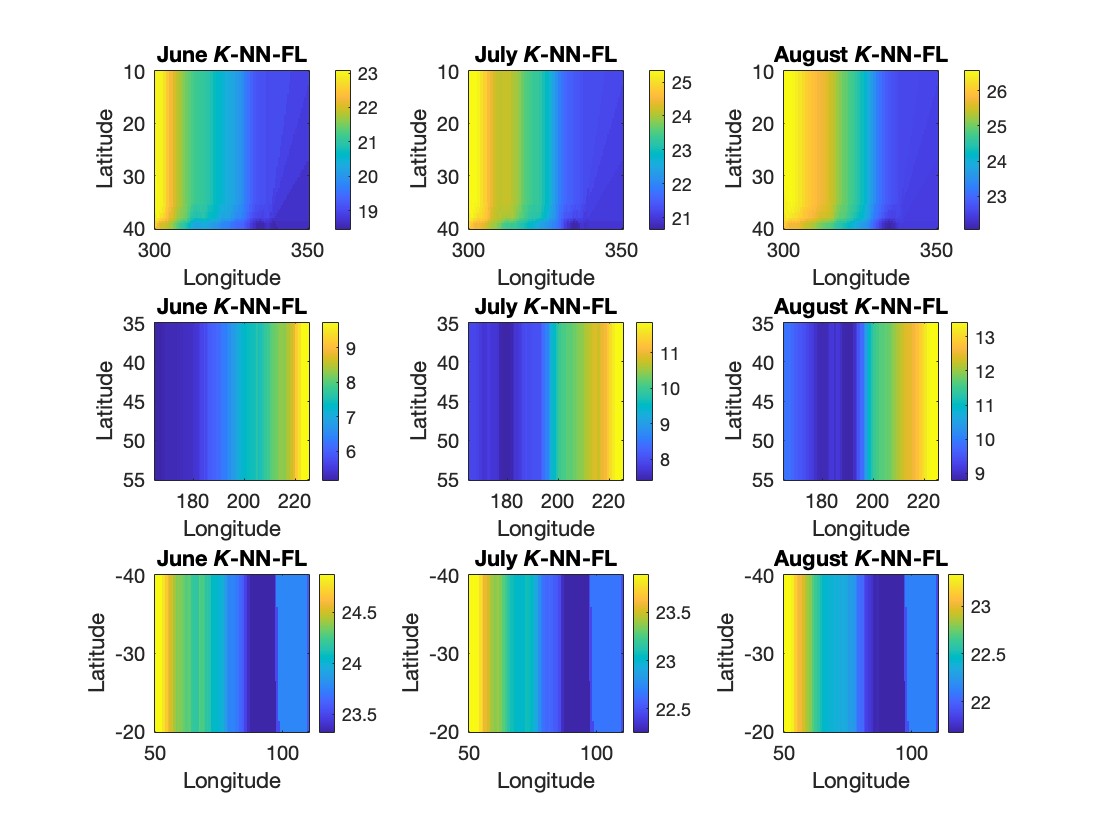}
	\caption{From top to bottom, each row presents Heatmaps of the SST function estimated by our proposed method in the regions $\mathbb{I}$, $\mathbb{II}$ and $\mathbb{III}$, respectively. Each column corresponds to the months of June, July, and August, respectively.} 
	\label{Est-Func-R1}
\end{figure}

\subsection{Real data application}\label{Real-data-ex}
We consider the COBE-SSTE dataset \citep{dataset}, which consists of monthly average sea surface temperature (SST) from 1850 to 2019, on three different design points of $1$ degree latitude by $1$ degree longitude $(30 \times 40)$, $(20 \times 60)$, and $(21 \times 61)$, denoted by $\mathbb{I}, \mathbb{II}$ and $\mathbb{III}$ respectively. The specific coordinates are latitude $10$S-$39$S and longitude $300$E-$340$E for a design located in the Atlantic Ocean (Region $\mathbb{I}$). A second design that we can find in the Pacific Ocean with coordinates latitude $35$S-$55$S and longitude $165$E-$225$E is considered (Region $\mathbb{II}$). Finally, we consider a grid design in the Indian Ocean with coordinates latitude -$40$S-$20$S and longitude $50$E-$110$E (Region $\mathbb{III}$). See Figure \ref{Regions} for a visualization of these three regions.


To make sea surface temperature predictions, taking into account seasonality, we have applied our proposed method along with the other methods CART, MARS, RF, KNN-R, LLE, and Voronoigram, using data from the months of June, July, and August for each year from 1850 to 2019, separately. The total number of years is $n=170$, and we have considered $m_i=1200$ for each year $i=1,...,n$ for the regions $\mathbb{I}$ and $\mathbb{II}$, while $m_i=1281 \,\forall i=1,...,n$ has been considered for the region $\mathbb{III}$.

To implement our analysis, we have followed the methodology described in Section \ref{simu-data-1}. First, the data is split into training and test sets using $30$ $75\%/25\%$ splits for all the regions $\mathbb{I}$, $\mathbb{II}$, and $\mathbb{III}$. All models have been fitted using the training data with $5$-fold cross-validation to determine the tuning parameter values and prediction performance has been evaluated on the test set.


To compare the performance of our proposed method against the competitors, we use the average test set prediction error (across the 30 test sets) as the performance metric. The results have been highly promising; our proposed method has not only performed well but has actually outperformed all other methods in this comparative study. Importantly, these results are consistent across different months and geographical regions. All these comparative details have been meticulously tabulated and can be found in Table \ref{RdTable}. To offer a visualization of the estimated SST function, we have plotted the predicted SST function using our $K$-NN-FL method. These
visualizations are provided in Figure \ref{Est-Func-R1} for each of the different geographical regions. The interested reader in a more detailed comparison, we have also included a section that compares our estimated SST function with that of the most competitive alternative method, MARS. This section is referenced as Appendix \ref{sec-DetailedRDataE}.

\section{Outline of the proof of Theorems \ref{thm:main tv} and \ref{Penalized}}
 \label{OutlineT1-2}
\ \ \ \
{\bf{Theorem \ref{thm:main tv}}}:
The proof of Theorem \ref{thm:main tv} proceeds in two main steps. The first step analyzes the estimation error in terms of the $L_2$ norm. Here, we rely on two key ingredients: a minimization property and a coupling between the empirical norm and the $L_2$ norm. The second step establishes deviation bounds, which are based on Rademacher complexities and account for both measurement errors and spatial noise.

A central difficulty in both steps comes from the interaction between time and space, what we refer to as spatio-temporal dependence. To address this, we introduce a blocking strategy described in Appendix~\ref{BetaMsection}. The idea is to convert the sum of the original $\beta$-mixing sequence into a sum of  approximately independent blocks. The total number of such blocks is of order $O(\log(n))$. This reduction allows us to analyze each block separately, treating them as independent, and to reduce the complexity of the full sequence to a process of effective length $O(\log(n))$.

{\bf{Decomposition of estimation error}}:
Let \(\eta^2 = \eta_1^2 + V_n^2 \eta_2^2 + \frac{1}{n}\) denote the target rate, where \(\eta_1^2 \asymp \frac{\log n}{n}\) and \(\eta_2^2 \asymp \frac{\log^{4k/(2k+1)}(n)}{(nm)^{2k/(2k+1)}}\). In this step, our goal is to control the probability \(\mathbb{P}\left(\|\widehat{f}_{k,V_n} - f^*\|_2^2 > 4\eta^2\right)\).

To proceed, we center the true function \( f^* \) by subtracting its empirical mean. This is needed because the estimator \(\widehat{g}_{k,V_n}\) is constructed to have empirical mean zero. We define
\[
\mathcal{A}_{f^*} := \frac{1}{\sum_{i=1}^n m_i} \sum_{i=1}^n \sum_{j=1}^{m_i} f^*\left(x_{i, j}\right),
\quad \widetilde{f}^* := f^* - \mathcal{A}_{f^*},
\]
so that \(\widetilde{f}^*\) lies in the same centered function space as the estimator.

We now introduce the set \(\Lambda\) consisting of all convex combinations of \(\widehat{g}_{k,V_n}\) and \(\widetilde{f}^*\). For any \(f \in \Lambda\), we define the deviation \(\Delta_f := f - \widetilde{f}^*\), which measures how far \(f\) is from the centered true function. Using the fact that \(\widehat{g}_{k,V_n}\) minimizes the empirical squared error over the total variation ball \(B_{TV}^{k-1}(V_n)\), and that this set is convex, we can derive the following inequality,
\begin{align}
\label{Equation1-maintext}
	\|f^* - \bar{y} - f\|_{nm}^2
&\lesssim  \underbrace{\frac{1}{n} \sum_{i=1}^n \frac{1}{m_i} \sum_{j=1}^{m_i} \left[\Delta_f(x_{i,j}) \delta_i(x_{i,j}) - \langle \Delta_f, \delta_i \rangle_{\mathcal{L}_2} \right]}_{\text{Term (i): spatial noise contribution}} \\
&\quad + \underbrace{\frac{1}{n} \sum_{i=1}^n \frac{1}{m_i} \sum_{j=1}^{m_i} \Delta_f(x_{i,j}) \epsilon_{i,j}}_{\text{Term (ii): measurement error contribution}} 
+  \underbrace{\frac{1}{n} \sum_{i=1}^n \frac{1}{m_i} \sum_{j=1}^{m_i} \langle \Delta_f, \delta_i \rangle_{\mathcal{L}_2}}_{\text{Term (iii): spatial mean effect}}. \nonumber
\end{align}

Armed with Inequality~\eqref{Equation1-maintext}, the next step is to relate the empirical norm to the $L_2$ norm. This transition is guided by Lemma~\ref{lemma2}, which applies to function classes that are bounded both in total variation and in the $L_\infty$ norm. To apply the lemma, we note that for any \(f \in \Lambda\), the deviation \(\Delta_f\) satisfies \(J_k(\Delta_f) \leq 2V_n\). What remains is to show that \(\Delta_f\) also lies in an $L_\infty$ ball. In fact, with high probability, we can guarantee that
\begin{equation}
\label{family}
\{\Delta_f : f \in \Lambda\} \subset B_{TV}^k(2V_n) \cap B_\infty(L_n) =: \mathcal{F},
\end{equation}
where \(L_n \asymp V_n\) is a quantity of the same order as the tuning parameter \(V_n\). Given this, Lemma~\ref{lemma2} allows us to control the $L_2$ norm of \(\Delta_f\) in terms of its empirical norm. Specifically, with high probability,
\[
\|\Delta_f\|_2^2 \lesssim \|\Delta_f\|_{nm}^2 + L_n^2 \eta_2^2.
\]

To derive the full bound, we further use that the response mean \(\bar{y}\) and the empirical mean \(\mathcal{A}_{f^*}\) of the true signal differ by at most \(\eta_1\) with high probability. This contributes the additional term \(\eta_1^2\). Moreover, Lemma~\ref{expected-value-b} helps to control the average of the inner products 
\(\langle \Delta_f, \delta_i \rangle_{\mathcal{L}_2}\), corresponding to Term~(iii) in 
Equation~\eqref{Equation1-maintext}, and contributes another term of order \(\frac{1}{n}\). 

Putting everything together, we obtain that for any \(f \in \Lambda\), with high probability,
\begin{align}
\label{u-b-2-maintext}
\|\Delta_f\|_2^2
&\lesssim 
\underbrace{\eta_1^2}_{\text{mean centering}} 
+ \underbrace{V_n^2 \eta_2^2}_{\text{empirical-to-}L_2\text{ bound}} 
+ \underbrace{\frac{1}{n}}_{\text{Term~(iii) bound}} 
 + \text{Term~(i)} +\text{Term~(ii)}.
\end{align}

To conclude this step, let \(\Omega\) denote the high-probability event where the previously mentioned
inequality holds. Then, for any \(\varepsilon > 0\), we can write
\[
\mathbb{P}\left(\|\widehat{f}_{k,V_n} - f^*\|_2^2 > 4\eta^2\right) \leq A_1 + A_2 + \frac{\varepsilon}{2}.
\]
Here, \(A_1\) is the probability that the supremum over the class \(\mathcal{F}\) of the empirical inner product between the spatial noise and \(\Delta_f\) exceeds \(\eta^2\). The term \(A_2\) represents the analogous probability involving measurement errors.

\textbf{Deviation bounds for measurement errors and spatial noise.} 
In the full argument presented in Appendix~\ref{sec-proof-thm2}, each of the terms $A_1$
  and $A_2$ is controlled using standard concentration inequalities from empirical process theory. In particular, the analysis relies critically on the fact that the deviations \(\Delta_f\) lie within the class \(\mathcal{F}\), defined as the intersection of a total variation ball and a uniform bound in Equation~\eqref{family}. This structural constraint is key to deriving uniform bounds over the function class.


{\bf{Theorem \ref{Penalized}}}:
The proof of Theorem \ref{Penalized} follows a similar path, with the incorporation of an additional step into the first phase of Theorem \ref{thm:main tv}. This relies on the fact that the constraint-based on $V_n$ is no longer available for Theorem \ref{Penalized}. We point out the main difference in Appendix \ref{OutLineThm2}.

\section{Outline of the proof of Theorems \ref{Cest-d>1} and \ref{Pest-d>1}} \label{OutlineT3-4}
In the course of proving both theorems, addressing the spatio-temporal dependence poses a significant challenge. We employ the strategy delineated in Section \ref{OutlineT1-2} to confront this complexity.

We shall now explain the basic ideas behind the proofs of both Theorems. We start with a discussion on embedding a grid graph into a $K$-NN graph, under Assumptions \ref{assume:tv functions}. This construction will then be utilized in the subsequent sections to derive an upper bound for the MSE of the 
$K$-NN-FL estimator.

In the flow-based proof of Theorem 4 presented by \cite{JMLR:v15:vonluxburg14a}, concerning commute distance on $K$-NN-graphs, a concept of embedding was introduced. Within this context, the authors explained the idea of a ``valid grid" as described in Definition 17. Given a predetermined set of design points, a grid graph $G$ is considered valid if it meets the following criteria:
\begin{itemize}
    \item The grid width must be sufficiently large such that every cell within the grid contains at least one design point.
    \item Conversely, the grid width should not be too large. Points situated within the same grid cell, or those in adjacent cells, maintain connectivity within the 
   $K$-NN graph.
\end{itemize}
The concept of a valid grid was originally formulated for fixed design points. However, through a minor modification, \cite{madrid2020adaptive} constructed a grid graph for random design points. This graph, with a high probability, meets the criteria for a valid grid as defined by \cite{JMLR:v15:vonluxburg14a}, assuming that the design points are independent. We expand on this construction by including the spatio-temporal dependence assumption as detailed in Assumption \ref{assume:tv functions}. Comprehensive details on this adaptation are provided in Appendix \ref{Lat-section}. We shall now delineate the methodological framework behind the grid embedding. This embedding, when applied to any particular signal, ensures a lower bound on the total variation along the $K$-NN graph.

Given $N \in \mathbb{N}$, we formulate a $d$-dimensional grid graph $G_{\text {lat }}=\left(V_{\text {lat }}, E_{\text {lat }}\right)$, which represents a lattice graph within the $[0,1]^d$ characterized by equal side lengths. This graph contains a total of $\left|V_{\text {lat }}\right|=N^d$ nodes. For simplification purposes, we assume that the nodes of the grid correspond to the points in $P_{\text {lat }}(N)$, where $P_{\text {lat }}(N)$ is the set of centers of the elements of the partition $\mathcal{P}_{N^{-1}}$, adhering to the notations established in Section \ref{sec:notation}. Furthermore, $z, z^{\prime} \in P_{\text {lat }}(N)$ share an edge in the graph $G_{\text {lat }}(N)$ if and only if $\left\|z-z^{\prime}\right\|_2=$ $N^{-1}$. If the nodes corresponding to $z, z^{\prime}$ share an edge, then we will write $\left(z, z^{\prime}\right) \in E_{\mathrm{lat}}(N)$.

Using $P_{\text {lat }}(N)$, for any signal $\theta \in \mathbb{R}^{\sum_{i=1}^nm_i}$, we can construct two vectors, specifically $\theta_I \in  \mathbb{R}^{\sum_{i=1}^nm_i}$ and $\theta^I \in \mathbb{R}^{N^d}$. The vector $\theta_I$ is a signal vector that is constant within lattice cells. While the vector $\theta^I$ comprises coordinates aligned with the various nodes of the lattice, which correspond to the centers of the cells. The precise definitions of $\theta_I$ and $\theta^I$ are given in Section \ref{Lat-section}. 

Since $\theta$ and $\theta_I$ have the same dimension, it is natural to ask how these two relate to each other, at least for the purpose of understanding the empirical process associated with the $K$-NN-FL estimator. Moreover, given that $\theta^I \in \mathbb{R}^{N^d}$, one can try to relate the total variation of $\theta^I$ along a $d$-dimensional grid with $N^d$ nodes, with the total variation of the original signal $\theta$ along the $K$-NN graph. 

Our next discussions will focus on clarifying these relationships. Notably, Lemma \ref{lemma8-Oscar} offers immediate insights on how to manage the empirical process tied to the $K$-NN-FL estimator. Specifically, by using the argument about minimizing properties, similar to what we talked about in Section \ref{OutlineT1-2}, it becomes crucial to obtain bounds for the expressions:
$$
 \sum_{i=1}^n \sum_{j=1}^{m_i}\left(\theta_{i, j}-\theta_{i, j}^*\right) \epsilon_{i,j},\ \sum_{i=1}^n \sum_{j=1}^{m_i}\left(\theta_{i, j}-\theta_{i, j}^*\right) \delta_i\left(x_{i,j}\right).
$$
To this end, let $\{\epsilon_{i,j}\}_{i=1,j=1}^{n,m_i}=:\epsilon\in\mathbb{R}^{\sum_{i=1}^nm_i}$ using the notation detailed in Section \ref{sec:notation}. The following relation can be expressed, 
$$
 \epsilon^T\left(\hat{\theta}-\theta^*\right)= \epsilon^T\left(\hat{\theta}-\hat{\theta}_I\right)+ \epsilon^T\left(\hat{\theta}_I-\theta_I^*\right)+ \epsilon^T\left(\theta_I^*-\theta^*\right) .
$$
A similar approach is used for handling spatial noise. As a result, in proving Theorem \ref{Cest-d>1} as detailed in Appendix \ref{proofT3}, our strategy focuses on setting bounds for each term on the right-hand side of the aforementioned equation. For the proof of Theorem \ref{Pest-d>1}, our methodology incorporates an auxiliary step, mirroring the strategy delineated in Section \ref{OutlineT1-2} pertinent to the proof of Theorem \ref{Penalized}. Comprehensive elaboration on this is available in Appendix \ref{proofT4}.

\section{Conclusion}\label{sec-conclusion} 
We tackle the problem of estimating the non-parametric regression function from data that exhibit both temporal and spatial dependence, which has not been studied in the literature. In the univariate case, our proposed estimator is built using the univariate {\color{black}{Locally Adaptive Regression Splines}} estimator, i.e., defined by penalizing according to the sum of the total variation operator of the weak derivative of the component function. We have demonstrated that this estimator is minimax optimal for estimating functions in the class of signals in a total variation ball. Furthermore, we have expanded our research to include the multivariate scenario. In this context, our proposed $K$-NN-FL estimator proves to be minimax optimal for estimating functions in the class of signals in a total variation along the $K$-NN graph.

\subsection{Future Directions}


  A limitation of our work relates to  
 Theorem \ref{Cest-d>1}, which deals with estimating a multivariate function from arbitrary design points with bounded variation. A natural open question is how to extend this to higher-order versions of bounded variation. However, this remains an open problem even in the case of independent data, not to mention in the spatio-temporal setting that we study.


Exploring the properties of mixing coefficients under more relaxed decay assumptions, like polynomial decay, opens up additional research opportunities. Relevant insights on this topic are provided next. We addressed the dependence assumption by creating independent block copies from a given $\beta$-mixing sequence. This process is fully explained in Appendix \ref{BetaMsection}.


In our analysis, we consistently applied two main techniques. For concrete examples where these two techniques were used see for instance Lemma \ref{lemma1} and Lemma \ref{lemma2}, respectively. Both techniques, when dealing with a
$\beta$-mixing sequence required rigorous analysis within designated blocks that inherently retain independence. The first approach centers on the ``Special Event", where block copies match exactly with their originals. In cases of exponential decay, as discussed in Appendix \ref{BetaMsection}, this ``Special Event" occurs with high probability. Nevertheless, when adopting polynomial decay assumptions, challenges arise regarding the block size, $L$, which must not exceed $n,$ as detailed in Appendix \ref{BetaMsection}. The plausible occurrence of this event with high probability is now contingent upon the condition  $L\asymp n$. Such conditions, combined with the cost implication of $O(L)$ in the rate, see Appendix \ref{sec-proof-thm2}, culminate in achieving a rate that is suboptimal.


Identifying this challenge highlighted the need for a different approach. Numerous parts of our proofs included a second technique. Instead of only focusing on the previously mentioned ``Special Event", we strategically incorporate processes linked to the block copies pertinent to the immediate problem at hand, as exemplified in Lemma \ref{lemma1} and the proof in Appendix \ref{sec-proof-thm2}. However, this technique comes with its own set of challenges. Specifically, there is the task of quantifying the distance between the process associated with the original blocks and the process related to their copies, see Equation (\ref{helpCon-Eq2}) in Lemma \ref{lemma1}, or Equation (\ref{eqn:e500}) in Appendix \ref{sec:aux_lemmas} for concrete examples. While general bounds can be easily established, deriving optimal bounds is still a complicated task.

\acks{ The second author was partially supported by 2023-2024 Hellman Fellowship.}

\newpage
\appendix
\section{Additional Notation}
\label{additionaNot}
Throughout the proof of Theorem \ref{Cest-d>1} and Theorem \ref{Pest-d>1}, some additional notation is needed. First, in the case of the space $\mathbb{R}^d$, we will use the notation $\|\cdot\|_{2}$ for the usual euclidean $l_2$ norm, and we will write $B_\varepsilon(x)$ for the ball $B_\varepsilon\left(x,\|\cdot\|_2\right)$. 
Furthermore, for $\varepsilon>0$ small enough, we denote by $\mathcal{P}_\varepsilon$ a rectangular partition of $(0,1)^d$ induced by $0, \varepsilon, 2 \varepsilon, \ldots, \varepsilon(\lfloor 1 / \varepsilon\rfloor-1), 1$, so that all the elements of $\mathcal{P}_\epsilon$ have volume of order $\epsilon^d$.
We write
$$
\Omega_\varepsilon=[0,1]^d \backslash B_\varepsilon\left(\partial[0,1]^d\right).
$$
Thus, $\Omega_\varepsilon$ is the set of points in the interior of $[0,1]^d$ such that balls of radius $\varepsilon$ with center in such points are also contained in $[0,1]^d$.

\section{Outline of the proof of Theorem \ref{Penalized}}
\label{OutLineThm2}
The proof of Theorem \ref{Penalized} follows a path similar to that of Theorem \ref{thm:main tv}, but it includes an extra step in the first phase. This relies on the fact that the constraint-based on $V_n$ is no longer available for Theorem \ref{Penalized}. We point out the main difference in the following, where some notation from Section \ref{OutlineT1-2} is used. 
\\
{\bf{Decomposition of estimation error}}
\\
Similar to the proof of Theorem \ref{thm:main tv}, for any $f\in\Lambda$, the expression attained in this case is 
{\small{\begin{align}
\label{equation3-maintext}
\vert\vert f-f^*+\mathcal{A}_{f^*}\vert\vert_{nm}^2
&\lesssim  \frac{1}{n} \sum_{i=1}^n \frac{1}{m_i} \sum_{j=1}^{m_i} \left[\Delta_f(x_{i,j}) \delta_i(x_{i,j}) - \langle \Delta_f, \delta_i \rangle_{\mathcal{L}_2} \right]+\frac{1}{n} \sum_{i=1}^n \frac{1}{m_i} \sum_{j=1}^{m_i} \Delta_f(x_{i,j}) \epsilon_{i,j}  \nonumber\\
&\quad + \left| \frac{1}{n} \sum_{i=1}^n \frac{1}{m_i} \sum_{j=1}^{m_i} \langle \Delta_f, \delta_i \rangle_{\mathcal{L}_2} \right|+ \lambda (J_{k}(f^*)-J_{k}(f)).
\end{align}}}
Recalling the proof strategy from Theorem \ref{thm:main tv}, the next step requires the application of Lemma \ref{lemma2}. The purpose of this is to transition from the empirical norm to
the $L_2$ norm within the context of Inequality (\ref{equation3-maintext}). Given that the constraint based on $V_n$ is no longer available for bounding $J_k(\Delta_{f})$, key factor that enables the use of Lemma \ref{lemma2}, an extra step is required.
To be specific, we consider the event $$\Omega^{\text{aux}}=\Big\{\sup_{f\in \Lambda: \vert\vert\Delta_f\vert\vert_2^2\le\eta^2}\Big(J_k(f)\Big)\le 5J_{k}(f^*)\Big\}.$$ Under the event $\Omega^{\text{aux}}$, it turns out that $J_k(\Delta_{f})$ is bounded by $6J_{k}(f^*)$ and we can now make use of Lemma \ref{lemma2} as in Theorem \ref{thm:main tv}. 
Consequently, under the event $\Omega^{\text{aux}}$ the analysis performed in the proof of Theorem \ref{thm:main tv} can be replicated. It remains to analyze the complement of the event $\Omega^{\text{aux}}$. To 
this end, observe that under $(\Omega^{\text{aux}})^c$, it is satisfied that
$\sup_{f\in \Lambda: \vert\vert\Delta_{f}\vert\vert_2^2\le\eta^2,J_{k}(\Delta_f)\le 5J_{k}(f^*)}\Big(J_k(\Delta_{f})\Big)\ge 4J_{k}(f^*).$
This, together with Equation (\ref{equation3-maintext}), 
and the choice of the tuning parameter $\lambda$
leads to
\begin{align*}
    \mathbb{P}\left(\vert\vert \widehat f_{k,\lambda}-f^*\vert\vert_2^2>2\eta^2\right)
    \lesssim & \widetilde{A_1}+\widetilde{A_2}+\frac{\varepsilon}{2}.
\end{align*}
where $\widetilde{A_1}$ and $\widetilde{A_2}$, differ from $A_1$ and $A_2$ by considering $B_{T V}^k\left(6J_{k}(f^*)\right)$ instead of $B_{T V}^k\left(2V_n\right).$
\\
{\bf{Deviation bounds for measurement errors and spatial noise}}
\\
This step is exactly the same as in the proof of Theorem \ref{thm:main tv}.

\section{Additional numerical results}\label{ExtraSim}
\subsection{Variation of the noises and dependence in the Multivariate Case}\label{ExtraSimD}

In Section \ref{Mult-Sim} (the multivariate simulations), we compared the performance of MARS, CART, RF, LLE, KNN-R, Voronoigram, and $K$-NN-FL in {\bf{Scenarios 5, 6, 7}} and {\bf{8}}. Here, the spatial noise was generated using the expression $\delta_i(x)=\sum_{t=1}^{50} t^{-1} b_{t,i} h_{t}(x)$, where $\{h_t(x)\}_{t=1}^{50}$ are basis functions, and $\{b_{t,i}\}_{t = 1, i=1}^{ 50,n}$ are considered to be i.i.d.~$\mathcal{N}(0,1)$. Moreover, for the measurement errors, the $\xi_{i}$ were considered to be $\mathcal{N}(0,0.5\mathbf{1}_{m_i})$ and the designs points were generated by a mixture with mixture probability $\phi=0.1.$

This section assesses the performance of our proposed algorithm and its competitors concerning the number of elements in the basis $\{h_t(x)\},$ the choice of the distribution of $\{b_{t,i}\}$, $\{\xi_i\}$ and the value of $\phi$. Specifically, we consider
\begin{itemize}
    \item {\bf{Scenario $5_{\text{sup}}$ }}: {\bf{Scenario 5}} with $\{b_{t,i}\}_{t = 1, i=1}^{ 50,n}\sim \mathcal{N}(0,1)$, $\xi_i\sim \mathcal{N}(0,0.5\mathbf{1}_{M})$ and $\phi=0.2$,
    \item {\bf{Scenario $6_{\text{sup}}$ }}: {\bf{Scenario 6}} with $\{b_{t,i}\}_{t = 1, i=1}^{ 250,n}\sim \mathcal{N}(0,1)$, $\xi_i\sim \mathcal{N}(0,0.5\mathbf{1}_{M})$ and $\phi=0.2$,
    \item {\bf{Scenario $7_{\text{sup}}$ }}: {\bf{Scenario 7}} with $\{b_{t,i}\}_{t = 1, i=1}^{ 50,n}\sim \mathcal{N}(0,2)$, $\xi_i\sim \mathcal{N}(0,1\mathbf{1}_{M})$ and $\phi=0.05$, and 
    \item {\bf{Scenario $8_{\text{sup}}$ }}: {\bf{Scenario 8}} with $\{b_{t,i}\}_{t = 1, i=1}^{ 150,n}\sim \mathcal{N}(0,3)$, $\xi_i\sim \mathcal{N}(0,2\mathbf{1}_{M})$ and $\phi=0.05$.
\end{itemize}

The results show that with respect to the MSE performance,  MARS, CART, Random Forest, Voronoigram, and $K$-NN-FL are the best five competitors, in all four scenarios. Furthermore, in all scenarios, $K$-NN-FL is the only method that demonstrates a decreasing MSE as the sample size increases.

Overall, $K$-NN-FL demonstrated a superior MSE performance in comparison to its competitors in every scenario. However, there were specific exceptions observed in {\bf{Scenario $7_{\text{sup}}$ }}. For $m_{\text{mult}}=0.5$ in dimensions $d=3$ and $4$, and for $m_{\text{mult}}=1$ in dimensions $d=9$ and $10$. In these exceptional cases, RF method slightly outperformed $K$-NN-FL. However, this advantage was exceedingly marginal and rapidly receded as $m_{\text{mult}}$ increases.

In {\bf{Scenario $7_{\text{sup}}$ }} and {\bf{Scenario $8_{\text{sup}}$ }}, MARS and RF are consistently the closest competitors to $K$-NN-FL. This aligns with the findings for {\bf{Scenario 7 }}and {\bf{Scenario 8 }} in Section \ref{Mult-Sim}, here
illustrated in Figure \ref{S3-FN} and Figure \ref{S4-FN}. Nevertheless, the results are quite different.

We now detail the key performance distinctions of the methods in {\bf{Scenario $7_{\text{sup}}$ }} and {\bf{Scenario $8_{\text{sup}}$ }} as compared to {\bf{Scenario 7 }} and {\bf{Scenario 8 }}, respectively.

In a side-by-side comparison of {\bf{Scenario $7_{\text{sup}}$ }} and {\bf{Scenario 7}}, the former presents increased variance in both spatial noise and measurement errors. Additionally, covariates in {\bf{Scenario $7_{\text{sup}}$ }} exhibit reduced dependence. In this scenario, the efficacy of both the MARS and CART models decreases. These models tend to amplify slight variances or disturbances, mistakenly recognizing them as significant trends, which distorts the genuine data architecture. Conversely, the RF model, employing a collective strategy of numerous decision trees, maintains its effectiveness. It manages to distinguish real data configurations from simple disturbances, avoiding the issues associated with overfitting. It is important to mention that in both contexts, the $K$-NN-FL consistently maintained a strong and consistent performance.

In a comparison between {\bf{Scenario $8_{\text{sup}}$}} and {\bf{Scenario 8}}, we see that even though the covariates in {\bf{Scenario $8_{\text{sup}}$}} are less closely related, the increased variation in spatial noise and measurement errors present more significant challenges. MARS shows an admirable capability in handling these challenges within {\bf{Scenario $8_{\text{sup}}$}}, while CART and RF encounter notable difficulties. The observed reduction in performance for CART and RF in {\bf{Scenario $8_{\text{sup}}$}} compared to {\bf{Scenario 8}} highlights their vulnerability to increased spatial noise, which obscures true patterns and leads to less precise model predictions. Importantly, $K$-NN-FL kept its performance consistent and managed well in both scenarios.

The performance of the top five competitors in {\bf{Scenario $5_{\text{sup}}$ }} and {\bf{Scenario $6_{\text{sup}}$ }} is presented in Figure \ref{S1-FN} and \ref{S2-FN}, respectively. Contrasting these with the results from {\bf{Scenario 5 }} and {\bf{Scenario 5 }}, detailed in Section \ref{Mult-Sim}, we note a marked distinction. Specifically, Random Forest sees a decline in its performance in Scenario {\bf{Scenario $5_{\text{sup}}$ }}, whereas MARS experiences a similar decrease in {\bf{Scenario $6_{\text{sup}}$}}. In both cases, the performance of $K$-NN-FL remains commendable. 

A reasonable explanation for the results observed in {\bf{Scenario $5_{\text{sup}}$}} stems from the interdependency among the covariates. The RF algorithm normally uses a mix of different decision trees, but when the covariates are more dependent on each other, this mix is not as varied. Consequently, the trees grow more similar and biased towards these interdependencies, which weakens the ability of RF to apply its rules broadly. This tendency makes the model more prone to overfitting, thus reducing its effectiveness.



Moving to {\bf{Scenario $6_{\text{sup}}$}}, the stronger connections between covariates significantly challenge the ability of MARS to detect real interactions. These complex relationships among covariates might lead MARS to recognize false associations, thus weakening its natural strength. Moreover, the addition of increased spatial noise adds to the complexity of modeling. This type of noise inherently introduces periodic fluctuations to the dataset, complicating the task of differentiating true data trends from these added waves.

\begin{figure}
    \centering
    \includegraphics[width=0.85\textwidth]{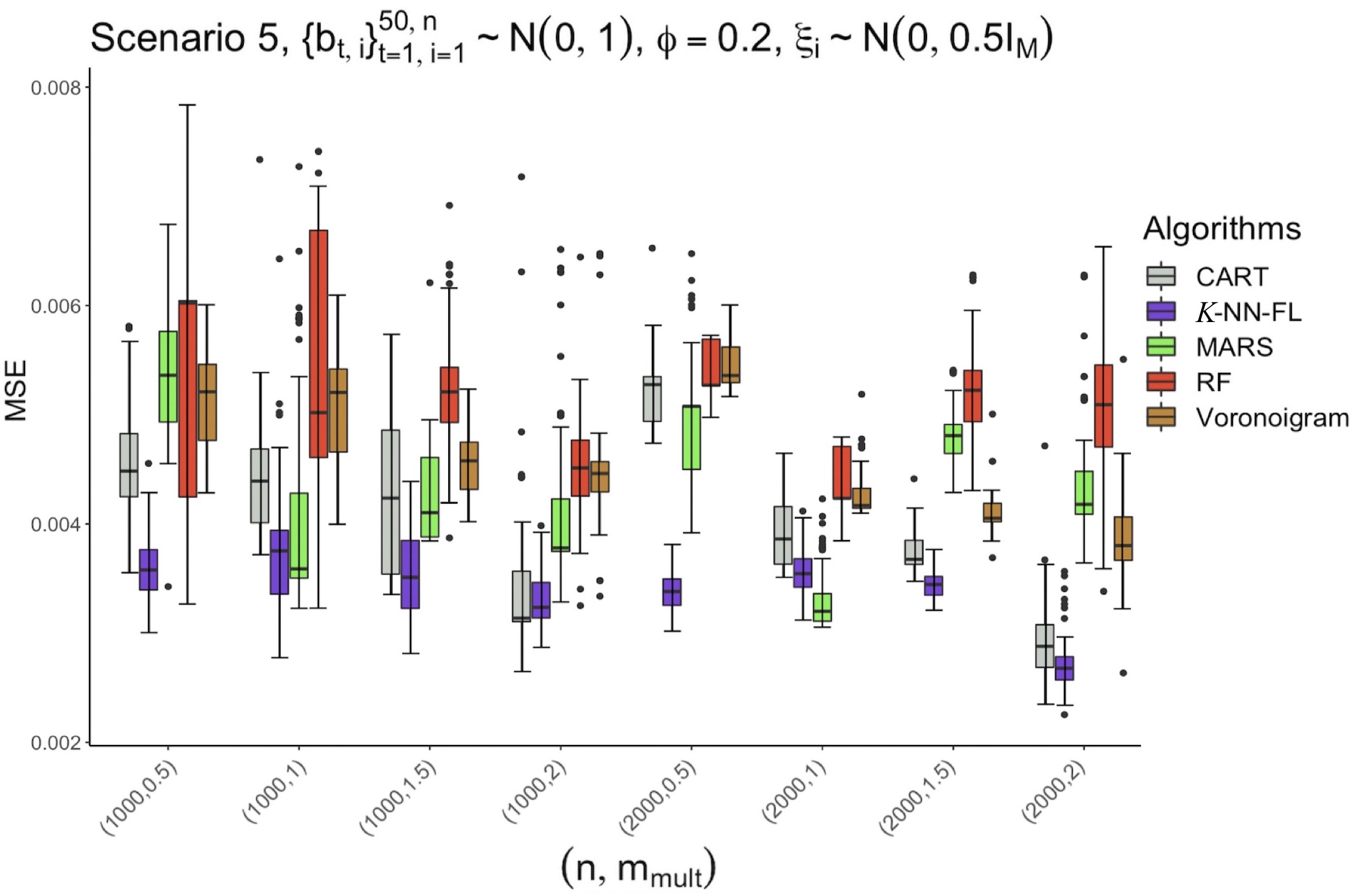}
    \caption{Best five competitors.} 
    \label{S1-FN}
\end{figure} 

\begin{figure}
    \centering
    \includegraphics[width=0.8\textwidth]{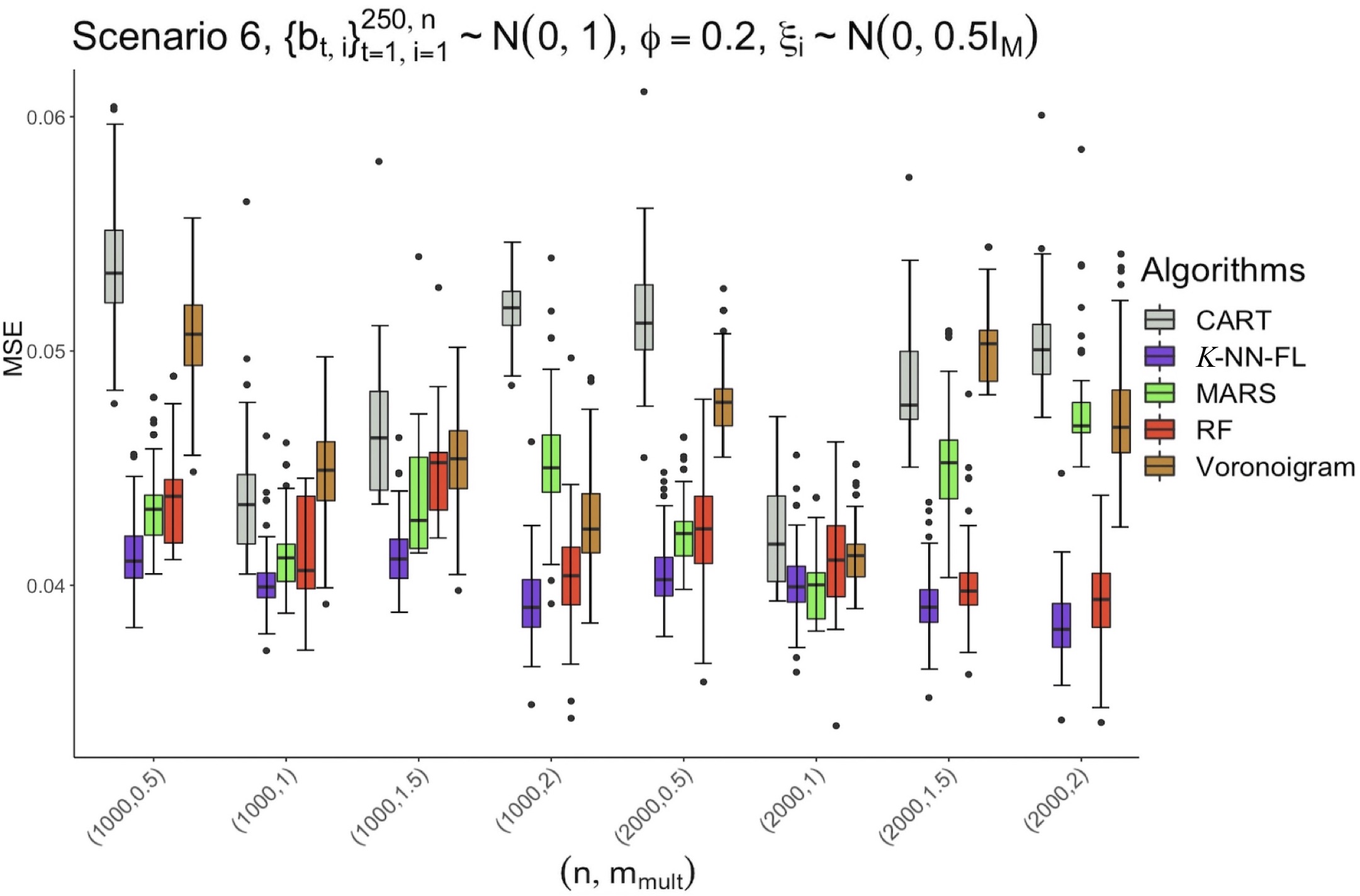}
    \caption{Best five competitors.} 
    \label{S2-FN}
\end{figure} 

\begin{figure}
    \centering
    \includegraphics[width=0.85\textwidth]{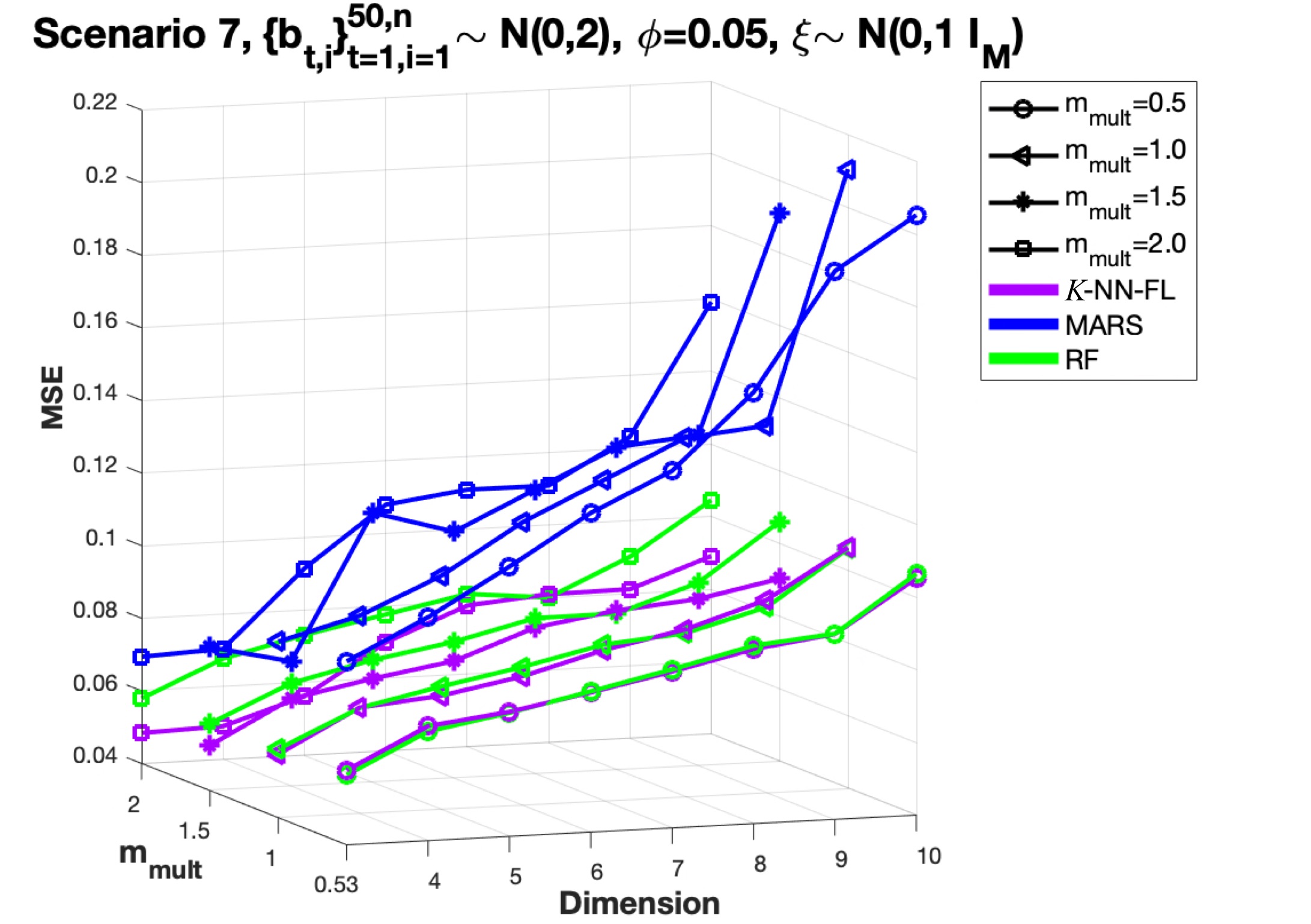}
    \caption{Best three competitors.} 
    \label{S3-FN}
\end{figure} 

\begin{figure}
    \centering
    \includegraphics[width=0.8\textwidth]{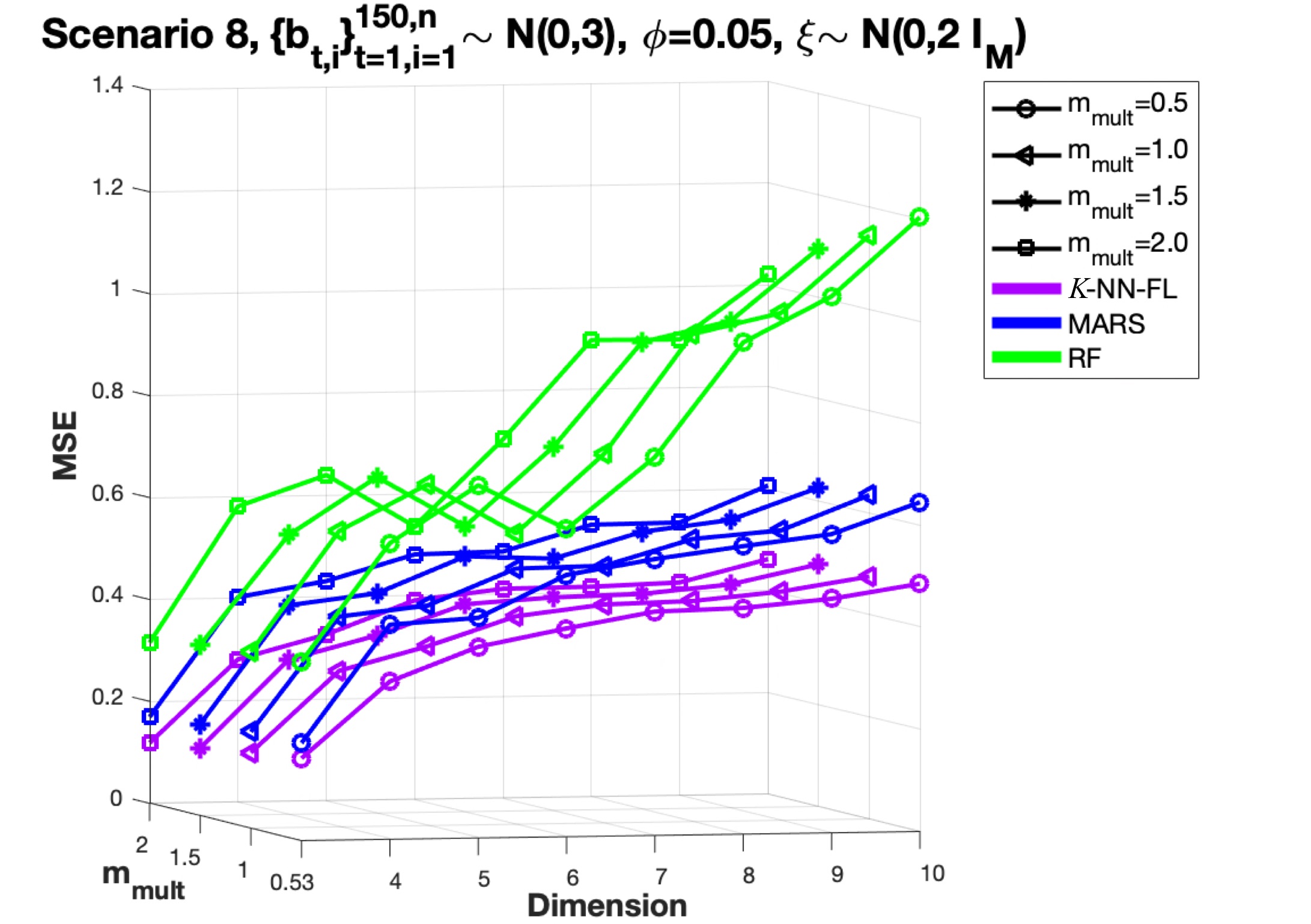}
    \caption{Best three competitors.}  
    \label{S4-FN}
\end{figure} 
\subsection{Comparison of Estimated SST functions}\label{sec-DetailedRDataE}

In this section, we delve into the comprehensive results of the Real Data example discussed in Section \ref{Real-data-ex}. We closely analyzed the estimated SST functions produced by $K$-NN-FL and MARS, noting that both methods exhibited superior performance in the real data application, as suggested by Table \ref{RdTable}. In Region $\mathbb{I}$, shown in Figure \ref{Comp-R1}, the two functions appear to be quite similar. This is consistent with the outcomes in Table \ref{RdTable}, where the performances of both 
$K$-NN-FL and MARS were very similar across the evaluated months. The contrast becomes more pronounced in Region $\mathbb{III}$. However, it is for Region $\mathbb{II}$, as showcased in Figure \ref{Comp-R2}, where the disparity in the estimated SST functions becomes most evident. This finding aligns with Table \ref{RdTable}, indicating a more regular and significant difference in the MSE of $K$-NN-FL and MARS for this particular Region.
\begin{figure}
    \centering
    \includegraphics[width=0.64\textwidth]{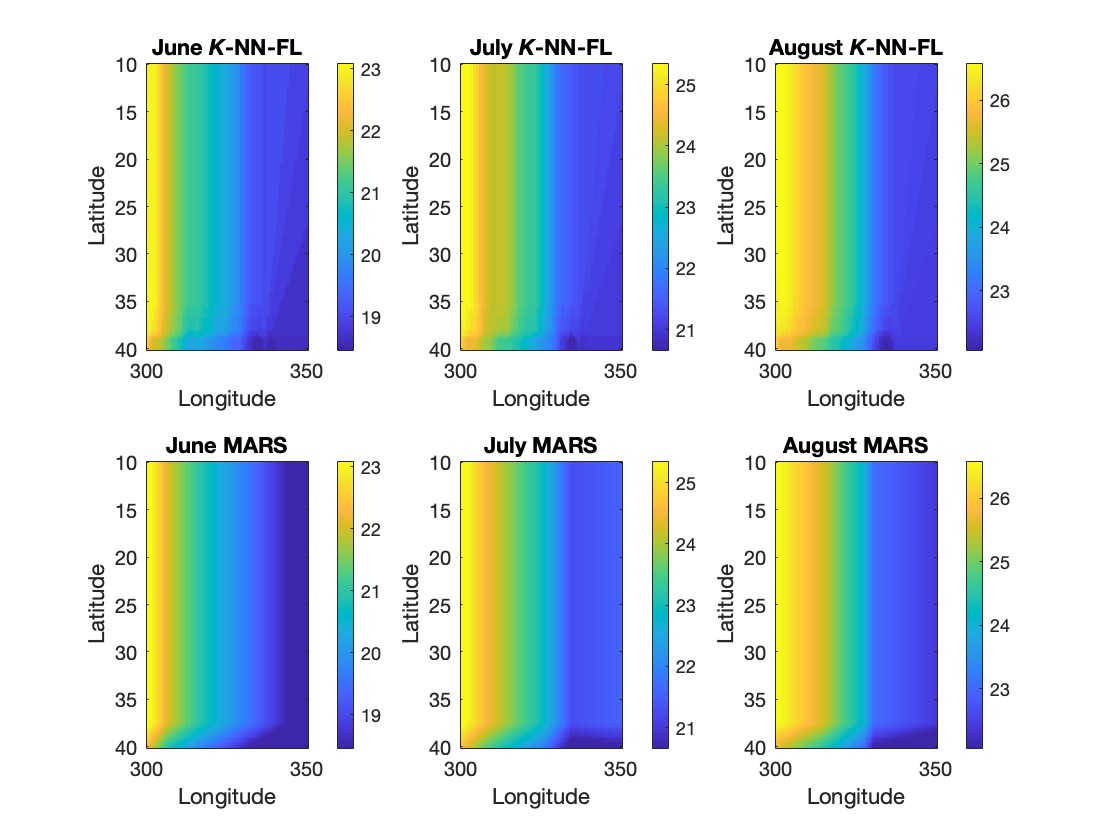}
    \caption{Region $\mathbb{I}$. The figures presented in the top row depict the estimated functions as determined by our proposed method, $K$-NN-FL. Conversely, the second row contains the functions estimated by MARS, which remains our most proximate competitor with regard to MSE. Sequentially, from left to right, each column of illustrations pertains to the successive months of June, July, and August.} 
    \label{Comp-R1}
\end{figure}

\begin{figure}
    \centering
    \includegraphics[width=0.64\textwidth]{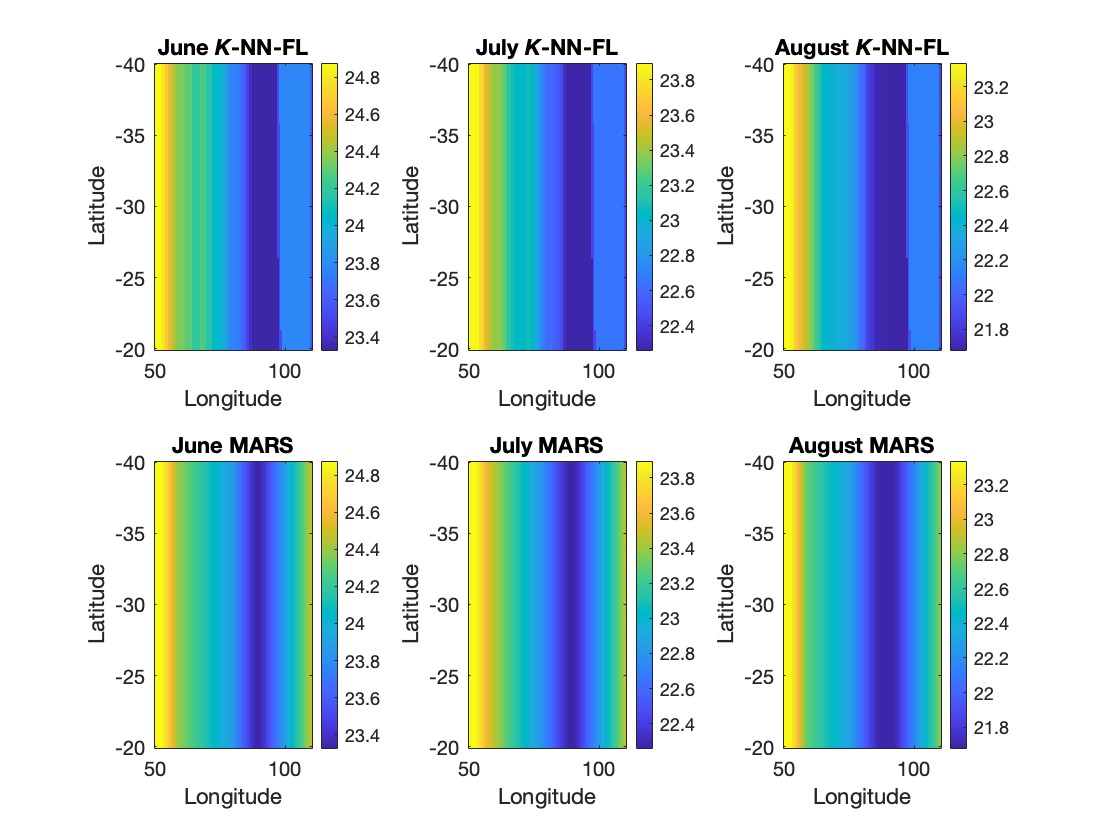}
    \caption{Region $\mathbb{III}$. The figures presented in the top row depict the estimated functions as determined by our proposed method, $K$-NN-FL. Conversely, the second row contains the functions estimated by MARS, which remains our most proximate competitor with regard to MSE. Sequentially, from left to right, each column of illustrations pertains to the successive months of June, July, and August. } 
    \label{Comp-R3}
\end{figure}

\begin{figure}
    \centering
    \includegraphics[width=.9\textwidth]{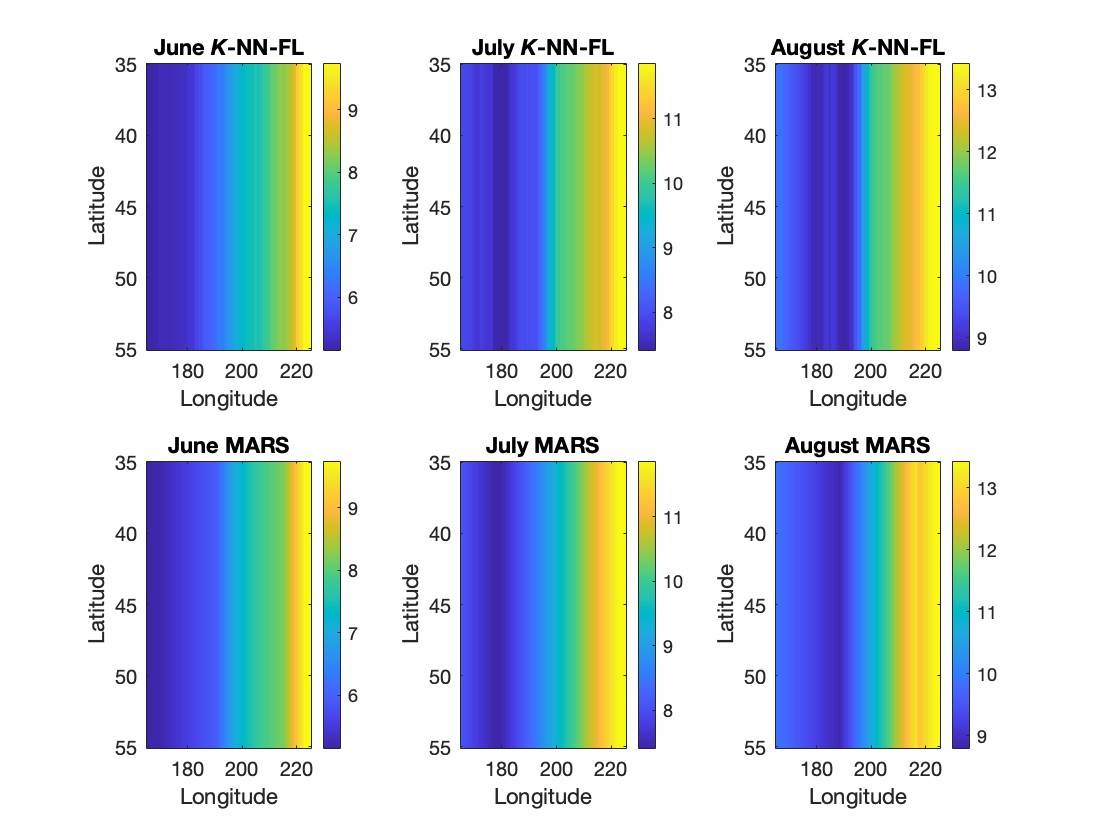}
    \caption{Region $\mathbb{II}$. The figures presented in the top row depict the estimated functions as determined by our proposed method, $K$-NN-FL. Conversely, the second row contains the functions estimated by MARS, which remains our most proximate competitor with regard to MSE. Sequentially, from left to right, each column of illustrations pertains to the successive months of June, July, and August.} 
    \label{Comp-R2}
\end{figure} 

\begin{figure}
    \centering
    \includegraphics[width=0.8\textwidth]{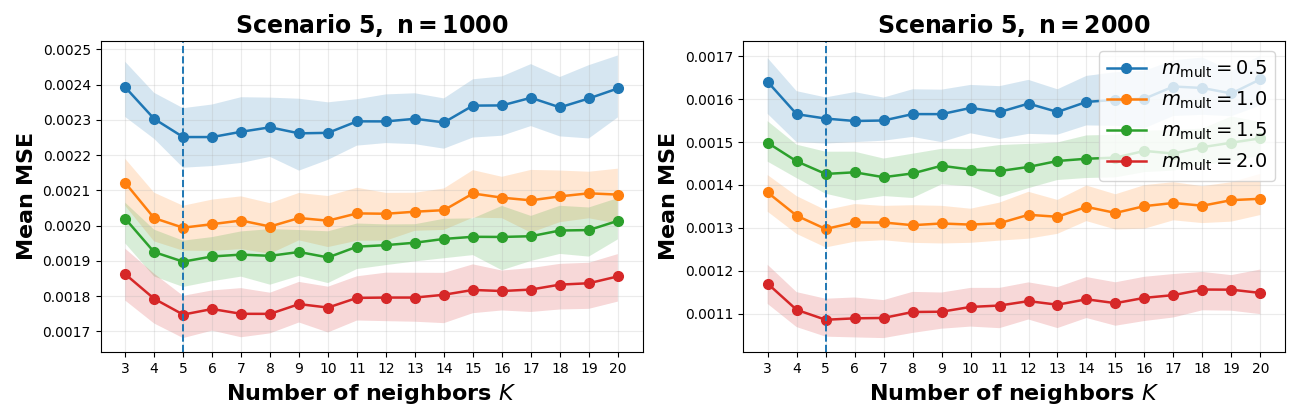}
    \caption{MSE average  over 100 Monte Carlo simulations for different values of $K$. The data has been generated as in Scenario~5 ($d=2$) with $n=1000$ (left) and $n=2000$ (right).
Shaded regions denote interquartile ranges across repetitions, and the vertical
dashed line indicates $K=5$, used in Section \ref{simu-data}.
}
    \label{fig:appendix-s5-K}
\end{figure}

{\color{black}{
\subsection{Sensitivity Analysis with Respect to the Number of Neighbors $K$}
\label{app:sensitivity-K}

In this appendix, we study the sensitivity of the proposed multidimensional method, $K$-NN-FL, to the choice of the number of nearest neighbors $K$. In the Experiments section (Section~\ref{simu-data}), simulation results for the multidimensional settings (Scenarios 5–8) are reported using a fixed value $K=5$, and the corresponding boxplots are constructed solely at this choice. To complement those results, we analyze how the mean squared error (MSE) varies as $K$ ranges from 3 to 20.

We focus on scenarios with fixed dimension $d=2$, namely \textbf{Scenario~5} and \textbf{Scenario~6}, which are representative of the low-dimensional multivariate settings considered in the simulation study. Following the simulation design described in Section~\ref{simu-data-1}, for each scenario we consider sample sizes $n \in \{1000,2000\}$. For each fixed value of $n$, we vary the multiplier parameter $m_{\mathrm{mult}} \in \{0.5,1,1.5,2\}$, which determines the corresponding values of $\{m_i\}_{i=1}^n$.

For each fixed value of the number of neighbors $K$, we report the performance of the $K$-NN-FL estimator by computing the Monte Carlo average of the test mean squared error (MSE) over $100$ independent simulation repetitions. In addition, we include shaded bands representing the interquartile range of the MSE across these repetitions, obtained from the $25$th and $75$th percentiles at each value of $K$. By examining these summaries as functions of $K$, we assess the sensitivity of the $K$-NN-FL estimator to the choice of the neighborhood size.

Figure~\ref{fig:appendix-s5-K} reports the results for Scenario~5. Across all values of
$m_{\mathrm{mult}}$ and both sample sizes, the MSE decreases when moving from very small
neighborhoods (e.g., $K=3$) to moderate values, and the curves are relatively flat around
$K=5$. In particular, the vertical dashed line at $K=5$ falls at or very near the minimum
of each curve, and the MSE changes only mildly for nearby choices of $K$. As $K$ increases
further, the error shows a gradual upward trend, consistent with mild oversmoothing, but
the overall magnitude of the change remains limited. Importantly, the ordering across
different values of $m_{\mathrm{mult}}$ is stable as $K$ varies.

Figure~\ref{fig:appendix-s6-K} shows a similar sensitivity pattern for Scenario~6. The method
exhibits stable performance over a broad range of $K$, with no abrupt deterioration as $K$
increases. The choice $K=5$ either attains the minimum MSE or lies very close to it across
all configurations, placing it firmly within a low-error region of the curves. As in
Scenario~5, very small neighborhood sizes such as $K=3$ tend to yield slightly larger
errors than $K=4$ or $5$, reflecting increased variability due to overly local
neighborhoods. For larger values of $K$, the MSE shows a mild upward trend, consistent with
moderate oversmoothing.

Overall, the changes across $K$ are modest in both scenarios, indicating that $K$-NN-FL is
not overly sensitive to the precise choice of the neighborhood size in these low-dimensional
settings. This supports $K=5$ as a reasonable default choice for the multivariate experiments
reported in the main text in Section \ref{simu-data}.

}}

\begin{figure}
    \centering
    \includegraphics[width=0.8\textwidth]{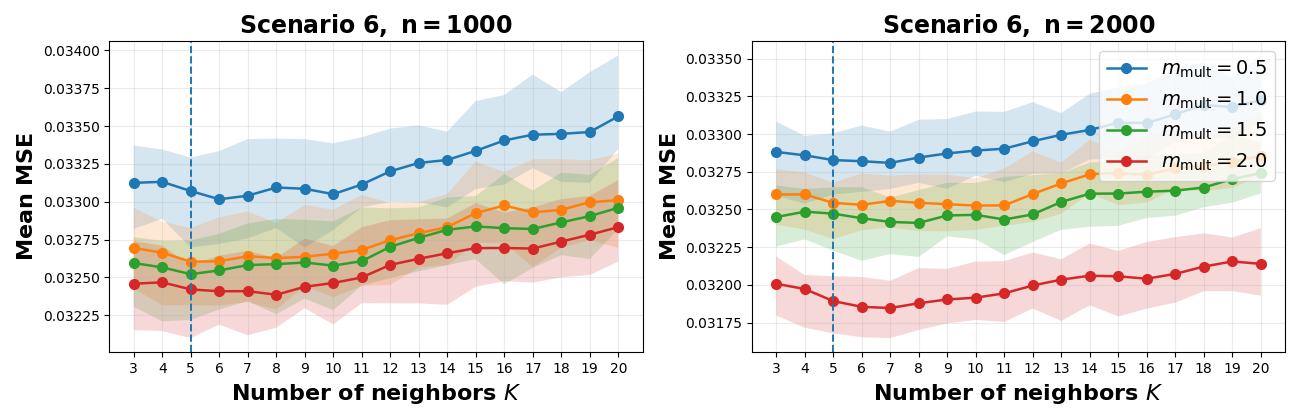}
    \caption{MSE average  over 100 Monte Carlo simulations for different values of $K$. The data has been generated as in Scenario~6 ($d=2$) with $n=1000$ (left) and $n=2000$ (right).
Shaded regions denote interquartile ranges across repetitions, and the vertical
dashed line indicates $K=5$, used in Section \ref{simu-data}}
    \label{fig:appendix-s6-K}
\end{figure}

\newpage

\section{Auxiliaries for proof of Theorem \ref{thm:main tv}}
\label{sec:aux_lemmas}

Throughout we assume that Assumption \ref{assume:tv functions} holds. We also invoke Assumption 1{\bf{f}}. Specifically, we rely on the condition that $cm\le m_i\le Cm$ for some positive constants $c,C.$ Additionally, we define \begin{equation}
     \label{Lchoicethm1}
     L=\frac{1}{C_{\beta}}2\log(n),
 \end{equation} leading us to the conclusion that 
 \begin{equation}
 \label{betamixcoef-ineq}
     \beta_x(L)\le\frac{1}{n^2}.
 \end{equation} This bound is equivalently satisfied by  $\beta_\delta(L)$ and $\beta_\epsilon(L).$ Subsequently, throughout our discussion, we make frequent references to the copies of the random variables $x_{i,j}$, $\epsilon_{i,j}$ and $\delta_{i}$ as they are constructed in Appendix \ref{BetaMsection}. These copies are denoted as $x_{i,j}^*$, $\epsilon_{i,j}^*$ and $\delta_{i}^*$, respectively. Furthermore, as delineated in Appendix \ref{BetaMsection}, during the formation of these copies the quantities $ | \mathcal J_{e,l}|$ and $ | \mathcal J_{o,l}|$, which represent the count of even and odd blocks respectively, satisfy $ | \mathcal J_{e,l}|,| \mathcal J_{o,l}|  \asymp n/L$. Thus, there exist positive constants $c_1$ and $c_2$ such that  
 \begin{equation}
     \label{BoundBlocksize}
      c_1n/L\le| \mathcal J_{e,l}|,| \mathcal J_{o,l}| \le  c_2n/L.
 \end{equation}

 Also, using a union bound argument and using Lemmas \ref{Ind-cop}, \ref{Ind-cop-measerrors}, and \ref{Ind-cop-func}, we note that the following sets, $\Omega_x=\cap_{i=1}^n\{x_{i,j}=x_{i,j}^*\ \forall\ j\in[m_i]\}$, $\Omega_\epsilon=\cap_{i=1}^n\{\epsilon_{i,j}=\epsilon_{i,j}^*\ \forall\ j\in[m_i]\}$ and $\Omega_\delta=\cap_{i=1}^n\{\delta_{i}=\delta_{i}^*\}$, satisfy the conditions $$\mathbb{P}(\Omega_x^c)\le n\beta_{x}(L)\le\frac{1}{n}, \ \mathbb{P}(\Omega_\epsilon^c)\le n\beta_{\epsilon}(L)\le\frac{1}{n},\  \text{and}\  \mathbb{P}(\Omega_\delta^c)\le n\beta_{\delta}(L)\le\frac{1}{n}.$$

\begin{lemma}
    \label{lem:aux1}
    The event 
$ \mathcal{E}_1\,:=\,  \left\{  \frac{1}{n}  \sum_{i=1}^n  \frac{1}{m_i}\sum_{j=1}^{m_i} \epsilon_{i,j}^2  \le 6c_2 \varpi _\epsilon^2  \right\},  $  
happens  with probability at least $$1-  8\frac{C\log^2(n)}{c^2c_2\varpi_\epsilon^2nm}\frac{1}{C_{\beta}^2}-\frac{\log(n)}{2\varpi_{\epsilon}^2 c_2n^2}-\frac{1}{ 2n^{1/\max{\{256\varpi_\epsilon^4,16\varpi_\epsilon^2\}}}\max{\{256\varpi_\epsilon^4,16\varpi_\epsilon^2\}}\varpi_{\epsilon}^2 c_2 },$$
 with $C_{\mathcal{B}} $ and $\varpi_{\epsilon}$ as in Assumption \ref{assume:tv functions}.
\end{lemma}

\begin{proof}
Consider $\epsilon_{i,j}^*$, the copies of $\epsilon_{i,j}$, as in Section \ref{BetaMsection}. Observe that for $c_2>0$, 
 \begin{align*}
    &\mathbb{P}\left(    \frac{1}{n}\sum_{i=1}^{n}\frac{1}{m_i} \sum_{j=1}^{m_i}    \epsilon_{i,j}^2   \geq    6\varpi_{\epsilon}^2  c_2  \right)
    \\
    \le
    &\mathbb{P}\left(    \frac{1}{n}\sum_{i=1}^{n}\frac{1}{m_i} \sum_{j=1}^{m_i}    (\epsilon_{i,j}^2-  (\epsilon_{i,j}^{*})^2) \geq    2\varpi_{\epsilon}^2 c_2  \right)
    +\mathbb{P}\left(   \frac{1}{n}\sum_{l=1}^{L} \sum_{i \in \mathcal{J}_{e,l}}\frac{1}{m_i}\sum_{j=1}^{m_i}(\epsilon_{i,j}^{*})^2\geq    2\varpi_{\epsilon}^2 c_2  \right)
    \\
     +&
    \mathbb{P}\left(    \frac{1}{n}\sum_{l=1}^{L} \sum_{i \in \mathcal{J}_{o,l}}\frac{1}{m_i}\sum_{j=1}^{m_i}(\epsilon_{i,j}^{*})^2 \geq    2\varpi_{\epsilon}^2c_2   \right),
\end{align*}
and we proceed to upper bound each of these terms. 
By Markov's inequality
\begin{equation}
    \label{eqn:e500}
    \begin{array}{lll}
\displaystyle     \mathbb{P}\left(    \frac{1}{n}\sum_{i=1}^{n}\frac{1}{m_i} \sum_{j=1}^{m_i}    (\epsilon_{i,j}^2-  (\epsilon_{i,j}^{*})^2) \geq    2\varpi_{\epsilon}^2 c_2  \right)&\le& \displaystyle \frac{\mathbb{E}\left(    \frac{1}{n}\sum_{i=1}^{n}\frac{1}{m_i} \sum_{j=1}^{m_i}    (\epsilon_{i,j}^2-  (\epsilon_{i,j}^{*})^2)  \right)}{   2\varpi_{\epsilon}^2 c_2 } \\
&  = &\displaystyle \frac{ \frac{1}{n}\sum_{i=1}^{n}\frac{1}{m_i} \sum_{j=1}^{m_i}  \mathbb{E}\left(   \mathbf{1}_{ \{ \epsilon_{i,j} \not=\epsilon_{i,j}^* \} }  (\epsilon_{i,j}^2-  (\epsilon_{i,j}^{*})^2)  \right)}{   2\varpi_{\epsilon}^2 c_2 }.\\
    \end{array}
\end{equation}
Notice that
\begin{align*}
    &\mathbb{E}\left(   \mathbf{1}_{ \{ \epsilon_{i,j} \not=\epsilon_{i,j}^* \} }  (\epsilon_{i,j}^2-  (\epsilon_{i,j}^{*})^2)  \right)
    \\
    \le& \int_{0}^{\log(n)}P(\{ \epsilon_{i,j} \not=\epsilon_{i,j}^* \})dt+\int_{\log(n)}^{\infty}P(\epsilon_{i,j}^2-  (\epsilon_{i,j}^{*})^2>t)dt, 
\end{align*}
and by Assumption 1 we have that $\epsilon_{i,j}^2$ and $(\epsilon_{i,j}^*)^2$ are both sub-exponential with parameter $(32\varpi_\epsilon^4,4\varpi_\epsilon^2)$, see for instance Appendix B in \cite{honorio2014tight}. Moreover using Cauchy Schwartz inequality we have that $\epsilon_{i,j}^2-(\epsilon_{i,j}^*)^2$ is sub-exponential of parameter $(128\varpi_\epsilon^4,8\varpi_\epsilon^2).$ Therefore, using Tail bound inequality for sub-exponential random variables, see Proposition 2.9 in \cite{wainwright2019high}, it follows that
\begin{align*}
    &\int_{0}^{\log(n)}P(\{ \epsilon_{i,j} \not=\epsilon_{i,j}^* \})dt+\int_{\log(n)}^{\infty}P(\epsilon_{i,j}^2-  (\epsilon_{i,j}^{*})^2>t)dt
    \\
    \le&
    \log(n)\beta_\epsilon(L)+\frac{1}{\max{\{256\varpi_\epsilon^4,16\varpi_\epsilon^2\}}}\exp(-\log(n)/\max{\{256\varpi_\epsilon^4,16\varpi_\epsilon^2\}}).
\end{align*}
Thus
\begin{equation}
    \label{eqn:501}
    \begin{array}{l}
        \arraycolsep=1.4pt\def\arraystretch{2}
 \displaystyle  \mathbb{P}\left(    \frac{1}{n}\sum_{i=1}^{n}\frac{1}{m_i} \sum_{j=1}^{m_i}    (\epsilon_{i,j}^2-  (\epsilon_{i,j}^{*})^2) \geq    2\varpi_{\epsilon}^2 c_2  \right) \\
 \leq  \displaystyle  \frac{\log(n)\beta_\epsilon(L)+\frac{1}{\max{\{256\varpi_\epsilon^4,16\varpi_\epsilon^2\}}}\exp(-\log(n)/\max{\{256\varpi_\epsilon^4,16\varpi_\epsilon^2\}})}{ 2\varpi_{\epsilon}^2 c_2 } \\
   \leq \displaystyle \,\,\,\,\,  \frac{\log(n)}{2\varpi_{\epsilon}^2 c_2n^2}+\frac{1}{ 2n^{1/\max{\{256\varpi_\epsilon^4,16\varpi_\epsilon^2\}}}\max{\{256\varpi_\epsilon^4,16\varpi_\epsilon^2\}}\varpi_{\epsilon}^2 c_2 }.\\          
    \end{array}
\end{equation}
Next, by union bound,
\begin{align*}
    &\mathbb{P}\left(   \frac{1}{n}\sum_{l=1}^{L} \sum_{i \in \mathcal{J}_{e,l}}\frac{1}{m_i}\sum_{j=1}^{m_i}(\epsilon_{i,j}^{*})^2\geq    2\varpi_{\epsilon}^2c_2   \right)
    \le \sum_{l=1}^{L} \mathbb{P}\left(   \frac{1}{n}\sum_{i \in \mathcal{J}_{e,l}}\frac{1}{m_i}\sum_{j=1}^{m_i}(\epsilon_{i,j}^{*})^2\geq    2\varpi_{\epsilon}^2c_2/L   \right).
\end{align*}
Given that $ | \mathcal J_{e,l}| \asymp n/L$ (see Inequality (\ref{BoundBlocksize})) and Assumption \ref{assume:tv functions}{\bf{f}} we further have that
\begin{align*}
    \mathbb{E}\Big(\frac{1}{n}\sum_{i \in \mathcal{J}_{e,l}}\frac{1}{m_i}\sum_{j=1}^{m_i}(\epsilon_{i,j}^{*})^2\Big)\le \frac{|\mathcal{J}_{e,l}|\varpi_\epsilon^2}{n}\leq \frac{\varpi_\epsilon^2c_2}{L},
\end{align*}
and
\[
\text{var}\Big(\frac{1}{n}\sum_{i \in \mathcal{J}_{e,l}}\frac{1}{m_i}\sum_{j=1}^{m_i}(\epsilon_{i,j}^{*})^2\Big) = \sum_{i \in \mathcal{J}_{e,l}}\sum_{j=1}^{m_i}\frac{1}{(nm_i)^2} \text{var} ((\epsilon_{i,j}^{*})^2) \leq   \frac{|\mathcal{J}_{e,l}|Cm\varpi_\epsilon^2}{n^2c^2m^2}\le \frac{c_2C\varpi_\epsilon^2}{nc^2mL}.
\]
From  Chebyshev's inequality it follows that
\begin{equation}
    \label{eqn:502}
   \begin{array}{l}
      \displaystyle \mathbb{P}\left(   \frac{1}{n}\sum_{i \in \mathcal{J}_{e,l}}\frac{1}{m_i}\sum_{j=1}^{m_i}(\epsilon_{i,j}^{*})^2\geq    2\varpi_{\epsilon}^2c_2/L   \right)  \\
        \displaystyle    \leq  \mathbb{P}\Big(  \Big| \frac{1}{n}\sum_{i \in \mathcal{J}_{e,l}}\frac{1}{m_i}\sum_{j=1}^{m_i}(\epsilon_{i,j}^{*})^2 -\mathbb{E}\Big(\frac{1}{n}\sum_{i \in \mathcal{J}_{e,l}}\frac{1}{m_i}\sum_{j=1}^{m_i}(\epsilon_{i,j}^{*})^2\Big)\Big|\geq    \varpi_{\epsilon}^2c_2/L   \Big)\\
     \displaystyle     \leq  \frac{c_2C\varpi_\epsilon^2}{nc^2mL}\frac{L^2}{\varpi_{\epsilon}^4c_2^2 }=\frac{CL}{c^2c_2\varpi_\epsilon^2nm}=\frac{2c_2C\log(n)}{c^2c_2^2\varpi_\epsilon^2nm}\frac{1}{C_{\beta}}.
   \end{array}
\end{equation}
The last equality is followed by the choice of $L$ in Equation (\ref{Lchoicethm1}).
Similarly we have that 
\begin{equation}
    \label{eqn:503}
       \mathbb{P}\left(   \frac{1}{n}\sum_{i \in \mathcal{J}_{o,l}}\frac{1}{m_i}\sum_{j=1}^{m_i}(\epsilon_{i,j}^{*})^2\geq    2\varpi_{\epsilon}^2c_2/L   \right)\le \frac{2C\log(n)}{c^2c_2\varpi_\epsilon^2nm}\frac{1}{C_{\beta}}.
\end{equation}
The claim follows combining (\ref{eqn:501}), (\ref{eqn:502}) and (\ref{eqn:503}).

\end{proof}

\begin{lemma}
    \label{lem:aux2}
    There exists a constant $\widetilde{C}_1>0$ such that  the event 
\[
   \mathcal{E}_2  \,:=\,\left\{ \frac{1}{n}  \sum_{i=1}^n  \frac{1}{m_i}\sum_{j=1}^{m_i} \delta_i ^2  (x_{i,j})   \le   6c_2 \widetilde C_1   \varpi_\delta^2\right\},
\]
 satisfies
$ \p( \mathcal{E}_2 )\geq  1- 8\frac{\log^2(n)}{c_2n}\frac{1}{C_{\beta}^2}-\frac{\log(n)}{2\varpi_{\delta}^2 c_2\widetilde C_1n^2}-\frac{1}{ 2n^{1/\max{\{256\varpi_\delta^4,16\varpi_\delta^2\}}}\max{\{256\varpi_\delta^4,16\varpi_\delta^2\}}\varpi_{\delta}^2 c_2 \widetilde C_1} $,  with $\varpi_{\delta}$ as in Assumption \ref{assume:tv functions}.
\end{lemma}

\begin{proof} By Assumption 1, for a positive constant $\widetilde C_1>0$, we have that 
\begin{equation}
    \label{MomentBoundDelta}
     \mathbb{E}  \big \{  \delta_i ^2  (x )  \big \} \le \widetilde C_1\varpi_\delta^2 \quad \text{and} \quad \mathbb{E}  \big \{  \delta_i ^4  (x )  \big \} \le \widetilde C_1^2 \varpi_\delta^4 ,
\end{equation}
for any $x\in [0,1]$.
Consider from Lemma \ref{Ind-cop-func} the corresponding $\{\delta_i^*\}_{i=1}^n$ copies of $\{\delta_i\}_{i=1}^n.$ 
Then
 \begin{align}
    &\displaystyle 	 \mathbb{P}\left(    \frac{1}{n}\sum_{i=1}^{n}\frac{1}{m_i} \sum_{j=1}^{m_i}    \delta_i ^2  (x_{i,j})    \geq    6\varpi_{\delta}^2  c_2\widetilde C_1 \right)\nonumber
    \\
    \le&
    \mathbb{P}\left(    \frac{1}{n}\sum_{i=1}^{n}\frac{1}{m_i} \sum_{j=1}^{m_i}    (\delta_i ^2  (x_{i,j}) -  \delta_i^{*} (x_{i,j} )^2) \geq    2\varpi_{\delta}^2 c_2\widetilde C_1  \right)\nonumber
    \\
    +&\mathbb{P}\left(   \frac{1}{n}\sum_{l=1}^{L} \sum_{i \in \mathcal{J}_{e,l}}\frac{1}{m_i}\sum_{j=1}^{m_i}\delta_i^{*} (x_{i,j} )^2\geq    2\varpi_{\delta}^2 c_2\widetilde C_1 \right)\nonumber
    \\
     +&
    \mathbb{P}\left(    \frac{1}{n}\sum_{l=1}^{L} \sum_{i \in \mathcal{J}_{o,l}}\frac{1}{m_i}\sum_{j=1}^{m_i}\delta_i^{*} (x_{i,j} )^2 \geq    2\varpi_{\delta}^2c_2\widetilde C_1   \right).
    \label{eqn:e506}
    \end{align}
Next  we analyze each term. By Markov's inequality,
\begin{align*}
   & \mathbb{P}\left(    \frac{1}{n}\sum_{i=1}^{n}\frac{1}{m_i} \sum_{j=1}^{m_i}    (\delta_i ^2  (x_{i,j}) -  \delta_i^{*} (x_{i,j} )^2) \geq    2\varpi_{\delta}^2 c_2\widetilde C_1  \right) 
   \\
    \le& \frac{ \frac{1}{n}\sum_{i=1}^{n}\frac{1}{m_i} \sum_{j=1}^{m_i}  \mathbb{E}\left(   \mathbf{1}_{ \{ \delta_{i} \not=\delta_{i}^* \} }  (\delta_i ^2  (x_{i,j}) -  \delta_i^{*} (x_{i,j} )^2)  \right)}{   2\varpi_{\delta}^2 c_2 \widetilde C_1}.
\end{align*}
Conditioning on $x_{i,j}$, we have that
\begin{align*}
    \mathbb{E}\left(   \mathbf{1}_{ \{ \delta_{i} \not=\delta_{i}^*\} }  (\delta_i ^2  (x_{i,j}) -  \delta_i^{*} (x_{i,j} )^2)  |x_{i,j}\right)\le& \int_{0}^{\log(n)}P(\{ \delta_{i} \not=\delta_{i}^* \})dt
    \\
    +&\int_{\log(n)}^{\infty}P(\delta_i ^2  (x_{i,j}) -  \delta_i^{*} (x_{i,j} )^2>t|x_{i,j} )dt.
\end{align*}
By Assumption 1, both $\delta_i^2(x_{i,j})$ and $\delta_i^*(x_{i,j})^2$ are both sub-exponential with a parameter $(32\varpi_\delta^4,4\varpi_\delta^2)$, see for instance Appendix B in \cite{honorio2014tight}. Then using Cauchy Schwartz inequality, we have that $\delta_i^2(x_{i,j})-\delta_{i}^{*}(x_{i,j})^2$ is sub-exponential of parameter $(128\varpi_\delta^4,8\varpi_\delta^2).$ Therefore, using Tail bound inequality for sub-exponential random variables, see Proposition 2.9 in \cite{wainwright2019high},
\begin{align*}
    &\int_{0}^{\log(n)}P(\{ \delta_{i} \not=\delta_{i}^* \})dt+\int_{\log(n)}^{\infty}P(\delta_i ^2  (x_{i,j}) -  \delta_i^{*} (x_{i,j} )^2>t|x_{i,j})dt
    \\
    \le&
    \log(n)\beta_\delta(L)+\frac{1}{\max{\{256\varpi_\delta^4,16\varpi_\delta^2\}}}\exp(-\log(n)/\max{\{256\varpi_\delta^4,16\varpi_\delta^2\}}).
\end{align*}
Thus, using Assumption 1{\bf{c}}
\begin{align*}
    &\mathbb{P}\left(    \frac{1}{n}\sum_{i=1}^{n}\frac{1}{m_i} \sum_{j=1}^{m_i}    (\delta_i ^2  (x_{i,j}) -  \delta_i^{*} (x_{i,j} )^2) \geq    2\varpi_{\delta}^2 c_2  \right)
    \\
    \le& \frac{\log(n)\beta_\delta(L)+\frac{1}{\max{\{256\varpi_\delta^4,16\varpi_\delta^2\}}}\exp(-\log(n)/\max{\{256\varpi_\delta^4,16\varpi_\delta^2\}})}{ 2\varpi_{\delta}^2 c_2 \widetilde C_1}
    \\
    \le&\frac{\log(n)}{2\varpi_{\delta}^2 c_2\widetilde C_1n^2}+\frac{1}{ 2n^{1/\max{\{256\varpi_\delta^4,16\varpi_\delta^2\}}}\max{\{256\varpi_\delta^4,16\varpi_\delta^2\}}\varpi_{\delta}^2 c_2 \widetilde C_1}.
    \label{eq:test}
\end{align*}
The last inequality is followed by Inequality (\ref{betamixcoef-ineq}).
To analyze the second term  in the right hand side of (\ref{eqn:e506}), we observe that
\begin{align*}
    &\mathbb{P}\left(   \frac{1}{n}\sum_{l=1}^{L} \sum_{i \in \mathcal{J}_{e,l}}\frac{1}{m_i}\sum_{j=1}^{m_i}\delta_i^{*} (x_{i,j} )^2\geq    2\varpi_{\delta}^2c_2 \widetilde C_1  \right)
    \\
    \le& \sum_{l=1}^{L} \mathbb{P}\left(   \frac{1}{n}\sum_{i \in \mathcal{J}_{e,l}}\frac{1}{m_i}\sum_{j=1}^{m_i}\delta_i^{*} (x_{i,j} )^2\geq    2\varpi_{\delta}^2c_2\widetilde C_1/L   \right).
\end{align*}
Moreover by Inequality (\ref{BoundBlocksize}) and Assumption 1{\bf{f}}
\begin{align*}
    \mathbb{E}\Big(\frac{1}{n}\sum_{i \in \mathcal{J}_{e,l}}\frac{1}{m_i}\sum_{j=1}^{m_i}\delta_i^{*} (x_{i,j} )^2\Big)\le \frac{|\mathcal{J}_{e,l}|\widetilde C_1\varpi_\delta^2}{n}=\frac{\widetilde C_1\varpi_\delta^2c_2}{L},
\end{align*}
and, Cauchy Schwarz inequality together with Inequality (\ref{BoundBlocksize}) and Assumption 1{\bf{f}} imply that
\begin{align*}
    &\text{var}\Big(\frac{1}{n}\sum_{i \in \mathcal{J}_{e,l}}\frac{1}{m_i}\sum_{j=1}^{m_i}\delta_i^{*} (x_{i,j} )^2\Big)=\sum_{i \in \mathcal{J}_{e,l}}\frac{1}{(nm_i)^2}\text{var}\Big(\sum_{j=1}^{m_i} \delta_i^{*} (x_{i,j} )^2\Big)
    \\
     \le&\sum_{i \in \mathcal{J}_{e,l}}\frac{1}{(nm_i)^2}\sum_{j, j' =1}^{m_i} \mathbb{E} \big\{ \delta_i^*(x_{i,j} )^2 \delta_i^*(x_{i,j'} )^2 \big\} 
     \\
    \le &
    \sum_{i \in \mathcal{J}_{e,l}}\frac{1}{(nm_i)^2}\sum_{j, j' =1}^{m_i} 
    \sqrt { \mathbb{E} \big\{ \delta_i^*(x_{i,j} )^4  \big \} \mathbb{E} \big\{  \delta_i^*(x_{i,j'} )^4  \big\}} 
\\
    \le&
    \frac{|\mathcal{J}_{e,l}|\widetilde C_1^2\varpi_\delta^4}{n^2}\le \frac{c_2\widetilde C_1^2\varpi_\delta^4}{nL}.
\end{align*}
Using Chebyshev's inequality we get that
\begin{align*}
    &\mathbb{P}\left(   \frac{1}{n}\sum_{i \in \mathcal{J}_{e,l}}\frac{1}{m_i}\sum_{j=1}^{m_i}\delta_i^{*} (x_{i,j} )^2\geq    2\varpi_{\delta}^2c_2\widetilde C_1/L   \right)
    \\
    &\le\mathbb{P}\left(   \frac{1}{n}\sum_{i \in \mathcal{J}_{e,l}}\frac{1}{m_i}\sum_{j=1}^{m_i}\delta_i^{*} (x_{i,j} )^2 -\mathbb{E}\Big(\frac{1}{n}\sum_{i \in \mathcal{J}_{e,l}}\frac{1}{m_i}\sum_{j=1}^{m_i}\delta_i^{*} (x_{i,j} )^2\Big)\geq    \varpi_{\delta}^2c_2\widetilde C_1/L   \right)
    \\
    &\le\frac{c_2\widetilde C_1^2\varpi_\delta^4}{nL}\frac{L^2}{\varpi_{\delta}^4c_2^2\widetilde C_1^2 }=\frac{L}{c_2n}=\frac{\log(n)}{c_2n}\frac{2}{C_{\beta}}.
\end{align*}
The last equality is followed by the choice of $L$ in Equation (\ref{Lchoicethm1}).
Similarly we have that 
\begin{align*}
    \mathbb{P}\left(   \frac{1}{n}\sum_{i \in \mathcal{J}_{o,l}}\frac{1}{m_i}\sum_{j=1}^{m_i}\delta_i^{*} (x_{i,j} )^2\geq    2\varpi_{\delta}^2c_2 \widetilde C_1/L   \right)\le \frac{\log(n)}{c_2n}\frac{2}{C_{\beta}},
\end{align*}  
and the claim is obtained.
\end{proof}

\begin{lemma}
    \label{lem:aux3}
With the notation from Lemma \ref{lem:aux2},    for $\eta > 0$ it holds that 
    \[
     \mathbb{P}( \vert \bar{y}   -   \mathcal{A}_{f^*}  \vert >  \eta  )    \,\leq\,  \frac{128 \log n }{  c C_{\beta} \eta^2  n }  c_2 C  \widetilde C_1(\varpi_\delta^2  +\varpi_\epsilon^2 )\,+\, \frac{3}{n},
    \]
    where $\bar{y} \,:=\, \frac{1}{ \sum_{i=1}^n m_i }    \sum_{i=1}^n \sum_{j=1}^{m_i} y_{i,j}$ and $\mathcal{A}_{f^*} =\frac{1}{\sum_{i=1}^n m_i}\sum_{i=1}^n\sum_{j=1}^{m_i}f^*(x_{i,j})$.
\end{lemma}

\begin{proof}
Using the notation from the beginning of Appendix \ref{sec:aux_lemmas}, 
\[   
\begin{array}{lll}
 \mathbb{P}( \vert \bar{y}   -   \mathcal{A}_{f^*}  \vert > \eta  )  & \leq  & \displaystyle \mathbb{P}\left(  \bigg\vert \frac{1}{  \sum_{i=1}^n m_i  }  \sum_{i=1}^n \sum_{j=1}^{m_i}  (  \epsilon_{i,j}^* + \delta_{i}^*(x_{i,j}^* )   )    \bigg\vert > \eta    \right)  \,+\, \frac{3}{n}.  
\end{array}
\]
Hence, by union bound,
\begin{equation}
    \label{eqn:511}
    \begin{array}{lll}
 \mathbb{P}( \vert \bar{y}   -   \mathcal{A}_{f^*}  \vert > \eta  )  & \leq  & \displaystyle \mathbb{P}\bigg(  \bigg\vert \sum_{l=1}^{L}\sum_{i\in\mathcal{J}_{e,l}}\sum_{j=1}^{m_i} \epsilon_{i,j}^*        \bigg\vert > \frac{\eta}{4}  \sum_{i=1}^n m_i   \bigg) + \\
  && \displaystyle \mathbb{P}\bigg(  \bigg\vert \sum_{l=1}^{L}\sum_{i\in\mathcal{J}_{o,l}}\sum_{j=1}^{m_i} \epsilon_{i,j}^*        \bigg\vert > \frac{\eta}{4}  \sum_{i=1}^n m_i   \bigg) + \\
    && \displaystyle   \mathbb{P}\bigg(  \bigg\vert \sum_{l=1}^{L}\sum_{i\in\mathcal{J}_{e,l}}\sum_{j=1}^{m_i} \delta_i^*(x_{i,j})       \bigg\vert > \frac{\eta}{4}  \sum_{i=1}^n m_i   \bigg) +\\
  && \displaystyle  \mathbb{P}\bigg(  \bigg\vert \sum_{l=1}^{L}\sum_{i\in\mathcal{J}_{o,l}}\sum_{j=1}^{m_i} \delta_i^*(x_{i,j})         \bigg\vert > \frac{\eta}{4}  \sum_{i=1}^n m_i   \bigg)  + \frac{3}{n}\\
 &\leq& \displaystyle \sum_{l=1}^L \mathbb{P}\bigg(  \bigg\vert \sum_{i\in\mathcal{J}_{e,l}}\sum_{j=1}^{m_i} \epsilon_{i,j}^*        \bigg\vert > \frac{\eta}{4L}  \sum_{i=1}^n m_i   \bigg) + \\
  & &\displaystyle  \sum_{l=1}^L\mathbb{P}\bigg(  \bigg\vert \sum_{i\in\mathcal{J}_{o,l}}\sum_{j=1}^{m_i} \epsilon_{i,j}^*        \bigg\vert > \frac{\eta}{4L}  \sum_{i=1}^n m_i   \bigg) +\\
    & &\displaystyle \sum_{l=1}^L\mathbb{P}\bigg(  \bigg\vert \sum_{i\in\mathcal{J}_{e,l}}\sum_{j=1}^{m_i} \delta_i^*(x_{i,j})       \bigg\vert > \frac{\eta}{4L}  \sum_{i=1}^n m_i   \bigg) +\\
     && \displaystyle \sum_{l=1}^{L} \mathbb{P}\bigg(  \bigg\vert \sum_{i\in\mathcal{J}_{o,l}}\sum_{j=1}^{m_i} \delta_i^*(x_{i,j})         \bigg\vert > \frac{\eta}{4L}  \sum_{i=1}^n m_i   \bigg)  + \frac{3}{n}\\
\end{array}
\end{equation}
Next, Then we observe that by Assumption 1 we have that $$\mathbb{E}\Big( \sum_{i\in \mathcal{J}_{e,l}} \sum_{j=1}^{m_i}\delta_i^*(x_{i,j}^*)\Big)=0.$$ 
Therefore, by Chebyshev's inequality
\[
\begin{array}{l}
\displaystyle  \mathbb{P}\Big(\Big| \sum_{i\in \mathcal{J}_{e,l}} \sum_{j=1}^{m_i}\delta_i^*(x_{i,j}^*)\Big|> \frac{\eta}{4L}  \sum_{i=1}^n m_i   \Big)\\
=   \displaystyle  \mathbb{P}\Big(\Big| \sum_{i\in \mathcal{J}_{e,l}} \sum_{j=1}^{m_i}\delta_i^*(x_{i,j}^*)-\mathbb{E}\Big(\sum_{i\in \mathcal{J}_{e,l}} \sum_{j=1}^{m_i}\delta_i^*(x_{i,j}^*)\Big)\Big|> \frac{\eta}{4L}  \sum_{i=1}^n m_i   \Big) \\
\displaystyle \leq \frac{16 L^2}{\eta^2 (\sum_{i=1}^n  m_i )^2   }   \text{var}\Big(  \sum_{i\in \mathcal{J}_{e,l}}\sum_{j=1}^{m_i}\delta_i^*(x_{i,j}^*)\Big).   
\end{array}
\]
However, due to the independence of the variables $\{\delta_i^*\}_{i\in\mathcal{J}_{e,l}},$
\[
\begin{array}{lll}
 \displaystyle     \text{var}\Big(  \sum_{i\in \mathcal{J}_{e,l}}\sum_{j=1}^{m_i}\delta_i^*(x_{i,j}^*)\Big)   &   = &    \displaystyle   \sum_{i\in \mathcal{J}_{e,l}}  \text{var}\Big(  \sum_{j=1}^{m_i}\delta_i^*(x_{i,j}^*)\Big)     \\
     &\leq &     \displaystyle   \sum_{i\in \mathcal{J}_{e,l}}     \sum_{j, j' =1}^{m_i} \mathbb{E} \big\{ \delta_i^*(x_{i,j}^* ) \delta_i^*(x_{i,j'}^* ) \big\} \\
          &\leq &     \displaystyle   \sum_{i\in \mathcal{J}_{e,l}}     \sum_{j, j' =1}^{m_i}   \sqrt { \mathbb{E} \big\{ \delta_i^*(x_{i,j}^* )^2  \big \} \mathbb{E} \big\{  \delta_i^*(x_{i,j'}^* )^2  \big\}} \\
           &\leq& \displaystyle C |\mathcal{J}_{e,l}|  m^2 \widetilde C_1\varpi_\delta^2
\end{array}
\]
where the second inequality follows from Cauchy-Schwarz inequality, and the last inequality follows from (\ref{MomentBoundDelta}) and $C>0$ as in the beginning of Appendix \ref{sec:aux_lemmas}.

Therefore, 
we obtain that 
\begin{align}
\label{eqn:512}
  \mathbb{P}\Big(\Big| \sum_{i\in \mathcal{J}_{e,l}} \sum_{j=1}^{m_i}\delta_i^*(x_{i,j}^*)\Big|> \frac{\eta}{4L}  \sum_{i=1}^n m_i   \Big)  \leq \frac{16 L^2}{\eta^2 (\sum_{i=1}^n  m_i )^2   }  C |\mathcal{J}_{e,l}|  m^2 \widetilde C_1\varpi_\delta^2
\end{align}
Similarly we get that
\begin{equation}
\label{eqn:513}
    \mathbb{P}\Big(\Big| \sum_{i\in \mathcal{J}_{o,l}} \sum_{j=1}^{m_i}\delta_i^*(x_{i,j}^*)\Big|> \frac{\eta}{4L}  \sum_{i=1}^n m_i   \Big)  \leq \frac{16 L^2}{\eta^2 (\sum_{i=1}^n  m_i )^2   }  C |\mathcal{J}_{e,l}|  m^2 \widetilde C_1\varpi_\delta^2.
\end{equation}
Proceeding in a similar way with the events involving $\{\epsilon_{i,j}^*\}$ and using (\ref{BoundBlocksize}) along with (\ref{Lchoicethm1}) we arrive to the claim. 


\end{proof}

\begin{lemma}
    \label{lem:aux4}
    Let $\left\{\xi_{i,j}\right\}_{\left\{ i\in\mathcal{J}_{e,l},j\in [m_i]\right\}}$ independent and identically distributed Rademacher random variables, independent of  $\left(\left\{x_{i,j}^*\right\}_{\left\{ i\in\mathcal{J}_{e,l},j\in [m_i]\right\}},\left\{\delta_i^*\right\}_{\left\{ i\in\mathcal{J}_{e,l},j\in [m_i]\right\}}\right)$. Let
    \[
    Z_{i,j}^*(\Delta):=\left\{\Delta\left(x_{i,j}^*\right)\right\} \delta_i^*\left(x_{i,j}^*\right)-\left\langle \Delta, \delta_i^*\right\rangle_{2}.
    \]
    Then, for any function class $\mathcal{F}$ and  for any $\eta^{\prime} > 0$ it holds that  
    \[
     \begin{array}{l}
    \displaystyle        \mathbb{E}\Big( \sup_{\Delta\in \mathcal{F},\, \vert\vert \Delta \vert\vert_2\le \eta^{\prime}}\frac{1}{n}\sum_{i\in\mathcal{J}_{e,l}}\frac{1}{m_i}\sum_{j=1}^{m_i}Z_{i,j}^*(\Delta)\Big)\\
     \displaystyle    \,\leq\, \frac{2Cc_2}{cL} \mathbb{E}\left(\mathbb{E}\left(\delta^*_{\mathcal{J}_{e,l}}\mathbb{E}\left(\sup_{\Delta \in \mathcal{F},\, \vert\vert \Delta \vert\vert_2\le \eta^{\prime}} \frac{1}{\sum_{i\in\mathcal{J}_{e,l}}m_i}  \sum_{i\in\mathcal{J}_{e,l}} \sum_{j=1}^{m_i} \xi_{i,j}\Delta\left(x_{i,j}^*\right)\mid\left\{\delta_i^*\right\},\left\{x_{i,j}^*\right\}\right) \mid\left\{\{\delta_i^*\}\right\}\right)\right), 
     \end{array}
    \]
    where $\delta^*_{\mathcal{J}_{e,l}} =\max_{i\in\mathcal{J}_{e,l},j\in[ m_i]}\vert\delta_{i}^*(x_{i,j}^*)\vert $.
\end{lemma}

\begin{proof}

\begin{equation}
    \label{eqn:535}
   \begin{array}{l}
     \displaystyle      \mathbb{E}\Big( \sup_{\Delta\in \mathcal{F},\, \vert\vert \Delta \vert\vert_2\le \eta^{\prime}}\frac{1}{n}\sum_{i\in\mathcal{J}_{e,l}}\frac{1}{m_i}\sum_{j=1}^{m_i}Z_{i,j}^*(\Delta)\Big) \\
   \displaystyle   \,=\, \mathbb{E}\left(\mathbb{E}\Big( \sup_{\Delta \in \mathcal{F},\, \vert\vert \Delta \vert\vert_2\le \eta^{\prime} }\frac{1}{n}\sum_{i\in\mathcal{J}_{e,l}}\frac{1}{m_i}\sum_{j=1}^{m_i}Z_{i,j}^*(\Delta)\mid\left\{\{\delta_i^*\}_{i\in\mathcal{J}_{e,l}}\right\}\Big) \right) \\
   \displaystyle \leq  2\mathbb{E}\left(\mathbb{E}\left( \sup_{\Delta \in \mathcal{F},\, \vert\vert \Delta \vert\vert_2\le \eta^{\prime} } \frac{1}{n} \sum_{i\in\mathcal{J}_{e,l}} \frac{1}{m_i}\sum_{j=1}^{m_i} \xi_{i,j} \delta_i^*\left(x_{i,j}^*\right)\Delta\left(x_{i,j}^*\right) \mid\left\{\{\delta_i^*\}_{i\in\mathcal{J}_{e,l}}\right\}\right)\right)\\
   \end{array}
\end{equation}
where the second equation follows from the law of iterated expectations, and the inequality by Lemma \ref{lemma6}.

Now, we observe that we can write the inner most term in the previous equation as 
{\small{\begin{align*}
    &\mathbb{E}\left(\sup_{ \Delta \in \mathcal{F},\, \vert\vert \Delta \vert\vert_2\le \eta^{\prime} } \frac{1}{m_i}\sum_{j=1}^{m_i} \xi_{i,j} \delta_i^*\left(x_{i,j}^*\right)\Delta\left(x_{i,j}^*\right)\mid\left\{\{\delta_i^*\}_{i\in\mathcal{J}_{e,l}}\right\}\right)
    \\
    &=\mathbb{E}\left(\sup_{\Delta \in \mathcal{F},\, \vert\vert \Delta \vert\vert_2\le \eta^{\prime} } \frac{1}{\sum_{i\in\mathcal{J}_{e,l}}m_i} \sum_{i\in\mathcal{J}_{e,l}} \sum_{j=1}^{m_i} \xi_{i,j} \frac{\sum_{i\in\mathcal{J}_{e,l}}m_i}{nm_i}\delta_i^*\left(x_{i,j}^*\right)\Delta\left(x_{i,j}^*\right) \mid\left\{\{\delta_i^*\}_{i\in\mathcal{J}_{e,l}}\right\}\right).
\end{align*}}}
Next, using Assumption 1{\bf{f}} and Inequality (\ref{BoundBlocksize}) we have that $\frac{\sum_{i\in\mathcal{J}_{e,l}}m_i}{nm_i}\le \frac{Cc_2}{cL}$. Let 
 $$A_{i,j}^{\Delta} =\delta_{i}^*(x_{i,j}^*)\Delta(x_{i,j}^*),\ \delta^*_{\mathcal{J}_{e,l}}=\max_{i\in\mathcal{J}_{e,l},j\in[ m_i]}\vert\delta_{i}^*(x_{i,j}^*)\vert .$$

  By Ledoux–Talagrand contraction inequality (see Proposition 5.28 in \cite{wainwright2019high}),
 {\footnotesize{
\begin{align*}
&\mathbb{E}\left(\sup_{ \Delta\in \mathcal{F},\, \vert\vert \Delta \vert\vert_2\le \eta^{\prime} } \frac{1}{\sum_{i\in\mathcal{J}_{e,l}}m_i} \sum_{i\in\mathcal{J}_{e,l}} \sum_{j=1}^{m_i} \xi_{i,j} \frac{\sum_{i\in\mathcal{J}_{e,l}}m_i}{nm_i}A_{i,j}^{\Delta} \mid\left\{\{\delta_i^*\}_{i\in\mathcal{J}_{e,l}}\right\}\right)
    \\
    =&\mathbb{E}\left(\mathbb{E}\left(\sup_{\Delta\in \mathcal{F},\, \vert\vert \Delta \vert\vert_2\le \eta^{\prime} } \frac{1}{\sum_{i\in\mathcal{J}_{e,l}}m_i} \sum_{i\in\mathcal{J}_{e,l}} \sum_{j=1}^{m_i} \xi_{i,j} \frac{\sum_{i\in\mathcal{J}_{e,l}}m_i}{nm_i}A_{i,j}^{\Delta} \mid\left\{\delta_i^*\right\},\left\{x_{i,j}^*\right\}\right) \mid\left\{\{\delta_i^*\}\right\}\right) \\
    \leq& \mathbb{E}\left(\frac{Cc_2}{cL}\delta^*_{\mathcal{J}_{e,l}}\mathbb{E}\left(\sup_{\Delta \in \mathcal{F},\, \vert\vert \Delta \vert\vert_2\le \eta^{\prime} } \frac{1}{\sum_{i\in\mathcal{J}_{e,l}}m_i}  \sum_{i\in\mathcal{J}_{e,l}} \sum_{j=1}^{m_i} \xi_{i,j}\Delta\left(x_{i,j}^*\right)\mid\left\{\delta_i^*\right\},\left\{x_{i,j}^*\right\}\right) \mid\left\{\{\delta_i^*\}\right\}\right)
\end{align*}
}}
and so the claim follows.

\end{proof}

\begin{lemma}
\label{lem:aux5}
Define
\begin{equation*}
Y_{i,j}^*(\Delta):=\left\{\Delta\left(x_{i,j}^*\right)\right\} \epsilon_{i,j}^*.
\end{equation*}
for any real valued function $\Delta$. Then there exists $\left\{\xi_{i,j}\right\}_{\left\{ i\in\mathcal{J}_{e,l},j\in [m_i]\right\}}$ independent Rademacher variables independent of $\left(\{x_{i,j}^*\}_{\{ i\in\mathcal{J}_{e,l},j\in [m_i]\}},\left\{\delta_i^*\right\}_{\left\{ i\in\mathcal{J}_{e,l},j\in [m_i]\right\}}\right),$  such that for any class of functions $\mathcal{F}$ with $0\in \mathcal{F}$, it holds that 
\[
  \begin{array}{l}
         \displaystyle       \mathbb{P}\left(\sup_{\Delta \in \mathcal{F} } \frac{1}{n} \sum_{i=1}^n \frac{1}{m_i}\sum_{j=1}^{m_i}Y_{i,j}^*(\Delta)>\eta^2\right) \leq \\
            \leq \displaystyle\sum_{l=1}^{L}  \frac{2 Cc_2}{c\eta^2}\mathbb{E}\left(  \mathbb{E}\left(\epsilon^*_{\mathcal{J}_{e,l}} \mathbb{E}\left(\sup_{\Delta \in \mathcal{F}} \frac{1}{\sum_{i\in\mathcal{J}_{e,l}}m_i}  \sum_{i\in\mathcal{J}_{e,l}} \sum_{j=1}^{m_i} \xi_{i,j}\Delta\left(x_{i,j}^*\right)\mid\left\{\epsilon_{i,j}^*\right\},\left\{x_{i,j}^*\right\}\right) \mid\{\epsilon_{i,j}^*\}\right)   \right)     \,+\,\\
 \,\,\,\,\,\,\,\, \displaystyle \sum_{l=1}^{L} \frac{2 Cc_2}{c\eta^2} \mathbb{E}\left(\mathbb{E}\left(\epsilon^*_{\mathcal{J}_{o,l}} \mathbb{E}\left(\sup_{\Delta \in \mathcal{F}} \frac{1}{\sum_{i\in\mathcal{J}_{0,l}}m_i}  \sum_{i\in\mathcal{J}_{0,l}} \sum_{j=1}^{m_i} \xi_{i,j}\Delta\left(x_{i,j}^*\right)\mid\left\{\epsilon_{i,j}^*\right\},\left\{x_{i,j}^*\right\}\right) \mid\{\epsilon_{i,j}^*\}\right)\right) \\
  \end{array}
\]
for any $\eta >0$, and where $\epsilon^*_{\mathcal{J}_{e,l}}=\max_{i\in\mathcal{J}_{e,l},j\in[ m_i]}\vert\epsilon_{i,j}^*\vert$ and $\epsilon^*_{\mathcal{J}_{0,l}}=\max_{i\in\mathcal{J}_{0,l},j\in[ m_i]}\vert\epsilon_{i,j}^*\vert$.

\end{lemma}

\begin{proof}
Notice that by union bound 
\[
\begin{array}{l}
 \displaystyle   \mathbb{P}\left(\sup_{\Delta \in \mathcal{F} } \frac{1}{n} \sum_{i=1}^n \frac{1}{m_i}\sum_{j=1}^{m_i}Y_{i,j}^*(\Delta)>\eta^2\right) \\
   \leq \displaystyle\sum_{l=1}^{L}\mathbb{P}\left(\sup_{\Delta \in \mathcal{F}} \frac{1}{n} \sum_{i\in\mathcal{J}_{e,l}} \frac{1}{m_i}\sum_{j=1}^{m_i}Y_{i,j}^*(\Delta)>\eta^2 / (2L)\right)\,+\,\\
 \,\,\,\,\,\,\,\, \displaystyle \sum_{l=1}^{L}\mathbb{P}\left(\sup_{\Delta\in \mathcal{F}}\frac{1}{n} \sum_{i\in\mathcal{J}_{o,l}} \frac{1}{m_i}\sum_{j=1}^{m_i}Y_{i,j}^*(\Delta)>\eta^2 / (2L)\right)\\
\end{array}
\]
By Markov's inequality,
\begin{align}
\label{A1bound3}
    &\mathbb{P}\left( \sup_{\Delta  \mathcal{F} }\frac{1}{n} \sum_{i\in\mathcal{J}_{e,l}} \frac{1}{m_i}\sum_{j=1}^{m_i}Y_{i,j}^*(\Delta)>\eta^2 / (2L)\right) 
    \\
    &\leq \frac{2L}{\eta^2}\mathbb{E}\left( \sup_{\Delta \in \mathcal{F}} \frac{1}{n} \sum_{i\in\mathcal{J}_{e,l}} \frac{1}{m_i}\sum_{j=1}^{m_i}Y_{i,j}^*(\Delta)\right)\nonumber,
\end{align}
and a similar inequality is obtained for the sum over $\mathcal{J}_{o,l}.$

Next,
\begin{align}
\label{Fourth-bound-1-1}
 & \frac{2L}{\eta^2}\mathbb{E}\Big( \sup_{\Delta \in \mathcal{F}}\frac{1}{n}\sum_{i\in\mathcal{J}_{e,l}}\frac{1}{m_i}\sum_{j=1}^{m_i}Y_{i,j}^*(\Delta)\Big) \nonumber\\ 
& \leq \frac{2L}{\eta^2} \mathbb{E}\left(\mathbb{E}\left( \sup_{\Delta \in \mathcal{F}} \frac{1}{n} \sum_{i\in\mathcal{J}_{e,l}} \frac{1}{m_i}\sum_{j=1}^{m_i} \xi_{i,j} \epsilon_{i,j}^*\Delta\left(x_{i,j}^*\right) \right)\right)\nonumber\\ 
 & =   \frac{2L}{\eta^2} \mathbb{E}\left(\mathbb{E}\left(\sup_{\Delta \in \mathcal{F}}\frac{1}{n} \sum_{i\in\mathcal{J}_{e,l}} \frac{1}{m_i}\sum_{j=1}^{m_i} \xi_{i,j} \epsilon_{i,j}^*\Delta\left(x_{i,j}^*\right)\mid\{\epsilon_{i,j}^*\}\right)\right)\nonumber\\
  & =   \frac{2L}{\eta^2}  \mathbb{E}\left( \mathbb{E}\left(\sup_{\Delta \in \mathcal{F}} \frac{1}{\sum_{i\in\mathcal{J}_{e,l}}m_i} \sum_{i\in\mathcal{J}_{e,l}} \sum_{j=1}^{m_i} \xi_{i,j} \frac{\sum_{i\in\mathcal{J}_{e,l}}m_i}{nm_i}\epsilon_{i,j}^*\Delta\left(x_{i,j}^*\right) \mid\{\epsilon_{i,j}^*\}\right)\right),\\
\end{align}
where the inequality follows by Lemma \ref{lemma6}.  

 Next, using Assumption 1{\bf{f}} and Inequality (\ref{BoundBlocksize}) we have that $\frac{\sum_{i\in\mathcal{J}_{e,l}}m_i}{nm_i}\le \frac{Cc_2}{cL}$. Let
 $$B_{i,j}^{\Delta} =\epsilon_{i,j}^*\Delta(x_{i,j}^*),\ \epsilon^*_{\mathcal{J}_{e,l}}=\max_{i\in\mathcal{J}_{e,l},j\in[ m_i]}\vert\epsilon_{i,j}^*\vert.$$
 Thus, by Ledoux–Talagrand contraction inequality (see Proposition 5.28 in \cite{wainwright2019high}),
 {\footnotesize	{
\begin{align*}
&\mathbb{E}\left(\sup_{\Delta \in \mathcal{F}} \frac{1}{\sum_{i\in\mathcal{J}_{e,l}}m_i} \sum_{i\in\mathcal{J}_{e,l}} \sum_{j=1}^{m_i} \xi_{i,j} \frac{\sum_{i\in\mathcal{J}_{e,l}}m_i}{nm_i}B_{i,j}^{\Delta} \mid\{\epsilon_{i,j}^*\}\right)
    \\
    =&\mathbb{E}\left(\mathbb{E}\left(\sup_{\Delta \in \mathcal{F}} \frac{1}{\sum_{i\in\mathcal{J}_{e,l}}m_i} \sum_{i\in\mathcal{J}_{e,l}} \sum_{j=1}^{m_i} \xi_{i,j} \frac{\sum_{i\in\mathcal{J}_{e,l}}m_i}{nm_i}B_{i,j}^{\Delta} \mid\left\{\epsilon_{i,j}^*\right\},\left\{x_{i,j}^*\right\}\right) \mid\{\epsilon_{i,j}^*\}\right) \\
    \leq& \mathbb{E}\left(\frac{Cc_2}{cL}\epsilon^*_{\mathcal{J}_{e,l}} \mathbb{E}\left(\sup_{\Delta \in \mathcal{F}} \frac{1}{\sum_{i\in\mathcal{J}_{e,l}}m_i}  \sum_{i\in\mathcal{J}_{e,l}} \sum_{j=1}^{m_i} \xi_{i,j}\Delta\left(x_{i,j}^*\right)\mid\left\{\epsilon_{i,j}^*\right\},\left\{x_{i,j}^*\right\}\right) \mid\{\epsilon_{i,j}^*\}\right),
\end{align*}
}}
the claim follows combining all of the above.

\end{proof}

\section{Proof of Theorem \ref{thm:main tv}}\label{sec-proof-thm2}

\begin{proof}
Let $\varepsilon >0$. We can assume without loss of generality that $\nu$ is the uniform distribution on the interval $[0,1]$. Let us define $\widetilde{f^*}=f^*- \mathcal{A}_{f^*} $ with $\mathcal{A}_{f^*} =\frac{1}{\sum_{i=1}^n m_i}\sum_{i=1}^n\sum_{j=1}^{m_i}f^*(x_{i,j})$,  and note that $J_{k}(\widetilde{f^*})\le V_n$. Let $\Lambda := \{ t\widehat g_{k,V_n} + (1-t)\widetilde{f^*}: t \in [0,1] \}$ be the set of functions defined as a convex combination of $\widehat g_{k,V_n}$ and $\widetilde{f^*}$. Given the minimization property of $\widehat g_{k,V_n}$ and the convexity of the $B_{TV}^{k-1}(V_n)$ ball, for any function $f \in \Lambda$, it follows that
\begin{align*}
\displaystyle	\frac{1}{ n } \sum_{i=1}^n  \frac{1}{m_i}\sum_{j=1}^{m_i} \big (  y_{i,j} -  \bar{y} -  f(x_{i,j})  \big)^2  
	\le   \frac{1}{ n }  \sum_{i=1}^n  \frac{1}{m_i}\sum_{j=1}^{m_i} \big (  y_{i,j} -  \bar{y} -  \widetilde{f^*}  (x_{i,j})  \big)^2 .
\end{align*}
This implies that for all $f \in \Lambda$ it holds that 
\begin{equation}
    \label{eqn:main_eq}
     \begin{array}{lll}
     \displaystyle   \frac{1}{2}\|  f^* -  \bar{y} - f\|_{nm}^2 & \leq&  \displaystyle\frac{1}{n }\sum_{i=1}^n   \frac{1}{m_i} \sum_{j=1}^{m_i} \Big(f(x_{i,j})-f^*(x_{i,j})+\mathcal{A}_{f^*}\Big) \Big(\delta_i(x_{i,j}) +\epsilon_{
i,j}\Big) \\ 
 & = &  T_1(f)+T_2(f)+T_3(f),
 \end{array}
\end{equation}
where 
{\small{\begin{align}
&T_1(f):=\frac{1}{n } \sum_{i=1}^n \frac{1}{m_i}\sum_{j=1}^{m_i}\left\{f\left(x_{i,j}\right)-f^*\left(x_{i,j}\right) +\mathcal{A}_{f^*}  \right\} \epsilon_{i,j}, \label{T1}\\
&T_2(f):=\frac{1}{n} \sum_{i=1}^n \frac{1}{m_i}\sum_{j=1}^{m_i}\left[\left\{f\left(x_{i,j}\right)-f^*\left(x_{i,j}\right)+\mathcal{A}_{f^*}\right\} \delta_i\left(x_{i,j}\right)-\left\langle f-f^*+\mathcal{A}_{f^*}, \delta_i\right\rangle_{\mathcal{L}_2}\right],\label{T2}
\\
&T_3(f):=\frac{1}{n} \sum_{i=1}^n \frac{1}{m_i}\sum_{j=1}^{m_i}\left\langle f-f^*+\mathcal{A}_{f^*}, \delta_i\right\rangle_{\mathcal{L}_2},
\label{T3}  
\end{align}}}

Inequality (\ref{eqn:main_eq}) is at the heart of our proof and it will be used to obtain different localization properties of our estimator.  
\\
{\bf Step 1.} Next, let $\Delta_f = f-  f^*   + \mathcal{A}_{f^*}$  for any function $f$. Also, with the notation of Lemmas \ref{lem:aux1} and \ref{lem:aux2}, if $\mathcal{E}_1\cap \mathcal{E}_2$ holds then  for any $\Delta \,:\, [0,1] \rightarrow \mathbb{R}$, Cauchy-Schwarz inequality, it holds that 
 \begin{equation}
\label{Bound2EmpNorm} \frac{1}{n  }  \sum_{i=1}^n   \frac{1}{m_i}\sum_{j=1}^{m_i}\Delta(x_{i,j}) \epsilon_{i,j}   \le C_{1}'\varpi_\epsilon \| \Delta\|_{nm },\end{equation}
 for a positive constant $C_1^{\prime }>0$ and 
\begin{equation}
\label{Bound3EmpNorm}\frac{1}{n  }  \sum_{i=1}^n  \frac{1}{m_i}\sum_{j=1}^{m_i} \Delta(x_{i,j}) \delta_i  (x_{i,j})    \le   C_1 ''   \varpi_\delta  \| \Delta\|_{nm } ,\end{equation}
for some constant  $C_1 '' >0 $.
Therefore in $\mathcal{E}_1\cap \mathcal{E}_2$, using Equations (\ref{Bound2EmpNorm}) and (\ref{Bound3EmpNorm}), we get that
{\small{\begin{equation}
\label{Bound1EmpNorm}
    \frac{1}{2}\|  f^* -  \bar{y} - f\|_{nm}^2  \le C_1 '\varpi_\epsilon\vert\vert \Delta_f\vert\vert_{nm}+C_1 ''   \varpi_\delta  \| \Delta_f\|_{nm }.
\end{equation}}}

{\bf Step 2.} Suppose that  $\mathcal{E}_1 \cap \mathcal{E}_2\cap\Omega_1 \cap \Omega_{y}(\eta_1)$ holds, where the event $\Omega_1$ is defined as
\[
	 \Omega_1=\Big\{ \| p\|_{\infty} ^2\,\leq\,  \frac{4CC_{k}^2}{cn}\sum_{i=1}^{n} \frac{1}{m_i} \sum_{j=1}^{m_i} p(x_{i,j})^2 \ \text{for all}\  p(x)  \ \text{polynomial of  degree} \ k-1\Big\}.
\]  
for $C_{k}$  an absolute positive constant, 
and
\[
  \Omega_y(\eta_1)\,:= \, \{   \vert \bar{y} - \mathcal{A}_{f^*}\vert  \leq \eta_1  \}
\]
for
\begin{equation}
    \label{eqn:eta_cond}
    \eta_1 \,\geq\,     26 \sqrt{\frac{ \log n}{n}\frac{c_2 C \widetilde C_1 (\varpi_\delta^2  +\varpi_\epsilon^2) }{ c C_{\beta}}     }.
\end{equation}
. We observe that,
\begin{align*}
    \mathbb{P}\Big(\Omega_1\Big)\ge& 1 - \frac{32\log^3(n)}{c_1^2c^2n^2m^2}\frac{1}{C_{\beta}^3}- \frac{\log(n)}{n^2}
\end{align*}
by Remark \ref{Rem1HPE}, and  $\mathbb{P}(\Omega_y(\eta)) \geq 1- \varepsilon/10$, for large enough $n$, by Lemma \ref{lem:aux3}.

{\bf Step 3.} In this step, we will show that for any function $f \in \Lambda$, its $\|\Delta_f\|_{\infty}$ is bounded with high probability.

For this purpose, let $f \in \Lambda$ let 
\begin{equation}
\label{normalizing}
\widetilde{\Delta}_f:=  \frac{\Delta_{f}}{ 2 V_n }.
\end{equation}
 Since $\mathcal{E}_1$ and $\mathcal{E}_2$  hold, (\ref{Bound2EmpNorm}), and (\ref{Bound3EmpNorm}) and (\ref{Bound1EmpNorm})  yield 
\begin{equation}
\label{auxequation-1-thm1}
     \frac{1}{2}\|  f^* -  \bar{y} - f\|_{nm}^2 
\le  C_1' \varpi_\epsilon \| \Delta_{f}\|_{nm } + C_1''   \varpi_\delta  \| \Delta_{f}\|_{nm }    
\end{equation}
Further,  considering the inequality $\vert\vert f_1+f_2\vert\vert_{nm}^2\ge \frac{\vert\vert f_1\vert\vert_{nm}^2}{2}-\vert\vert f_2\vert\vert_{nm}^2$, it follows that
\begin{equation}
    \label{eqn:ineq_2}
    \frac{1}{2}(\bar{y} -  \mathcal{A}_{f^*})^2 +\frac{1}{2}\|  f^* -  \bar{y} - f \|_{ nm} ^2\ge \frac{1}{4}\| \Delta_f\|_{nm}^2 .
\end{equation}
Therefore, using Inequality (\ref{auxequation-1-thm1}), and the definition of $\Omega_y(\eta_1)$ we get that
\begin{equation}
    \label{eqn:_new_basic}
  \begin{array}{lll}
\displaystyle         \frac{1}{4}\| \Delta_f\|_{nm}^2 &\leq&\displaystyle      C_1' \varpi_\epsilon \| \Delta_{f}\|_{nm } + C_1''   \varpi_\delta  \| \Delta_{f}\|_{nm }     +  \frac{\eta_1^2}{2}\\
    &\leq & \displaystyle   16( C_1' )^2 \varpi_\epsilon^2   +   \frac{\|\Delta_{f}\|_{nm } ^2 }{16 }  +   16( C_1'' )^2\varpi _\delta^2  +   \frac{\| \Delta_{f}\|_{nm } ^2 }{16 } +\frac{\eta_1^2}{2}
  \end{array}
\end{equation}
for  $\eta_1 = O(1)$ satisfying (\ref{eqn:eta_cond}).

As a result, using Inequality (\ref{eqn:_new_basic}), we obtain that 
\begin{align}
	\label{eq:l infty bound in step 1}\|  \Delta_{f}\|_{ nm} ^2  \le    \widetilde C_2 \big ( \varpi_\epsilon^2 + \varpi _\delta^ 2 +1\big   ).
\end{align} 
where $\widetilde C_2>0$ is a constant. Notice that $ V_n \ge 1 $ by assumption. Then, Equation (\ref{eq:l infty bound in step 1}) implies that  
\begin{equation}
\label{eq:l infty bound in step 1-aux}
   \| \widetilde{\Delta}_{f} \|_{nm }  ^2  \le   \widetilde C_2 \big ( \varpi_\epsilon^2 + \varpi _\delta^2 + 1  \big   ) .
\end{equation}
Moreover, by definition of $\widetilde{\Delta}_{f}$ in Equation (\ref{normalizing}) it holds that 
$$ J_{k}( \widetilde{\Delta}_{f}) = \frac{J_{k}(\Delta_{f})}{ 2V_n }\le 1.$$
Observe that $$\sum_{i=1}^{n}\sum_{j=1}^{m_i}{\widetilde{\Delta}}_{f}(x_{i,j})=0.$$ Therefore 
\begin{equation}
\label{Equality-projection}
{\widetilde{\Delta}}_{f}=\mathcal{P}({\widetilde{\Delta}}_{f})   +  \mathcal{G}({\widetilde{\Delta}}_{f}),
\end{equation}
where  $\mathcal{P}(\widetilde\Delta_f)  $ is a polynomial with  $\mathcal{P}(\widetilde\Delta_f)   \in \text{span}\left\{1,x,x^2,\ldots,x^{k-1}\right\}$ and  $\mathcal{G}(\widetilde\Delta_f)$ a  function orthogonal to $\mathcal{P}(\widetilde\Delta_f)  $  in the inner product in $L_2([0,1])$. Then, 
\begin{equation}
\label{Equality-projection-1}
 \|   \widetilde{\Delta}_{f}\|_{\infty} \,\leq \,      \|\mathcal{P}( \widetilde{\Delta}_{f}) \|_{\infty}  +  \| \mathcal{G}(\widetilde{\Delta}_{f})\|_{\infty}, 
\end{equation}
and we proceed to bound each term.

First, since we have that $J_{k}(\widetilde{\Delta}_{f}) \leq 1$, by  Lemma 5 in   \cite{sadhanala2019additive}, there exists $R_0>0$ such that
\begin{equation}
\label{Ineq-Lemma5-tibs}
    \| \mathcal{G}(\widetilde\Delta_{f})\|_{\infty}   \,\leq  \,   R_0.
\end{equation}

Since $\Omega_1$ holds, by Lemma \ref{lemma1}, we have that 
\[
   \begin{array}{lll} 
   	\|  \mathcal{P}(\widetilde{\Delta}_{f}) \|_{\infty}&\leq& \sqrt{4\frac{CC_{k}^2}{c}} \|   \mathcal{P}(\widetilde{\Delta}_{f}) \|_{nm}.
   \end{array}
\]
In consequence, using Equation (\ref{Equality-projection}) and Inequality (\ref{Ineq-Lemma5-tibs})
\[
   \begin{array}{lll} 
   	\|  \mathcal{P}(\widetilde{\Delta}_{f}) \|_{\infty}
   	&\leq&\sqrt{4\frac{CC_{k}^2}{c}} \|   \widetilde{\Delta}_{g_f}\|_{nm}\,+\, \sqrt{4\frac{CC_{k}^2}{c}} \|   \mathcal{G}(\widetilde{\Delta}_{f}) \|_{nm}\\
   	&\leq&\sqrt{4\frac{CC_{k}^2}{c}} \|   \widetilde{\Delta}_{f}\|_{nm}\,+\, \sqrt{4\frac{CC_{k}^2}{c}} R_0\\
   \end{array}
\]
Further, Inequality (\ref{eq:l infty bound in step 1-aux}) leads to
\[
   \begin{array}{lll} 
   	\|  \mathcal{P}(\widetilde{\Delta}_{f}) \|_{\infty}
   	&\leq&  \sqrt{4\frac{CC_{k}^2}{c}  \widetilde C_2 \big ( \varpi_\epsilon^2 + \varpi_\delta^ 2  +1\big   ) } \,+\, \sqrt{4\frac{CC_{k}^2}{c}} R_0,\\
   \end{array}
\]
Now we make use of the above inequality together with Inequality (\ref{Equality-projection-1}) and Inequality (\ref{Ineq-Lemma5-tibs}) to obtain
\begin{equation}
	\label{eqn:e8}
\begin{array}{lll}
 \|   \widetilde{\Delta}_{f}\|_{\infty}&\leq &      \|\mathcal{P}(\widetilde{\Delta}_{f}) \|_{\infty}  +  \| \mathcal{G}(\widetilde{\Delta}_{f})\|_{\infty}\\
 &\leq&	\sqrt{C_4}\left(  \sqrt{ \varpi_\epsilon^2 + \varpi _\delta^ 2+1 }+ 1  \right), \\ 
\end{array}
\end{equation}
for some sufficiently large constant $C_4>0$ such that 
\begin{equation}
\label{LinfConst}
\sqrt{C_4}\left(  \sqrt{ \sigma_\epsilon^2 + \sigma _\delta^ 2 +C_{f^*}^2}+ 1  \right)\ge1.
\end{equation}
Inequality (\ref{eqn:e8}) and the definition of $\widetilde{{\widetilde{\Delta}}}_{g_f}$ in Equation (\ref{normalizing}) imply that
\begin{equation}
    \label{auxequation-3-thm1}
     \|  \Delta_f\|_{\infty}\le L_n := 2V_{n}\sqrt{C_4}\left(  \sqrt{ \sigma_\epsilon^2 + \sigma _\delta^ 2 +1}+ 1  \right).
\end{equation}
\

{\bf{Step 4.}}

Next, notice that 
\begin{equation}
    \label{eqn:f_tilde}
\{\widetilde{\Delta}_f=\frac{\Delta_f}{L_n}:f\in\Lambda\}\subset   \widetilde{\mathcal{F}} \,:=\,  B_{T V}^{k-1}\left(1\right) \cap B_{\infty}\left(1\right).
\end{equation}
Also, for this step, we utilize Lemma \ref{lemma2} to show that the event
\begin{align*}
    \Omega(\eta_2):=&\left\{
    \begin{array}{l}
        \frac{1}{C}\|\widetilde{\Delta}_f\|_{nm}^2-\|\widetilde{\Delta}_f\|_2^2\leq \frac{1}{2}\|\widetilde{\Delta}_f\|_2^2+\frac{1}{2}\eta_2^2, \\
        \|\widetilde{\Delta}_f\|_2^2-\frac{1}{c}\|\widetilde{\Delta}_f\|_{nm}^2\leq \frac{1}{2}\|\widetilde{\Delta}_f\|_2^2+\frac{1}{2}\eta_2^2,\ \forall \ f\in\Lambda
    \end{array}
    \right\}
\end{align*}
holds with high probability, for any $\eta_2\ge C_1\log(n)^{\frac{2k}{2k+1}}(nm)^{-\frac{k}{2k+1}}$. 
In fact, from Lemma \ref{lemma2}, we can conclude that
\begin{equation*}
    \mathbb{P}\Big(\Omega(\eta_2)\Big)\ge 1-C_2\log(n) \exp \left(-C_3 nm \eta_2^2/\log(n)\right)-\frac{1}{n},
\end{equation*}
where $C_1$, $C_2$ and $C_3$ are positive constants. 
\\
\\

{\bf{Step 5.}} On this step we provide a basic inequality for   $\left\|\Delta_f\right\|_{2}$ 
for any $f\in{\Lambda}$. To this end, assume that $\Omega(\eta_2)\cap\mathcal{E}_1 \cap \mathcal{E}_2\cap\Omega_1 \cap \Omega_y(\eta_1)$ holds with $\eta >0$. Since $\Omega(\eta_2)$ holds, we have that 
\begin{equation*}
    \frac{\vert\vert \widetilde{\Delta}_f\vert\vert_2^2}{2}\le \frac{1}{c}\vert\vert \widetilde{\Delta}_f\vert\vert_{nm}^2+\frac{1}{2}\eta_2^2.
\end{equation*}
Using that $\widetilde{\Delta}_f=\frac{\Delta_f}{L_n}$ we obtain
\begin{equation*}
    \frac{\vert\vert \Delta_f\vert\vert_2^2}{2}\le \frac{1}{c}\vert\vert \Delta_f\vert\vert_{nm}^2+\frac{L_n^2}{2}\eta_2^2.
\end{equation*}
However,  from the (\ref{eqn:ineq_2}) and (\ref{eqn:main_eq}),
we have that
\begin{equation}
    \label{eqn:520}
    \frac{\vert\vert \Delta_f\vert\vert_2^2}{2}   \,\leq\,\frac{4}{c} [T_1(f)+T_2(f)+T_3(f)] \,+\, \frac{2 }{c} \eta_1^2  +\frac{L_n^2}{2}  \eta_2^2
\end{equation}

In the following the term $T_3(f)$ defined in Equation (\ref{T3}) is bounded. 
First, observe that by Lemma \ref{expected-value-b}
$$
\left\|\frac{1}{n} \sum_{i=1}^n \delta_i\right\|_{2} \leq \frac{ C_5}{\sqrt{n}}
$$
holds for a positive constant $C_5$. Hence,
$$
T_3(f) \leq\left\|\frac{1}{n} \sum_{i=1}^n \delta_i\right\|_{2}\left\|\Delta_f\right\|_{2} \leq \frac{ C_5\left\|\Delta_f\right\|_{2}}{\sqrt{n}} \leq \frac{16C_5 }{cn}+c\frac{\left\|\Delta_f\right\|_{2}^2}{16}.
$$
 Therefore, from (\ref{eqn:520}), it follows that 
\begin{align}
\label{u-b-1}
 \frac{\vert\vert \Delta_f\vert\vert_2^2}{4}   \,\leq\,\frac{4}{c} [T_1(f)+T_2(f)] \,+\, \frac{2 }{c}\eta_1^2 +\frac{L_n^2}{2} \eta_2^2 \,+\,  \frac{64 C_5}{c^2 n},\,\,\,\,\ \forall f \in \Lambda.
\end{align}

{\bf{Step 6}.} With the notation from the beginning of Appendix \ref{sec:aux_lemmas}, let $\Omega=\mathcal{E}_1 \cap \mathcal{E}_2\cap\Omega_1\cap \Omega_x\cap\Omega_\delta\cap\Omega_\epsilon\cap \Omega(\eta_2) \cap \Omega_y(\eta_1)$. Also let $\widetilde{\eta} = \sqrt{8}\eta_3$ where  
\[
\eta_3^2 \,:=\, \frac{2 }{c}\eta_1^2 +\frac{L_n^2}{2} \eta_2^2 \,+\,  \frac{64 C_5}{c^2 n},
\]
where $\eta_1$ and $\eta_2$ will be  chosen later.

Assume that the event
$
\{\| \Delta_{\widehat g_{k,V_n}}\|_{2}>\widetilde{\eta}\}
$
holds.
Then, $\kappa=\frac{\widetilde{\eta}}{\vert\vert \Delta_{\widehat g_{k,V_n}}\vert\vert_2}$, satisfies  $\kappa \in [0,1].$ Moreover observe that for $f=\kappa \widehat g_{k,V_n}+(1-\kappa)f^*$ we have that $\vert\vert \Delta_f \vert\vert_2^2=\widetilde{\eta}^2.$ Additionally, if $\Omega$ holds we have that $\mathcal{E}_1 \cap \mathcal{E}_2\cap\Omega_1$ holds. Therefore,
\begin{align}
\label{Req-bound4}
    &\mathbb{P}\left(\left\{\left\|\widehat g_{k,V_n}-f^* + \mathcal{A}_{f^*}\right\|_{2}^2>\widetilde{\eta}^2\right\}\cap \Omega\right)\nonumber
    \\
    \le&
    \mathbb{P}\left(\left\{\exists \Delta_f \in \mathcal{F} \  \text{with}\ f\in\Lambda :\left\|\Delta_f\right\|_{2}^2=\widetilde{\eta}^2\right\} \cap \Omega\right)\nonumber
    \\
    \le&\mathbb{P}\Big(\Big\{\exists \Delta_f \in \mathcal{F} \  \text{with}\ f\in\Lambda  : T_1(f)+T_2(f)> \widetilde{\eta}^2 \cdot \frac{c}{32} ,\nonumber
    \\
      & \text { and } \  \left\|\Delta_f\right\|_{2}^2\le\widetilde{\eta}^2\Big\} \cap \Omega\Big)\nonumber
    \\
    \le& A_{1}+A_{2},
\end{align}
where 
$$
A_j:=\mathbb{P}\left(\left\{\sup_{ f\,:\,\Delta_f \in \mathcal{F}, \left\|\Delta_f\right\|_{2}^2 \leq \widetilde{\eta}^2} T_{j}(f)>\widetilde{\eta}^2 \cdot \frac{c}{64} \right\} \cap \Omega\right),
$$
for $j=1,2$. We proceed to set bounds on each term 
\\
{\bf{Step 7}.} Let let us now bound $A_2$.  Towards that end, define $C^{\prime} =  c/64$ and 
$$
Z_{i,j}^*(f):=\left\{\Delta_f\left(x_{i,j}^*\right)\right\} \delta_i^*\left(x_{i,j}^*\right)-\left\langle \Delta_f, \delta_i^*\right\rangle_{2}.
$$
Observe that since $\Omega$ holds, $\Omega_x$ and $\Omega_\delta$ hold. In consequence, letting $$T_2^*(f)=\frac{1}{n} \sum_{i=1}^n \frac{1}{m_i}\sum_{j=1}^{m_i}\left[\left\{\Delta_f\left(x_{i,j}^*\right)\right\} \delta_i^*\left(x_{i,j}^*\right)-\left\langle \Delta_f, \delta_i^*\right\rangle_{\mathcal{L}_2}\right]=\frac{1}{n} \sum_{i=1}^n \frac{1}{m_i}\sum_{j=1}^{m_i}Z_{i,j}^*(f),$$
we have that
\begin{align}
\label{A2bound2}
    A_2
    =&\mathbb{P}\left(\left\{\sup_{ f\,:\,\Delta_f \in \mathcal{F},\left\|\Delta_f\right\|_{2}^2 \leq \widetilde{\eta}^2} \frac{1}{n} \sum_{i=1}^n \frac{1}{m_i}\sum_{j=1}^{m_i}Z_{i,j}^*(f)>\widetilde{\eta}^2  C' \right\} \cap \Omega\right)\nonumber\\
    \le&\sum_{l=1}^{L}\mathbb{P}\left(\left\{\sup_{f\,:\, \Delta_f \in \mathcal{F},\left\|\Delta_f\right\|_{2}^2 \leq \widetilde{\eta}^2} \frac{1}{n} \sum_{i\in\mathcal{J}_{e,l}} \frac{1}{m_i}\sum_{j=1}^{m_i}Z_{i,j}^*(f)>\widetilde{\eta}^2 C'/ (2L)\right\} \cap \Omega\right)\nonumber
    \\
    +&\sum_{l=1}^{L}\mathbb{P}\left(\left\{\sup_{f\,:\,\Delta_f \in \mathcal{F},\left\|\Delta_f\right\|_{2}^2 \leq \widetilde{\eta}^2}\frac{1}{n} \sum_{i\in\mathcal{J}_{o,l}} \frac{1}{m_i}\sum_{j=1}^{m_i}Z_{i,j}^*(f)>\widetilde{\eta}^2 C'/ (2L)\right\} \cap \Omega\right) \nonumber\\
\end{align}
\begin{align}
          \le& \sum_{l=1}^{L}\frac{2L C'}{\widetilde{\eta}^2}\mathbb{E}\Big( \sup_{f\,:\,\Delta_f \in \mathcal{F},\left\|\Delta_f\right\|_{2}^2 \leq \widetilde{\eta}^2}\frac{1}{n}\sum_{i\in\mathcal{J}_{e,l}}\frac{1}{m_i}\sum_{j=1}^{m_i}Z_{i,j}^*(f)\Big)\nonumber
    \\
    +&\sum_{l=1}^{L}\frac{2L C'}{\widetilde{\eta}^2}\mathbb{E}\Big( \sup_{f\,:\,\Delta_f \in \mathcal{F},\left\|\Delta_f\right\|_{2}^2 \leq \widetilde{\eta}^2}\frac{1}{n}\sum_{i\in\mathcal{J}_{o,l}}\frac{1}{m_i}\sum_{j=1}^{m_i}Z_{i,j}^*(f)\Big) \nonumber
\end{align}
where the first inequality follows from union bound, and the second by Markov's inequality. Next we proceed to bound each term in the last equation. 

Furthermore,  there exists $\left\{\xi_{i,j}\right\}_{\left\{ i\in\mathcal{J}_{e,l},j\in [m_i]\right\}}$ independent and identically distributed Rademacher random variables, independent of  $\left(\left\{x_{i,j}^*\right\}_{\left\{ i\in\mathcal{J}_{e,l},j\in [m_i]\right\}},\left\{\delta_i^*\right\}_{\left\{ i\in\mathcal{J}_{e,l},j\in [m_i]\right\}}\right)$, such that by Lemma \ref{lem:aux4}, it holds that 
{\small{\begin{align}
\label{Fourth-bound-1}
 & \frac{2L  }{\widetilde{\eta}^2 C' }\mathbb{E}\Big( \sup_{f\,:\,\Delta_f \in \mathcal{F},\left\|\Delta_f\right\|_{2}^2 \leq \widetilde{\eta}^2}\frac{1}{n}\sum_{i\in\mathcal{J}_{e,l}}\frac{1}{m_i}\sum_{j=1}^{m_i}Z_{i,j}^*(f)\Big)  \leq  \nonumber\\ 
&\frac{4Cc_2}{\widetilde{\eta}^2 cC' } \mathbb{E}\left( \mathbb{E}\left(\delta^*_{\mathcal{J}_{e,l}}\mathbb{E}\left(\sup_{f\,:\,\Delta_f \in \mathcal{F},\left\|\Delta_f\right\|_{2}^2 \leq \widetilde{\eta}^2} \frac{1}{\sum_{i\in\mathcal{J}_{e,l}}m_i}  \sum_{i\in\mathcal{J}_{e,l}} \sum_{j=1}^{m_i} \xi_{i,j}\Delta_f\left(x_{i,j}^*\right)\mid\left\{\delta_i^*\right\},\left\{x_{i,j}^*\right\}\right) \mid\left\{\{\delta_i^*\}\right\}\right)\right).
\end{align}}}

Then, we have that $\widetilde{\Delta}_f=\frac{\Delta_f}{L_n}$ satisfies that, $L_n\ge 2V_n$,   $\vert\vert \widetilde{\Delta}_f\vert\vert_{\infty}\le 1$ and $J_k(\widetilde{ \Delta}_f)\le 1.$ Moreover, for some constant $\widetilde{C}$, we have that $L_n\le\widetilde{C}V_n$.
Furthermore, since $L_n\ge 2V_n$ it is follows that $\vert\vert \widetilde{\Delta}_f\vert\vert_{2}\le \frac{1}{2V_n}\widetilde{\eta}$. In consequence, the term in the right hand side of (\ref{Fourth-bound-1}) can be  bounded by
{\footnotesize{\begin{align*}
\frac{\widetilde{V}_n}{ \widetilde{\eta}^2 C'   }\mathbb{E}\left(\delta^*_{\mathcal{J}_{e,l}}\mathbb{E}\left(\sup _{f\,:\,\Delta_f \in \mathcal{F},\left\|\Delta_f\right\|_{2} \leq \eta^{\prime}} \frac{1}{\sum_{i\in\mathcal{J}_{e,l}}m_i}  \sum_{i\in\mathcal{J}_{e,l}} \sum_{j=1}^{m_i} \xi_{i,j}\Delta_f\left(x_{i,j}^*\right)\mid\left\{\delta_i^*\right\},\left\{x_{i,j}^*\right\}\right) \mid\left\{\{\delta_i^*\}_{i\in\mathcal{J}_{e,l}}\right\}\right), 
\end{align*}}}
where, $\eta^{\prime}=\frac{1}{2V_n}\widetilde{\eta}$, $\widetilde{V}_n=\frac{4\widetilde{C}Cc_2V_n}{c}$. 
Performing an analysis akin to the one carried out in the proof of 
Lemma \ref{lemma2}, we get that
\begin{align*}
&\mathbb{E}\left(\sup _{f\,:\,\Delta_f \in \mathcal{F},\left\|\Delta_f\right\|_{2} \leq \eta^{\prime}} \frac{1}{\sum_{i\in\mathcal{J}_{e,l}}m_i}  \sum_{i\in\mathcal{J}_{e,l}} \sum_{j=1}^{m_i} \xi_{i,j}\Delta_f\left(x_{i,j}^*\right)\mid\left\{\delta_i^*\right\},\left\{x_{i,j}^*\right\}\right)
\\
\le&\frac{C_{12}\sqrt{\log(n)}}{\sqrt{nm}}(\eta^{\prime})^{1-1/(2k)},
\end{align*}
for a positive constant $C_{12}.$ By Inequality (\ref{Fourth-bound-1}) it follows that
\begin{align*}
     &  \frac{2L  }{\widetilde{\eta}^2 C' }\mathbb{E}\Big( \sup_{\Delta_f \in \mathcal{F}:\left\|\Delta_f\right\|_{2}^2 \leq \widetilde{\eta}^2}\frac{1}{n}\sum_{i\in\mathcal{J}_{e,l}}\frac{1}{m_i}\sum_{j=1}^{m_i}Z_{i,j}^*(f)\Big) 
     \\ 
     \le& \frac{L}{\widetilde{\eta}^2   C' }\frac{\widetilde{V}_nC_{13}\sqrt{\log(n)}}{\sqrt{nm}L}(\eta^{\prime})^{1-1/(2k)} \mathbb{E}\Big(\max_{i\in\mathcal{J}_{e,l},1\le j\le m_i}\vert\delta_{i}^*(x_{i,j})\vert\Big),
\end{align*}
for a positive constant $C_{13}.$
Then, by Assumption 1, we have that 
\begin{align}
\label{eqn:delta_inf}
\mathbb{E}\Big(\max_{i\in\mathcal{J}_{e,l},1\le j\le m_i}\vert\delta_{i}^*(x_{i,j})\vert\Big)&\le \varpi_\delta\sqrt{2\log(\sum_{i\in\mathcal{J}_{e,l}}m_i)}\nonumber
\\
&\le\varpi_\delta\sqrt{2\log(\frac{C_{\beta}Cc_2nm}{2\log(n)})}.
\end{align}

Therefore, 
\begin{equation}
    \label{eqn:532}
    A_2 \,\leq\, C_{14}    \frac{L}{\widetilde{\eta}^2   C' }\frac{\widetilde{V}_n\sqrt{\log(n)}}{\sqrt{nm}}(\eta^{\prime})^{1-1/(2k)} \sqrt{
    \log(nm)}, 
\end{equation}
for some positive constant $C_{14}>0$.\\

{\bf{Step 8}.} Now we bound $A_1$. Recall part of the analysis employed in {\bf{Step 5}}. Here, we have that $\widetilde{\Delta}_f=\frac{\Delta_f}{L_n}$ satisfies that  $\vert\vert \widetilde{\Delta}_f\vert\vert_{\infty}\le 1$ and $J_k(\widetilde{ \Delta}_f)\le 1.$ Moreover, we have that $L_n\le\widetilde{C}V_n$.
Furthermore, since $L_n\ge 2V_n$ it is follows that $\vert\vert \widetilde{\Delta}_f\vert\vert_{2}\le \frac{1}{2V_n}\widetilde{\eta}$. In consequence, based on Lemma \ref{lem:aux5}, we have that it is enough to bound 
{\small{\begin{align*}
\frac{\widetilde{V}_n}{L}\mathbb{E}\left(\epsilon^*_{\mathcal{J}_{e,l}}\mathbb{E}\left(\sup _{f\,:\,\Delta_f \in \mathcal{F},\left\|\Delta_f\right\|_{2} \leq \eta^{\prime}} \frac{1}{\sum_{i\in\mathcal{J}_{e,l}}m_i}  \sum_{i\in\mathcal{J}_{e,l}} \sum_{j=1}^{m_i} \xi_{i,j}\Delta_f\left(x_{i,j}^*\right)\mid\{\epsilon_{i,j}^*\},\left\{x_{i,j}^*\right\}\right) \mid\{\epsilon_{i,j}^*\}\right), 
\end{align*}}}
where, $\eta^{\prime}=\frac{1}{2V_n}\widetilde{\eta}$, $\widetilde{V}_n=\frac{\widetilde{C}Cc_2V_n}{c}$.

Performing an analysis akin to the one carried out in the proof of 
Lemma \ref{lemma2}, we get that
\begin{align*}
&\mathbb{E}\left(\sup _{\Delta_f \in \widetilde{\mathcal{F}}:\left\|\Delta_f\right\|_{2} \leq \eta^{\prime}} \frac{1}{\sum_{i\in\mathcal{J}_{e,l}}m_i}  \sum_{i\in\mathcal{J}_{e,l}} \sum_{j=1}^{m_i} \xi_{i,j}\Delta_f\left(x_{i,j}^*\right)\mid\{\epsilon_{i,j}^*\},\left\{x_{i,j}^*\right\}\right)
\\
\le&\frac{C_{15}\sqrt{\log(n)}}{\sqrt{nm}}(\eta^{\prime})^{1-1/(2k)},
\end{align*} 
for a positive constant $C_{15}.$ And the same inequality holds replacing $\mathcal{J}_{e,l}$ with $\mathcal{J}_{o,l}$.

Proceeding as in (\ref{eqn:delta_inf}) and invoking Lemma \ref{lem:aux5} we arrive at 
\begin{equation}
    \label{eqn:533}
    A_1 \,\leq\, C_{15}    \frac{L}{\widetilde{\eta}^2   C' }\frac{\widetilde{V}_n\sqrt{\log(n)}}{\sqrt{nm}}(\eta^{\prime})^{1-1/(2k)} \sqrt{
    \log(nm)}, 
\end{equation}
for a positive constant $C_{15}$.

{\bf{Step 9}.}

Finally, from the inequality 
\[
\|  f_{k,V_n}-f^* \|_2^2 \,\leq \,  2 ( \bar{y} -  \mathcal{A}_{f^*} )^2 \,+\, 2 \|  g_{k,V_n} + \mathcal{A}_{f^*} -f^*  \|_2^2  
\]
and combining (\ref{Req-bound4}), the fact that $\Omega$ holds with probability approaching one as $n$ goes to $\infty$,  (\ref{eqn:532}) and (\ref{eqn:533}), we obtain that for a given $\epsilon$ there exists a choice $\eta$ satisfying
\[
\eta^2 \asymp
\left( V_n^2\log^\frac{4k}{2k+1}(n)(nm)^{-\frac{2k}{2k+1}}+\frac{\log^{5/2}(n)}{n}\right),
\]
by choosing appropriate $\eta_1$ and $\eta_2$ with
\[
\eta_1 \asymp \sqrt{\frac{\log n}{n}}  \,\,\,\text{and}\,\,\, \eta_2 \asymp \log(n)^{\frac{2k}{2k+1}}(nm)^{-\frac{k}{2k+1}},
\]
such that 
\[
    \mathbb{P}(\|  f_{k,V_n}-f^* \|_2^2  \geq  C_{5} \eta^2 ) \,\leq \,\varepsilon,
\]
for some positive constant $C_5$

\end{proof}
\newpage
\section{Proof of Theorem \ref{Penalized}}

\begin{proof}
Throughout, let $\varepsilon>0$ be given. We use the notation from the proof of Theorem \ref{thm:main tv}. Specifically,  $\widetilde{f^*}=f^*- \mathcal{A}_{f^*} $ with $\mathcal{A}_{f^*} =\frac{1}{\sum_{i=1}^n m_i}\sum_{i=1}^n\sum_{j=1}^{m_i}f^*(x_{i,j})$,  $\Lambda := \{ t\widehat{g}_{k,\lambda}  + (1-t)\widetilde{f^*}: t \in [0,1] \}$. Hence, by the basic inequality,  
\begin{align*}
\displaystyle	\frac{1}{ n } \sum_{i=1}^n  \frac{1}{m_i}\sum_{j=1}^{m_i} \big (  y_{i,j} -  \bar{y} -  f(x_{i,j})  \big)^2  \,+\, J_k(f ) 
	\le   \frac{1}{ n }  \sum_{i=1}^n  \frac{1}{m_i}\sum_{j=1}^{m_i} \big (  y_{i,j} -  \bar{y} -  \widetilde{f^*}  (x_{i,j})  \big)^2  \,+\, J_k( f^* ) .
\end{align*}

This implies that for all $f \in \Lambda$ it holds that 
\begin{equation}
    \label{eqn:main_eq2}
     \begin{array}{lll}
     \displaystyle   \frac{1}{2}\|  f^* -  \bar{y} - f\|_{nm}^2 & \leq&  \displaystyle\frac{1}{n }\sum_{i=1}^n   \frac{1}{m_i} \sum_{j=1}^{m_i} \Big(f(x_{i,j})-f^*(x_{i,j})+\mathcal{A}_{f^*}\Big) \Big(\delta_i(x_{i,j}) +\epsilon_{
i,j}\Big) \\ 
 & = &  T_1(f)+T_2(f)+T_3(f)+    \lambda( J_k( f^* ) -  J_k( f ) ) ,
 \end{array}
\end{equation}
where 
{\small{\begin{align}
&T_1(f):=\frac{1}{n } \sum_{i=1}^n \frac{1}{m_i}\sum_{j=1}^{m_i}\left\{f\left(x_{i,j}\right)-f^*\left(x_{i,j}\right) +\mathcal{A}_{f^*}  \right\} \epsilon_{i,j}, \label{T12}\\
&T_2(f):=\frac{1}{n} \sum_{i=1}^n \frac{1}{m_i}\sum_{j=1}^{m_i}\left[\left\{f\left(x_{i,j}\right)-f^*\left(x_{i,j}\right)+\mathcal{A}_{f^*}\right\} \delta_i\left(x_{i,j}\right)-\left\langle f-f^*+\mathcal{A}_{f^*}, \delta_i\right\rangle_{\mathcal{L}_2}\right],\label{T22}
\\
&T_3(f):=\frac{1}{n} \sum_{i=1}^n \frac{1}{m_i}\sum_{j=1}^{m_i}\left\langle f-f^*+\mathcal{A}_{f^*}, \delta_i\right\rangle_{\mathcal{L}_2},
\label{eqn:570}  
\end{align}}}
and we let $\Delta_f = f -f^* +\mathcal{A}_{f^*}  $.

Now, denoting $  \Omega_y(\eta_1)$  for \begin{equation}
    \label{eqn:eta_cond2}
    \eta_1 \,\geq\,     26 \sqrt{\frac{ \log n}{n}\frac{c_2 C \widetilde C_1 (\varpi_\delta^2  +\varpi_\epsilon^2) }{ c C_{\beta}}     }
\end{equation} 
as in Step 2 of the proof of Theorem \ref{thm:main tv}, from (\ref{eqn:main_eq2}), we obtain that 
\begin{equation}
    \label{eqn:main_eq3}
     \begin{array}{lll}
     \displaystyle   \frac{1}{4}\|  \Delta_f\|_{nm}^2 
 & \leq &  T_1(f)+T_2(f)+T_3(f) +    \lambda( J_k( f^* ) -  J_k( f ) )  \,+\,\frac{\eta_1^2}{2}.
 \end{array}
\end{equation}

 {\bf{Step 2}}

 Now, let  $f \in \Lambda$ such that 
$\left\|\Delta_f\right\|_{2}^2 \leq \widetilde{\eta}^2$ 
and 
$J_k\left(\Delta_{f}\right) \geq 4 J_{k}(f^*)$.  Then, there exists $ h\in \Lambda$ such that $J_{k}(h ) \leq J_k\left(\Delta_{h}\right)+J_{k}(f^*)=5 J_{k}(f^*)$ and $\left\|\Delta_h\right\|_{2}^2 \leq \widetilde{\eta}^2$. This follows by letting $t:=\frac{4 J_k\left(f^*\right)}{J_k\left(\Delta_{f}\right)} \in[0,1]$ and $ h:=t  f+(1-t)f^* .$ In fact
$$
J_k\left(\Delta_{ h}\right)=J_k\left(t\left(f-f^*\right)\right)=t J_k\left(f-f^*\right)=\frac{4 J_k\left(f^*\right)}{J_k\left(f-f^*\right)} J_k\left(f-f^*\right)=4 J_k\left(f^*\right),
$$
from where $J_{k}( h) \leq 5 J_{k}(f^*)$. Moreover $ h=tt_1\widehat{g}_{k,\lambda}+(1-tt_1)f^*$ where $t_1\in[0,1]$ is such that $f=t_1 \widehat{g}_{k,\lambda}+(1-t_1)f^*.$ Therefore $ h_f\in\Lambda.$ It is also satisfied that
\[
\begin{array}{lll}
     \|\Delta_h\|_2^2 & = &\| h - f^* +  \mathcal{A}_{f^*} \|_2^2  \\
     & =&\| (t f+  (1-t)(f^* -  \mathcal{A}_{f^*} )    )- f^* +  \mathcal{A}_{f^*} \|_2^2  \\
&= &t^2 \| f+  f^* -  \mathcal{A}_{f^*}  \|_2^2  \\
&\leq & \| \Delta_f\|_2^2\\
 & \leq&\widetilde{\eta}^2.
\end{array}
\]
\\
  \\

 {\bf{Step 3}}
 
For $\widetilde{\eta}>0$, with $\eta = O(1)$, let
\[
\mathcal{E}_3\,:=\, \left\{ \underset{ f\,:\, \|\Delta_f\|_2 \leq \widetilde{\eta}  }{ \sup}\, J_k(f) \,\leq\, 5J_k(f^*)  \right\}.
\]
Then, with $\Omega_1$,  $\mathcal{E}_1$, and $\mathcal{E}_2$ as in the proof of Theorem \ref{thm:main tv}, suppose that $ \Omega_1 \cap \mathcal{E}_1\cap\mathcal{E}_2 \cap \mathcal{E}_3$
holds. Then  for any $f\in \Lambda$, as in Step 3 of the proof of Theorem \ref{thm:main tv}, see (\ref{eqn:_new_basic}), it holds that 
\begin{equation}
    \label{eqn:_new_basic2}
  \begin{array}{lll}
\displaystyle         \frac{1}{4}\| \Delta_f\|_{nm}^2 &\leq&\displaystyle      C_1' \varpi_\epsilon \| \Delta_{f}\|_{nm } + C_1''   \varpi_\delta  \| \Delta_{f}\|_{nm }     +  \frac{\eta_1^2}{2} \,+\, \lambda J_k(f^*)\\
    &\leq & \displaystyle   16( C_1' )^2 \varpi_\epsilon^2   +   \frac{\|\Delta_{f}\|_{nm } ^2 }{16 }  +   16( C_1'' )^2\varpi _\delta^2  +   \frac{\| \Delta_{f}\|_{nm } ^2 }{16 } +\frac{\eta_1^2}{2}\,+\, \lambda J_k(f^*)\\
  \end{array}
\end{equation}
which implies  
\begin{equation*}
  \begin{array}{lll}
\displaystyle   \|\Delta_f\|_{nm}^2\,\leq \, \widetilde{C}_2(\varpi_\epsilon^2 \,+\,  \varpi_\delta^2  \,+\, \lambda J_k(f^*) )      
  \end{array}
\end{equation*}
for a constant $\widetilde{C}_2>0$. Hence, if we choose $\lambda$ such that $\lambda J_K(f^*) =  O(1)$ we obtain that  for all $f\in \Lambda$,
\begin{equation}
    \label{eqn:_new_basic22}
  \begin{array}{lll}
\displaystyle   \|\Delta_f\|_{nm}^2\,\leq \, \widetilde{C}_2(\varpi_\epsilon^2 \,+\,  \varpi_\delta^2  \,+\, 1)      
  \end{array}
\end{equation}
by setting $\widetilde{C}_2$ large enough.

Next, define $\widetilde{\Delta}_f =  \Delta_f/(6 J_k(f^*))$. Then,  
\[
   J_k( \widetilde{\Delta}_f  ) \,\leq\,  \frac{ J_k(f) +  J_k(f^*)  }{  6 J_k(f^*)  }\,\leq\,  \frac{ 5J_k(f^*) +  J_k(f^*)  }{  6 J_k(f^*)  } \,\leq \, 1.
\]
Also, (\ref{eqn:_new_basic22}) implies 
\begin{equation}
    \label{eqn:_new_basic23}
  \begin{array}{lll}
\displaystyle   \|\widetilde\Delta_f\|_{nm}^2\,\leq \, \frac{1}{36J_k(f^*)^2 }\widetilde{C}_2(\varpi_\epsilon^2 \,+\,  \varpi_\delta^2  \,+\, 1).      
  \end{array}
\end{equation}

Now observe that $\sum_{i=1}^{n}\sum_{j=1}^{m_i}\Delta_{f}(x_{i,j})=0$. In consequence,
\begin{equation}
    \label{DecompPolOrt}
  \Delta_{f}=\mathcal{P}(\Delta_f)   +  \mathcal{G}(\Delta_{f}),
\end{equation}
where  $\mathcal{P}(\Delta_{f})  $ is a polynomial with  $\mathcal{P}(\Delta_{f})   \in \text{span}\left\{1,x,x^2,\ldots,x^{k-1}\right\}$ and  $\mathcal{G}(\Delta_{f})$ a  function orthogonal to $\mathcal{P}(f)  $  in the inner product in $L_2(U)$, with $U$ is the uniform distribution on $[0, 1]$. Then we  have that
\begin{equation}
    \label{eqn:600}
     \|   \widetilde{\Delta_{f}}\|_{\infty} \,\leq \,      \|\mathcal{P}(\widetilde{\Delta}_{f}) \|_{\infty}  +  \| \mathcal{G}(\widetilde{\Delta}_{f})\|_{\infty}, 
\end{equation}
and we proceed to bound each term. First using that $J_{k}(\widetilde{\Delta}_{f}) \leq 1$, by  Lemma 5 in \cite{sadhanala2019additive} there exists $R_0>0$ such that
\begin{equation}
    \label{Lemma5TibsBound}
    \| \mathcal{G}(\widetilde{\Delta}_{f})\|_{\infty}   \,\leq  \,   R_0.
\end{equation}
Since $\Omega_1$ holds we have that 
   \begin{align}
      \label{PolyBound}
       \|  \mathcal{P}(\widetilde{\Delta}_{f}) \|_{\infty}&\leq \sqrt{\frac{4CC_{k}^2}{c}} \|   \mathcal{P}(\widetilde{\Delta}_{f}) \|_{nm}.
   \end{align}
   Therefore, 
\[
   \begin{array}{lll} 
   	\|  \mathcal{P}(\widetilde{\Delta}_{f}) \|_{\infty}&\leq& \sqrt{\frac{4CC_{k}^2}{c}} \|   \mathcal{P}(\widetilde{\Delta}_{f}) \|_{nm}\\
   	&\leq&\sqrt{\frac{4CC_{k}^2}{c}} \|  \widetilde{\Delta}_{f}\|_{nm}\,+\, \sqrt{\frac{4CC_{k}^2}{c}} \|   \mathcal{G}(\widetilde{\Delta}_{f}) \|_{nm}\\
   	&\leq&\sqrt{\frac{4CC_{k}^2}{c}} \|  \widetilde{\Delta}_{f}\|_{nm}\,+\, \sqrt{\frac{4CC_{k}^2}{c}} R_0.
   \end{array}
\]
As a result,  from (\ref{eqn:_new_basic23})  and (\ref{eqn:600}),
\begin{equation}
\label{AuxThm2-1}
\begin{array}{lll}
 \| \widetilde{\Delta}_{f}\|_{\infty}&\leq &      \|\mathcal{P}(\widetilde{\Delta}_{f}) \|_{\infty}  +  \| \mathcal{G}(\widetilde{\Delta}_{f})\|_{\infty}\\
 &\leq&\sqrt{\frac{4CC_{k}^2}{c}}\frac{1}{6J_k(f^*)}\sqrt{ \widetilde{C}_2 \left(\varpi_\epsilon^2+\varpi_\delta^2+1\right)}+\sqrt{\frac{4CC_{k}^2}{c}} R_0 + R_0
 \\
 &\leq&	\sqrt{\widetilde C_4}\frac{1}{6J_k(f^*)}\left(  \sqrt{ \varpi_\epsilon^2 + \varpi _\delta^ 2 +1}  \right)+ \sqrt{\widetilde C_4}, 
\end{array}
\end{equation}
for some constant $\widetilde C_4>1$.

Therefore, for any $f\in \Lambda$, 
\begin{equation}
\label{AuxThm2-3}
\begin{array}{lll}
 \| \Delta_{f}\|_{\infty}
 &\leq&	\sqrt{\widetilde C_4}\left(  \sqrt{ \varpi_\epsilon^2 + \varpi _\delta^ 2 +1}  \right)+  6J_k(f^*)\sqrt{\widetilde C_4}.
\end{array}
\end{equation}

In consequence
{\small{\begin{equation}
\label{LambdaFm}
 \{\Delta_f:f\in\Lambda\  \text{with} \ J_{k}(f)\le5J_{k}(f^*), \vert\vert \Delta_f\vert\vert_2^2\le \widetilde{\eta}^2\}\subset B_{T V}^{k-1}\left(6 J_{k}(f^*)\right) \cap B_{\infty}\left(L_n\right)=:\mathcal{F},
\end{equation}}}
with
\begin{equation}
    \label{L_nchoiceT2}
    L_n:=6J_{k}(f^*)\sqrt{\widetilde C_4}+\sqrt{\widetilde C_4}\left(  \sqrt{ \varpi_\epsilon^2 + \varpi _\delta^ 2 + 1}  \right)\ge 6J_{k}(f^*),
\end{equation}
for large enough $\widetilde C_4$.

We now proceed to achieve the second goal of this step. To this end define $\widetilde{\widetilde{\Delta}}_f:=\frac{\Delta_f}{L_n}\in B_\infty(1)$. Making use of the above analysis, we obtain that $\vert\vert \widetilde{\widetilde{\Delta}}_f\vert\vert_{\infty}\le 1$, and by Inequality (\ref{AuxThm2-3}) and the definition of $L_n$, we further have that $J_{k}(\widetilde{\widetilde{\Delta}}_f)\le 1.$ Thus
\begin{equation}
\label{LambdaTilde}
 \{\widetilde{\widetilde{\Delta}}_f:f\in\Lambda\  \text{with} \ J_{k}(f)\le5J_{k}(f^*), \vert\vert \Delta_f\vert\vert_2^2\le\widetilde{\eta}^2\}\subset B_{T V}^{k-1}\left(1\right) \cap B_{\infty}\left(1\right)=:\widetilde{\mathcal{F}}.
\end{equation}

 Suppose also that the event   $\Omega(\eta_2)$ defined in Step 4 of the proof of Theorem \ref{thm:main tv} holds, for $\eta_2$  satisfying $\eta_2\ge C_1\log(n)^{\frac{2k}{2k+1}}(nm)^{-\frac{k}{2k+1}}$ with $C_1$ a positive constant also as in Theorem \ref{thm:main tv}. 

Then, proceeding as in Step 5 of the proof of Theorem \ref{thm:main tv}, and using (\ref{eqn:main_eq3}), we have that
\begin{equation}
    \label{eqn:700}
    \frac{\vert\vert \Delta_f\vert\vert_2^2}{2}   \,\leq\,\frac{4}{c} [T_1(f)+T_2(f)] \,+\,  \frac{64 C_5}{c^2 n}  \,+\, \lambda( J_k( f^* ) -  J_k( f ) ) ] \,+\,    \frac{2 }{c} \eta_1^2  +\frac{L_n^2}{2}  \eta_2^2.
\end{equation}

{\bf{Step 4}.} With the notation from the beginning of Appendix \ref{sec:aux_lemmas} and that of the proof of Theorem \ref{thm:main tv}, let $\Omega=\mathcal{E}_1 \cap \mathcal{E}_2\cap\Omega_1\cap \Omega_x\cap\Omega_\delta\cap\Omega_\epsilon\cap \Omega(\eta_2) \cap \Omega_y(\eta_1)$. Also let $\widetilde{\eta} = 2 \eta_3$ where  
\[
\eta_3^2 \,:=\, \frac{2 }{c}\eta_1^2 +\frac{L_n^2}{2} \eta_2^2 \,+\,  \frac{64 C_5}{c^2 n} +   \frac{4}{c}\lambda J_k( f^* ) ,
\]
with $\eta_1$ and $\eta_1$ to be chosen later and where 
\begin{equation}
    \label{eqn:lamb}
     \lambda\,:=\, \frac{c}{ 2J_k(f^*)  }\left[ \frac{2 }{c}\eta_1^2 +\frac{L_n^2}{2} \eta_2^2 \,+\,  \frac{64 C_5}{c^2 n}  \right].
\end{equation}
 Then, as in the proof of Theorem 1
\begin{align}
\label{zero-bound-thm2}
    \mathbb{P}\Big(\left\{ \left\|\widehat g_{k,\lambda}-f^* + \mathcal{A}_{f^*}\right\|_{2}^2>\widetilde{\eta}^2\right\}\Big)
    \le&\mathbb{P}\Big(\left\{ \left\|\widehat g_{k,\lambda}-f^* + \mathcal{A}_{f^*}\right\|_{2}^2>\widetilde{\eta}^2\right\}\cap \Omega\Big)+  \mathbb{P}(\Omega^c) \nonumber
    \\
    =&\mathbb{P}\Big(\left\{ \left\|\widehat g_{k,\lambda}-f^* + \mathcal{A}_{f^*}\right\|_{2}^2>\widetilde{\eta}^2\right\}\cap  \mathcal{E}_3\cap \Omega\Big)\nonumber
    \\
    +&\mathbb{P}\Big(\left\{ \left\|\widehat g_{k,\lambda}-f^* + \mathcal{A}_{f^*}\right\|_{2}^2>\widetilde{\eta}^2\right\}\cap  \mathcal{E}_3^c\cap\Omega\Big)+\mathbb{P}(\Omega^c),
\end{align}
and we proceed to bound the right hand side of (\ref{zero-bound-thm2}).

{\bf{Step 5}.} For the first term, we observe that 
\begin{align}
\label{eqn:702}
    &\mathbb{P}\left(\left\{\left\|\widehat g_{k,\lambda}-f^* + \mathcal{A}_{f^*}\right\|_{2}^2>\widetilde{\eta}^2\right\}\cap \mathcal{E}_3\cap \Omega\right)\nonumber
    \\
    \le&
    \mathbb{P}\left(\left\{\exists \Delta_f \in \mathcal{F} \  \text{with}\ f\in\Lambda :\left\|\Delta_f\right\|_{2}^2=\widetilde{\eta}^2\right\} \cap \mathcal{E}_3\cap \Omega\right)\nonumber
    \\
    \le&\mathbb{P}\Big(\Big\{\exists \Delta_f \in \mathcal{F} \  \text{with}\ f\in\Lambda  : T_1(f)+T_2(f)> \tilde{\eta}^2 \cdot \frac{c}{16} ,\nonumber
    \\
      & \text { and } \  \left\|\Delta_f\right\|_{2}^2\le\widetilde{\eta}^2\Big\} \cap \mathcal{E}_3\cap \Omega\Big)\nonumber
    \\
    \le& A_{1}+A_{2},
\end{align}
where 
$$
A_j:=\mathbb{P}\left(\left\{\sup_{ f\,:\,\Delta_f \in \mathcal{F}, \left\|\Delta_f\right\|_{2}^2 \leq \widetilde{\eta}^2} T_{j}(f)>\widetilde{\eta}^2 \cdot \frac{c}{32} \right\} \cap \Omega\right),
$$
for $j=1,2$. 

Then, as in the proof of Theorem \ref{thm:main tv}, we can obtain that 
\begin{equation}
    \label{eqn:aj1}
    A_j \,\leq\, \frac{\varepsilon}{8}
\end{equation}
for $j=1,2$ for  choices of $\eta_1$ and $\eta_2$ satisfying
\begin{equation}
    \label{eqn:choices}
    \eta_1 = C^{(1)}  \sqrt{\frac{\log n}{n}}  \,\,\,\text{and}\,\,\, \eta_2 =   C^{(2)}\log(n)^{\frac{2k}{2k+1}}(nm)^{-\frac{k}{2k+1}},
\end{equation}
with $C^{(1)}$ and $C^{(2)}$ large enough constants.

{\bf{Step 6}.} 
In this step we analyze  the second term in (\ref{eqn:702}). For this, by definition of the event $\mathcal{E}_3$ we obtain that
\begin{align*}
    &\mathbb{P}\Big(\left\{ \left\|\widehat g_{k,\lambda}-f^* + \mathcal{A}_{f^*}\right\|_{2}^2>\widetilde{\eta}^2\right\}\cap  \mathcal{E}_3^c\cap\Omega\Big)\\
    \le& \mathbb{P}\left(\left\{\sup _{f \in \Lambda:\left\|\Delta_f\right\|_{2} \leq  \widetilde\eta^2 }J_k(f)> 5 J_k\left(f^*\right)\right\}\cap \Omega\right)\\
    \leq&  \mathbb{P}\Big( \Big\{ \sup _{f\in \Lambda:\left\|\Delta_f\right\|_{2} \leq \widetilde\eta}J_{k}(f-f^*)\ge4J_{k}(f^*)\Big\}\cap \Omega\Big)\\
    \leq &   \mathbb{P}\Big( \Big\{ \sup _{f\in \Lambda:\left\|\Delta_f\right\|_{2} \leq \widetilde\eta, J_{k}(f)\le5J_{k}(f^*))}J_{k}(f-f^*)\ge4J_{k}(f^*)\Big\}\cap \Omega\Big)
\end{align*}
where the second  inequality follows from the inequality   $J_{k}(f-f^*)\ge J_{k}(f)-J_{k}(f^*)$, and the third from {\bf{Step 2}}.

Now notice that $\Omega \subset \mathcal{E}_1 \cap \mathcal{E}_2 \cap \mathcal{E}_4$. From this, using the study conducted in
{\bf{Step 3}}, we obtain that 
\begin{align*}
    &\mathbb{P}\Big( \Big\{ \sup _{f\in \Lambda:\left\|\Delta_f\right\|_{2} \leq \widetilde\eta, J_{k}(f)\le5J_{k}(f^*))}J_{k}(f-f^*)\ge4J_{k}(f^*)\Big\}\cap \Omega\Big)
    \\
= &
    \mathbb{P}\Big( \Big\{ \sup _{f\in \Lambda:\left\|\Delta_f\right\|_{2} \leq \widetilde\eta, J_{k}(f)\le5J_{k}(f^*),\vert\vert \Delta_f\vert\vert_{\infty}\le L_n}J_{k}(f-f^*)\ge4J_{k}(f^*))\Big\}\cap \Omega\Big).
\end{align*}
However, noticing that if $\Omega$ holds then from (\ref{eqn:700}), for any $f\in \Lambda$, it holds that 
$$
J_{k}(f) \leq \frac{4(T_1(f)+T_2(f))}{c\lambda}+J_{k}(f^*) \,+\,  \frac{64 C_5}{c^2 n  \lambda }  \,+\,\frac{2\eta_1^2}{c \lambda}+\frac{L_n^2\eta_2^2}{2\lambda}  .
$$
Hence, 
\begin{equation}
    \label{eqn:main_eq4}
    J_k(f-f^*) \,\leq\,  J_k(f)  +   J_k(f^*) \,\leq\, \frac{4(T_1(f)+T_2(f))}{c\lambda}+2J_{k}(f^*) \,+\,  \frac{64 C_5}{c^2 n  \lambda }  \,+\,\frac{2\eta_1^2}{c \lambda}+\frac{L_n^2\eta_2^2}{2\lambda}. 
\end{equation}
Therefore, from our choice of $\lambda$ in (\ref{eqn:lamb}), 
\begin{align*}
    &\mathbb{P}\Big( \Big\{ \sup _{f\in \Lambda:\left\|\Delta_f\right\|_{2} \leq \widetilde\eta, J_{k}(f)\le5J_{k}(f^*))}J_{k}(f-f^*)\ge4J_{k}(f^*)\Big\}\cap \Omega\Big)
    \\
\leq  &
    \mathbb{P}\Big( \Big\{ \sup _{f\in \Lambda:\left\|\Delta_f\right\|_{2} \leq \widetilde\eta, J_{k}(f)\le5J_{k}(f^*),\vert\vert \Delta_f\vert\vert_{\infty}\le L_n}  T_1(f) + T_2(f)   \ge  \frac{c}{4}\left[\frac{64 C_5}{c^2 n   }  \,+\,\frac{2\eta_1^2}{c }+\frac{L_n^2\eta_2^2}{2}\right]  \Big\}\cap \Omega\Big)\\
    \leq  &
    \mathbb{P}\Big( \Big\{ \sup _{f\,:\,\Delta_f \in \mathcal{F}, \left\|\Delta_f\right\|_{2} \leq \widetilde{\eta}}  T_1(f) + T_2(f)   \ge  \frac{c}{4}\left[\frac{64 C_5}{c^2 n   }  \,+\,\frac{2\eta_1^2}{c }+\frac{L_n^2\eta_2^2}{2}\right]  \Big\}\cap \Omega\Big)\\
    \leq & B_1+B_2,
\end{align*}
where 
\[
B_j \,=\,     \mathbb{P}\Big( \Big\{ \sup _{f\,:\,\Delta_f \in \mathcal{F}, \left\|\Delta_f\right\|_{2} \leq \widetilde{\eta}}  T_j(f)  \ge  \frac{c}{8}\left[\frac{64 C_5}{c^2 n   }  \,+\,\frac{2\eta_1^2}{c }+\frac{L_n^2\eta_2^2}{2}\right]  \Big\}\cap \Omega\Big)
\]
Then, as in the proof of Theorem \ref{thm:main tv} where we provided upper bounds for the corresponding $A_1$ and $A_2$, we can obtain that 
\begin{equation}
    \label{eqn:aj}
    B_j \,\leq\, \frac{\varepsilon}{8}
\end{equation}
for $j=1,2$ for $\eta_1$ and $\eta_2$ as in (\ref{eqn:choices}) by making $C^{(1)}$ and $C^{(2)}$ large enough.

{\bf{Step 7}.} 

Finally, as in the proof of Theorem \ref{thm:main tv}, we can arrive at  $\mathbb{P}(\Omega^c) \leq\varepsilon/3 $. Then,  from (\ref{zero-bound-thm2}), and {\bf{Steps 5 and 6}},  and as in {\bf{Step 9}} of the proof of Theorem \ref{thm:main tv},  it follows that  
\[
 \mathbb{P}\Big(\left\{\|  f_{k,V_n}-f^* \|_2^2  \geq \tilde{C}\widetilde{\eta}^2\right\}\Big)  \,\leq\, \varepsilon,
\]
for some constant $\tilde{C}>0$ and the claim follows.

\end{proof}
\newpage

\section{Proof of Theorem \ref{Cest-d>1}}\label{proofT3}
\begin{proof}
The proof of Theorem 3 is divided into two main steps. Throughout the proof we make use of the abuse of notation for vectors in $\mathbb{R}^{\sum_{i=1}^nm_i}$ as explained in Section \ref{sec:notation} in Chapter 1.
\\
{\bf{Step A.}} In this step the assertion of Equation (\ref{eqn:e41}) in Chapter 2 is shown. First 
let $\varepsilon>0.$ By Assumption 1{\bf{f}} we have that
\begin{align*}
    \frac{1}{cnm}\sum_{i=1}^n \sum_{j=1}^{m_i}\left((\widehat{\theta}_{V_n})_{i, j}-\theta_{i, j}^*\right)^2\ge&\frac{1}{n}\sum_{i=1}^n \frac{1}{m_i}\sum_{j=1}^{m_i}\left((\widehat{\theta}_{V_n})_{i, j}-\theta_{i, j}^*\right)^2.
\end{align*}
It suffices to show that for some sufficiently large constant $C_{ d}$ only depending on  $d$, it holds that
$$
\mathbb{P}\left(\left\|\widehat{\theta}_{V_n}-\theta^*\right\|_2^2>C_{ d}\left\{m K^{2}+ K^{2} \log^3 (n m)+K \log ^{\frac{14}{4}}(n m)+ K^2 \log^{\frac{1}{2}} (n m) V_n\right\}\right) \leq \varepsilon,
$$
where 
$
\left\|\widehat{\theta}_{V_n}-\theta^*\right\|_2^2=\sum_{i=1}^n \sum_{j=1}^{m_i}\left((\widehat{\theta}_{V_n})_{i, j}-\theta_{i, j}^*\right)^2 .
$
Let
\begin{equation}
\label{eta-def}
\eta^2=C_{ d}\left\{m K^{2}+ K^{2} \log^3 (n m)+K \log ^{\frac{14}{4}}(n m)+ K^2 \log^{\frac{1}{2}} (n m) V_n\right\}.
\end{equation}\break
{\bf{Step 0A.}}
To proceed denote 
$$
\Theta=\left\{\theta \in \mathbb{R}^{\sum_{i=1}^n m_i}: \theta=t\widehat{\theta}_{V_n}+(1-t) \theta^*, t \in[0,1]\right\},
$$
and for any $\theta_1,\theta_2\in\mathbb{R}^{\sum_{i=1}^nm_i}$ consider
$\vert\vert \theta_1-\theta_2\vert\vert_{nm}^2=\frac{1}{n}\sum_{i=1}^n\frac{1}{m_i}\sum_{j=1}^{m_i}((\theta_1)_{i,j}-(\theta_2)_{i,j})^2.$
Observe that any $\theta=t\widehat{\theta}_{V_n}+(1-t)\theta^*\in \Theta$ satisfies,
\begin{equation}
    \label{aux-eq-T3-Vn-prop}
    \vert\vert \nabla_{G_K}\theta\vert\vert_1\le t\vert\vert \nabla_{G_K}\widehat{\theta}_{V_n}\vert\vert_1+ (1-t)\vert\vert \nabla_{G_K}\theta^*\vert\vert_1\le V_n,
\end{equation}
where the last inequality is followed by the constraint assumption. Using Inequality (\ref{aux-eq-T3-Vn-prop}) and the constraint assumption,
\begin{equation}
    \label{aux-eq-T3-Vn-prop-1}
    \vert\vert \nabla_{G_K}(\theta-\theta^*)\vert\vert_1\le \vert\vert \nabla_{G_K}\theta\vert\vert_1+ \vert\vert \nabla_{G_K}\theta^*\vert\vert_1\le2V_n.
\end{equation}
By convexity of the constrain, the objective functions and the minimizer property, for any $\theta \in \Theta$,
\begin{align}
\label{aux-eq-T3-0}
    \vert\vert y_{i,j}-\theta\vert\vert_{nm}^2\le\vert\vert y_{i,j}-\theta^*\vert\vert_{nm}^2.
\end{align}
Expanding in Inequality (\ref{aux-eq-T3-0}) leads to
\begin{align}
\label{aux-eq-T3-3}
\frac{1}{2}\left\|\theta-\theta^*\right\|_{nm}^2 \leq& \frac{1}{n}\sum_{i=1}^n \frac{1}{m_i}\sum_{j=1}^{m_i}\left(\theta_{i, j}-\theta_{i, j}^*\right) \epsilon_{i,j}+\frac{1}{n}\sum_{i=1}^n \frac{1}{m_i}\sum_{j=1}^{m_i}\left(\theta_{i, j}-\theta_{i, j}^*\right) \delta_i\left(x_{i,j}\right).
\end{align}
Using Assumption 1{\bf{f}} we have that
\begin{equation}
\label{aux-eq-T3-1}
    \frac{1}{Cnm}\left\|\theta-\theta^*\right\|_{2}^2\le\left\|\theta-\theta^*\right\|_{nm}^2,
\end{equation}
where $C$ is positive constants from Assumption 1{\bf{f}}. From Inequalities (\ref{aux-eq-T3-3}) and (\ref{aux-eq-T3-1})
\begin{equation}
    \label{aux-eq-T3-4}
     \frac{1}{2}\left\|\theta-\theta^*\right\|_{2}^2 \leq\sum_{i=1}^n \sum_{j=1}^{m_i}\frac{Cm}{m_i}\left(\theta_{i, j}-\theta_{i, j}^*\right) \epsilon_{i,j}+\sum_{i=1}^n \sum_{j=1}^{m_i} \frac{Cm}{m_i}\left(\theta_{i, j}-\theta_{i, j}^*\right) \delta_i\left(x_{i,j}\right)
\end{equation}
From now along the proof we define 
\begin{equation}
\label{Tildevariables}
    \widetilde{\epsilon}_{i,j}=\frac{Cm}{m_i}\epsilon_{i,j} \ \text{and}\  \widetilde{\delta}_{i}=\frac{Cm}{m_i}\delta_{i}.
\end{equation}
 Some observations concerning these random variables are presented. First, by Assumption 1{\bf{d}} and {\bf{f}} we have that $\widetilde{\epsilon}_{i,j}$ are sub-Gaussians of parameter $\frac{C^2}{c^2}\varpi_\epsilon^2$. Similarly using Assumption 1{\bf{d}} and {\bf{f}} we have that $\widetilde{\delta}_{i}$ are sub-Gaussians of parameter $\frac{C^2}{c^2}\varpi_\delta^2$.
To conclude {\bf{Step 0A}} observe the following.
If $\left\|\widehat{\theta}_{V_n}-\theta^*\right\|_2>\eta$, then $\widetilde{\theta}=\frac{\eta}{\left\|\widehat{\theta}_{V_n}-\theta^*\right\|_2}\widehat{\theta}_{V_n}+\Big(1-\frac{\eta}{\left\|\widehat{\theta}_{V_n}-\theta^*\right\|_2}\Big)\theta^* \in \Theta$ satisfies
$$
\left\|\widetilde{\theta}-\theta^*\right\|_2=\eta
.$$
Therefore,
\begin{align*}
    \frac{\eta^2}{2}\le\sum_{i=1}^n \sum_{j=1}^{m_i}\left(\widetilde{\theta}_{i, j}-\theta_{i, j}^*\right) \widetilde\epsilon_{i,j}+\sum_{i=1}^n \sum_{j=1}^{m_i}\left(\widetilde{\theta}_{i, j}-\theta_{i, j}^*\right) \widetilde\delta_i\left(x_{i,j}\right),
\end{align*}
and in consequence, the event $\left\{\left\|\widehat{\theta}_{V_n}-\theta^*\right\|_2>\eta\right\},$ satisfies
{\footnotesize{\begin{align*}
&\left\{\left\|\widehat{\theta}_{V_n}-\theta^*\right\|_2>\eta\right\} \\
& \subset\left\{\exists \widetilde{\theta} \in \Theta,\left\|\widetilde{\theta}-\theta^*\right\|_2=\eta, \sum_{i=1}^n \sum_{j=1}^{m_i}\left(\widetilde{\theta}_{i, j}-\theta_{i, j}^*\right) \widetilde\epsilon_{i,j}+\sum_{i=1}^n \sum_{j=1}^{m_i}\left(\widetilde{\theta}_{i, j}-\theta_{i, j}^*\right) \widetilde\delta_i\left(x_{i,j}\right) \geq \frac{\eta^2}{2}\right\} \\
& \subset\left\{\sup _{\|\Delta\|_2 \leq \eta, \vert\vert \nabla_{G_K}(\Delta)\vert\vert_1\le \vert\vert \nabla_{G_K}\theta\vert\vert_1+ \vert\vert \nabla_{G_K}\theta^*\vert\vert_1\le2V_n} \sum_{i=1}^n \sum_{j=1}^{m_i} \Delta_{i,j} \widetilde\epsilon_{i,j}+\sum_{i=1}^n \sum_{j=1}^{m_i} \Delta_{i,j} \widetilde\delta_i\left(x_{i,j}\right) \geq \frac{\eta^2}{2}\right\} ,
\end{align*}}}
where $\Delta=\theta-\theta^*.$ The first contention is followed by Inequality (\ref{aux-eq-T3-4}) and Equation (\ref{Tildevariables}). The second one is due to Inequality (\ref{aux-eq-T3-Vn-prop-1}).
It follows that,
{\footnotesize{\begin{align}
\label{theequation-Thm3}
 &\mathbb{P}\left\{\left\|\widehat{\theta}_{V_n}-\theta^*\right\|_2>\eta\right\} \nonumber
 \\
\leq& \mathbb{P}\left\{\sup _{\|\Delta\|_2 \leq \eta,\vert\vert \nabla_{G_K}(\Delta)\vert\vert_1\le \vert\vert \nabla_{G_K}\theta\vert\vert_1+ \vert\vert \nabla_{G_K}\theta^*\vert\vert_1\le2V_n} \sum_{i=1}^n \sum_{j=1}^{m_i} \Delta_{i,j} \widetilde\epsilon_{i,j}+\sum_{i=1}^n \sum_{j=1}^{m_i} \Delta_{i,j} \widetilde\delta_i\left(x_{i,j}\right) \geq \frac{\eta^2}{2}\right\}\nonumber \\
 \leq& \mathbb{P}\left\{\sup _{\|\Delta\|_2 \leq \eta,\vert\vert \nabla_{G_K}(\Delta)\vert\vert_1\le \vert\vert \nabla_{G_K}\theta\vert\vert_1+ \vert\vert \nabla_{G_K}\theta^*\vert\vert_1\le2V_n} \sum_{i=1}^n \sum_{j=1}^{m_i} \Delta_{i,j} \widetilde\epsilon_{i,j} \geq \frac{\eta^2}{4}\right\} \nonumber\\
+&\mathbb{P}\left\{\sup _{\|\Delta\|_2 \leq \eta,\vert\vert \nabla_{G_K}(\Delta)\vert\vert_1\le2V_n} \sum_{i=1}^n \sum_{j=1}^{m_i} \Delta_{i,j}\widetilde \delta_i\left(x_{i,j}\right) \geq \frac{\eta^2}{4}\right\} \text {. }
\end{align}}}
{\bf{Step 1A.}} In this step the term $$\mathbb{P}\left\{\sup _{\|\Delta\|_2 \leq \eta,\vert\vert \nabla_{G_K}(\Delta)\vert\vert_1\le \vert\vert \nabla_{G_K}\theta\vert\vert_1+ \vert\vert \nabla_{G_K}\theta^*\vert\vert_1\le2V_n} \sum_{i=1}^n \sum_{j=1}^{m_i} \Delta_{i,j} \widetilde\epsilon_{i,j} \geq \frac{\eta^2}{4}\right\},$$ in Inequality (\ref{theequation-Thm3}) is analyzed.
From Lemma \ref{lemma11-oscar} and the observation after Equation (\ref{Tildevariables}), with a probability of at least 
{\footnotesize{$$1-2 \frac{1}{\sqrt{\log(nm)}}-\frac{4}{C_{\beta}}C\log(n)nm\exp(-\frac{\widetilde{C}}{24}\frac{K}{\log(n)})-3\frac{1}{n}-C_2\log(n)\frac{nm}{K}\exp\Big(-\frac{C_{\beta}}{48} \widetilde{b}_2\frac{K}{\log(n)}\Big)-\frac{1}{nm},$$}} it is satisfied that
\begin{align*}
& \sup _{\|\Delta\|_2 \leq \eta,\left\|\nabla_{G_K} \Delta\right\|_1 \leq 2 V_n} \sum_{i=1}^n \sum_{j=1}^{m_i} \Delta_{i,j} \widetilde\epsilon_{i,j} \leq C_3\varpi_\epsilon K\log^{\frac{1}{4}}(nm)\eta+ C_4\varpi_\epsilon K \log (n m)  V_n \leq \frac{\eta^2}{ 4},
\end{align*}
where $C_2>0$ is a constant depending on $d, C_\beta$ and $v_{\max }$. Moreover, $\widetilde{b}_2>0$ is a constant depending on $d, v_{\min }$, and $v_{\max }$. Furthermore, $\widetilde{C}=\frac{C_\beta c_1 c}{c_2 C}$ with $c_1$ and $c_2$ positive constants. Here $v_{\max }, v_{\min }$, $C, C_\beta$ and $c$ are from Assumption 1. The positive constants $C_3$ and $C_4$ only depend on $d$ and the last inequality follows from the choice of $\eta$ in Equation \ref{eta-def} for sufficiently large $C_d$. Hence,
{{\begin{align}
    &\mathbb{P}\left\{\sup _{\|\Delta\|_2 \leq \eta,\left\|\nabla_{G_K} \Delta\right\|_1 \leq 2 V_n} \sum_{i=1}^n \sum_{j=1}^{m_i} \Delta_{i,j} \widetilde\epsilon_{i,j} \geq \frac{\eta^2}{4}\right\} \nonumber
    \\
    \leq &2 \frac{1}{\sqrt{\log(nm)}}+\frac{4}{C_{\beta}}C\log(n)nm\exp(-\frac{\widetilde{C}}{24}\frac{K}{\log(n)})\nonumber
    \\
    -&3\frac{1}{n}+C_2\log(n)\frac{nm}{K}\exp\Big(-\frac{C_{\beta}}{48} \widetilde{b}_2\frac{K}{\log(n)}\Big)+\frac{1}{nm}\nonumber
    \\
    <&\frac{\varepsilon}{2}.\label{eps-bound}
\end{align}}}
\\
\\
{\bf{Step 2A}}. In this step the term $$\mathbb{P}\left\{\sup _{\|\Delta\|_2 \leq \eta,\vert\vert \nabla_{G_K}(\Delta)\vert\vert_1\le2V_n} \sum_{i=1}^n \sum_{j=1}^{m_i} \Delta_{i,j} \widetilde\delta_i\left(x_{i,j}\right) \geq \frac{\eta^2}{4}\right\},$$ in Inequality (\ref{theequation-Thm3}) is analyzed. To this end, we consider $\left\{\mathcal{I}_r\right\}_{r=1}^{N^d}$ the collection of the cells in the lattice graph $G_{\text {lat }}$ as explained in Section \ref{Lat-section}. The number of the nodes of this graph is $N^d$ where 
\begin{equation}
\label{choiceN-T3}
N=\left\lceil\frac{3 \sqrt{d}  (nm)^{1 / d}}{ aK^{1 / d}}\right\rceil,
\end{equation}
where  $a=1 /\left(  4Cc_2c_{2,d}v_{\text{max} }\right)^{1 / d}$ with $c_{2,d}>0$ is a constant depending on $d.$ The remaining constants found in the definition of 
$a$ were initially introduced in \textbf{Step 1A}.
For any $\Delta \in \mathbb{R}^{\sum_{i=1}^n m_i}$ such that $\|\Delta\|_2 \leq \eta,\vert\vert \nabla_{G_K}(\Delta)\vert\vert_1\le 2V_n,$ the vector $\Delta^I \in \mathbb{R}^{N^d}$, defined in Definition \ref{def-1-E}, has a unique value in each of the cell $\mathcal{I}_r$ in $G_{\text {lat }}$. Denote $\widetilde{\Delta}:[0,1]^d \rightarrow \mathbb{R}$ to be such that
\begin{equation}
\label{Deltatildefunct-def}
\widetilde{\Delta}(x)=\Delta_r^I \text { when } x \in \mathcal{I}_r.
\end{equation}
For convenience, denote
$$
\widetilde{\Delta}\left(\mathcal{I}_r\right)=\widetilde{\Delta}(x)=\Delta_r^I \text { when } x \in \mathcal{I}_r.
$$
Thus
$$
\widetilde{\Delta}\left(x_{i,j}\right)=\Delta_r^I=\left(\Delta_I\right)_{ i, j } \text { when } x_{i,j} \in \mathcal{I}_r,
$$
where the definition of $\Delta_I$ can be found in Definition \ref{def-2-E}.
Then for any $\Delta \in \mathbb{R}^{\sum_{i=1}^n m_i}$ such that $\|\Delta\|_2 \leq \eta,\left\|\nabla_{G_K} \Delta\right\|_1 \leq 2 V_n$, it holds that
\begin{align}
\sum_{i=1}^n \sum_{j=1}^{m_i} \Delta_{i,j} \widetilde\delta_i\left(x_{i,j}\right) & =\sum_{i=1}^n \sum_{j=1}^{m_i}\left\{\Delta_{i,j} \widetilde\delta_i\left(x_{i,j}\right)-\left\langle\widetilde{\Delta}, \widetilde\delta_i\right\rangle_{2}\right\}+ \sum_{i=1}^nm_i\left\langle\widetilde{\Delta}, \widetilde\delta_i\right\rangle_{2}\nonumber \\
& =\sum_{i=1}^n \sum_{j=1}^{m_i}\left\{\Delta_{i,j} \widetilde\delta_i\left(x_{i,j}\right)-\widetilde{\Delta}\left(x_{i,j}\right) \widetilde\delta_i\left(x_{i,j}\right)\right\} \nonumber\\
& +\sum_{i=1}^n \sum_{j=1}^{m_i}\left\{\widetilde{\Delta}\left(x_{i,j}\right) \widetilde\delta_i\left(x_{i,j}\right)-\left\langle\widetilde{\Delta}, \widetilde\delta_i\right\rangle_{2}\right\} \nonumber\\
& + \sum_{i=1}^nm_i\left\langle\widetilde{\Delta}, \widetilde\delta_i\right\rangle_{2}.\nonumber
\end{align}
In consequence,
\begin{align}
    &\mathbb{P}\left\{\sup _{\|\Delta\|_2 \leq \eta,\left\|\nabla_{G_K} \Delta\right\|_1 \leq 2 V_n} \sum_{i=1}^n \sum_{j=1}^{m_i} \Delta_{i,j} \widetilde\delta_i\left(x_{i,j}\right) \geq \frac{\eta^2}{4}\right\} \nonumber
    \\
    \le&\mathbb{P}\Big(\sup _{\|\Delta\|_2 \leq \eta,\left\|\nabla_{G_K} \Delta\right\|_1 \leq 2 V_n}\sum_{i=1}^n \sum_{j=1}^{m_i}\left\{\Delta_{i,j} \widetilde\delta_i\left(x_{i,j}\right)-\widetilde{\Delta}\left(x_{i,j}\right) \widetilde\delta_i\left(x_{i,j}\right)\right\} >\frac{\eta^2}{12}\Big)\label{delta-1}
    \\
    +&\mathbb{P}\Big(\sup _{\|\Delta\|_2 \leq \eta,\left\|\nabla_{G_K} \Delta\right\|_1 \leq 2 V_n}\sum_{i=1}^n \sum_{j=1}^{m_i}\left\{\widetilde{\Delta}\left(x_{i,j}\right) \widetilde\delta_i\left(x_{i,j}\right)-\left\langle\widetilde{\Delta}, \widetilde\delta_i\right\rangle_{2}\right\}>\frac{\eta^2}{12}\Big)\label{delta-2}
    \\
    +&\mathbb{P}\Big(\sup _{\|\Delta\|_2 \leq \eta,\left\|\nabla_{G_K} \Delta\right\|_1 \leq 2 V_n}\sum_{i=1}^nm_i\left\langle\widetilde{\Delta}, \widetilde\delta_i\right\rangle_{2}>\frac{\eta^2}{12}\Big)\label{delta-3}.
\end{align}
\\
{\bf{Step 2A-a.}} In this step we bound the term $$\mathbb{P}\Big(\sup _{\|\Delta\|_2 \leq \eta,\left\|\nabla_{G_K} \Delta\right\|_1 \leq 2 V_n}\sum_{i=1}^nm_i\left\langle\widetilde{\Delta}, \widetilde\delta_i\right\rangle_{2}>\frac{\eta^2}{12}\Big)$$ in Equation
(\ref{delta-3}). Observe that with probability at least $1-\frac{1}{\log(n)}$, for a positive constant $\widetilde{C}$, it holds that
\begin{align*}
\sum_{i=1}^nm_i\left\langle\widetilde{\Delta},\widetilde \delta_i\right\rangle_{2} & =\sum_{i=1}^n\left\langle\widetilde{\Delta}, m_i\widetilde\delta_i\right\rangle_{2}=\sqrt{n}\left\langle\widetilde{\Delta}, \frac{1}{\sqrt{n}} \sum_{i=1}^n m_i\widetilde\delta_i\right\rangle_{2} \\
& \le\sqrt{n}\|\widetilde{\Delta}\|_{2}\left\|\frac{1}{\sqrt{n}} \sum_{i=1}^n m_i\widetilde\delta_i\right\|_{2} \\
& =\sqrt{n}\|\widetilde{\Delta}\|_{2}\left\|\frac{1}{\sqrt{n}} \sum_{i=1}^n Cm\delta_i\right\|_{2} \\
& \leq \widetilde{C}m \sqrt{n}\log^{1/2}(n)\|\widetilde{\Delta}\|_{2},
\end{align*}
where the last inequality follows from Lemma \ref{lemma7}. Moreover the first inequality is followed by Holder’s inequality while the third equality due to Equation (\ref{Tildevariables}).  Note that since $$
N=\left\lceil\frac{3 \sqrt{d}  (nm)^{1 / d}}{ aK^{1 / d}}\right\rceil,
$$
we get that
$$
\|\widetilde{\Delta}\|_{2}^2=\frac{1}{N^d} \sum_{r=1}^{N^d} \widetilde{\Delta}^2\left(\mathcal{I}_r\right)=\frac{1}{N^d} \sum_{r=1}^{N^d}\left(\Delta_r^I\right)^2 \leq C_5\frac{K}{n m}\|\Delta\|_2^2\le C_5\frac{K}{n m} \eta^2,
$$
where the first inequality follows from Lemma \ref{lemma8}, and $C_5>0$ is a positive constant. So with probability at least $1-\frac{1}{\log(n)}$,
$$
\sum_{i=1}^nm_i\left\langle\widetilde{\Delta}, \widetilde\delta_i\right\rangle_{2} \leq C_5 \log^{1/2}(n)\sqrt{K m} \eta.
$$
Thus,
\begin{align}
\label{T3Cons-5}
    &\mathbb{P}\Big(\sup _{\|\Delta\|_2 \leq \eta,\left\|\nabla_{G_K} \Delta\right\|_1 \leq 2 V_n}\sum_{i=1}^nm_i\left\langle\widetilde{\Delta}, \widetilde\delta_i\right\rangle_{2}>\frac{\eta^2}{12}\Big)\nonumber
    \\
    \le& \mathbb{P}\Big(\sup _{\|\Delta\|_2 \leq \eta,\left\|\nabla_{G_K} \Delta\right\|_1 \leq 2 V_n}\sum_{i=1}^nm_i\left\langle\widetilde{\Delta}, \widetilde\delta_i\right\rangle_{2}>C_5 \log^{1/2}(n)\sqrt{K m} \eta\Big)\nonumber
    \\
    \le& \frac{1}{\log(n)}\le\frac{\varepsilon}{6}.
\end{align}
{\bf{Step 2A-b.}} For Equation (\ref{delta-1}), observe that for any $\Delta \in \mathbb{R}^{\sum_{i=1}^n m_i}$ such that $\|\Delta\|_2 \leq \eta,\left\|\nabla_{G_K} \Delta\right\|_1 \leq 2 V_n,$ we have that
\begin{align*}
& \sum_{i=1}^n \sum_{j=1}^{m_i}\left\{\Delta_{i,j} \widetilde\delta_i\left(x_{i,j}\right)-\widetilde{\Delta}\left(x_{i,j}\right) \widetilde\delta_i\left(x_{i,j}\right)\right\} \\
= & \sum_{i=1}^n \sum_{j=1}^{m_i}\left\{\Delta_{i,j}-\left(\Delta_I\right)_{ i, j }\right\} \widetilde\delta_i\left(x_{i,j}\right) \\
\leq & 2\max_{1\le i\le n,1\le j\le m_i}\vert\widetilde\delta_{i}(x_{i,j})\vert\left\|\nabla_{G_K} \Delta\right\|_1 \leq \frac{2C}{c}\varpi_\delta \sqrt{2\log(nm)} V_n,
\end{align*}
where the first inequality follows from  Lemma \ref{lemma8-Oscar}, and the second one by Assumption 1 together with the fact that,
the event $\Big\{ \max_{1\le i\le n,1\le j\le m_i}\vert \delta_i(x_{i,j})\vert\le\varpi_\delta \sqrt{2\log(nm)} \Big\}$ happens with probability at least $1-\frac{1}{nm}.$
Thus,
{{\begin{align}
\label{T3Cons-4}
&\mathbb{P}\Big(\sup _{\|\Delta\|_2 \leq \eta,\left\|\nabla_{G_K} \Delta\right\|_1 \leq 2 V_n}\sum_{i=1}^n \sum_{j=1}^{m_i}\left\{\Delta_{i,j} \widetilde\delta_i\left(x_{i,j}\right)-\widetilde{\Delta}\left(x_{i,j}\right) \widetilde\delta_i\left(x_{i,j}\right)\right\} >\frac{\eta^2}{12}\Big)\nonumber
    \\
\le&\mathbb{P}\Big(\sup _{\|\Delta\|_2 \leq \eta,\left\|\nabla_{G_K} \Delta\right\|_1 \leq 2 V_n}\sum_{i=1}^n \sum_{j=1}^{m_i}\left\{\Delta_{i,j} \widetilde\delta_i\left(x_{i,j}\right)-\widetilde{\Delta}\left(x_{i,j}\right) \widetilde\delta_i\left(x_{i,j}\right)\right\}\nonumber
\\
& >\frac{2C}{c}\varpi_\delta \sqrt{2\log(nm)} V_n\Big)\nonumber
    \\
    \le&
    \frac{1}{nm}\le \frac{\varepsilon}{6}.
\end{align}}}
\\
\\
{\bf{Step 2A-c}}. In this step we bound the term $$\mathbb{P}\Big(\sup _{\|\Delta\|_2 \leq \eta,\left\|\nabla_{G_K} \Delta\right\|_1 \leq 2 V_n}\sum_{i=1}^n \sum_{j=1}^{m_i}\left\{\Delta_{i,j} \widetilde\delta_i\left(x_{i,j}\right)-\widetilde{\Delta}\left(x_{i,j}\right) \widetilde\delta_i\left(x_{i,j}\right)\right\} >\frac{\eta^2}{12}\Big),$$ in Equation (\ref{delta-2}). For this, we introduce some notation. First, let $\mathcal{F}_{\text {lat }}$ denote the class of functions $f:[0,1]^d \rightarrow \mathbb{R}$ such that $f$ takes constant value in each cell $\mathcal{I}_r$ of the lattice graph $G_{\text {lat }}$. Thus $f$ is a  piecewise constant function. Denote $v(f) \in \mathbb{R}^{N^d}$ to be such that
\begin{equation}
\label{f-Flat-def}
v(f)_r=f\left(\mathfrak{c}_r\right)=f\left(\mathcal{I}_r\right),
\end{equation}
where $\mathfrak{c}_r$ is the center of $\mathcal{I}_r$. For more details, we refer to Section \ref{Lat-section}. Let $D$ be the incident matrix of $G_{\text {lat }}$. Then $\|D(v(f))\|_1$ can be viewed as the discrete total variation of $f$ and $\|v(f)\|_2$ is the discrete $\mathcal{L}_2$ norm of $f$. Denote
$$
\mathcal{F}_{\text {lat }}(A, B)=\left\{f \in \mathcal{F}_{\text {lat }}:\|v(f)\|_2 \leq A,\|D(v(f))\|_1 \leq B\right\},
$$
where $\|v(f)\|_2$ is the Euclidean norm of the vector $v(f) \in \mathbb{R}^{N^d}$. Now we are ready to continue analyzing the quantity in (\ref{delta-2}). Observe that
for any $\Delta \in \mathbb{R}^{\sum_{i=1}^n m_i}$ such that $\|\Delta\|_2 \leq \eta,\left\|\nabla_{G_K} \Delta\right\|_1 \leq 2 V_n$, note that  the function $\widetilde{\Delta}$ defined in Equation (\ref{Deltatildefunct-def}) satisfies $v(\widetilde{\Delta})=\Delta^I \in \mathbb{R}^{N^d}$. Therefore,
$$
\|v(\widetilde{\Delta})\|_2^2=\left\|\Delta^I\right\|_2^2 \leq\|\Delta\|_2^2 \leq \eta^2,
$$
where the first inequality follows from Lemma \ref{lemma8}. In addition, it holds that
$$
\|D(v(\widetilde{\Delta}))\|_1=\left\|D \Delta^I\right\|_1 \leq\left\|\nabla_{G_K} \Delta\right\|_1 \leq 2 V_n.
$$
The first inequality follows from Lemma \ref{lemma8-Oscar}. To achieve this inequality we need the following two events to hold,
$$
\mathcal{E}_3=\left\{\max _{r=1,...,N^d }|I_r| \geq  \frac{1}{2} \widetilde{b}_2 K\right\}
$$
and $\Omega$ defined as ``If $x_{i,j} \in C(x_{i,j}^{\prime})$ and $x_{s,t} \in C(x_{s,t}^{\prime})$ for $x_{i,j}^{\prime}, x_{s,t}^{\prime} \in P_{\text{lat}}(N)$ with $\left\|x_{i,j}^{\prime}-x_{s,t}^{\prime}\right\|_2 \leq$ $N^{-1}$, then $x_{i,j}$ and $x_{s,t}$ are connected in the $K$-NN graph". By Lemma \ref{Omega-Oscar-Paper}, the events happen with probability at least,
$$1-\frac{4}{C_{\beta}}\log(n)N^d\exp\Big(-\frac{C_{\beta}}{48} \widetilde{b}_2\frac{K}{\log(n)}\Big)-\frac{1}{n},$$
and
$$1-\frac{4}{C_{\beta}}C\log(n)nm\exp(-\frac{\widetilde{C}}{24}\frac{K}{\log(n)})-\frac{1}{n},$$
Respectively. All the constants above were previously introduced in {\bf{Step 1A}}. Then, for any $\Delta \in \mathbb{R}^{\sum_{i=1}^n m_i}$ such that $\|\Delta\|_2 \leq \eta,\left\|\nabla_{G_K} \Delta\right\|_1 \leq 2 V_n,$ we have that
$
\widetilde{\Delta} \in \mathcal{F}_{\text {lat }}(\eta, 2 V_n),
$
under the event event $\mathcal{E}_3\cap\Omega.$
Now, we perform an analysis that combines what we have done in the proof of Theorem 1 and 
the proof of Lemma \ref{Omega-Oscar-Paper}. Specifically,  observe that the events $\Omega_x=\bigcap_{i=1}^{n}\Big\{x_{i,j}=x_{i,j}^* \forall j=1,...,m_i\Big\}$
and $\Omega_\delta=\bigcap_{i=1}^{n}\Big\{\delta_{i}=x_{i}^* \Big\}$ hold with probability at least $1-\frac{2}{\log(n)}$.
Here $L=\frac{2}{C_{\beta}}\log(n)$ is the choice of the size of each block, as in Section \ref{BetaMsection}. We also use the fact that $c_3\frac{n}{L}\le \vert \mathcal{J}_{e,l}\vert,\vert\mathcal{J}_{o,l}\vert\le c_2\frac{n}{L}$ for some positive constants $c_3$ and $c_2$. 
Then, we have that
{\footnotesize{\begin{align}
\label{T3Cons-0}
    &\mathbb{P}\Big(\sup _{\|\Delta\|_2 \leq \eta,\left\|\nabla_{G_K} \Delta\right\|_1 \leq 2 V_n}\sum_{i=1}^n \sum_{j=1}^{m_i}\left\{\widetilde{\Delta}\left(x_{i,j}\right) \widetilde\delta_i\left(x_{i,j}\right)-\left\langle\widetilde{\Delta}, \widetilde\delta_i\right\rangle_{2}\right\}>\frac{\eta^2}{12}\Big)\nonumber
    \\
    \le&\mathbb{P}\Big(\sup _{f \in \mathcal{F}_{\text {lat }}(\eta, 2 V_n)} \sum_{i=1}^n \sum_{j=1}^{m_i}\left\{f\left(x_{i,j}\right) \widetilde\delta_i\left(x_{i,j}\right)-\left\langle f, \widetilde\delta_i\right\rangle_{2}\right\}>\frac{\eta^2}{12}\Big)+\frac{\varepsilon}{24}\nonumber
    \\
    \le& \mathbb{P}\Big(\sup _{f \in \mathcal{F}_{\text {lat }}(\eta, 2 V_n)} \sum_{i=1}^n \sum_{j=1}^{m_i}\left\{f\left(x_{i,j}^*\right) \widetilde\delta_i^*\left(x_{i,j}\right)-\left\langle f, \widetilde\delta_i^*\right\rangle_{2}\right\}>\frac{\eta^2}{12}\Big)+\frac{\varepsilon}{24}+\frac{\varepsilon}{24}\nonumber
    \\
    =& A_1+\frac{\varepsilon}{12},
\end{align}}}
where,
\begin{equation}
    \label{tildevariables1}
    \widetilde\delta_i^*=\frac{Cm}{m_i}\delta_i^*,\ \text{for} \ i\in[n].
\end{equation}
The first inequality is followed by using that $\Omega\cap\mathcal{E}_3$ holds with high probability while the second one due to the event $\Omega_x\cap\Omega_\delta.$
We proceed to bound the term $A_1$. Using a union bound argument,
\begin{align}
\label{last-step-t3-aux1}
    &\mathbb{P}\Big(\sup _{f \in \mathcal{F}_{\text {lat }}(\eta, 2 V_n)} \sum_{i=1}^n \sum_{j=1}^{m_i}\left\{f\left(x_{i,j}^*\right) \widetilde\delta_i^*\left(x_{i,j}\right)-\left\langle f, \widetilde\delta_i^*\right\rangle_{2}\right\}>\frac{\eta^2}{12}\Big)\nonumber
    \\
    \le&
    \sum_{l=1}^L\mathbb{P}\Big(\sup _{f \in \mathcal{F}_{\text {lat }}(\eta, 2 V_n)} \sum_{i\in \mathcal{J}_{e,l}} \sum_{j=1}^{m_i}\left\{f\left(x_{i,j}^*\right) \widetilde\delta_i^*\left(x_{i,j}\right)-\left\langle f, \widetilde\delta_i^*\right\rangle_{2}\right\}>\frac{\eta^2}{24L}\Big)\nonumber
    \\
    +&\sum_{l=1}^L\mathbb{P}\Big(\sup _{f \in \mathcal{F}_{\text {lat }}(\eta, 2 V_n)} \sum_{i\in \mathcal{J}_{o,l}} \sum_{j=1}^{m_i}\left\{f\left(x_{i,j}^*\right) \widetilde\delta_i^*\left(x_{i,j}\right)-\left\langle f, \widetilde\delta_i^*\right\rangle_{2}\right\}>\frac{\eta^2}{24L}\Big),
\end{align}
By Markov's inequality we have that
\begin{align}
    \label{T3Cons-1}&\mathbb{P}\Big(\sup _{f \in \mathcal{F}_{\text {lat }}(\eta, 2 V_n)} \sum_{i=1}^n \sum_{j=1}^{m_i}\left\{f\left(x_{i,j}^*\right) \widetilde\delta_i^*\left(x_{i,j}\right)-\left\langle f, \widetilde\delta_i^*\right\rangle_{2}\right\}>\frac{\eta^2}{12}L\Big)\nonumber
    \\
    \le&
    \sum_{l=1}^L\frac{24L\mathbb{E}\Big(\sup _{f \in \mathcal{F}_{\text {lat }}(\eta, 2 V_n)} \sum_{i\in \mathcal{J}_{e,l}} \sum_{j=1}^{m_i}\left\{f(x_{i,j}^*) \widetilde\delta_i^*(x_{i,j}^*)-\left\langle f, \widetilde\delta_i^*\right\rangle_{2}\right\}\Big)}{\eta^2}\nonumber
    \\
    +&\sum_{l=1}^L\frac{24L\mathbb{E}\Big(\sup _{f \in \mathcal{F}_{\text {lat }}(\eta, 2 V_n)} \sum_{i\in \mathcal{J}_{o,l}} \sum_{j=1}^{m_i}\left\{f(x_{i,j}^*) \widetilde\delta_i^*(x_{i,j}^*)-\left\langle f, \widetilde\delta_i^*\right\rangle_{2}\right\}\Big)}{\eta^2}.
\end{align}
Therefore, it is enough to control the terms $$\mathbb{E}\Big(\sup _{f \in \mathcal{F}_{\text {lat }}(\eta, 2 V_n)} \sum_{i\in \mathcal{J}_{e,l}} \sum_{j=1}^{m_i}\left\{f(x_{i,j}^*) \widetilde\delta_i^*(x_{i,j}^*)-\left\langle f, \widetilde\delta_i^*\right\rangle_{2}\right\}\Big),$$ and using the same line of arguments the terms $$\mathbb{E}\Big(\sup _{f \in \mathcal{F}_{\text {lat }}(\eta, 2 V_n)} \sum_{i\in \mathcal{J}_{o,l}} \sum_{j=1}^{m_i}\left\{f(x_{i,j}^*) \widetilde\delta_i^*(x_{i,j}^*)-\left\langle f, \widetilde\delta_i^*\right\rangle_{2}\right\}\Big),$$ are bounded.
To this end observe that conditioning on $\left\{\delta_i^*\right\}_{i=1}^n$, the collection \break$\left\{\delta_i^*\left(x_{i,j}^*\right)\right\}_{i\in \mathcal{J}_{e,l}, 1\le j\le m_i}$ are independent. Using Equation (\ref{tildevariables1}), it follows that
\begin{align}
\label{Expectedvaluewithrespectx-t3}
&  \mathbb{E}_{x,\omega}\left(\sup _{f \in \mathcal{F}_{\text {lat }}(\eta, 2 V_n)} \sum_{i\in \mathcal{J}_{e,l}} \sum_{j=1}^{m_i}\left\{f\left(x_{i,j}^*\right) \widetilde\delta_i^*\left(x_{i,j}^*\right)-\left\langle f, \widetilde\delta_i\right\rangle_{2}\right\}\right)\nonumber \\
\leq & 2 \mathbb{E}_{x,\omega}\left(\sup _{f \in \mathcal{F}_{\text {lat }}(\eta, 2 V_n)} \sum_{i\in \mathcal{J}_{e,l}} \sum_{j=1}^{m_i} \omega_{i,j} f\left(x_{i,j}^*\right) \widetilde\delta_i^*\left(x_{i,j}^*\right)\right),
\end{align}
where the inequality follows from a symmetrization argument, as in Lemma \ref{lemma6} and Lemma \ref{lemma2}. Here, $\left\{\omega_{i,j}\right\}_{i\in \mathcal{J}_{e,l}, 1\le j\le m_i}$ are i.i.d. Rademacher random variables, independent of $$\left(\left\{x_{i,j}^*\right\}_{i\in \mathcal{J}_{e,l}, 1\le j\le m_i},\left\{\delta_{i}^*\right\}_{i\in \mathcal{J}_{e,l}}\right).$$ 
Now, consider the event $\mathcal{E}_0=\{\sup _{1 \leq r \leq N^d}\left|\left\{x_{i,j}^*\right\}_{i\in\mathcal{J}_{e,l}, 1\le j\le m_i} \cap \mathcal{I}_r\right|\le \frac{3}{2}  \widetilde{b}_1 K\}$, where $\widetilde{b}_1$ is a constant that only depends on $d.$
Observe that by the proof of Lemma \ref{Omega-Oscar-Paper},
$$\mathbb{P}(\mathcal{E}_0)\ge1-\frac{4}{C_{\beta}}\log(n)N^d\exp\Big(-\frac{C_{\beta}}{48} \widetilde{b}_2\frac{K}{\log(n)}\Big).$$
Then,
\begin{align}
\label{T3Cons-3-ad3}
&\mathbb{E}_{x,\omega}\left(\sup _{f \in \mathcal{F}_{\text {lat }}(\eta, 2 V_n)} \sum_{i\in \mathcal{J}_{e,l}} \sum_{j=1}^{m_i} \omega_{i,j} f\left(x_{i,j}^*\right) \widetilde\delta_i^*\left(x_{i,j}^*\right)\right) \nonumber\\
=& \mathbb{E}_{x,\omega}\left(\Big(\sup _{f \in \mathcal{F}_{\text {lat }}(\eta, 2 V_n)} \sum_{i\in \mathcal{J}_{e,l}} \sum_{j=1}^{m_i} \omega_{i,j} f\left(x_{i,j}^*\right) \widetilde\delta_i^*\left(x_{i,j}^*\right)\Big)\mathbf{1}_{\mathcal{E}_0}\right)\nonumber
\\
+&\mathbb{E}_{x,\omega}\left(\Big(\sup _{f \in \mathcal{F}_{\text {lat }}(\eta, 2 V_n)} \sum_{i\in \mathcal{J}_{e,l}} \sum_{j=1}^{m_i} \omega_{i,j} f\left(x_{i,j}^*\right) \widetilde\delta_i^*\left(x_{i,j}^*\right)\Big)\mathbf{1}_{\mathcal{E}_0^c}\right).
\end{align}
For the term $\mathbb{E}_{x,\omega}\left(\Big(\sup _{f \in \mathcal{F}_{\text {lat }}(\eta, 2 V_n)} \sum_{i\in \mathcal{J}_{e,l}} \sum_{j=1}^{m_i} \omega_{i,j} f(x_{i,j}^*) \widetilde\delta_i^*(x_{i,j}^*)\Big)\mathbf{1}_{\mathcal{E}_0^c}\right)$. Using Equation (\ref{f-Flat-def}) observe that any $f\in\mathcal{F}_{\text {lat }}(\eta, 2 V_n) $ satisfies $\vert\vert f\vert\vert_{\infty}\le \vert\vert v(f)\vert\vert_2\le \eta$. Therefore using the fact that the sum of the supremum is smaller than the supremum of the sum together with Cauchy Schwartz inequality, the choice of $L$, Assumption 1{\bf{c}} and {\bf{f}}, we get that
\begin{align}
\label{complement_E0}
&\mathbb{E}\left(\Big(\sup _{f \in \mathcal{F}_{\text {lat }}(\eta, 2 V_n)} \sum_{i\in \mathcal{J}_{e,l}} \sum_{j=1}^{m_i} \omega_{i,j} f\left(x_{i,j}^*\right) \widetilde\delta_i^*\left(x_{i,j}^*\right)\Big)\mathbf{1}_{\mathcal{E}_0^c}\right)\nonumber
\\
\le& \eta\frac{c_2CC_{\beta}nm}{2\log(n)} \frac{4}{C_{\beta}}\log(n)N^d\exp\Big(-\frac{C_{\beta}}{48} \widetilde{b}_2\frac{K}{\log(n)}\Big).
\end{align}
Next we analyze $\mathbb{E}_{x,\omega}\left(\Big(\sup _{f \in \mathcal{F}_{\text {lat }}(\eta, 2 V_n)} \sum_{i\in \mathcal{J}_{e,l}} \sum_{j=1}^{m_i} \omega_{i,j} f\left(x_{i,j}^*\right) \widetilde\delta_i^*\left(x_{i,j}^*\right)\Big)\mathbf{1}_{\mathcal{E}_0}\right).$
Let $\zeta_{i ,j}=\omega_{i,j} \widetilde\delta_i^*\left(x_{i,j}^*\right)$.  Then for any $f \in \mathcal{F}_{\text {lat }}(\eta, 2 V_n)$, under the event $\mathcal{E}_0,$
\begin{align}
\label{equation-final-T3-aux}
 &\sum_{i\in\mathcal{J}_{e,l}} \sum_{j=1}^{m_i} \omega_{i,j} f\left(x_{i,j}^*\right) \widetilde\delta_i^*\left(x_{i,j}\right)\nonumber
 \\
&=  \sum_{i=1}^n \sum_{j=1}^{m_i} \zeta_{i ,j} f\left(x_{i,j}^*\right) \nonumber\\
& =\sum_{r=1}^{N^d} \sum_{\left\{(i ,j) \in \{\mathcal{J}_{e,l}\}\times [m_i] : x_{i,j}^* \in \mathcal{I}_r\right\}} \zeta_{i ,j} f\left(\mathcal{I}_r\right) \nonumber\\
& =\frac{3}{2}  \widetilde{b}_1 K\Big(\max_{i\in\mathcal{J}_{e,l},1\le j\le m_i}\vert \widetilde\delta_i^*(x_{i,j}^*)\vert+1\Big)\sum_{l=1}^{N^d} f\left(\mathcal{I}_r\right) \breve{\zeta}_r\nonumber
\\
&=\frac{3}{2}  \widetilde{b}_1 K\Big(\max_{i\in\mathcal{J}_{e,l},1\le j\le m_i}\vert \widetilde\delta_i^*(x_{i,j}^*)\vert+1\Big)
\langle\breve{\zeta}, v(f)\rangle,
\end{align}
where $\breve{\zeta} \in \mathbb{R}^{N^d}$ has the form
$$
\breve{\zeta}_r=\Big(\frac{3}{2}  \widetilde{b}_1 K\Big)^{-1}\Big(\max_{i\in\mathcal{J}_{e,l},1\le j\le m_i}\vert \widetilde\delta_i^*(x_{i,j}^*)\vert+1\Big)^{-1}\sum_{\left\{(i ,j)\in \mathcal{J}_{e,l}\times[m_i]: x_{i,j}^* \in \mathcal{I}_r\right\}} \zeta_{i ,j}.
$$
The third equality is obtained by the definition of $f$ in Equation (\ref{f-Flat-def}). 
Observe that $\left\{\zeta_{i ,j}\right\}_{i\in \mathcal{J}_{e,l}, 1\le j\le m_i}$ are independent conditioning on
\break
$\left\{\delta_i^*\right\}_{i\in \mathcal{J}_{e,l}, 1\le j\le m_i},\left\{x_{i,j}^*\right\}_{i\in \mathcal{J}_{e,l}, 1\le j\le m_i}$. Therefore, using basic addition and scalar multiplication properties for independent sub-Gaussian random variables we have that $\left\{\breve{\zeta}_r\right\}_{r=1}^{N^d}$ is a collection of sub-Gaussian random variables with parameters bounded uniformly by $1$. 
Let $\Pi$ denote the projection matrix onto the span of $\mathbf{1}_{\mathbb{R}^{N^d}} \in \mathbb{R}^{N^d}$. Denote by $D^{\dagger}$  the pseudoinverse
of the incidence matrix $D$.  By Holder’s inequality and triangle inequality,
\begin{align}
\label{aux-thm3-eq1}
\langle\breve{\zeta}, v(f)\rangle & =\langle(\Pi+(I-\Pi)) \breve{\zeta}, v(f)\rangle \nonumber\\
& =\langle\Pi \breve{\zeta}, v(f)\rangle+\langle(I-\Pi) \breve{\zeta}, v(f)\rangle \nonumber\\
& \le\|\Pi \breve{\zeta}\|_2\|v(f)\|_2+\left\|\left(D^{\dagger}\right)^{\top} \breve{\zeta}\right\|_{\infty}\|D(v(f))\|_1.
\end{align}
 By the proof of Theorem 2 in \cite{hutter2016optimal}, with probability at least $1-\frac{1}{(nm)^2},$
$$
\left\|\left(D^{\dagger}\right)^{\top} \breve{\zeta}\right\|_{\infty} \leq  \rho \sqrt{2 \log (e N^d(nm)^2)},$$
and
$$
 \quad\|\Pi \breve{\zeta}\|_2 \leq 2  \sqrt{2 \log (e (nm)^2)},
$$
where
$\rho=\max _{1 \leq v \leq\left|E_{\text {lat }}\right|}\left\|D_{, v}^{\dagger}\right\|_2. $ Additionally, using Lemmas \ref{H-R-prop4} and \ref{H-R-prop6} (correspondingly Propositions 4 and 6 in \cite{hutter2016optimal}), we have that $\rho\le C(d)\sqrt{\log (N)},$ where $C(d)$ is a positive constant depending on $d.$
Using our choice of $N$ in Equation (\ref{choiceN-T3}), it is satisfied that
\begin{equation}
\label{H-R-T3-1}
\left\|\left(D^{\dagger}\right)^{\top} \breve{\zeta}\right\|_{\infty} \leq C(d) \sqrt{ \log(\frac{nm}{K})\log ( \frac{(nm)^3}{K} )} ),\end{equation}
and,
\begin{equation}
\label{H-R-T3-2} \quad\|\Pi \breve{\zeta}\|_2 \leq 2  \sqrt{ \log((nm)^2)}.
\end{equation}
with probability at least $1-\frac{1}{(nm)^2}.$
Let
\begin{align*}
\mathcal{E}_1=&\left\{ \frac{3}{2}  \widetilde{b}_1 K\Big(\max_{i\in\mathcal{J}_{e,l},1\le j\le m_i}\vert \widetilde\delta_i^*(x_{i,j}^*)\vert+1\Big)\left\|\left(D^{\dagger}\right)^{\top} \breve{\zeta}\right\|_{\infty} \right.
\\
\leq& \left.\frac{3}{2}  \widetilde{b}_1 K\Big(\max_{i\in\mathcal{J}_{e,l},1\le j\le m_i}\vert \widetilde\delta_i^*(x_{i,j}^*)\vert+1\Big)C(d) \sqrt{ \log(\frac{nm}{K})\log ( \frac{(nm)^3}{K} )} )\right\},
\end{align*}
and,
\begin{align*}
  \mathcal{E}_2=&\left\{\frac{3}{2}  \widetilde{b}_1 K\Big(\max_{i\in\mathcal{J}_{e,l},1\le j\le m_i}\vert \widetilde\delta_i^*(x_{i,j}^*)\vert+1\Big)\|\Pi \breve{\zeta}\|_2 \right.
 \\
 \leq& \left.2 \frac{3}{2}  \widetilde{b}_1 K\Big(\max_{i\in\mathcal{J}_{e,l},1\le j\le m_i}\vert \widetilde\delta_i^*(x_{i,j}^*)\vert+1\Big) \sqrt{ \log((nm)^2)}\right\}.
\end{align*}
By Inequalities (\ref{H-R-T3-1}) and (\ref{H-R-T3-2}) 
 these two events happen with probability at least $1-\frac{1}{(nm)^2}.$
Under events $\mathcal{E}_1$ and $\mathcal{E}_2$, using Inequality (\ref{equation-final-T3-aux}) and Inequality (\ref{aux-thm3-eq1}), uniformly for all $f \in \mathcal{F}_{\text {lat }}(\eta, 2 V_n)$, it holds that
\begin{align}
\label{Aux-T3-equation-Last}
    &\sup _{f \in \mathcal{F}_{\text {lat }}(\eta, 2 V_n)} \sum_{i=1}^n \sum_{j=1}^m \omega_{i,j} f\left(x_{i,j}\right) \widetilde\delta_i\left(x_{i,j}\right) \nonumber
    \\
    \leq& 2\eta X_{\widetilde{\delta}^*} \sqrt{ \log((nm)^2)}
    +2V_n X_{\widetilde{\delta}^*} C(d) \sqrt{ \log(\frac{nm}{K})\log ( \frac{(nm)^3}{K} )} ).
\end{align}
where we have denoted by 
\begin{equation}
\label{X-delta-aux}
    X_{\widetilde{\delta}^*}:=\frac{3}{2}  \widetilde{b}_1 K\Big(\max_{i\in\mathcal{J}_{e,l},1\le j\le m_i}\vert \widetilde\delta_i^*(x_{i,j}^*)\vert+1\Big)
\end{equation}
Hence, 
\begin{align}
\label{T3Cons-3}
& \mathbb{E}_\omega\left(\sup _{f \in \mathcal{F}_{\text {lat }}(\eta, 2 V_n)} \sum_{i\in \mathcal{J}_{e,l}} \sum_{j=1}^{m_i} \omega_{i,j} f\left(x_{i,j}^*\right) \widetilde\delta_i^*\left(x_{i,j}^*\right)\mathbf{1}_{\mathcal{E}_0}\right)\nonumber \\
=&  \mathbb{E}_\omega\left(\sup _{f \in \mathcal{F}_{\text {lat }}(\eta, 2 V_n)} \sum_{i\in \mathcal{J}_{e,l}} \sum_{j=1}^{m_i} \omega_{i,j} f\left(x_{i,j}^*\right) \widetilde\delta_i^*\left(x_{i,j}^*\right)\mathbf{1}_{\mathcal{E}_1\cap\mathcal{E}_2}\mathbf{1}_{\mathcal{E}_0}\right)\nonumber
\\
+&\mathbb{E}_\omega\left(\sup _{f \in \mathcal{F}_{\text {lat }}(\eta, 2 V_n)} \sum_{i\in \mathcal{J}_{e,l}} \sum_{j=1}^{m_i} \omega_{i,j} f\left(x_{i,j}^*\right) \widetilde\delta_i^*\left(x_{i,j}^*\right)\mathbf{1}_{\mathcal{E}_1^c\cup\mathcal{E}_2^c}\mathbf{1}_{\mathcal{E}_0}\right).
\end{align}
From Inequality (\ref{Aux-T3-equation-Last}) and Inequality (\ref{T3Cons-3}),
\begin{align}
\label{T3Cons-3-ad1}
& \mathbb{E}_\omega\left(\sup _{f \in \mathcal{F}_{\text {lat }}(\eta, 2 V_n)} \sum_{i\in \mathcal{J}_{e,l}} \sum_{j=1}^{m_i} \omega_{i,j} f\left(x_{i,j}^*\right) \widetilde\delta_i^*\left(x_{i,j}^*\right)\mathbf{1}_{\mathcal{E}_0}\right) \nonumber\\
\leq& \widetilde{C}(d) X_{\widetilde{\delta}^*}\Big(\eta  \sqrt{ \log (n m)}+V_n \sqrt{K} \log(n m)\Big)\nonumber
\\
+&\mathbb{E}_\omega\left(\sup _{f \in \mathcal{F}_{\text {lat }}(\eta, 2 V_n)} \sum_{i\in \mathcal{J}_{e,l}} \sum_{j=1}^{m_i} \omega_{i,j} f\left(x_{i,j}^*\right) \widetilde\delta_i^*\left(x_{i,j}^*\right)\mathbf{1}_{\mathcal{E}_1^c\cup\mathcal{E}_2^c}\mathbf{1}_{\mathcal{E}_0}\right).
\end{align}
Using Equation (\ref{f-Flat-def}) observe that any $f\in\mathcal{F}_{\text {lat }}(\eta, 2 V_n) $ satisfies $\vert\vert f\vert\vert_{\infty}\le \vert\vert v(f)\vert\vert_2\le \eta$. Therefore,
\begin{align*}
    &\mathbb{E}_\omega\left(\sup _{f \in \mathcal{F}_{\text {lat }}(\eta, 2 V_n)} \sum_{i\in \mathcal{J}_{e,l}} \sum_{j=1}^{m_i} \omega_{i,j} f\left(x_{i,j}^*\right) \widetilde\delta_i^*\left(x_{i,j}^*\right)\mathbf{1}_{\mathcal{E}_1^c\cup\mathcal{E}_2^c}\mathbf{1}_{\mathcal{E}_0}\right)
    \\
    \le& \frac{Cc_2C_{\beta}\eta \Big(\max_{i\in\mathcal{J}_{e,l},1\le j\le m_i}\vert \widetilde\delta_i^*(x_{i,j}^*)\vert\Big)}{nm\log(n)}.
\end{align*}
Here we have used Assumption 1{\bf{f}}, the choice of $L$ and the fact that $\mathcal{E}_1\cap\mathcal{E}_2$ holds with probability at least $1-\frac{1}{(nm)^2}$.
From the last display and Inequality (\ref{T3Cons-3-ad1})
\begin{align}
\label{T3Cons-3-ad2}
& \mathbb{E}_\omega\left(\sup _{f \in \mathcal{F}_{\text {lat }}(\eta, 2 V_n)} \sum_{i\in \mathcal{J}_{e,l}} \sum_{j=1}^{m_i} \omega_{i,j} f\left(x_{i,j}^*\right) \widetilde\delta_i^*\left(x_{i,j}^*\right)\mathbf{1}_{\mathcal{E}_0}\right) \nonumber\\
\leq&  \widetilde{C}(d) X_{\widetilde{\delta}^*}\Big(\eta  \sqrt{ \log (n m)}+V_n \sqrt{K} \log(n m)\Big)\nonumber
\\
+&
\frac{Cc_2C_{\beta}\eta \Big(\max_{i\in\mathcal{J}_{e,l},1\le j\le m_i}\vert \widetilde\delta_i^*(x_{i,j}^*)\vert\Big)}{nm\log(n)}.
\end{align}
Finally, we observe that $X_{\widetilde{\delta}^*}$ defined in Equation (\ref{X-delta-aux}), satisfies
$$
    X_{\widetilde{\delta}^*}:=\frac{3}{2}  \widetilde{b}_1 K\Big(\max_{i\in\mathcal{J}_{e,l},1\le j\le m_i}\vert \widetilde\delta_i^*(x_{i,j}^*)\vert+1\Big)\le \frac{3}{2}  \widetilde{b}_1 K\Big(\frac{C}{c}\max_{i\in\mathcal{J}_{e,l},1\le j\le m_i}\vert \delta_i^*(x_{i,j}^*)\vert+1\Big),
$$
by Assumption 1{\bf{f}} and Equation (\ref{tildevariables1}). Moreover, using Assumption 1{\bf{c}}, we obtain that $\mathbb{E}( X_{\widetilde{\delta}^*})\le \widetilde{\widetilde{C}}(d)K\sqrt{\log(nm)}$ with $\widetilde{\widetilde{C}}(d)>0$ a constant that depends on $d.$ Thus, by Inequalities (\ref{last-step-t3-aux1}), (\ref{T3Cons-1}), (\ref{Expectedvaluewithrespectx-t3}), (\ref{T3Cons-3-ad3}), (\ref{complement_E0}) and (\ref{T3Cons-3-ad2})
{{\begin{align}
    \label{T3-f1}
    &\mathbb{P}\Big(\sup _{f \in \mathcal{F}_{\text {lat }}(\eta, 2 V_n)} \sum_{i=1}^n \sum_{j=1}^{m_i}\left\{f\left(x_{i,j}^*\right) \widetilde\delta_i^*\left(x_{i,j}\right)-\left\langle f, \widetilde\delta_i^*\right\rangle_{2}\right\}>\frac{\eta^2}{12}L\Big)\nonumber
    \\
    \le&
    \frac{24 \widetilde{C}(d)\widetilde{\widetilde{C}}(d)\log^2(n)K\sqrt{\log(nm)} \Big(\eta  \sqrt{ \log (n m)}+V_n \sqrt{K} \log(n m)\Big)}{\eta^2}\nonumber
    \\
    +&
    \frac{24C_6\varpi_{\delta}\eta\log^2(n)\sqrt{\log(nm)}}{\eta^2nm\log(n)}\nonumber
    \\
    +&\frac{\eta}{\eta^2}\frac{c_2CC_{\beta}nm}{2\log(n)} \frac{4}{C_{\beta}}\log(n)N^d\exp\Big(-\frac{C_{\beta}}{48} \widetilde{b}_2\frac{K}{\log(n)}\Big).
\end{align}}}
Here $C_6>0$ is a constant. By our choice of $\eta$ in equation (\ref{eta-def}), the choice of $N$ in Equation (\ref{choiceN-T3}) and Inequality (\ref{T3-f1}), we conclude that
\begin{equation}
    \label{T3Final2}
    \mathbb{P}\Big(\sup _{f \in \mathcal{F}_{\text {lat }}(\eta, 2 V_n)} \sum_{i=1}^n \sum_{j=1}^{m_i}\left\{f\left(x_{i,j}^*\right) \widetilde\delta_i^*\left(x_{i,j}\right)-\left\langle f, \widetilde{\delta}_i^*\right\rangle_{2}\right\}>\frac{\eta^2}{12}L\Big)
    \le \frac{\varepsilon}{12}.
\end{equation}
Using Inequality (\ref{T3Cons-0}) and Inequality (\ref{T3Final2}),
\begin{align}
\label{T3-f3}
    &\mathbb{P}\Big(\sup _{\|\Delta\|_2 \leq \eta,\left\|\nabla_{G_K} \Delta\right\|_1 \leq 2 V_n}\sum_{i=1}^n \sum_{j=1}^{m_i}\left\{\widetilde{\Delta}\left(x_{i,j}\right) \widetilde\delta_i\left(x_{i,j}\right)-\left\langle\widetilde{\Delta}, \widetilde\delta_i\right\rangle_{2}\right\}>\eta^2/12\Big)\nonumber
    \\
    \le& \frac{\varepsilon}{12}+\frac{\varepsilon}{12}=\frac{\varepsilon}{6}.
\end{align}
{\bf{Step 2A-d.}} We conclude the {\bf{Step 2A}}. From Equation (\ref{T3Cons-5}), (\ref{T3Cons-4}) and (\ref{T3-f3}) we conclude that,
\begin{align}
    \label{delta-bound}
    \mathbb{P}\left\{\sup _{\|\Delta\|_2 \leq \eta,\left\|\nabla_{G_K} \Delta\right\|_1 \leq 2 V_n} \sum_{i=1}^n \sum_{j=1}^m \Delta_{i,j} \widetilde\delta_i\left(x_{i,j}\right) \geq \frac{\eta^2}{4}\right\} \leq \frac{\varepsilon}{2} .
\end{align}
\\
\\
{\bf{Step 3A}}
The desired result follows from Equation (\ref{eps-bound}) and Equation (\ref{delta-bound}).
\\
\\

{\bf{Step B.}} Now the assertion of Equation (\ref{eqn:e5}) in Chapter 2 is demonstrated. Through this step, notation from Section \ref{Lat-section} is used.
We start noticing that,
\begin{align*}
& \vert\vert f^*- \widehat{f}_{V_n}\vert\vert_2^2=\int\left\{f^*(x)- \widehat{f}_{V_n}(x)\right\}^2 v(x) d x.
\end{align*}
Then, adding and subtracting $\frac{1}{\left|\mathcal{N}_K(x)\right|} \sum_{x_{i,j} \in \mathcal{N}_K(x)} f^*\left(x_{i,j}\right)$, it can be written as
\begin{align*} 
&\int\Big\{f^*(x)-\frac{1}{\Big|\mathcal{N}_K(x)\Big|} \sum_{x_{i,j} \in \mathcal{N}_K(x)} f^*\left(x_{i,j}\right)+\frac{1}{\left|\mathcal{N}_K(x)\right|} \sum_{x_{i,j} \in \mathcal{N}_K(x)} f^*\left(x_{i,j}\right)
\\
-&\frac{1}{\Big|\mathcal{N}_K(x)\Big|} \sum_{x_{i,j} \in \mathcal{N}_K(x)}  \widehat{f}_{V_n}\left(x_{i,j}\right)\Big\}^2 v(x) d x \\
& \leq 2 \int\left\{f^*(x)-\frac{1}{\left|\mathcal{N}_K(x)\right|} \sum_{x_{i,j} \in \mathcal{N}_K(x)} f^*\left(x_{i,j}\right)\right\}^2 v(x) d x \\
& +2 \int\left\{\frac{1}{\left|\mathcal{N}_K(x)\right|} \sum_{x_{i,j} \in \mathcal{N}_K(x)} f^*\left(x_{i,j}\right)-\frac{1}{\left|\mathcal{N}_K(x)\right|} \sum_{x_{i,j} \in \mathcal{N}_K(x)}  \widehat{f}_{V_n}\left(x_{i,j}\right)\right\}^2 v(x) d x .
\end{align*}
Now we proceed to bound each of the terms 
\begin{equation}
    \label{Sec-T3-1}
    \text{Int}_1=\int\left\{f^*(x)-\frac{1}{\left|\mathcal{N}_K(x)\right|} \sum_{x_{i,j} \in \mathcal{N}_K(x)} f^*\left(x_{i,j}\right)\right\}^2 v(x) d x,
\end{equation}
and,
\begin{equation}
     \label{Sec-T3-2}
     \text{Int}_2=\int\left\{\frac{1}{\left|\mathcal{N}_K(x)\right|} \sum_{x_{i,j} \in \mathcal{N}_K(x)} f^*\left(x_{i,j}\right)-\frac{1}{\left|\mathcal{N}_K(x)\right|} \sum_{x_{i,j} \in \mathcal{N}_K(x)}  \widehat{f}_{V_n}\left(x_{i,j}\right)\right\}^2 v(x) d x.
\end{equation}
For the second term $\text{Int}_2$ in Equation (\ref{Sec-T3-2}), we observe that
\begin{align*}
&\int\left\{\frac{1}{\left|\mathcal{N}_K(x)\right|} \sum_{x_{i,j} \in \mathcal{N}_K(x)} f^*\left(x_{i,j}\right)-\frac{1}{\left|\mathcal{N}_K(x)\right|} \sum_{x_{i,j} \in \mathcal{N}_K(x)}  \widehat{f}_{V_n}\left(x_{i,j}\right)\right\}^2 v(x) d x \\
& =\sum_{r=1}^{N^d} \int_{\mathcal{I}_r}\left\{\frac{1}{\left|\mathcal{N}_K(x)\right|} \sum_{x_{i,j} \in \mathcal{N}_K(x)} f^*\left(x_{i,j}\right)-\frac{1}{\left|\mathcal{N}_K(x)\right|} \sum_{x_{i,j} \in \mathcal{N}_K(x)}  \widehat{f}_{V_n}\left(x_{i,j}\right)\right\}^2 v(x) d x .
\end{align*}
Moreover, for $x'\in[0,1]^d$, there exists $\mathfrak{c}_r(x') \in P_{\text{lat}}(N)$ that satisfies $x' \in C(\mathfrak{c}_r(x'))$ and
$$
\left\|x'-P_I(x')\right\|_2 \leq \frac{1}{2 N} .
$$
Even more, by Lemma \ref{Omega-Oscar-Paper}, with high probability, there exists $(i\left(x^{\prime}\right),j\left(x^{\prime}\right) )\in[n]\times[m_i]$ such that $x_{i\left(x^{\prime}\right),j\left(x^{\prime}\right)} \in C(\mathfrak{c}_l(x'))$. Hence, 
$$
\left\|x^{\prime}-x_{i\left(x^{\prime}\right),j\left(x^{\prime}\right)}\right\|_2 \leq \frac{1}{N} .
$$
Denote by  $x_{i_1,j_1}, \ldots, x_{i_K,j_K}, i_1, \ldots, i_K \in[n]$ and $(j_1, \ldots, j_K)
\in[m_{i_1}]\times...\times [m_{i_K}]$ $K$-nearest neighbors corresponding to $x_{i\left(x^{\prime}\right),j\left(x^{\prime}\right)}$. Observe that
$$
\vert\vert x^{\prime}- x_{i_r,j_r}\vert\vert_2 \leq \frac{1}{ N}+R_{K, \max }, \quad r=1, \ldots, K
$$
where 
\begin{align*}
& R_{K, \max }=\max _{1 \leq i \leq n,1\le j\le m_i} R_K\left(x_{i,j}\right), 
\end{align*}
with 
$R_K(x)$ denoting  the distance from $x \in [0,1]^d$ to its $K$th nearest neighbor in the set $\left\{x_{1,1}, \ldots, x_{n,m_n}\right\}$. 
Thus, with high probability, for any $x^{\prime} \in [0,1]^d$,
\begin{align*}
\mathcal{N}_K\left(x^{\prime}\right) & \subset\left\{x_{i,j}: \vert\vert x^{\prime}- x_{i,j} \vert\vert_2\leq \frac{1}{ N}+R_{K, \max }\right\} \\
& \subset\left\{x_{i,j}:\left\|\mathfrak{c}_r(x')-x_{i,j}\right\|_2 \leq\left(\frac{1}{2}+1\right) \frac{1}{N}+ R_{K, \max }\right\} .
\end{align*}

Hence, we set
$$
\widetilde{\mathcal{N}}\left(\mathfrak{c}_r(x')\right)=\left\{x_{i,j}:\left\|\mathfrak{c}_r(x')-x_{i,j}\right\|_2 \leq\left(\frac{1}{2}+1\right) \frac{1}{N}+R_{K, \max }\right\} .
$$
As a result, for $r \in\left[N^d\right]$,
\begin{align*}
& \int_{C\left(\mathfrak{c}_r\right)}\left\{\frac{1}{\left|\mathcal{N}_K(x)\right|} \sum_{x_{i,j} \in \mathcal{N}_K(x)} f^*\left(x_{i,j}\right)-\frac{1}{\left|\mathcal{N}_K(x)\right|} \sum_{x_{i,j} \in \mathcal{N}_K(x)}  \widehat{f}_{V_n}\left(x_{i,j}\right)\right\}^2 v(x) d x  ,
\end{align*}
can  be bounded as,
\begin{align*}
&\int_{C\left(\mathfrak{c}_r\right)} \frac{1}{\left|\mathcal{N}_K(x)\right|} \sum_{x_{i,j} \in \mathcal{N}_K(x)}\left\{f^*\left(x_{i,j}\right)- \widehat{f}_{V_n}\left(x_{i,j}\right)\right\}^2 v(x) d x \\
& \leq \frac{1}{K} \int_{C\left(\mathfrak{c}_r\right)} \sum_{x_{i,j} \in \widetilde{\mathcal{N}}\left(\mathfrak{c}_r\right)}\left\{f^*\left(x_{i,j}\right)- \widehat{f}_{V_n}\left(x_{i,j}\right)\right\}^2 v(x)d x \\
& =\frac{1}{K}\left[\sum_{x_{i,j} \in \widetilde{\mathcal{N}}\left(\mathfrak{c}_r\right)}\left\{f^*\left(x_{i,j}\right)- \widehat{f}_{V_n}\left(x_{i,j}\right)\right\}^2\right] \int_{C\left(\mathfrak{c}_r\right)} v(x) d x .
\end{align*}
The above leads to a upper bound for the term 
$$\int\left\{\frac{1}{\left|\mathcal{N}_K(x)\right|} \sum_{x_{i,j} \in \mathcal{N}_K(x)} f^*\left(x_{i,j}\right)-\frac{1}{\left|\mathcal{N}_K(x)\right|} \sum_{x_{i,j} \in \mathcal{N}_K(x)}  \widehat{f}_{V_n}\left(x_{i,j}\right)\right\}^2 v(x) d x.$$
This upper bound has the form,
\begin{align*}
    \sum_{r=1}^{N^d}\left(\frac{1}{K}\left[\sum_{x_{i,j} \in \widetilde{\mathcal{N}}\left(\mathfrak{c}_r\right)}\left\{f^*\left(x_{i,j}\right)- \widehat{f}_{V_n}\left(x_{i,j}\right)\right\}^2\right] \int_{C\left(\mathfrak{c}_r\right)} v(x) d x\right),
\end{align*}
which can be bounded by
\begin{align*}
& \left[\sum_{r=1}^{N^d} \frac{1}{K} \sum_{x_{i,j} \in \widetilde{\mathcal{N}}\left(\mathfrak{c}_r\right)}\left\{f^*\left(x_{i,j}\right)- \widehat{f}_{V_n}\left(x_{i,j}\right)\right\}^2\right] \max _{1 \leq r \leq N^d} \int_{C\left(\mathfrak{c}_r\right)} v(x) d x \\
& \leq\left[\sum_{i=1}^n\sum_{j=1}^{m_i} \frac{1}{K} \sum_{r \in\left[N^d\right]:\left\|x_{i,j}-\mathfrak{c}_r\right\|_2 \leq\left(\frac{1}{2}+1\right) \frac{1}{N}+ R_{K, \max }}\left\{f^*\left(x_{i,j}\right)- \widehat{f}_{V_n}\left(x_{i,j}\right)\right\}^2\right] 
\\
&\cdot\Big[\max _{1 \leq r \leq N^d} \int_{C\left(\mathfrak{c}_r\right)} v(x) d x \Big]\\
& \leq\left[\sum_{i=1}^n\sum_{j=1}^{m_i}\left\{f^*\left(x_{i,j}\right)- \widehat{f}_{V_n}\left(x_{i,j}\right)\right\}^2\right]\Big[\frac{1}{K} \max _{1 \leq r \leq N^d} \int_{C\left(\mathfrak{c}_r\right)} v(x) d x\Big] \\
& \cdot\left[\max _{(i,j) \in[n]\times[m_i]}\left|\left\{k:\left\|x_{i,j}-\mathfrak{c}_r\right\|_2 \leq\left(\frac{1}{2}+1\right) \frac{1}{N}+ R_{K, \max }\right\}\right|\right].
\end{align*}
Next, for a set $A \subset \mathbb{R}^d$ and positive constant $\widetilde{r}$, we define the packing and external covering numbers as
\begin{align*}
\mathrm{N}_{\widetilde{r}}^{\text {pack }}(A) & :=\max \left\{r \in \mathbb{N}: \exists q_1, \ldots, q_r \in A, \right. \\
&\quad \left. \text { such that }\left\|q_j-q_{j^{\prime}}\right\|_2>\widetilde{r} \quad \forall j \neq j^{\prime}\right\}, \\
\mathrm{N}_{\widetilde{r}}^{\text {ext }} & :=\min \left\{r \in \mathbb{N}: \exists q_1, \ldots, q_r \in \mathbb{R}^d, \right. \\
&\quad \left. \text { such that } \forall x \in A \text { there exists } r_x \text { with }\left\|x-q_{r_x}\right\|_2<\widetilde{r}\right\}.
\end{align*}
Furthermore, by Lemma \ref{dist-bet-Knn}, there exists a constant $\widetilde{c}$ such that $R_{K, \max } \leq \widetilde{c} / N$ with high probability. This implies that with high probability, for a positive constant $\widetilde{C}$, we have that
\begin{align*}
& \max _{(i,j) \in[n]\times[m_i]}\left|\left\{r \in\left[N^d\right]:\left\|x_{i,j}-\mathfrak{c}_r\right\|_2 \leq\left(\frac{1}{2}+1\right) \frac{1}{N}+ R_{K, \max }\right\}\right| \\
& \leq  \max _{(i,j) \in[n]\times[m_i]}\left|\left\{r \in\left[N^d\right]:\left\|x_{i,j}-\mathfrak{c}_r\right\|_2 \leq \frac{\widetilde{C}}{N}\right\}\right| \\
& \leq \max _{(i,j) \in[n]\times[m_i]} \mathrm{N}_{\frac{1}{N}}^{\mathrm{pack}}\left(B_{\frac{\widetilde{C}}{N}}\left(x_{i,j}\right)\right) \\
& \leq \mathrm{N}_{\frac{1}{N}}^{\mathrm{ext}}\left(B_{\frac{\widetilde{C}}{N}}(0)\right) \\
& =\mathrm{N}_1^{\mathrm{ext}}\left(B_{\widetilde{C}}^{\widetilde{N}}(0)\right) \\
& <C^{\prime} \text {, } 
\end{align*}
for some positive constant $C^{\prime}$, where the first inequality follows from $R_{K, \max } \leq \widetilde{c} / N$, the second from the definition of packing number, and the remaining inequalities from basic properties of packing and external covering numbers.
Therefore, there exists a constant $\widetilde{C}_1>0$ such that
\begin{align*}
& \text{Int}_2=\int\left\{\frac{1}{\left|\mathcal{N}_K(x)\right|} \sum_{x_{i,j} \in \mathcal{N}_K(x)} f^*\left(x_{i,j}\right)-\frac{1}{\left|\mathcal{N}_K(x)\right|} \sum_{x_{i,j} \in \mathcal{N}_K(x)}  \widehat{f}_{V_n}\left(x_{i,j}\right)\right\}^2 v(x) \mu(d x) \\
& \leq\left[\sum_{i=1}^n\sum_{j=1}^{m_i}\left\{f^*\left(x_{i,j}\right)- \widehat{f}_{V_n}\left(x_{i,j}\right)\right\}^2\right] \frac{C^{\prime}}{K} \max _{1 \leq r \leq N^d} \int_{C\left(\mathfrak{c}_r\right)} v(x) d x \\
& \leq\left[\sum_{i=1}^n\sum_{j=1}^{m_i}\left\{f^*\left(x_{i,j}\right)- \widehat{f}_{V_n}\left(x_{i,j}\right)\right\}^2\right] \frac{C^{\prime} c_{2,d}v_{\max }}{K} Vol\left(B_{\frac{\left(\tilde{b}_1\right)^{1 / d}}{(2c_{2,d}v_{\max }Cc_2)^{1 / d}}\left(\frac{K}{nm}\right)^{1 / d}}(x)\right)
\\
&\le \frac{1}{nm}\left[\sum_{i=1}^n\sum_{j=1}^{m_i}\left\{\theta_{i,j}^*- (\widehat{\theta}_{V_n})_{i,j}\right\}^2\right] \widetilde{C}_1 v_{\max },
\end{align*}
where, the second inequality is followed by Equation (\ref{Omega-Oscar-paper-0}). Then, making use of {\bf{Step A}}, we conclude that
\begin{align}
    \label{first-boudn-L2}
   & \text{Int}_2\nonumber
   =\int\left\{\frac{1}{\left|\mathcal{N}_K(x)\right|} \sum_{x_{i,j} \in \mathcal{N}_K(x)} f^*\left(x_{i,j}\right)-\frac{1}{\left|\mathcal{N}_K(x)\right|} \sum_{x_{i,j} \in \mathcal{N}_K(x)}  \widehat{f}_{V_n}\left(x_{i,j}\right)\right\}^2 v(x) d x\nonumber\\
    =&O_\mathbb{P}(r_{n,m}),
\end{align}
where
$$r_{n,m}=\frac{\log^{8+2l}(nm)}{n}+\left\{\frac{\log ^{5+l}(n m)}{n m}\right\}  V_n.$$ It remains to analyze the expression in Equation (\ref{Sec-T3-1}), this is,
$$\text{Int}_1=\int\left\{f^*(x)-\frac{1}{\left|\mathcal{N}_K(x)\right|} \sum_{x_{i,j} \in \mathcal{N}_K(x)} f^*\left(x_{i,j}\right)\right\}^2 v(x) d x.$$
To this end,
we use the notation from Section \ref{plsection}. Observe that, if $\mathcal{S}$ corresponds to the piecewise definition of $f^*,$ the term 
{\small{$$\int\left\{f^*(x)-\frac{1}{\left|\mathcal{N}_K(x)\right|} \sum_{x_{i,j} \in \mathcal{N}_K(x)} f^*\left(x_{i,j}\right)\right\}^2 v(x) d x,$$}}
can be written as,
{\small{$$\sum_{r \in\left[N^d\right]} \int_{C\left(\mathfrak{c}_r\right)}\left\{f^*(x)-\frac{1}{K} \sum_{x_{i,j} \in \mathcal{N}_K(x)} f^*\left(x_{i,j}\right)\right\}^2 v(x) d x,$$}}
which by convexity can be bounded by,
{\footnotesize{\begin{align*} 
&\sum_{r \in\left[N^d\right]} \int_{C\left(\mathfrak{c}_r\right)} \sum_{x_{i,j} \in \mathcal{N}_K(x)} \frac{1}{K}\left\{f^*(x)-f^*\left(x_{i,j}\right)\right\}^2 v(x) d x \\
& = \sum_{r \in\left[N^d\right]:C\left(\mathfrak{c}_r\right) \cap \mathcal{S}=\emptyset} \int_{C\left(\mathfrak{c}_r\right)} \sum_{x_{i,j} \in \mathcal{N}_K(x)} \frac{1}{K}\left\{f^*(x)-f^*\left(x_{i,j}\right)\right\}^2 v(x) d x \\
&+ \sum_{r \in\left[N^d\right]:C\left(\mathfrak{c}_r\right) \cap \mathcal{S}\neq\emptyset} \int_{C\left(\mathfrak{c}_r\right)} \sum_{x_{i,j} \in \mathcal{N}_K(x)} \frac{1}{K}\left\{f^*(x)-f^*\left(x_{i,j}\right)\right\}^2 v(x) d x.
\end{align*}}}
Therefore, the term 
$$\int\left\{f^*(x)-\frac{1}{\left|\mathcal{N}_K(x)\right|} \sum_{x_{i,j} \in \mathcal{N}_K(x)} f^*\left(x_{i,j}\right)\right\}^2 v(x) d x,$$
is bounded by
\begin{align*}
 & 4\left|\left\{r \in\left[N^d\right]: C\left(\mathfrak{c}_r\right) \cap \mathcal{S} \neq \emptyset\right\}\right|\left\|f_0\right\|_{\infty}^2 \max _{1 \leq r \leq N^d} \int_{C\left(\mathfrak{c}_r\right)} v(x) d x \\ & +\sum_{r\in\left[N^d\right]:C\left(\mathfrak{c}_r\right) \cap \mathcal{S}=\emptyset} \int_{C\left(\mathfrak{c}_r\right)} \sum_{x_{i,j} \in \mathcal{N}_K(x)} \frac{1}{K}\left\{f_0(x)-f_0\left(x_{i,j}\right)\right\}^2 v(x) d x ,
 \end{align*}
which by Equation (\ref{Omega-Oscar-paper-0}) is smaller than 
 \begin{align*}
  &Vol\left(B_{\frac{\widetilde{b}_1^{1 / d}}{(2c_{2,d}Cc_2v_{\max })^{1 / d}}\left(\frac{K}{nm}\right)^{1 / d}}(x)\right)\left(4 v_{\max }\left\|f_0\right\|_{\infty}^2\right)\left|\left\{r\in\left[N^d\right]: C\left(\mathfrak{c}_r\right) \cap \mathcal{S} \neq \emptyset\right\}\right|
  \\
   & +\sum_{r \in\left[N^d\right]:C\left(\mathfrak{c}_r\right) \cap \mathcal{S}=\emptyset} \int_{C\left(\mathfrak{c}_r\right)} \sum_{x_{i,j} \in \mathcal{N}_K(x)} \frac{1}{K}\left\{f^*(x)-f^*\left(x_{i,j}\right)\right\}^2 v(x) d x. 
 \end{align*}
 Then, using the bound for volume and the piecewise lipschitz  condition, it is bounded by
\begin{align*}&\left( 4 \widetilde{b}_1\left\|f^*\right\|_{\infty}^2\right) \frac{K}{nm}\left|\left\{r \in\left[N^d\right]: C\left(\mathfrak{c}_r\right) \cap \mathcal{S} \neq \emptyset\right\}\right| \\ & +\sum_{r \in\left[N^d\right]:C\left(\mathfrak{c}_r\right) \cap \mathcal{S}=\emptyset} \int_{C\left(\mathfrak{c}_l\right)} \sum_{x_{i,j} \in \mathcal{N}_K(x)} \frac{L_0}{K}\left\|x-x_{i,j}\right\|_2^2 v(x) d x \\ & \leq\left( 4 \widetilde{b}_1^d\left\|f^*\right\|_{\infty}^2\right) \frac{K}{nm}\left|\left\{r \in\left[N^d\right]: C\left(\mathfrak{c}_r\right) \cap \mathcal{S} \neq \emptyset\right\}\right| \\ & +L_0\left\{\sum_{r \in\left[N^d\right]:C\left(\mathfrak{c}_r\right) \cap \mathcal{S}=\emptyset}  \int_{\left(\mathfrak{c}_r\right)} v(x) d x\right\}\left( \frac{1}{N}+R_{K, \max }\right)^2 \\ & \leq\left( 4 \widetilde{b}_1^d\left\|f^*\right\|_{\infty}^2\right) \frac{K}{nm}\left|\left\{r\in\left[N^d\right]: C\left(\mathfrak{c}_r\right) \cap \mathcal{S} \neq \emptyset\right\}\right| +L_0\left( \frac{1}{N}+ R_{K, \max }\right)^2.
\end{align*}
Finally, by the choice of $N$, Lemma \ref{dist-bet-Knn} and following the conclusion on Lemma 15 of \cite{madrid2020adaptive}, we get that
\begin{align}
\label{boudforAERR}
    \text{Int}_2=\int\left\{f^*(x)-\frac{1}{\left|\mathcal{N}_K(x)\right|} \sum_{x_{i,j} \in \mathcal{N}_K(x)} f^*\left(x_{i,j}\right)\right\}^2 v(x) d x=O_\mathbb{P}\Big(\Big(\frac{K}{nm}\Big)^{\frac{1}{d}}\Big).
\end{align}
From Equation (\ref{first-boudn-L2}) and (\ref{boudforAERR}) we conclude the claim.
\end{proof}

\newpage
\section{Proof of Theorem \ref{Pest-d>1}}\label{proofT4}
\begin{proof}
The proof of Theorem 4 is divided into two main steps.
\\
{\bf{Step A.}} In this step, the assertion of Equation (\ref{PT-eq1}) in Chapter 2 is shown.
Let $\varepsilon>0.$
By Assumption 1{\bf{f}}, we have that
\begin{align*}
    \frac{1}{cnm}\sum_{i=1}^n \sum_{j=1}^{m_i}\left(\widehat{\theta}_{i, j}-\theta_{i, j}^*\right)^2\ge&\frac{1}{n}\sum_{i=1}^n \frac{1}{m_i}\sum_{j=1}^{m_i}\left(\widehat{\theta}_{i, j}-\theta_{i, j}^*\right)^2.
\end{align*}
It suffices to show that for some sufficiently large constant $C_{d}$ only depending on  $d$, it holds that
$$
\mathbb{P}\left(\left\|\widehat{\theta}_{\widetilde\lambda}-\theta^*\right\|_2^2>\eta^2\right) \leq \varepsilon,
$$
where 
$
\left\|\widehat{\theta}-\theta^*\right\|_2^2=\sum_{i=1}^n \sum_{j=1}^{m_i}\left(\widehat{\theta}_{i, j}-\theta_{i, j}^*\right)^2 ,
$ and
\begin{align}
\label{etaT4}
\eta^2=C_{ d}\left\{m K^{2}+ K^{2} \log^3 (n m)+K \log ^{\frac{14}{4}}(n m)+ K^2 \log^{\frac{1}{2}} (n m) \left\|\nabla_{G_K} \theta^*\right\|_1\right\}.
\end{align}
The tuning parameter is considered to be
\begin{equation}
    \label{lambda}
    \widetilde{\lambda}=\frac{\eta^2}{4Cnm\left\|\nabla_{G_K} \theta^*\right\|_1},
\end{equation}
and let
\begin{equation}
    \label{lambda-T4}
    \lambda=\frac{\eta^2}{4\left\|\nabla_{G_K} \theta^*\right\|_1}.
\end{equation}
\break
{\bf{Step 0A.}}
To proceed denote 
$$
\Theta=\left\{\theta \in \mathbb{R}^{\sum_{i=1}^n m_i}: \theta=t\widehat{\theta}_{\widetilde{\lambda}}+(1-t) \theta^*, t \in[0,1]\right\},
$$
and for any $\theta_1,\theta_2\in\mathbb{R}^{\sum_{i=1}^nm_i}$ consider
$\vert\vert \theta_1-\theta_2\vert\vert_{nm}^2=\frac{1}{n}\sum_{i=1}^n\frac{1}{m_i}\sum_{j=1}^{m_i}((\theta_1)_{i,j}-(\theta_2)_{i,j})^2.$
By convexity of the constraint, the objective functions and the minimizer property, for any $\theta \in \Theta$,
\begin{align}
\label{aux-eq-T4-0}
    \vert\vert y_{i,j}-\theta\vert\vert_{nm}^2\le\vert\vert y_{i,j}-\theta^*\vert\vert_{nm}^2+\widetilde{\lambda}\left\|\nabla_{G_K} \theta^*\right\|_1-\widetilde{\lambda}\left\|\nabla_{G_K} \theta\right\|_1.
\end{align}
Expanding in Inequality (\ref{aux-eq-T4-0}) leads to
\begin{align}
\label{aux-eq-T4-3}
&\frac{1}{2}\left\|\theta-\theta^*\right\|_{nm}^2 \nonumber
\\
\leq& \frac{1}{n}\sum_{i=1}^n \frac{1}{m_i}\sum_{j=1}^{m_i}\left(\theta_{i, j}-\theta_{i, j}^*\right) \epsilon_{i,j}+\frac{1}{n}\sum_{i=1}^n \frac{1}{m_i}\sum_{j=1}^{m_i}\left(\theta_{i, j}-\theta_{i, j}^*\right) \delta_i\left(x_{i,j}\right)\nonumber
\\
+&\widetilde{\lambda}\left\|\nabla_{G_K} \theta^*\right\|_1-\widetilde{\lambda}\left\|\nabla_{G_K} \theta\right\|_1.
\end{align}
Using Assumption 1{\bf{f}} we have that
\begin{equation}
\label{aux-eq-T4-1}
    \frac{1}{Cnm}\left\|\theta-\theta^*\right\|_{2}^2\le\left\|\theta-\theta^*\right\|_{nm}^2,
\end{equation}
where $C$ is positive constants from Assumption 1{\bf{f}}. From Inequalities (\ref{aux-eq-T4-3}) and (\ref{aux-eq-T4-1})
\begin{align}
    \label{aux-eq-T4-4}
     &\frac{1}{2}\left\|\theta-\theta^*\right\|_{2}^2\nonumber \\
     \leq&\sum_{i=1}^n \sum_{j=1}^{m_i}\frac{Cm}{m_i}\left(\theta_{i, j}-\theta_{i, j}^*\right) \epsilon_{i,j}+\sum_{i=1}^n \sum_{j=1}^{m_i} \frac{Cm}{m_i}\left(\theta_{i, j}-\theta_{i, j}^*\right) \delta_i\left(x_{i,j}\right)\nonumber
     \\
     +&\lambda\left\|\nabla_{G_K} \theta^*\right\|_1-\lambda\left\|\nabla_{G_K} \theta\right\|_1
\end{align}
From now along the proof we define 
\begin{equation}
\label{T4ildevariables}
    \widetilde{\epsilon}_{i,j}=\frac{Cm}{m_i}\epsilon_{i,j} \ \text{and}\  \widetilde{\delta}_{i}=\frac{Cm}{m_i}\delta_{i}.
\end{equation}
 Some observations concerning these random variables are presented. First, by Assumption 1{\bf{d}} and {\bf{f}} we have that $\widetilde{\epsilon}_{i,j}$ are sub-Gaussians of parameter $\frac{C^2}{c^2}\varpi_\epsilon^2$. Similarly using Assumption 1{\bf{d}} and {\bf{f}} we have that $\widetilde{\delta}_{i}$ are sub-Gaussians of parameter $\frac{C^2}{c^2}\varpi_\delta^2$.
To conclude {\bf{Step 0A}} observe the following.
From Inequality (\ref{aux-eq-T4-4}) and Equation (\ref{T4ildevariables}) we also have that
$$
\left\|\nabla_{G_K} \theta\right\|_1 \leq \frac{\sum_{i=1}^n \sum_{j=1}^{m_i}\left(\theta_{i, j}-\theta_{i, j}^*\right) \widetilde\epsilon_{i,j}+\sum_{i=1}^n \sum_{j=1}^{m_i}
\left(\theta_{i, j}-\theta_{i, j}^*\right) \widetilde\delta_i\left(x_{i,j}\right)}{\lambda}+\left\|\nabla_{G_K} \theta^*\right\|_1.
$$
Therefore,
\begin{align*}
&\left\|\nabla_{G_K} \theta-\nabla_{G_K} \theta^*\right\|_1
\\
\leq& \frac{\sum_{i=1}^n \sum_{j=1}^{m_i}\left(\theta_{i, j}-\theta_{i, j}^*\right) \widetilde\epsilon_{i,j}+\sum_{i=1}^n \sum_{j=1}^{m_i}\left(\theta_{i, j}-\theta_{i, j}^*\right) \widetilde\delta_i\left(x_{i,j}\right)}{\lambda}+2\left\|\nabla_{G_K} \theta^*\right\|_1 .
\end{align*}
Denote
$$
\mathcal{T}_\epsilon\left(\theta-\theta^*\right)=\sum_{i=1}^n \sum_{j=1}^{m_i}\left(\theta_{i, j}-\theta_{i, j}^*\right) \widetilde\epsilon_{i,j} \quad \text { and } \quad \mathcal{T}_\delta\left(\theta-\theta^*\right)=\sum_{i=1}^n \sum_{j=1}^{m_i}\left(\theta_{i, j}-\theta_{i, j}^*\right) \widetilde\delta_i\left(x_{i,j}\right).
$$
Then for any $\theta \in \Theta$,
\begin{equation}
\label{eq-2}
\left\|\nabla_{G_K} \theta-\nabla_{G_K} \theta^*\right\|_1 \leq \frac{\mathcal{T}_\epsilon\left(\theta-\theta^*\right)+\mathcal{T}_\delta\left(\theta-\theta^*\right)}{\lambda}+2\left\|\nabla_{G_K} \theta^*\right\|_1.
\end{equation}
Similarly to the proof of Theorem 2, the rest of the proof consists of a two-step procedure. Specifically, we separately analyze {\footnotesize{$
\mathcal{E}_1=\left\{\sup _{\theta \in \Theta:\left\|\theta-\theta^*\right\|_2^2 \leq \eta^2}\left\|\nabla_{G_K} \theta\right\|_1 \geq 5\left\|\nabla_{G_K} \theta^*\right\|_1\right\}.
$}} and its complement by using Inequality (\ref{eq-2}) and Inequality (\ref{aux-eq-T4-3}), respectively.
\\
\\
{\bf{Step 1.}} In this step we analyze the event $\mathcal{E}_1.$ To this end, some observations are presented in the following. Let $\theta \in \Theta$ be such that $\left\|\theta-\theta^*\right\|_2^2 \leq \eta^2$. Suppose in addition that
$$
\left\|\nabla_{G_K} \theta\right\|_1 \geq 5\left\|\nabla_{G_K} \theta^*\right\|_1.
$$
By triangle inequality
$$
\left\|\nabla_{G_K} \theta-\nabla_{G_K} \theta^*\right\|_1 \geq\left\|\nabla_{G_K} \theta\right\|_1-\left\|\nabla_{G_K} \theta^*\right\|_1 \geq 4\left\|\nabla_{G_K} \theta^*\right\|_1.
$$
Let $\widetilde{\theta}=t \theta+(1-t) \theta^*$, where
$$
t=\frac{4\left\|\nabla_{G_K} \theta^*\right\|_1}{\left\|\nabla_{G_K} \theta-\nabla_{G_K} \theta^*\right\|_1} \in[0,1].
$$
Then it follows that $\widetilde{\theta} \in \Theta$ and that $\left\|\widetilde{\theta}-\theta^*\right\|_2^2 \leq \eta^2$. In addition, it holds that
\begin{align}
    \label{eq-1}
&\left\|\nabla_{G_K} \widetilde{\theta}-\nabla_{G_K} \theta^*\right\|_1\nonumber
\\
=&t\left\|\nabla_{G_K} \theta-\nabla_{G_K} \theta^*\right\|_1=\frac{4\left\|\nabla_{G_K} \theta^*\right\|_1}{\left\|\nabla_{G_K} \theta-\nabla_{G_K} \theta^*\right\|_1}\left\|\nabla_{G_K} \theta-\nabla_{G_K} \theta^*\right\|_1\nonumber
\\
=&4\left\|\nabla_{G_K} \theta^*\right\|_1 .
\end{align}
Equation (\ref{eq-2}) and (\ref{eq-1}) imply that,
$$
2\left\|\nabla_{G_K} \theta^*\right\|_1 \leq \frac{\mathcal{T}_\epsilon\left(\widetilde{\theta}-\theta^*\right)+\mathcal{T}_\delta\left(\widetilde{\theta}-\theta^*\right)}{\lambda}.
$$
By the choice of $\lambda$ in Equation (\ref{lambda}), the above display leads to
\begin{equation}
    \label{eq-3}
\frac{\eta^2}{2} \leq \mathcal{T}_\epsilon\left(\widetilde{\theta}-\theta^*\right)+\mathcal{T}_\delta\left(\widetilde{\theta}-\theta^*\right).
\end{equation}
We are now ready to complete the analysis for the event
$$
\mathcal{E}_1=\left\{\sup _{\theta \in \Theta:\left\|\theta-\theta^*\right\|_2^2 \leq \eta^2}\left\|\nabla_{G_K} \theta\right\|_1 \geq 5\left\|\nabla_{G_K} \theta^*\right\|_1\right\}.
$$
The properties of $\widetilde{\theta}$ in  Equation (\ref{eq-1}) and Equation (\ref{eq-3}) imply that
\begin{align*}
 \mathcal{E}_1 &\subset\left\{\sup _{\widetilde{\theta} \in \Theta:\left\|\widetilde{\theta}-\theta^*\right\|_2^2 \leq \eta^2,\left\|\nabla_{G_K} \widetilde{\theta}-\nabla_{G_K} \theta^*\right\|_1 \leq 4\left\|\nabla_{G_K} \theta^*\right\|_1} \mathcal{T}_\epsilon\left(\widetilde{\theta}-\theta^*\right)+\mathcal{T}_\delta\left(\widetilde{\theta}-\theta^*\right) \geq \frac{\eta^2}{2}\right\} \\
& \subset\left\{\sup _{\|\Delta\|_2^2 \leq \eta^2,\left\|\nabla_{G_K} \Delta\right\|_1 \leq 4\left\|\nabla_{G_K} \theta^*\right\|_1} \mathcal{T}_\epsilon(\Delta)+\mathcal{T}_\delta(\Delta) \geq \frac{\eta^2}{2}\right\} . 
\end{align*}
Therefore
\begin{align*}
 \mathbb{P}\left(\mathcal{E}_1\right)& \leq \mathbb{P}\left(\left\{\sup _{\|\Delta\|_2^2 \leq \eta^2,\left\|\nabla_{G_K} \Delta\right\|_1 \leq 4\left\|\nabla_{G_K} \theta^*\right\|_1} \mathcal{T}_\epsilon(\Delta)+\mathcal{T}_\delta(\Delta) \geq \frac{\eta^2}{2}\right\}\right) \\
& \leq \mathbb{P}\left(\left\{\sup _{\|\Delta\|_2^2 \leq \eta^2,\left\|\nabla_{G_K} \Delta\right\|_1 \leq 4\left\|\nabla_{G_K} \theta^*\right\|_1} \mathcal{T}_\epsilon(\Delta) \geq \frac{\eta^2}{4}\right\}\right) \\
 &+\leq \mathbb{P}\left(\left\{\sup _{\|\Delta\|_2^2 \leq \eta^2,\left\|\nabla_{G_K} \Delta\right\|_1 \leq 4\left\|\nabla_{G_K} \theta^*\right\|_1} \mathcal{T}_\delta(\Delta) \geq \frac{\eta^2}{4}\right\}\right) .
\end{align*}
Using $V_n=4\left\|\nabla_{G_K} \theta^*\right\|_1$ in Lemma \ref{lemma11-oscar}, for some positive constants $C_3$ and $C_4$  we have that   
{\footnotesize{\begin{align*}
& \sup _{\|\Delta\|_2 \leq \eta,\left\|\nabla_{G_K} \Delta\right\|_1 \leq 4\left\|\nabla_{G_K} \theta^*\right\|_1} \sum_{i=1}^n \sum_{j=1}^m \Delta_{i,j} \widetilde\epsilon_{i,j} \leq C_3\varpi_\epsilon K\log^{1/4}(nm)\eta+ C_4\varpi_\epsilon K \log (n m)  \left\|\nabla_{G_K} \theta^*\right\|_1, 
\end{align*}}}
with a probability of at least 
{\footnotesize{$$1-2 \frac{1}{\sqrt{\log(nm)}}-\frac{4}{C_{\beta}}C\log(n)nm\exp(-\frac{\widetilde{C}}{24}\frac{K}{\log(n)})-3\frac{1}{n}-C_2\log(n)\frac{nm}{K}\exp\Big(-\frac{C_{\beta}}{48} \widetilde{b}_2\frac{K}{\log(n)}\Big)-\frac{1}{nm},$$}} 
where $C_2>0$ is a constant depending on $d, C_\beta$ and $v_{\max }$. Moreover, $\widetilde{b}_2>0$ is a constant depending on $d, v_{\min }$, and $v_{\max }$. Furthermore, $\widetilde{C}=\frac{C_\beta c_1 c}{c_2 C}$ with $c_1$ and $c_2$ positive constants. Here $v_{\max }, v_{\min }$, $C, C_\beta$ and $c$ are from Assumption 1. The positive constants $C_3$ and $C_4$ only depend on $d$.
By the choice of $\eta$ in Equation (\ref{etaT4}) it therefore holds that
\begin{equation}
\label{eps-bound-P}
\mathbb{P}\left(\left\{\sup _{\|\Delta\|_2^2 \leq \eta^2,\left\|\nabla_{G_K} \Delta\right\|_1 \leq 4\left\|\nabla_{G_K} \theta^*\right\|_1} \mathcal{T}_\epsilon(\Delta) \geq \frac{\eta^2}{4}\right\}\right) \leq \frac{\varepsilon}{4} .
\end{equation}
Then, following the exact same line of arguments as in {\bf{Step 2A}} in the proof of Theorem 3 with the use of $V_n=4\left\|\nabla_{G_K} \theta^*\right\|_1$, by our choice of $\eta$ in Equation (\ref{etaT4}), we get that
\begin{equation}
    \label{delta-b-p}
    \mathbb{P}\left(\left\{\sup _{\|\Delta\|_2^2 \leq \eta^2,\left\|\nabla_{G_K} \Delta\right\|_1 \leq 4\left\|\nabla_{G_K} \theta^*\right\|_1} \mathcal{T}_\delta(\Delta) \geq \frac{\eta^2}{4}\right\}\right) \leq \frac{\varepsilon}{4} .
\end{equation}
Therefore Inequality (\ref{eps-bound-P}) and Inequality (\ref{delta-b-p}) give that
\begin{align}
 \mathbb{P}\left(\mathcal{E}_1\right) &\leq \mathbb{P}\left(\left\{\sup _{\left\|\theta-\theta^*\right\|_2^2 \leq \eta^2,\left\|\nabla_{G_K} \theta-\nabla_{G_K} \theta^*\right\|_1 \leq 4\left\|\nabla_{G_K} \theta^*\right\|_1} \mathcal{T}_\epsilon\left(\theta-\theta^*\right) \geq \frac{\eta^2}{4}\right\}\right) \nonumber\\
& +\mathbb{P}\left(\left\{\sup _{\left\|\theta-\theta^*\right\|_2^2 \leq \eta^2,\left\|\nabla_{G_K} \theta-\nabla_{G_K} \theta^*\right\|_1 \leq 4\left\|\nabla_{G_K} \theta^*\right\|_1} \mathcal{T}_\delta\left(\theta-\theta^*\right) \geq \frac{\eta^2}{4}\right\}\right)\nonumber \\
& \leq \frac{\varepsilon}{2} \text {. }\label{f1-T4}
\end{align}
{\bf{Step 2}}.  In this step, the analysis for the event $\mathcal{E}_1^c$ is presented. We start by noticing that if $\left\|\widehat{\theta}_{\widetilde{\lambda}}-\theta^*\right\|_2>\eta$, there exists $\breve{\theta}=t \widehat{\theta}_{\widetilde{\lambda}}+(1-t) \theta^* \in \Theta$ with $t=\frac{\eta}{\left\|\widehat{\theta}_{\widetilde{\lambda}}-\theta^*\right\|_2}$ such that
\begin{equation}
\label{theta_aux}
    \left\|\breve{\theta}-\theta^*\right\|_2=\eta.
\end{equation}
Therefore Inequality (\ref{aux-eq-T4-4}) and  Equation (\ref{theta_aux}) imply that
$$
\sum_{i=1}^n \sum_{j=1}^{m_i}\left(\breve{\theta}_{i, j}-\theta_{i, j}^*\right) \widetilde\epsilon_{i,j}+\sum_{i=1}^n \sum_{j=1}^{m_i}\left(\breve{\theta}_{i, j}-\theta_{i, j}^*\right) \widetilde\delta_i\left(x_{i,j}\right)+\lambda\left\|\nabla_{G_K} \theta^*\right\|_1-\lambda\left\|\nabla_{G_K} \breve{\theta}\right\|_1 \geq \frac{\eta^2}{2}.
$$
The above display together with the choice of $\lambda$ in Equation (\ref{lambda-T4}) and the fact that $-\lambda\left\|\nabla_{G_K} \breve{\theta}\right\|_1\le 0$, implies that
\begin{equation}
\label{aux-T4-1}
\sum_{i=1}^n \sum_{j=1}^{m_i}\left(\breve{\theta}_{i, j}-\theta_{i, j}^*\right) \widetilde\epsilon_{i,j}+\sum_{i=1}^n \sum_{j=1}^{m_i}\left(\breve{\theta}_{i, j}-\theta_{i, j}^*\right) \widetilde\delta_i\left(x_{i,j}\right) \geq \frac{\eta^2}{4}.
\end{equation}
Therefore,
{\footnotesize{\begin{align*}
& \left\{\left\|\widehat{\theta}_{\widetilde\lambda}-\theta^*\right\|_2>\eta\right\} \cap \mathcal{E}_1^c \\
& \subset\left\{\exists \breve{\theta} \in \Theta,\left\|\breve{\theta}-\theta^*\right\|_2=\eta, \sum_{i=1}^n \sum_{j=1}^{m_i}\left(\breve{\theta}_{i, j}-\theta_{i, j}^*\right) \widetilde\epsilon_{i,j}+\sum_{i=1}^n \sum_{j=1}^{m_i}\left(\breve{\theta}_{i, j}-\theta_{i, j}^*\right) \widetilde\delta_i\left(x_{i,j}\right) \geq \frac{\eta^2}{4}\right\} \cap \mathcal{E}_1^c \\
& \subset\left\{\sup _{\|\Delta\|_2 \leq \eta,\left\|\nabla_{G_K} \Delta\right\|_1 \leq 6\left\|\nabla_{G_K} \theta^*\right\|_1} \mathcal{T}_\epsilon(\Delta) +\mathcal{T}_\delta(\Delta)  \geq \frac{\eta^2}{4}\right\} .
\end{align*}}}
The first contagion is directly followed by Inequality (\ref{aux-T4-1}). For the second one observe that that in $\mathcal{E}_1^c$, using the triangle inequality, $\left\|\nabla_{G_K} \Delta\right\|_1 \leq 6\left\|\nabla_{G_K} \theta^*\right\|_1$. 
In consequence,
\begin{align*}
& \mathbb{P}\left(\left\{\left\|\widehat{\theta}_{\widetilde\lambda}-\theta^*\right\|_2>\eta\right\} \cap \mathcal{E}_1^c\right) \\
& \leq \mathbb{P}\left\{\sup _{\|\Delta\|_2 \leq \eta,\left\|\nabla_{G_K} \Delta\right\|_1 \leq 6\left\|\nabla_{G_K} \theta^*\right\|_1} \mathcal{T}_\epsilon(\Delta)+\mathcal{T}_\delta(\Delta) \geq \frac{\eta^2}{4}\right\} \\
& \leq\mathbb{P}\left\{\sum_{\|\Delta\|_2 \leq \eta,\left\|\nabla_{G_K} \Delta\right\|_1 \leq 6\left\|\nabla_{G_K} \theta^*\right\|_1} \mathcal{T}_\epsilon(\Delta)  \geq \frac{\eta^2}{8}\right\} \\
&+\mathbb{P}\left\{\sum_{\|\Delta\|_2 \leq \eta,\left\|\nabla_{G_K} \Delta\right\|_1 \leq 6\left\|\nabla_{G_K} \theta^*\right\|_1} \mathcal{T}_\delta(\Delta)  \geq \frac{\eta^2}{8}\right\}.
\end{align*}
By the same argument that leads to Inequality \ref{eps-bound-P} and Inequality \ref{delta-b-p}, it holds that
\begin{align*}
& \mathbb{P}\left\{\sup _{\|\Delta\|_2 \leq \eta,\left\|\nabla_{G_K} \Delta\right\|_1 \leq 6\left\|\nabla_{G_K} \theta^*\right\|_1} \mathcal{T}_\epsilon(\Delta) \geq \frac{\eta^2}{8}\right\} \leq \frac{\varepsilon}{4} , \\
& \mathbb{P}\left\{\sup _{\|\Delta\|_2 \leq \eta,\left\|\nabla_{G_K} \Delta\right\|_1 \leq 6\left\|\nabla_{G_K} \theta^*\right\|_1} \mathcal{T}_\delta(\Delta) \geq \frac{\eta^2}{8}\right\} \leq \frac{\varepsilon}{4} .
\end{align*}
Therefore
\begin{equation}
\label{f2-T4}
\mathbb{P}\left(\left\{\left\|\widehat{\theta}_{\widetilde\lambda}-\theta^*\right\|_2>\eta\right\} \cap \mathcal{E}_1^c\right) \leq  \frac{\varepsilon}{2} .
\end{equation}
{\bf{Step 3.}} It follows from Inequality (\ref{f1-T4}) and Inequality (\ref{f2-T4}) that
$$
\mathbb{P}\left(\left\{\left\|\widehat{\theta}_{\widetilde\lambda}-\theta^*\right\|_2>\eta\right\}\right) \leq \mathbb{P}\left(\left\{\left\|\widehat{\theta}_{\widetilde\lambda}-\theta^*\right\|_2>\eta\right\} \cap \mathcal{E}_1^c\right)+\mathbb{P}\left(\mathcal{E}_1\right) \leq  \varepsilon ,
$$
concluding the claim.
\\
{\bf{Step B.}} Now the assertion of Equation (\ref{PT-eq2}) in Chapter 2 is demonstrated. The proof is completely the same as the one for {\bf{Step B}} of Theorem 3 and is omitted.

\end{proof}

\newpage
\section{Proof of Lemma \ref{Mlb-M}}\label{ProofLemma1}
\begin{proof} First, we consider the simpler model where $\delta_i(x) = 0$ for all $i\in[n],$ and $x\in[0,1]$, i.e., there is no spatial noise. This is, 
$$y_{i,j}=f^*(x_{i,j})+\epsilon_{i,j}.$$
Denote by $\mathcal{F}(L),$ the family of piecewise Lipschitz functions defined in Appendix \ref{plsection}. By the discussion in Appendix \ref{plsection}, there exists $C_{\text{opt,1}}$, a positive constant, such that
$$
\limsup _{n \rightarrow \infty} \inf _{\widetilde{f}} \sup _{f^*\in \mathcal{F}(L)} \mathbb{P}\left(\left\|\widetilde{f}-f^*\right\|_{_2}^2>C_{\text{opt,1}}(nm)^{-1/d}\right)>0.
$$
 It now suffices to show that there exist a constant $C_{\text{opt,2}}>0,$ such that
$$
\limsup _{n \rightarrow \infty} \inf _{\widetilde{f}} \sup _{f^*\in \mathcal{F}(L)} \mathbb{P}\left(\left\|\widetilde{f}-f^*\right\|_{\mathcal{L}_2}^2>C_{\text{opt,2}} n^{-1}\right)>0.
$$
To this end, we assume that $\epsilon_{i,j}=0, \ \forall i=1,...,n;j=1,...,m_i.$ Therefore, we obtain a simple functional mean model given by,
\begin{equation}
\label{model-only-spnoi}
y_{i,j}(x_{i,j})=X(x_{i,j})=f^*(x_{i,j})+\delta_i(x_{i,j}), i\in[n],\ j\in [m_i].
\end{equation}
Here $\{\delta_{i}\}_{i=1}^{n}$ denotes the spatial noise and follow Assumption \ref{assume:tv functions}. Consider $L=\frac{2}{C_{\beta}}\log(n)$ with $C_{\beta}>0$ a constant stated in Assumption \ref{assume:tv functions}. Take into account the copies of the design points and measurement errors, denoted as $\{x_{i,j}^*\}_{i=1,j=1}^{n,m_i}$ and $\{\epsilon_{i,j}^*\}_{i=1,j=1}^{n,m_i}$, which were generated in Appendix \ref{BetaMsection}.  During the formation of these copies in Appendix \ref{BetaMsection}, the quantities $ | \mathcal J_{e,l}|$ and $ | \mathcal J_{o,l}|$, which represent the count of even and odd blocks respectively, satisfy $ | \mathcal J_{e,l}|,| \mathcal J_{o,l}|  \asymp n/L$. Thus, there exist positive constants $c_1$ and $c_2$ such that  
 \begin{equation}
     \label{BoundBlocksize-K-ap-lemma1}
      c_1n/L\le| \mathcal J_{e,l}|,| \mathcal J_{o,l}| \le  c_2n/L.
 \end{equation}
By Lemma \ref{Ind-cop}, we have that $\{x_{i,j}^*\}_{i\in\mathcal{J}_{e,l},j\in[m_i]}$ and $\{\epsilon_{i,j}^*\}_{i\in\mathcal{J}_{o,l},j\in[m_i]}$ are independent. Further, for any $i\in[n]$ it follows that $\mathbb{P}(x_{i,j}\neq x_{i,j}^*\ \text{for some}\ j\in[m_i])\le \beta_{x}(L).$ Moreover, by Assumption $\ref{assume:tv functions}$, we note that $\mathbb{P}(x_{i,j}\neq x_{i,j}^*\ \text{for some}\ j\in[m_i])\le \beta_{x}(L)\le \frac{1}{n^2}.$  Therefore, the event $\Omega_x=\bigcap_{i=1}^{n}\{x_{i,j}=x_{i,j}^*\ \forall\ j\in[m_i]\}$ happens with probability at least $\frac{1}{n}$. The same is satisfied for  $\{\delta_{i}^*\}_{i=1}^{n}$ with $\Omega_\delta=\bigcap_{i=1}^{n}\{\delta_{i}=\delta_{i}^*\}$. We denote by $\Omega$ the event $\Omega_x\cap\Omega_\delta.$

Let $\widetilde{f}$ be an estimator for the model described in Equation (\ref{model-only-spnoi}) based on the observations $\{(x_{i,j},y_{i,j})\}_{i=1,j=1}^{n,m_i}$. We let $\theta^*\in\mathbb{R}^{\sum_{i=1}^nm_i}$ given by $\theta_{i,j}^*=f^*(x_{i,j})$ for any $(i,j)\in [n]\times[m_i]$, where the abuse of notation explained in Section \ref{sec:notation} is used. Similarly, we consider $\mu^*,\widetilde{\theta},\widetilde{\mu}\in\mathbb{R}^{\sum_{i=1}^nm_i}$ with each entry given by $\mu_{i,j}^*=f^*(x_{i,j}^*)$, $\widetilde{\theta}_{i,j}=\widetilde{f}(x_{i,j})$, and $\widetilde{\mu}_{i,j}=\widetilde{f}(x_{i,j}^*)$ for any $(i,j)\in [n]\times[m_i]$, respectively.

Let $l\in[L]$. Under the event $\Omega$ we notice that $\widetilde{f}$ is an estimator for the model
\begin{equation}
\label{aux-model-lemma1}
y_{i,j}=f^*(x_{i,j}^*)+\delta_i^*(x_{i,j}^*),\ i\in{\mathcal{J}_{e,l}},\ j\in[m_i].
\end{equation}
Specifically, we have that under $\Omega$ the truncated vector $\widetilde{\mu}^{\mathcal{J}_{e,l}}\in \mathbb{R}^{\sum_{i\in\mathcal{J}_{e,l}}m_i}$, given by $\widetilde{\mu}^{\mathcal{J}_{e,l}}_{i,j}=\widetilde{\mu}_{i,j}$ for any $(i,j)\in\mathcal{J}_{e,l}\times[m_i]$, is an estimator of the truncated vector $\mu^{*,\mathcal{J}_{e,l}}\in \mathbb{R}^{\sum_{i\in\mathcal{J}_{e,l}}m_i}$ given by $\mu^{*,\mathcal{J}_{e,l}}_{i,j}=\mu^*_{i,j}$ for any $(i,j)\in\mathcal{J}_{e,l}\times[m_i]$.
Now observe that when looking for the minimax lower bound for the model described in Equation (\ref{aux-model-lemma1}), for each time $i$ only one observation among the $m_i$ available observations contributes to the determination of such minimax lower bound. To see this, suppose that conditioning in $x_{i,j}^*$, we examine the distributions
$P_1$ and $P_2$ originating from $Z_{i,1}=(z_{i,1},...,z_{i,1})^{t}\in \mathbb{R}^{m_i}$ and $Z_{i,2}=(z_{i,1},...,z_{i,1})^{t}\in \mathbb{R}^{m_i}$, where $z_{i,l} \sim \mathcal{N}(\mu_l,1)$ for $l=1$ and $l=2.$ Note that the Kullback-Leibler divergence  $KL(P_1,P_2)$ is solely dependent on $KL(\mathcal{N}(\mu_1,1),\mathcal{N}(\mu_2,1))$. By invoking Le Cam's lemma, the problem reduces to the estimation of the mean based on $\vert\mathcal{J}_{e,l}\vert$ observations. Given the independence of the copies of the design points and spatial noises, it is established that the minimax lower bound for this task is $\frac{1}{\vert \mathcal{J}_{e,l}\vert}$. Thus,
\[
    \underset{ n \rightarrow \infty}{\lim\sup}\,\underset{f^* \in \mathcal{F}     }{ \sup     }\,\mathbb{P}\left(  \| \widetilde{\mu}^{\mathcal{J}_{e,l}}-\mu^{*,\mathcal{J}_{e,l}} \|_{2}^2 \,\geq \,      \frac{C_{\text{opt,2},\mathcal{J}_{e,l}}}{\vert \mathcal{J}_{e,l}\vert }  \right)\,>0,\,
   \]
   for a positive constant $C_{\text{opt,2},\mathcal{J}_{e,l}}.$ 
   Assume that $Cm\ge m_i\ge cm,$ for a positive constant $c,C$, and $m=\Big(\frac{1}{n}\sum_{i=1}^n\frac{1}{m_i}\Big)^{-1}$. From this fact and Inequality (\ref{BoundBlocksize-K-ap}) we obtain that 
   \begin{equation}
   \label{AppendixK-eq1-lemma1}
    \underset{ n \rightarrow \infty}{\lim\sup}\,\underset{f^* \in \mathcal{F}     }{ \sup     }\,\mathbb{P}\left( \| \widetilde{\mu}^{\mathcal{J}_{e,l}}-\mu^{*,\mathcal{J}_{e,l}} \|_{2}^2 \,\geq \,       \frac{C_{\text{opt,2},\mathcal{J}_{e,l}}L}{  Cc_2n } \right)\,>C_0>0.\,
   \end{equation}
   for a constant $C_0.$
Similarly, under the event $\Omega$, $\widetilde{f}$ is an estimator for the model
$$y_{i,j}=f^*(x_{i,j}^*)+\delta_{i}^*(x_{i,j}^*),\ i\in{\mathcal{J}_{o,l}},\ j\in[m_i],$$ 
with
\begin{equation}
   \label{AppendixK-eq2-lemma1}
    \underset{ n \rightarrow \infty}{\lim\sup}\,\underset{f^* \in \mathcal{F}     }{ \sup     }\,\mathbb{P}\left(  \| \widetilde{\mu}^{\mathcal{J}_{o,l}}-\mu^{*,\mathcal{J}_{o,l}} \|_{2}^2 \,\geq \,       \frac{C_{\text{opt,2},\mathcal{J}_{o,l}}L}{  Cc_2n } \right)\,>C_0>0.\,
   \end{equation}
   for a positive constant $C_{\text{opt,2},\mathcal{J}_{o,l}}.$ 

   We are now going to elaborate on the final observations of this discussion. Let $C_{\text{opt,2}}=\frac{1}{Cc_2}\min_{l\in[L]}\{{C_{\text{opt,2},\mathcal{J}_{e,l}},C_{\text{opt,2},\mathcal{J}_{o,l}}}\}$. Then,
\begin{align*}
     \underset{ n \rightarrow \infty}{\lim\sup}\,\underset{f^* \in \mathcal{F}     }{ \sup     }\,\mathbb{P}\left(  \| \widetilde{\theta}-\theta^* \|_{2}^2 \,\geq \,       \frac{C_{\text{opt,2}}}{  n } \right)\ge \underset{ n \rightarrow \infty}{\lim\sup}\,\underset{f^* \in \mathcal{F}     }{ \sup     }\,\mathbb{P}\left( \Big\{ \|\widetilde{\theta}-\theta^* \|_{2}^2 \,\geq \,       \frac{C_{\text{opt,2}}}{  n } \Big\}\cap\Omega\right) >0.
\end{align*}
The logic behind the second inequality is as follows. For any $f^*\in\mathcal{F},$
\begin{align}
\label{AppendixK-eq3-lemma1}
      &\mathbb{P}\left( \Big\{ \| \widetilde{\theta}-\theta^* \|_{2}^2 \,\geq \,       \frac{C_{\text{opt,2}}}{  n } \Big\}\cap\Omega\right)\nonumber
      \\
      \ge &\mathbb{P}\left( \Big\{ \| \widetilde{\theta}-\theta^* \|_{2}^2 \,\geq \,       \frac{LC_{\text{opt,2}}}{  n } \Big\}\cap\Omega\right)\nonumber
      \\
      \ge&\mathbb{P}\left(\bigcup_{l=1}^L \Big\{\Big\{ \| \widetilde{\mu}^{\mathcal{J}_{e,l}}-\mu^{*,\mathcal{J}_{e,l}} \|_{2}^2 \,\geq \,       \frac{LC_{\text{opt,2},\mathcal{J}_{e,l}}}{  Cc_2n } \Big\}\bigcup \Big\{ \| \widetilde{\mu}^{\mathcal{J}_{o,l}}-\mu^{*,\mathcal{J}_{o,l}} \|_{2}^2 \,\geq \,       \frac{LC_{\text{opt,2},\mathcal{J}_{o,l}}}{  Cc_2n } \Big\}\Big\}\cap\Omega\right)\nonumber
      \\
      >&C_0,
\end{align}
where the first inequality is followed by the fact that $L=\frac{1}{C_\beta}\log(n)\ge1.$ Moreover the third inequality is achieved from Inequality (\ref{AppendixK-eq1-lemma1}) and (\ref{AppendixK-eq2-lemma1}). The derivation of the second inequality is explained below.
Under the event $\Omega$ the relation, $$\| \widetilde{\theta}-\theta^* \|_{2}^2=\|\widetilde{\mu}-\mu^*\|_2^2=\sum_{l=1}^L\Big(\| \widetilde{\mu}^{\mathcal{J}_{e,l}}-\mu^{*,\mathcal{J}_{e,l}} \|_{2}^2 +\| \widetilde{\mu}^{\mathcal{J}_{o,l}}-\mu^{*,\mathcal{J}_{o,l}} \|_{2}^2\Big),$$ is satisfied. Therefore, $$\| \widetilde{\theta}-\theta^*  \|_{2}^2\ge \Big\{\| \widetilde{\mu}^{\mathcal{J}_{e,l}}-\mu^{*,\mathcal{J}_{e,l}} \|_{2}^2,\| \widetilde{\mu}^{\mathcal{J}_{o,l}}-\mu^{*,\mathcal{J}_{o,l}} \|_{2}^2\Big\},$$ for any $l\in[L].$ These specific conditions allow us to conclude that the event $\{\| \widetilde{\theta}-\theta^* \|_{2}^2< \frac{LC_{\text{opt,2}}}{  n } \}$ is contained in the event
$$\bigcap_{l=1}^L \Big\{\Big\{ \| \widetilde{\mu}^{\mathcal{J}_{e,l}}-\mu^{*,\mathcal{J}_{e,l}} \|_{2}^2 \,< \,       \frac{LC_{\text{opt,2},\mathcal{J}_{e,l}}}{  Cc_2n } \Big\}\bigcap \Big\{ \| \widetilde{\mu}^{\mathcal{J}_{o,l}}-\mu^{*,\mathcal{J}_{o,l}} \|_{2}^2 \,< \,       \frac{LC_{\text{opt,2},\mathcal{J}_{o,l}}}{  Cc_2n } \Big\}\Big\},$$
and consequently the second inequality in Inequality (\ref{AppendixK-eq3-lemma1}) is satisfied. Specifically, suppose that  $\|\widetilde{\theta}-\theta^*  \|_{2}^2 \,< \,       \frac{LC_{\text{opt,2}}}{  n } $. For any $l\in[L]$ by the definition of $C_{\text{opt,2}}$ we have that $\| \widetilde{\theta}-\theta^*  \|_{2}^2 \,< \,       \frac{LC_{\text{opt,2},\mathcal{J}_{e,l}}}{  Cc_2n }$. It follows that $\| \widetilde{\mu}^{\mathcal{J}_{e,l}}-\mu^{*,\mathcal{J}_{e,l}} \|_{2}^2 \,< \,       \frac{LC_{\text{opt,2},\mathcal{J}_{e,l}}}{  Cc_2n }$. Similarly, $\| \widetilde{\mu}^{\mathcal{J}_{o,l}}-\mu^{*,\mathcal{J}_{o,l}}\|_{2}^2 \,< \,       \frac{LC_{\text{opt,2},\mathcal{J}_{o,l}}}{  Cc_2n }$, concluding the contention of the aforementioned events.

The claim is then followed by taking $C_{\text{opt}}=\frac{1}{2}\min\{C_{\text{opt,1}},C_{\text{opt,2}}\}$.
\end{proof}
\newpage
\section{Additional Technical Results for proof of Theorem \ref{thm:main tv} and Theorem \ref{Penalized}}\label{AppendixI}

\begin{lemma}\label{lemma1}
Let $\{x_{i,j}\}_{i=1,j=1}^{n,m_i}$ be $\beta$-mixing random variables in $[0,1]$ with coefficients having exponential decay as in Assumption \ref{assume:tv functions}. Assume that $cm\le m_i\le Cm$ for any $i\in\{1,...,n\}$ with $m=\Big(\frac{1}{n}\sum_{i=1}^{n}\frac{1}{m_i}\Big)^{-1}$ and some constants $c,C>0$.
The event $\Omega$  such that 
	\[
	  \| p\|_{\infty} ^2\,\leq\,  \frac{4C_k^2}{cnm}\sum_{i=1}^{n}  \sum_{j=1}^{m_i} p(x_{i,j})^2,
	\]
	for all  $p(x)$  polynomial of  degree $k-1$ happens, with probability at least $1 - \frac{32\log^3(n)}{C_{\beta}^3c_1^2c^2n^2m^2}- \frac{\log(n)}{n^2}$, where $c_1>0$ and $C_k>0$ that only depends on $k,$ are absolute constants.
\end{lemma}
\begin{proof} Let
\begin{equation}
\label{Lchoice-lemma1-t1-t2}
L=\frac{2}{C_{\beta}}\log(n).
\end{equation} 
Using the exponential decay in Assumption 1{\bf{f}} we have that
 \begin{equation}
 \label{betamixcoef-ineq-aux-lemma1-t1-t2}
     \beta_x(L)\le\frac{1}{n^2}.
 \end{equation} Subsequently throughout our discussion we make frequent references to the copies of the random variables $x_{i,j}$ as they are constructed in Section \ref{BetaMsection}. These copies are denoted as $x_{i,j}^*$. Furthermore as delineated in Section \ref{BetaMsection} during the formation of these copies the quantities $ | \mathcal J_{e,l}|$ and $ | \mathcal J_{o,l}|$, which represent the count of even and odd blocks respectively, satisfy $$ | \mathcal J_{e,l}|,| \mathcal J_{o,l}|  \asymp n/L.$$ Thus there exist positive constants $c_1$ and $c_2$ such that  
\begin{equation}
\label{aux-lemma1-t1-t2-block}    
 c_1 \frac{n}{L}\le|\mathcal{J}_{o,l}|,|\mathcal{J}_{e,l}|\le c_2 \frac{n}{L}.
\end{equation}
	Let  $\{ q_l\}_{l=0}^{k-1}$  be  an orthonormal  basis of the set of polynomials in $[0,1]$  with respect to the usual inner product in $\mathcal{L}_2([0,1])$. Thus,\[
	 \int_{0}^1    q_l(x) q_{  l^{\prime} }(x)dx  \,=\,  1_{  \{l= l^{\prime}\} },
	\]
	for all $l,l^{\prime} \in \{0,\ldots,k-1\}$.  Then any polynomial $p(x)$ of degree $k-1$ can be written as
      \[
      p\,=\,  \sum_{l=0}^{k-1}  a_l q_l,
      \]
	 for some constants $a_0,\ldots, a_{k-1} \in \mathbb{R}$. We notice that 
    \begin{equation}
	\label{Leqn:e0}
	 \sum_{l=0}^{k-1}  a_l^2\,=\,\int_0^1  p(x)^2 dx.
 \end{equation}
By applying the triangle inequality and adding and subtracting the term $$\frac{1}{Cnm} \sum_{i=1}^{n} \sum_{j=1}^{m_i} p(x_{i,j})^2,$$ we obtain the following
\begin{align*}
    &\displaystyle \left	\vert	\frac{1}{Cnm}\sum_{i=1}^{n}   \sum_{j=1}^{m_i}\int_{0}^{1}  p(x)^2       \,-\,    \frac{1}{Cnm}\sum_{i=1}^{n}   \sum_{j=1}^{m_i}   p(x_{i,j})^2 \right\vert 
    \\
    \le&
    \displaystyle \left	\vert	\frac{1}{Cnm}\sum_{l=1}^{L}\sum_{i\in \mathcal{J}_{e,l}}   \sum_{j=1}^{m_i}\int_{0}^{1}  p(x)^2       \,-\,    \frac{1}{Cnm}\sum_{l=1}^{L}\sum_{i\in \mathcal{J}_{e,l}}   \sum_{j=1}^{m_i}  p(x_{i,j}^*)^2 \right\vert 
    \\
    +&\displaystyle \left	\vert	\frac{1}{Cnm}\sum_{l=1}^{L}\sum_{i\in \mathcal{J}_{o,l}}   \sum_{j=1}^{m_i}\int_{0}^{1}  p(x)^2       \,-\,    \frac{1}{Cnm}\sum_{l=1}^{L}\sum_{i\in \mathcal{J}_{o,l}}   \sum_{j=1}^{m_i}  p(x_{i,j}^*)^2 \right\vert 
    \\
    +&\left\vert\frac{1}{Cnm}\sum_{i=1}^{n}   \sum_{j=1}^{m_i}  \{ p(x_{i,j})^2-p(x_{i,j}^*)^2\}\right\vert=: I_1+I_2+I_3.
\end{align*}
Now we analyze each of the terms $I_1,I_2$ and $I_3.$ We start with $I_1.$ The analysis for $I_2$ is similar and is omitted. Applying triangle inequality and adding and subtracting $\frac{1}{Cnm}\sum_{l=1}^{L}\sum_{i\in \mathcal{J}_{e,l}}   \sum_{j=1}^{m_i}   p(t_{i,j})^2,$
\begin{equation}
	\label{Leqn:e1}
		\begin{array}{lll}
		&\displaystyle \left	\vert	\frac{1}{Cnm}\sum_{l=1}^{L}\sum_{i\in \mathcal{J}_{e,l}}   \sum_{j=1}^{m_i}\int_{0}^{1}  p(x)^2       \,-\,    \frac{1}{Cnm}\sum_{l=1}^{L}\sum_{i\in \mathcal{J}_{e,l}}   \sum_{j=1}^{m_i}  p(x_{i,j}^*)^2 \right\vert  
  \\
  \leq&	\displaystyle \left	\vert	\frac{1}{Cnm}\sum_{l=1}^{L}\sum_{i\in \mathcal{J}_{e,l}}   \sum_{j=1}^{m_i}\int_{0}^{1}  p(x)^2       \,-\,    \frac{1}{Cnm}\sum_{l=1}^{L}\sum_{i\in \mathcal{J}_{e,l}}   \sum_{j=1}^{m_i}   p(t_{i,j})^2 \right\vert \\
		+ &  \displaystyle \,\,\left	\vert	 \frac{1}{Cnm}\sum_{l=1}^{L}\sum_{i\in \mathcal{J}_{e,l}}   \sum_{j=1}^{m_i}   p(x^*_{i,j})^2      \,-\,    \frac{1}{Cnm}\sum_{l=1}^{L}\sum_{i\in \mathcal{J}_{e,l}}   \sum_{j=1}^{m_i}  p(t_{i,j})^2 \right\vert\\
		 =:& A+B,
	\end{array}
\end{equation}
where for each $l\in\{1,...,L\}$, $\{t_{i,j}\}_{i\in{\mathcal{J}_{e,l}},j=1,...,m_i}$ are in increasing order and evenly spaced in $(0,1)$. We now proceed to bound $A$ and $B$. 
First, we conduct the analysis for $B$.
Denote by $\{x^*_{(i),(j)} \}_{i\in{\mathcal{J}_{e,l}},j=1,...,m_i}$ be the order statistics of $\{x^*_{i,j} \}_{i\in{\mathcal{J}_{e,l}},j=1,...,m_i}$. Since $\{x^*_{i,j} \}_{i\in{\mathcal{J}_{e,l}},j=1,...,m_i}$ are independent, by Dvoretzky–Kiefer–Wolfowitz–Massart inequality the event
\[
  \underset{i\in{\mathcal{J}_{e,l}},j=1,...,m_i}{\max}\, \vert      x^*_{(i) ,(j)}  - t_{i,j}\vert \,\leq\,    \sqrt{\frac{\log (\sum_{i\in\mathcal{J}_{e,l}}m_i)}{\sum_{i\in\mathcal{J}_{e,l}}m_i}}\le \sqrt{\frac{\log(\sum_{i\in \mathcal{J}_{e,l}}m_i)L}{cc_1nm}},
\]
happens with a probability of at least 
\begin{align*}
    &1 - 2\exp\left(   - 2 \log\Big(\sum_{i\in\mathcal{J}_{e,l}}m_i\Big) \right)
    \\
    \ge& 1 - 2\exp\left(   - 2 \log\Big(cc_1nm/L\Big) \right)
    \\
    =&1-\frac{2L^2}{c_1^2c^2n^2m^2}=1-\frac{8\log^2(n)}{C_{\beta}^2c_1^2c^2n^2m^2}.
\end{align*}
The inequality is followed by Equation (\ref{Lchoice-lemma1-t1-t2}), Inequality (\ref{aux-lemma1-t1-t2-block}), and Assumption 1{\bf{f}}. Hence, by union bound and our choice of $L$ in Equation (\ref{Lchoice-lemma1-t1-t2}) we get that
the event
\[
  \underset{l=0,...,L-1}{\max}\, \underset{i\in{\mathcal{J}_{e,l}},j=1,...,m_i}{\max}\, \vert      x^*_{(i), (j)}  - t_{i,j}\vert \,\leq\,    \sqrt{\frac{\log (\sum_{i\in\mathcal{J}_{e,l}}m_i)}{\sum_{i\in\mathcal{J}_{e,l}}m_i}}\le \sqrt{\frac{\log(\sum_{i\in \mathcal{J}_{e,l}}m_i)L}{cc_1nm}},
\]
holds with a probability of at least 
$$1-\frac{16\log^3(n)}{C_{\beta}^3c_1^2c^2n^2m^2}.$$
Denote this event by $\widetilde{\Omega}$. Therefore, if  $\widetilde{\Omega}$ holds, it follows that
\begin{equation}
	\label{Leqn:e2}
	   \begin{array}{lll}
		B    & \leq&     \displaystyle  \frac{1}{Cnm}\sum_{l=1}^{L}\sum_{i\in \mathcal{J}_{e,l}}   \sum_{j=1}^{m_i} \left\vert   p(x_{(i),(j)}^*)^2 - p(t_{i,j})^2 \right\vert  \\
		&\leq &\displaystyle  \frac{1}{Cnm}\sum_{l=1}^{L}\sum_{i\in \mathcal{J}_{e,l}}   \sum_{j=1}^{m_i}   \| (p^2){\prime} \|_{\infty}   \left\vert  x_{(i),(j)}^* - t_{i,j} \right\vert  \\
		&\leq &\displaystyle  \frac{ \| (p^2){\prime} \|_{\infty}   \sqrt{\log (Cc_2nm/L)L} }{Cnm\sqrt{cc_1nm}}\sum_{l=1}^{L}\sum_{i\in \mathcal{J}_{e,l}}   \sum_{j=1}^{m_i}   1
  \\
		& \leq&\displaystyle  \frac{ \| (p^2){\prime} \|_{\infty}   \sqrt{\log (Cc_2nm/L)L} }{Cnm\sqrt{cc_1nm}}Lc_2\frac{n}{L}Cm
  \\
		& \leq&\displaystyle      c_2\| (p^2){\prime} \|_{\infty}     \sqrt{\frac{2\log (Cc_2nm)\log(n)}{C_{\beta}cc_1nm}},
	\end{array}
\end{equation}
where the second inequality follows from the mean value theorem, and the third one is due to the event $\widetilde{\Omega}$. The fourth and the last inequality are obtained by Equation (\ref{Lchoice-lemma1-t1-t2}), Inequality (\ref{aux-lemma1-t1-t2-block}), and Assumption 1{\bf{f}}.
\\
Now the analysis for $A$ is presented. Let $l\in[L].$ For $i \in\mathcal{J}_{e,l}$ and $j\in\{1,...,m_i\}$ define  $A_1^{1} =(0, t_{1,1} )$, $A_j^i =  ( t_{i,j-1 }, t_{i,j}  )  $, and $A_{m_{L\lfloor n/L\rfloor+l}}^{L\lfloor n/L\rfloor+l}=(t_{L\lfloor n/L\rfloor+l,m_{L\lfloor n/L\rfloor+l}},1)$. Then we observe the following. First we use triangle inequality and multiply and divide by $\sum_{i\in{\mathcal{J}_{e,l}}}m_i$. Next observe that $$\frac{1}{\sum_{i\in{\mathcal{J}_{e,l}}}m_i} \sum_{i\in \mathcal{J}_{e,l}}   \sum_{j=1}^{m_i}   \int_{0}^1  p(x)^2 dx=\int_{0}^1  p(x)^2 dx.$$ Further by our choice of the intervals $A_{i}^j$ we have that $$\frac{1}{\sum_{i\in{\mathcal{J}_{e,l}}}m_i}=\int_{A_{i}^j}1dx.$$ Thus, 
\begin{equation}
	\label{Leqn:e31}
	\begin{array}{lll}
		 A & \le&  \displaystyle \sum_{l=1}^{L}\frac{\sum_{i\in{\mathcal{J}_{e,l}}}m_i}{Cnm}  \left\vert\frac{1}{\sum_{i\in{\mathcal{J}_{e,l}}}m_i} \sum_{i\in \mathcal{J}_{e,l}}   \sum_{j=1}^{m_i}   \{ \int_{0}^1  p(x)^2 dx   \,-\,      p(t_{i,j})^2  \}  \right\vert
   \\
&=& \displaystyle \sum_{l=1}^{L}\frac{\sum_{i\in{\mathcal{J}_{e,l}}}m_i}{Cnm}  \left\vert   \int_{0}^1  p(x)^2 dx   \,-\,\frac{1}{\sum_{i\in{\mathcal{J}_{e,l}}}m_i} \sum_{i\in \mathcal{J}_{e,l}}   \sum_{j=1}^{m_i}      \{ p(t_{i,j})^2  \}  \right\vert
  \\
&=& \displaystyle \sum_{l=1}^{L}\frac{\sum_{i\in{\mathcal{J}_{e,l}}}m_i}{Cnm}  \left\vert   \int_{0}^1  p(x)^2 dx   \,-\, \sum_{i\in \mathcal{J}_{e,l}}   \sum_{j=1}^{m_i}      \{ \int_{A_j^i}p(t_{i,j})^2 dx \}  \right\vert
  \\
&=& \displaystyle \sum_{l=1}^{L}\frac{\sum_{i\in{\mathcal{J}_{e,l}}}m_i}{Cnm}  \left\vert   \sum_{i\in \mathcal{J}_{e,l}}   \sum_{j=1}^{m_i}      \{ \int_{A_j^i}p(x)^2 -p(t_{i,j})^2dx  \}  \right\vert.
 	\end{array}
\end{equation}
From Inequality (\ref{Leqn:e31}) and using the same inequalities as for $B$ we have that
\begin{equation}
	\label{Leqn:e3}
	\begin{array}{lll}
		 A 
&\le& \displaystyle \sum_{l=1}^{L}\frac{\sum_{i\in{\mathcal{J}_{e,l}}}m_i}{Cnm}     \sum_{i\in \mathcal{J}_{e,l}}   \sum_{j=1}^{m_i}      \{ \int_{A_j^i}\left\vert p(x)^2 -p(t_{i,j})^2\right\vert dx \}  

\\
&\le& \displaystyle \sum_{l=1}^{L}\frac{\sum_{i\in{\mathcal{J}_{e,l}}}m_i}{Cnm}     \sum_{i\in \mathcal{J}_{e,l}}   \sum_{j=1}^{m_i}       \| (p^2){\prime} \|_{\infty}    \int_{A_j^i}        \vert   x -  t_{i,j} \vert  dx 

\\
&\le& \displaystyle \sum_{l=1}^{L}\frac{\sum_{i\in{\mathcal{J}_{e,l}}}m_i}{Cnm}     \sum_{i\in \mathcal{J}_{e,l}}   \sum_{j=1}^{m_i}   \frac{1}{\sum_{i\in{\mathcal{J}_{e,l}}}m_i}    \| (p^2){\prime} \|_{\infty}    \int_{A_j^i}          dx  

\\
&=& \displaystyle \sum_{l=1}^{L}\frac{\sum_{i\in{\mathcal{J}_{e,l}}}m_i}{Cnm}     \sum_{i\in \mathcal{J}_{e,l}}   \sum_{j=1}^{m_i}   \frac{1}{\Big(\sum_{i\in{\mathcal{J}_{e,l}}}m_i\Big)^2}    \| (p^2){\prime} \|_{\infty}       
\\
		  	  	  	   & =& 
              \displaystyle  \sum_{l=1}^{L}\frac{ \| (p^2){\prime} \|_{\infty}  }{Cnm}
               \\
		& =&\displaystyle   \| (p^2){\prime} \|_{\infty}     \frac{2\log(n)}{C_\beta Cnm}.
 	\end{array}
\end{equation}
Then combining (\ref{Leqn:e1})--(\ref{Leqn:e3}) we obtain 
\begin{align}
	\label{Leqn:e4}
	\begin{array}{lll}
		&\displaystyle \left	\vert	\frac{1}{Cnm}\sum_{l=1}^{L}\sum_{i\in \mathcal{J}_{e,l}}   \sum_{j=1}^{m_i}\int_{0}^{1}  p(x)^2       \,-\,    \frac{1}{Cnm}\sum_{l=1}^{L}\sum_{i\in \mathcal{J}_{e,l}}   \sum_{j=1}^{m_i}  p(x_{i,j}^*)^2 \right\vert  \\&\leq  \displaystyle      \displaystyle      c_2\| (p^2){\prime} \|_{\infty}     \sqrt{\frac{2\log (Cc_2nm)\log(n)}{C_{\beta}cc_1nm}}+\displaystyle   \| (p^2){\prime} \|_{\infty}     \frac{2C_{\beta}\log(n)}{Cnm}.
	\end{array}  
\end{align}
Similarly we have that
\begin{align}
	\label{Leqn:e5}
	\begin{array}{lll}
	&	\displaystyle \left	\vert	\frac{1}{Cnm}\sum_{l=1}^{L}\sum_{i\in \mathcal{J}_{o,l}}   \sum_{j=1}^{m_i}\int_{0}^{1}  p(x)^2       \,-\,    \frac{1}{Cnm}\sum_{l=1}^{L}\sum_{i\in \mathcal{J}_{o,l}}   \sum_{j=1}^{m_i}  p(x_{i,j}^*)^2 \right\vert \\
 &\leq \displaystyle      \displaystyle      c_2\| (p^2){\prime} \|_{\infty}     \sqrt{\frac{2\log (Cc_2nm)\log(n)}{C_{\beta}cc_1nm}}+\displaystyle   \| (p^2){\prime} \|_{\infty}     \frac{2\log(n)}{C_{\beta}Cnm},
	\end{array}  
\end{align}
with a probability of at least
$$1-\frac{16\log^3(n)}{C_{\beta}^3c_1^2c^2n^2m^2}.$$
Now we analyze the term $I_3.$ We observe that, 
\begin{align}
\label{helpCon-Eq2}
    \left\vert\frac{1}{Cnm}\sum_{i=1}^{n}   \sum_{j=1}^{m_i}   \{p(x_{i,j})^2-p(x_{i,j}^*)^2\}\right\vert\le\| (p^2){\prime} \|_{\infty}\frac{1}{Cnm}\sum_{i=1}^{n}   \sum_{j=1}^{m_i}   \left\vert x_{i,j}-x_{i,j}^*\right\vert,
\end{align}
and by Markov's inequality
\begin{align*}
    \mathbb{P}\Big(\frac{1}{Cnm}\sum_{i=1}^{n}   \sum_{j=1}^{m_i}   \left\vert x_{i,j}-x_{i,j}^*\right\vert\ge2\frac{1}{\log(n)}\Big)\le \frac{\mathbb{E}\Big(\frac{1}{Cnm}\sum_{i=1}^{n}   \sum_{j=1}^{m_i}   \left\vert x_{i,j}-x_{i,j}^*\right\vert\Big)\log(n)}{2},
\end{align*}
which can be bounded as,
\begin{align*}
  \frac{1}{2Cnm}\sum_{i=1}^{n}   \sum_{j=1}^{m_i}    \mathbb{E}\Big(\vert x_{i,j}-x_{i,j}^*\vert\Big)\log(n).
\end{align*}
Then by Lemma \ref{Ind-cop} and Inequality (\ref{betamixcoef-ineq-aux-lemma1-t1-t2}) we have that $\mathbb{P}(\{ x_{i,j} \not=x_{i,j}^* \})\le \frac{1}{n^2}$. It follows that,
\begin{align*}
    &\frac{1}{2Cnm}\sum_{i=1}^{n}   \sum_{j=1}^{m_i}    \mathbb{E}\Big(\vert x_{i,j}-x_{i,j}^*\vert\Big)\log(n)
    \\
    =&\frac{1}{2Cnm}\sum_{i=1}^{n}   \sum_{j=1}^{m_i}    \mathbb{E}\Big(\vert x_{i,j}-x_{i,j}^*\vert   \mathbf{1}_{ \{ x_{i,j} \not=x_{i,j}^* \} }  \Big)\log(n)
    \\
    +&\frac{1}{2Cnm}\sum_{i=1}^{n}   \sum_{j=1}^{m_i}    \mathbb{E}\Big(\vert x_{i,j}-x_{i,j}^*\vert   \mathbf{1}_{ \{ x_{i,j} =x_{i,j}^* \} }  \Big)\log(n)
    \\
    =&\frac{1}{2Cnm}\sum_{i=1}^{n}   \sum_{j=1}^{m_i}    \mathbb{E}\Big(\vert x_{i,j}-x_{i,j}^*\vert   \mathbf{1}_{ \{ x_{i,j} \not=x_{i,j}^* \} }  \Big)\log(n)+0
    \\
    \le&\frac{1}{2Cnm}\sum_{i=1}^{n}   \sum_{j=1}^{m_i}    \mathbb{E}\Big(2  \mathbf{1}_{ \{ x_{i,j} \not=x_{i,j}^* \} }  \Big)\log(n)
    \\
    \le&\frac{\log(n)}{Cnm}\sum_{i=1}^{n}   \sum_{j=1}^{m_i} \mathbb{P}(\{ x_{i,j} \not=x_{i,j}^* \})
    \\
    \le& \frac{\log(n)}{n^2}.
\end{align*}
Therefore the event 
$$\frac{1}{Cnm}\sum_{i=1}^{n}   \sum_{j=1}^{m_i}   \left\vert x_{i,j}-x_{i,j}^*\right\vert\le2\frac{1}{\log(n)},$$
holds with probability at least $1- \frac{\log(n)}{n^2}.$ In consequence, 
\begin{align}
\label{Leqn:e6}
    \| (p^2){\prime} \|_{\infty}\frac{1}{Cnm}\sum_{i=1}^{n}   \sum_{j=1}^{m_i}   \left\vert x_{i,j}-x_{i,j}^*\right\vert\le 2\| (p^2){\prime} \|_{\infty}\frac{1}{\log(n)},
\end{align}
holds with probability at least $1- \frac{\log(n)}{n^2}.$
Furthermore, 
\[
    \begin{array}{lll}
    (p(x)^2)^{\prime} &=& \displaystyle \left[ \sum_{l=0}^{k-1}  a_l^2 q_l(x)^2  \,+\,  \sum_{l_1 \neq  l_2}  a_{l_1} a_{l_2}  q_{l_1}(x) q_{l_2}(x) \right]^{\prime} \\ 
    & &\displaystyle \left[ \sum_{l=0}^{k-1}  a_l^2 [q_l(x)^2]^{\prime}  \,+\,  \sum_{l_1 \neq  l_2}  a_{l_1} a_{l_2}  [q_{l_1}(x) q_{l_2}(x)]^{\prime}  \right],\\ 
    \end{array}
\]
and so 
\[
\|  (p(x)^2){\prime}\|_{\infty}  \,\leq \, c_k' \|a\|_{\infty}^2,
\]
for some constant $c_k'>0$ that only depends on $k$. Therefore, from (\ref{Leqn:e4}),  (\ref{Leqn:e5}) and  (\ref{Leqn:e6})
\begin{equation*}
\begin{array}{lll}
	& &	\displaystyle \left \vert \frac{1}{Cnm}\sum_{i=1}^{n}   \sum_{j=1}^{m_i} \sum_{l=0}^{k-1} a_l^2    \,-\,    \frac{1}{Cnm}\sum_{i=1}^{n}   \sum_{j=1}^{m_i}  p(x_{i,j})^2\right \vert\\ & \leq& \displaystyle 2c_k'c_2 \|a\|_{\infty}^2  \sqrt{\frac{2\log (Cc_2nm)\log(n)}{C_{\beta}cc_1nm}}+2\displaystyle   \| (p^2){\prime} \|_{\infty}     \frac{2\log(n)}{Cnm}+2\| (p^2){\prime} \|_{\infty}\frac{1}{\log(n)}\\
		 & \leq& \displaystyle  2c_k'c_2 \left(  \sum_{l=0}^{k-1} a_l^2    \right)\sqrt{\frac{2\log (Cc_2nm)\log(n)}{C_{\beta}cc_1nm}}
   \\
   &+& 2\displaystyle  c_k' \left(  \sum_{l=0}^{k-1} a_l^2    \right)   \frac{2\log(n)}{Cnm}+\displaystyle  2 c_k'\left(  \sum_{l=0}^{k-1} a_l^2    \right)\frac{1}{\log(n)}.\\
\end{array}
\end{equation*}
This implies
{\small{\begin{equation}
\label{Leqn:e7}
\begin{array}{lll}
		& &\displaystyle \sum_{l=0}^{k-1} a_l^2    
		\\
		& \leq&    \displaystyle \frac{1}{\frac{\sum_{i=1}^nm_i}{Cnm}  -  2c_k'c_2 \sqrt{\frac{2\log (Cc_2nm)\log(n)}{C_{\beta}cc_1nm}}-4c_k^{'} \frac{\log(n)}{Cnm} -2c_k'\frac{1}{\log(n)} }   \frac{1}{Cnm}\sum_{i=1}^{n}  \sum_{j=1}^{m_i}   p(x_{i,j})^2\\
		 & \leq&  \displaystyle \frac{4C}{cCnm}\sum_{i=1}^{n}  \sum_{j=1}^{m_i}   p(x_{i,j})^2,
   \end{array}
\end{equation}}}
since $(nm)^b \gtrsim \log^a (nm)$ for any positive numbers $a$ and $b$, and  $\log(n) \gtrsim 1$. Hence,
\[
   \begin{array}{lll}
       \| p\|_{\infty}   & \leq&  \displaystyle \|a\|_{\infty}     \underset{x\in [0,1]}{\max} \sum_{l=0}^{k-1} \vert  q_l(x) \vert \\
        & \leq& \displaystyle C_k \left(\sum_{l=0}^{k-1} a_l^2 \right)^{1/2} \\
         & \leq& \displaystyle C_k\sqrt{\frac{4C}{cCnm}\sum_{i=1}^{n}  \sum_{j=1}^{m_i}   p(x_{i,j})^2},
   \end{array}
  \]
  for a positive constant $C_k.$

\end{proof}
\begin{remark}
\label{Rem1HPE}
From Assumption \ref{assume:tv functions} it follows that 
$$\frac{4CC_{k}^2}{Ccnm}\sum_{i=1}^{n}  \sum_{j=1}^{m_i} p(x_{i,j})^2\le \frac{4CC_{k}^2}{cn}\sum_{i=1}^{n}  \frac{1}{m_i}\sum_{j=1}^{m_i} p(x_{i,j})^2,$$
and in consequence, from Lemma \ref{lemma1}, the event 
\[
	  \| p\|_{\infty} ^2\,\leq\,  \frac{4C_{k}^2C}{cn}\sum_{i=1}^{n}  \frac{1}{m_i}\sum_{j=1}^{m_i} p(x_{i,j})^2,
	\]
	for all  $p(x)$  polynomial of  degree $k-1$, holds with probability at least $1 - \frac{32\log^3(n)}{C_{\beta}^3c_1^2c^2n^2m^2}- \frac{\log(n)}{n^2}$.
\end{remark}

The following lemma is a simple consequence of Talagrand’s inequality in \cite{wainwright2019high}.
\begin{lemma}\label{lemma0}
    Let $\left\{x_{i,j}\right\}_{i\in \mathcal{J},1\le j\le m_i}$ be a collection of i.i.d. random variables where $\mathcal{J}\subset \mathbb{N}$ is a finite set of indexes. Consider a countable class of functions $\mathcal{F}$ uniformly bounded by $b$. Then for all $\delta>0$, the random variable 
    $$Z=\sup _{f \in \mathcal{F} }\frac{1}{\sum_{i\in \mathcal{J}}m_i}\sum_{i\in\mathcal{J}}\sum_{j=1}^{m_i}f(x_{i,j}),$$
    satisfies the upper tail bound
$$
\mathbb{P}[Z \geq \mathbb{E}[Z]+\delta] \leq 2 \exp \left(\frac{-\sum_{i\in \mathcal{J}}m_i \delta^2}{8 e \left[\Sigma^2\right]+16eb\mathbb{E}(Z)+4 b \delta}\right),
$$
where $\Sigma^2=\sup _{f \in \mathcal{F}}  f^2\left(X\right)$.
\end{lemma}
\begin{proof}
    See Theorem 3.27 in \cite{wainwright2019high}.
\end{proof}
The following is a modified version of Theorem $14.1$ in \cite{wainwright2019high}.

\begin{lemma}\label{lemma2}
 Let $\left\{x_{i,j}\right\}_{i=1,j=1}^{n,m_i}$ be random variables satisfying Assumption \ref{assume:tv functions}, and let $\rho>0$. Moreover consider positive constants $C$ and $c$ such that $cm\le m_i\le Cm$. For any $t \geq C_1\rho^{\frac{2k}{2k+1}}\log(n)^{\frac{2k}{2k+1}}(nm)^{\frac{-k}{2 k+1}}$, with probability at least $1-C_2\log(n) \exp \left(-C_3 nm t^2/(\log(n)\rho^2)\right)-\frac{1}{n}$, it holds that
$$
\frac{1}{C}\|g\|_{nm}^2-\|g\|_2^2\leq \frac{1}{2}\|g\|_2^2+\frac{1}{2}t^2 \quad \text { for all } \quad g \in B_{T V}^{(k-1)}(1) \cap B_{\infty}(\rho) \text {,}
$$
and
$$
\|g\|_2^2-\frac{1}{c}\|g\|_{nm}^2\leq \frac{1}{2}\|g\|_2^2+\frac
{1}{2}t^2 \quad \text { for all } \quad g \in B_{T V}^{(k-1)}(1) \cap B_{\infty}(\rho) \text {,}
$$
where $C_1>0$ is a constant depending on $\rho$, and $C_2$ and $C_3$ are positive constants.
\end{lemma}
\begin{proof}
    Let $\mathcal{F}= B_{T V}^{(k-1)}(1) \cap B_{\infty}(\rho)$. Consider $\left\{x_{i,j}^*\right\}_{i=1,j=1}^{n,m_i}$ the copies $\left\{x_{i,j}\right\}_{i=1,j=1}^{n,m_i}$ defined as in Section \ref{BetaMsection}. For the construction of these copies, we use 
    \begin{equation}
        \label{L-block-lemma3-t1-t2}
        L=\frac{2}{C_{\beta}}\log(n),\ c_1 \frac{n}{L}\le|\mathcal{J}_{o,l}|,|\mathcal{J}_{e,l}|\le c_2 \frac{n}{L},
    \end{equation}
    for some positive constants $c_1$ and $c_2$.
    Denote by $\delta_{nm}$ any common positive solution of the inequality
\begin{align}
\label{eq-Ev-proof}
\widetilde{\mathcal{R}}_l\left(\mathcal{F},\delta\right)\leq \frac{c\delta^2}{32c_2C\rho},\ \forall \ l=1,...,L,
\end{align}
where $
\widetilde{\mathcal{R}}_l\left(\mathcal{F},\delta\right)$ is defined as
{\footnotesize{\begin{align}
\label{l-RademacherComp}
  & \mathbb{E}\left[\frac{1}{\sum_{i\in \mathcal{I}_{e,l}}m_i}\sup _{g \in \mathcal{F} \cap B_2(\delta)}\left|\sum_{i\in\mathcal{I}_{e,l}}\sum_{j=1}^{m_i} \sigma_{i,j} g\left(x_{i,j}^*
\right)\right|\right] \nonumber
\\
+&\mathbb{E}\left[\frac{1}{\sum_{i\in \mathcal{I}_{o,l}}m_i}\sup _{g \in \mathcal{F} \cap B_2(\delta)}\left|\sum_{i\in\mathcal{I}_{o,l}}\sum_{j=1}^{m_i} \sigma_{i,j} g\left(x_{i,j}^*
\right)\right|\right].
\end{align}}}
Here $\{\sigma_{i,j}\}_{i=1,j=1}^ {i=n,m_i}$ are independent Rademacher variables, independent of $\left\{x_{i,j}^*\right\}_{i=1,j=1}^{n,m_i}$. Now,  for any $t\ge \delta_{nm}$, let $Z_t=\sup _{g \in \mathcal{F} \cap B_2(t)}\Big\vert \frac{1}{\sum_{i=1}^{n}m_i}\sum_{i=1}^{n}\sum_{j=1}^{m_i}(g(x_{i,j}^*)^2-\vert\vert g\vert\vert_{2}^2)\Big\vert$ and $\Omega=\bigcap_{i=1}^n\Big\{x_{i,j}=x_{i,j}^*, \forall j=1,...,m_i\Big\}$. By Lemma \ref{Ind-cop} and the choice of $L$ in Equation (\ref{L-block-lemma3-t1-t2}) we get that $$\mathbb{P}(\Omega^c)\le n\beta(L)\le \frac{1}{n}.$$ Therefore, letting $$I(t)=\Big\{\Big\vert \frac{1}{\sum_{i=1}^{n}m_i}\sum_{i=1}^{n}\sum_{j=1}^{m_i}(g(x_{i,j})^2-\vert\vert g\vert\vert_{2}^2)\Big\vert\le \frac{1}{2}\vert\vert g\vert\vert_2^2+\frac{1}{2}t^2,\ \text{for all} \ g \in \mathcal{F}\Big\}^c,$$ we get that 
\begin{equation}
\label{eq-proof-0}
    \mathbb{P}(I(t))=\mathbb{P}(I(t)\cap \Omega)+\mathbb{P}(I(t)\cap \Omega^c)\le \mathbb{P}(I(t)\cap \Omega)+\frac{1}{n}.
\end{equation}
Now we follow the proof of Theorem 14.1 in \cite{wainwright2019high}. First, observe that, 
\begin{align*}
    &\Omega\cap\Big\{\Big\vert \frac{1}{\sum_{i=1}^{n}m_i}\sum_{i=1}^{n}\sum_{j=1}^{m_i}(g(x_{i,j})^2-\vert\vert g\vert\vert_{2}^2)\Big\vert\le \frac{1}{2}\vert\vert g\vert\vert_2^2+\frac{1}{2}t^2,\ \text{for all} \ g \in \mathcal{F}\Big\}^c
    \\
    &\subset \Big\{\Big\vert \frac{1}{\sum_{i=1}^{n}m_i}\sum_{i=1}^{n}\sum_{j=1}^{m_i}(g(x_{i,j}^*)^2-\vert\vert g\vert\vert_{2}^2)\Big\vert\le \frac{1}{2}\vert\vert g\vert\vert_2^2+\frac{1}{2}t^2,\ \text{for all} \ g \in \mathcal{F}\Big\}^c
    \\
    &\subset \Big\{Z_t\ge \frac{1}{2}t^2 \Big\}
    =\mathcal{A}_0(t).
\end{align*}
For the last contention we divide the analysis into two cases. First suppose that there exists some function with norm $\|g\|_2 \leq t$ such that 
$$\Big\vert \frac{1}{\sum_{i=1}^{n}m_i}\sum_{i=1}^{n}\sum_{j=1}^{m_i}(g(x_{i,j})^2-\vert\vert g\vert\vert_{2}^2)\Big\vert> \frac{1}{2}\vert\vert g\vert\vert_2^2+\frac{1}{2}t^2.$$
For this function, we must have $\left|\|g\|_{\sum_{i=1}^{n}m_i}^2-\|g\|_2^2\right|>\frac{1}{2}t^2$, showing that $Z_t>\frac{1}{2}t^2.
$ So that $\mathcal{A}_0(t)$ must hold. Otherwise, suppose that the inequality $$\Big\vert \frac{1}{\sum_{i=1}^{n}m_i}\sum_{i=1}^{n}\sum_{j=1}^{m_i}(g(x_{i,j})^2-\vert\vert g\vert\vert_{2}^2)\Big\vert> \frac{1}{2}\vert\vert g\vert\vert_2^2+\frac{1}{2}t^2,$$ is satisfied by some function with $\|g\|_2>t$. Any such function satisfies the inequality $$\left|\|g\|_2^2-\|g\|_{\sum_{i=1}^{n}m_i}^2\right|>\|g\|_2^2 / 2.$$ We may then define the rescaled function $\widetilde{g}=\frac{t}{\|g\|_2} g$. By construction it holds that $\|\widetilde{g}\|_2=t $, and also belongs to $\mathcal{F}$ since $\frac{t}{\|g\|_2}<1$ . Hence reasoning as before we find that $\mathcal{A}_0(t)$ must also hold in this case. In fact the function $\widetilde{g}$ satisfies that $\vert\vert \widetilde{g}\vert\vert_2\le t$ and 
$$\left|\|\widetilde{g}\|_2^2-\|\widetilde{g}\|_{\sum_{i=1}^{n}m_i}^2\right|>\|\widetilde{g}\|_2^2 / 2=t^2/2.$$
Now we control the event $\mathcal{A}_0(t)$. Using triangle inequality,
\begin{align}
\label{eq-proof-6}
    \mathbb{P}(\mathcal{A}_0(t))&=\mathbb{P}\Big(\sup _{g \in \mathcal{F} \cap B_2(t)}\Big\vert \frac{1}{\sum_{i=1}^{n}m_i}\sum_{i=1}^{n}\sum_{j=1}^{m_i}(g(x_{i,j}^*)^2-\vert\vert g\vert\vert_{2}^2)\Big\vert> \frac{1}{2}t^2\Big)
    \\
    &\le\mathbb{P}\Big(\sup _{g \in \mathcal{F} \cap B_2(t)}\Big\vert \frac{1}{\sum_{i=1}^{n}m_i}\sum_{l=1}^{L}\sum_{i\in \mathcal{J}_{e,l}}\sum_{j=1}^{m_i}(g(x_{i,j}^*)^2-\vert\vert g\vert\vert_{2}^2)\Big\vert> \frac{1}{4}t^2\Big)\nonumber
    \\
    &+\mathbb{P}\Big(\sup _{g \in \mathcal{F} \cap B_2(t)}\Big\vert \frac{1}{\sum_{i=1}^{n}m_i}\sum_{l=1}^{L}\sum_{i\in \mathcal{J}_{o,l}}\sum_{j=1}^{m_i}(g(x_{i,j}^*)^2-\vert\vert g\vert\vert_{2}^2)\Big\vert> \frac{1}{4}t^2\Big).\nonumber
    \end{align}
    Since the supremum of the sum is smaller or equal to the sum of the supremums, using a union bound argument and Inequality (\ref{eq-proof-6}) we get that
\begin{align}  
\label{eq-proof-7}
    \mathbb{P}(\mathcal{A}_0(t))&\le \sum_{l=1}^{L}\mathbb{P}\Big(\sup _{g \in \mathcal{F} \cap B_2(t)}\Big\vert \frac{1}{\sum_{i\in \mathcal{J}_{e,l}}^{n}m_i}\sum_{i\in \mathcal{J}_{e,l}}\sum_{j=1}^{m_i}(g(x_{i,j}^*)^2-\vert\vert g\vert\vert_{2}^2)\Big\vert> \frac{1}{4L}t^2\Big)
    \\
    &+ \sum_{l=1}^{L}\mathbb{P}\Big(\sup _{g \in \mathcal{F} \cap B_2(t)}\Big\vert \frac{1}{\sum_{i=1}^{n}m_i}\sum_{i\in \mathcal{J}_{o,l}}\sum_{j=1}^{m_i}(g(x_{i,j}^*)^2-\vert\vert g\vert\vert_{2}^2)\Big\vert> \frac{1}{4L}t^2\Big).
\end{align}
The analysis of each of the terms $\mathbb{P}\Big(\sup _{g \in \mathcal{F} \cap B_2(t)}\Big\vert \frac{1}{\sum_{i=1}^{n}m_i}\sum_{i\in \mathcal{J}_{o,l}}\sum_{j=1}^{m_i}(g(x_{i,j}^*)^2-\vert\vert g\vert\vert_{2}^2)\Big\vert> \frac{1}{4L}t^2\Big)$ follows the same line of the arguments as for \break $\mathbb{P}\Big(\sup _{g \in \mathcal{F} \cap B_2(t)}\Big\vert \frac{1}{\sum_{i\in \mathcal{J}_{e,l}}^{n}m_i}\sum_{i\in \mathcal{J}_{e,l}}\sum_{j=1}^{m_i}(g(x_{i,j}^*)^2-\vert\vert g\vert\vert_{2}^2)\Big\vert> \frac{1}{4L}t^2\Big)$, and therefore it is omitted below. To analyze $$\mathbb{P}\Big(\sup _{g \in \mathcal{F} \cap B_2(t)}\Big\vert \frac{1}{\sum_{i=1}^{n}m_i}\sum_{i\in \mathcal{J}_{e,l}}\sum_{j=1}^{m_i}(g(x_{i,j}^*)^2-\vert\vert g\vert\vert_{2}^2)\Big\vert> \frac{1}{4L}t^2\Big),$$
we observe that,
\begin{align}
\label{eq-proof-8}
  &\mathbb{P}\Big(\sup _{g \in \mathcal{F} \cap B_2(t)}\Big\vert \frac{1}{\sum_{i=1}^{n}m_i}\sum_{i\in \mathcal{J}_{e,l}}\sum_{j=1}^{m_i}(g(x_{i,j}^*)^2-\vert\vert g\vert\vert_{2}^2)\Big\vert> \frac{1}{4L}t^2\Big)
\\
=&
  \mathbb{P}\Big(\sup _{g \in \mathcal{F} \cap B_2(t)}\Big\vert \frac{1}{\sum_{i\in \mathcal{J}_{e,l}}m_i}\sum_{i\in \mathcal{J}_{e,l}}\sum_{j=1}^{m_i}(g(x_{i,j}^*)^2-\vert\vert g\vert\vert_{2}^2)\Big\vert> \frac{\sum_{i=1}^{n}m_i}{\sum_{i\in \mathcal{J}_{e,l}}m_i}\frac{1}{4L}t^2\Big)\nonumber
  \\
  \le&\mathbb{P}\Big(\sup _{g \in \mathcal{F} \cap B_2(t)}\Big\vert \frac{1}{\sum_{i\in \mathcal{J}_{e,l}}m_i}\sum_{i\in \mathcal{J}_{e,l}}\sum_{j=1}^{m_i}(g(x_{i,j}^*)^2-\vert\vert g\vert\vert_{2}^2)\Big\vert> \frac{c}{4c_2C}t^2\Big),\nonumber
\end{align}
where the inequality is followed by the fact that $m_i\asymp m$ in the statement of the Lemma, and Equation (\ref{L-block-lemma3-t1-t2}).
We now control
$$\mathbb{P}\Big(\sup _{g \in \mathcal{F} \cap B_2(t)}\Big\vert \frac{1}{\sum_{i\in \mathcal{J}_{e,l}}m_i}\sum_{i\in \mathcal{J}_{e,l}}\sum_{j=1}^{m_i}(g(x_{i,j}^*)^2-\vert\vert g\vert\vert_{2}^2)\Big\vert> \frac{c}{4c_2C}t^2\Big).$$
To this end, we make use of a classical symmetrization argument and Ledoux–Talagrand contraction inequality (see \cite{wainwright2019high}) to bound $$\mathbb{E}\Big(\sup _{g \in \mathcal{F} \cap B_2(t)}\Big\vert \frac{1}{\sum_{i\in \mathcal{J}_{e,l}}m_i}\sum_{i\in \mathcal{J}_{e,l}}\sum_{j=1}^{m_i}(g(x_{i,j}^*)^2-\vert\vert g\vert\vert_{2}^2)\Big\vert\Big).$$ Let $\left\{\widetilde{x}_{i,j}\right\}_{i=1,j=1}^{n,m_i}$ to be an i.i.d. sequence, independent copy  of $\left\{x_{i,j}^*\right\}_{i=1,j=1}^{n,m_i}$;  and $\left\{\sigma_{i,j}\right\}_{i=1,j=1}^{n,m_i}$ independent Rademacher variables independent of $\left\{\widetilde{x}_{i,j}\right\}_{i=1,j=1}^{n,m_i}$  and \break $\left\{x_{i,j}^*\right\}_{i=1,j=1}^{n,m_i}$. Then 
\begin{align*}
&\mathbb{E}\Big(\sup _{g \in \mathcal{F} \cap B_2(t)}\Big\vert \frac{1}{\sum_{i\in \mathcal{J}_{e,l}}m_i}\sum_{i\in \mathcal{J}_{e,l}}\sum_{j=1}^{m_i}(g(x_{i,j}^*)^2-\vert\vert g\vert\vert_{2}^2)\Big\vert\Big)
    \\
    =&\mathbb{E}_x\Big[\sup _{g \in \mathcal{F} \cap B_2(t)}\Big\vert \frac{1}{\sum_{i\in \mathcal{J}_{e,l}}m_i}\sum_{i\in \mathcal{J}_{e,l}}\sum_{j=1}^{m_i}(g(x_{i,j}^*)^2-\mathbb{E}_{\widetilde{x}}(g(\widetilde{x}_{i,j}))^2)\Big\vert\Big]
    \\
    \le&\mathbb{E}_{x}\Big[\sup _{g \in \mathcal{F} \cap B_2(t)}\mathbb{E}_{\widetilde{x}}\Big(\Big\vert \frac{1}{\sum_{i\in \mathcal{J}_{e,l}}m_i}\sum_{i\in \mathcal{J}_{e,l}}\sum_{j=1}^{m_i}(g(x_{i,j}^*)^2-g(\widetilde{x}_{i,j})^2)\Big\vert\Big)\Big].
\end{align*}
Since expected value of the supremum is an upper bound of the supremum of the expected value, 
\begin{align}
\label{aux-lemma3-t1-t2-eq0}
&\mathbb{E}\Big(\sup _{g \in \mathcal{F} \cap B_2(t)}\Big\vert \frac{1}{\sum_{i\in \mathcal{J}_{e,l}}m_i}\sum_{i\in \mathcal{J}_{e,l}}\sum_{j=1}^{m_i}(g(x_{i,j}^*)^2-\vert\vert g\vert\vert_{2}^2)\Big\vert\Big)\nonumber
\\
\le&\mathbb{E}_{x,\widetilde{x}}\Big[\sup _{g \in \mathcal{F} \cap B_2(t)}\Big(\Big\vert \frac{1}{\sum_{i\in \mathcal{J}_{e,l}}m_i}\sum_{i\in \mathcal{J}_{e,l}}\sum_{j=1}^{m_i}(g(x_{i,j}^*)^2-g(\widetilde{x}_{i,j})^2)\Big\vert\Big)\Big]\nonumber
    \\
=&\mathbb{E}_{x,\widetilde{x},\sigma}\Big[\sup _{g \in \mathcal{F} \cap B_2(t)}\Big(\Big\vert \frac{1}{\sum_{i\in \mathcal{J}_{e,l}}m_i}\sum_{i\in \mathcal{J}_{e,l}}\sum_{j=1}^{m_i}\sigma_{i,j}(g(x_{i,j}^*)^2-g(\widetilde{x}_{i,j})^2)\Big\vert\Big)\Big]\nonumber
\\
\le& 2\mathbb{E}_{x,\sigma}\Big[\sup _{g \in \mathcal{F} \cap B_2(t)}\Big\vert\frac{1}{\sum_{i\in \mathcal{J}_{e,l}}m_i}\sum_{i\in \mathcal{J}_{e,l}}\sum_{j=1}^{m_i}\sigma_{i,j}g(x_{i,j}^*)^2\Big\vert\Big],
\end{align}
in the last inequality, the independence between $\sigma_{i,j}$, $x_{i,j}^*$ and $\widetilde{x}_{i,j}$ is used.
Then we make use of Ledoux–Talagrand contraction inequality (see Proposition 5.28 in \cite{wainwright2019high}) to bound $$\mathbb{E}_{x,\sigma}\Big[\sup _{g \in \mathcal{F} \cap B_2(t)}\Big\vert \frac{1}{\sum_{i\in \mathcal{J}_{e,l}}m_i}\sum_{i\in \mathcal{J}_{e,l}}\sum_{j=1}^{m_i}\sigma_{i,j}g(x_{i,j}^*)^2\Big\vert\Big].$$ Specifically consider $\phi_\rho:\mathbb{R}\rightarrow \mathbb{R}$ given by $\phi_\rho(t):= \begin{cases}z^2 /(2 \rho) & \text { if }|z| \leq \rho \\ \rho / 2 & \text { otherwise }\end{cases}$, and let $$\mathcal{F}^2(x^*)=\Big\{(g^2(x_{l,1}^*),...,g^2(x_{2KL+l,m_{2KL+l}}^*)):g\in\mathcal{F}\cap B_2(t)\Big\}\subset \mathbb{R}^{\sum_{i\in \mathcal{J}_{e,l}}m_i},$$
with $K=\lfloor n/L\rfloor.$ It is straightforward to check that $\phi_\rho$ is a $1$-Lipschitz function. Therefore, a trivial application of Ledoux–Talagrand contraction inequality (see Proposition 5.28 in \cite{wainwright2019high}) yields to
\begin{align}
\label{aux-lemma3-t1-t2-eq1}
   &\mathbb{E}_{x,\sigma}\Big[\sup _{g \in \mathcal{F} \cap B_2(t)}\Big\vert \frac{1}{\sum_{i\in \mathcal{J}_{e,l}}m_i}\sum_{i\in \mathcal{J}_{e,l}}\sum_{j=1}^{m_i}\sigma_{i,j}g(x_{i,j}^*)^2\Big\vert\Big]\nonumber
   \\
   \le& 2\rho\mathbb{E}_{x,\sigma}\Big[\sup _{g \in \mathcal{F} \cap B_2(t)}\Big\vert \frac{1}{\sum_{i\in \mathcal{J}_{e,l}}m_i}\sum_{i\in \mathcal{J}_{e,l}}\sum_{j=1}^{m_i}\sigma_{i,j}g(x_{i,j}^*)\Big\vert\Big].
\end{align}
In consequence, from Inequality (\ref{aux-lemma3-t1-t2-eq0}) and Inequality (\ref{aux-lemma3-t1-t2-eq1})
\begin{align}
\label{eq-proof-2}
&\mathbb{E}\Big(\sup _{g \in \mathcal{F} \cap B_2(t)}\Big\vert \frac{1}{\sum_{i\in \mathcal{J}_{e,l}}m_i}\sum_{i\in \mathcal{J}_{e,l}}\sum_{j=1}^{m_i}(g(x_{i,j}^*)^2-\vert\vert g\vert\vert_{2}^2)\Big\vert\Big) \nonumber
    \\
    \le&4\rho\mathbb{E}_{x,\sigma}\Big[\sup _{g \in \mathcal{F} \cap B_2(t)}\Big\vert \frac{1}{\sum_{i\in \mathcal{J}_{e,l}}m_i}\sum_{i\in \mathcal{J}_{e,l}}\sum_{j=1}^{m_i}\sigma_{i,j}g(x_{i,j}^*)\Big\vert\Big].
    \end{align}
Even more, we observe the following. Let $g\in \mathcal{F}\cap B_2(t)$. Since $t\ge \delta_{nm}$, then the function $\widetilde{g}=\frac{\delta_{nm}}{t}g$ satisfies $\widetilde{g}\in \mathcal{F}\cap B_2(\delta_{nm})$. Therefore for any $g\in \mathcal{F}\cap B_2(t)$ we have that
\begin{align*}
  \frac{\delta_{nm}}{t} \left[\frac{1}{\sum_{i\in \mathcal{I}_{e,l}}m_i}\left|\sum_{i\in\mathcal{I}_{e,l}}\sum_{j=1}^{m_i} \sigma_{i,j} g\left(x_{i,j}^*
\right)\right|\right] 
&\le\left[\frac{1}{\sum_{i\in \mathcal{I}_{e,l}}m_i}\sup _{g \in \mathcal{F} \cap B_2(\delta_{nm})}\left|\sum_{i\in\mathcal{I}_{e,l}}\sum_{j=1}^{m_i} \sigma_{i,j} g\left(x_{i,j}^*
\right)\right|\right] ,
\end{align*}
from where, taking the supremum followed by the expectation, we obtain that
\begin{align}
\label{eq-proof-3}
  &\frac{\delta_{nm}}{t} \mathbb{E}\left[\frac{1}{\sum_{i\in \mathcal{I}_{e,l}}m_i}\sup _{g \in \mathcal{F} \cap B_2(t)}\left|\sum_{i\in\mathcal{I}_{e,l}}\sum_{j=1}^{m_i} \sigma_{i,j} g\left(x_{i,j}^*
\right)\right|\right]  \nonumber
\\
&\le\mathbb{E}\left[\frac{1}{\sum_{i\in \mathcal{I}_{e,l}}m_i}\sup _{g \in \mathcal{F} \cap B_2(\delta_{nm})}\left|\sum_{i\in\mathcal{I}_{e,l}}\sum_{j=1}^{m_i} \sigma_{i,j} g\left(x_{i,j}^*
\right)\right|\right] .
\end{align}
As a result, by Inequality (\ref{eq-Ev-proof}), Inequality (\ref{eq-proof-2}), and Inequality (\ref{eq-proof-3})
we have that
\begin{align}
\label{eq-proof-4}
    \mathbb{E}\Big(\sup _{g \in \mathcal{F} \cap B_2(t)}\Big\vert \frac{1}{\sum_{i\in \mathcal{J}_{e,l}}m_i}\sum_{i\in \mathcal{J}_{e,l}}\sum_{j=1}^{m_i}(g(x_{i,j}^*)^2-\vert\vert g\vert\vert_{2}^2)\Big\vert\Big)\le \frac{c\delta_{nm}t}{8c_2C}.
\end{align}
Next, our purpose is to use Lemma \ref{lemma0}. To this end, it remains to bound the expression 
$$\Sigma^2= \sup _{g \in \mathcal{F} \cap B_2(t)}\mathbb{E}\Big((g(x_{i,j}^*)^2-\vert\vert g\vert\vert_{2}^2)^2\Big).$$
 Let $g$ be an arbitrary member of $\mathcal{F} \cap B_2(t)$. Since $\|g\|_{\infty} \le \rho$ for all $g \in \mathcal{F}$, the  function $f=g^2-\mathbb{E}\left[g^2(x_{i,j}^*)\right]$ is bounded as $\|f\|_{\infty} \le 2\rho^2$, and moreover
$$
\operatorname{var}(f) \leq \rho^4\mathbb{E}\left[\frac{g^4}{\rho^4}\right] \leq \rho^4\mathbb{E}\left[\frac{g^2}{\rho^2}\right] \leq \rho^2t^2.
$$
It follows that $\Sigma^2$ is upper bounded and we are ready to use Lemma \ref{lemma0}. First, by Inequality (\ref{eq-proof-4}) we have that
\begin{align}
    \label{aux-lemma3-eq-proof-5}
    &\mathbb{P}\Big[\sup _{g \in \mathcal{F} \cap B_2(t)}\Big\vert \frac{1}{\sum_{i\in \mathcal{J}_{e,l}}m_i}\sum_{i\in \mathcal{J}_{e,l}}\sum_{j=1}^{m_i}(g(x_{i,j}^*)^2-\vert\vert g\vert\vert_{2}^2)\Big\vert> \frac{c}{4c_2C}t^2\Big]\nonumber
    \\
    \le&
    \mathbb{P}\Big[\sup _{g \in \mathcal{F} \cap B_2(t)}\Big\vert \frac{1}{\sum_{i\in \mathcal{J}_{e,l}}m_i}\sum_{i\in \mathcal{J}_{e,l}}\sum_{j=1}^{m_i}(g(x_{i,j}^*)^2-\vert\vert g\vert\vert_{2}^2)\Big\vert> \frac{c}{8c_2C}t^2\nonumber
    \\
    +&\mathbb{E}\Big(\sup _{g \in \mathcal{F} \cap B_2(t)}\Big\vert  \frac{1}{\sum_{i\in \mathcal{J}_{e,l}}m_i}\sum_{i\in \mathcal{J}_{e,l}}\sum_{j=1}^{m_i}(g(x_{i,j}^*)^2-\vert\vert g\vert\vert_{2}^2)\Big\vert\Big)\Big]
\end{align}
By Lemma \ref{lemma0}, Inequality (\ref{eq-proof-4}), using that $t\ge\delta_{nm}$, and the upper bound for $\Sigma^2$ we find that there is a universal constant $\widetilde{c}$ such that
\begin{align}
    \label{eq-proof-5}
    &\mathbb{P}\Big[\sup _{g \in \mathcal{F} \cap B_2(t)}\Big\vert \frac{1}{\sum_{i\in \mathcal{J}_{e,l}}m_i}\sum_{i\in \mathcal{J}_{e,l}}\sum_{j=1}^{m_i}(g(x_{i,j}^*)^2-\vert\vert g\vert\vert_{2}^2)\Big\vert> \frac{c}{4c_2C}t^2\Big]
    \\
    \le&2 \exp \left(-\frac{\Big(\sum_{i\in \mathcal{J}_{e,l}}m_i\Big) t^4}{\widetilde{c}\rho^2\left(t^2+t\delta_{nm}\right)}\right) \leq 2e^{-c_3\frac{1}{\rho^2} \Big(\sum_{i\in \mathcal{J}_{e,l}}m_i\Big) t^2}, \nonumber
\end{align}
for some positive constant $c_3.$ Finally use Equation (\ref{L-block-lemma3-t1-t2}) to get that
\begin{equation}
    \label{eq-proof-9}
    e^{-c_3\frac{1}{\rho^2} \Big(\sum_{i\in \mathcal{J}_{e,l}}m_i\Big) t^2}\le e^{-C_{\beta}c_3c_1nm  t^2/(2\log(n)\rho^2)}.
\end{equation}
Consequently, by Inequality (\ref{eq-proof-6}), (\ref{eq-proof-7}), (\ref{eq-proof-8}), (\ref{eq-proof-5}) and (\ref{eq-proof-9}) we obtain that 
\begin{equation}
    \label{eq-proof-10}
    \mathbb{P}\left[Z_t \geq \frac{t^2}{2}\right]  \leq \frac{8}{C_{\beta}}\log(n)e^{-C_{\beta}c_3c_1nm  t^2/(2\log(n)\rho^2)}.
\end{equation}
Thus, by Inequality (\ref{eq-proof-10}) and (\ref{eq-proof-0}) we obtain that for any $t\ge \delta_{nm}$
\begin{equation}
\label{main-eq-proof}
    \mathbb{P}(I(t))\le \frac{8}{C_{\beta}}\log(n)e^{-C_{\beta}c_3c_1nm  t^2/(2\log(n))}+\frac{1}{n}=C_2\log(n)e^{-C_3nm  t^2/(\log(n)\rho^2)}+\frac{1}{n},
\end{equation}
for positive constants $C_2$ and $C_3.$
It turns out that it suffices to show that
 $\delta_{nm} \leq C_1 (nm)^{-\frac{k}{2k+1}}$ for a constant $C_1>0$ to conclude the desired result.
To this end, let $l\in[L]$ and consider the empirical local Rademacher complexity
\begin{equation}
\label{EmpiricalRadCom}
\widetilde{\mathcal{R}}_{l,nm}\left(\mathcal{F} \cap B_2(t)\right)=\mathbb{E}_\sigma\left[\frac{1}{\sum_{i\in \mathcal{I}_{e,l}}m_i}\sup _{g \in \mathcal{F} \cap B_2(\delta)}\left|\sum_{i\in\mathcal{I}_{e,l}}\sum_{j=1}^{m_i} \sigma_{i,j} g\left(x_{i,j}^*
\right)\right|\right] .
\end{equation}
Denote by $\vert\vert g\vert\vert_{\sum_{i\in \mathcal{I}_{e,l}}m_i}=\frac{1}{\sum_{i\in \mathcal{I}_{e,l}}m_i}\sum_{i\in \mathcal{I}_{e,l}}\sum_{j=1}^{m_i} g\left(x_{i,j}^*
\right)^2.$
As we are considering $t \geq \delta_{nm}$, Corollary 2.2 of \cite{bartlett2005local} gives
$$
\mathcal{F} \cap B_2(t) \subseteq \mathcal{F} \cap B_{\vert\vert \cdot\vert\vert_{\sum_{i\in \mathcal{I}_{e,l}}m_i}}(\sqrt{2}t), 
$$
with probability at least $1-\frac{1}{\sum_{i\in \mathcal{I}_{e,l}}m_i}\ge 1-\frac{1}{C_{\beta}}\frac{2\log(n)}{cc_1nm}$. 
Define this good event as $\mathcal{E}$. 
On $\mathcal{E}$ using Dudley's entropy integral (see \cite{dudley1967sizes}), Equation (\ref{L-block-lemma3-t1-t2}), and the assumption $m_i\asymp m$ specified in the statement of the Lemma, we bound the empirical Rademacher complexity defined in Equation (\ref{EmpiricalRadCom}) as follows
\begin{align*}
\widetilde{\mathcal{R}}_{l,nm}\left(\mathcal{F} \cap B_2(t)\right) & \leq \widetilde{\mathcal{R}}_{l,nm}\left(\Big\{ g\in\mathcal{F}: \frac{1}{\sum_{i\in \mathcal{J}_{e,l}}m_i}\sum_{i\in\mathcal{J}_{e,l}}\sum_{j=1}^{m_i} g\left(x_{i,j}^*
\right)^2\le2t^2 \Big\}\right) \\
& \leq \frac{C_{\text{Dud}}}{\sqrt{\sum_{i\in \mathcal{J}_{e,l}}m_i}} \int_0^{\sqrt{2} t} \sqrt{\log N\left(\epsilon,\|\cdot\|_{\sum_{i\in \mathcal{J}_{e,l}}m_i}, \mathcal{F}\right)} d \epsilon
\\
&\le
\frac{C_{\text{Dud}}\sqrt{2\log(n)}}{\sqrt{C_{\beta}cc_1nm}} \int_0^{\sqrt{2} t} \sqrt{\log N\left(\epsilon,\|\cdot\|_{\sum_{i\in \mathcal{J}_{e,l}}m_i}, \mathcal{F}\right)} d \epsilon. 
\end{align*}
Using that $\mathcal{F}= B_{T V}^{(k-1)}(1) \cap B_{\infty}(\rho)$, by Lemma \ref{lemma5}, we obtain that
\begin{align*}
&\frac{C_{\text{Dud}}\sqrt{2\log(n)}}{\sqrt{C_{\beta}cc_1nm}} \int_0^{\sqrt{2} t} \sqrt{\log N\left(\epsilon,\|\cdot\|_{\sum_{i\in \mathcal{J}_{e,l}}m_i}, \mathcal{F}\right)} d \epsilon
\\
& \leq \frac{AC_{\text{Dud}}\rho^{\frac{1}{2k}}\sqrt{2\log(n)}}{\sqrt{C_{\beta}cc_1nm}} \int_0^{\sqrt{2} t} \epsilon^{-1 /(2 k)} d \epsilon \\
& =c_4\frac{\rho^{\frac{1}{2k}}\sqrt{2\log(n)}}{\sqrt{C_{\beta}cc_1nm}} t^{1-1 /(2 k)},
\end{align*}
for some positive constant $c_4,$ and $A>0$ from Lemma \ref{lemma5}.
Under the event $\mathcal{E}^c,$ given that $\mathcal{F}\subset B_{\infty}(\rho),$ we have that
$ \widetilde{\mathcal{R}}_{l,nm}\left(\mathcal{F} \cap B_2(t)\right) \leq \rho$. Therefore, splitting the expectation over $\mathcal{E}$ and $\mathcal{E}^c,$ the rademacher complexity defined in Equation (\ref{l-RademacherComp}) is bounded by
$$
\widetilde{\mathcal{R}}_l\left(\mathcal{F} \cap B_2(t)\right)=E_x \widetilde{\mathcal{R}}_{l,nm}\left(\mathcal{F} \cap B_2(t)\right) \leq c_4\frac{\rho^{\frac{1}{2k}}\sqrt{2\log(n)}}{\sqrt{C_{\beta}cc_1nm}} t^{1-1 /(2 k)}+\frac{2\rho\log(n)}{C_{\beta}cc_1nm}.
$$
So $\delta_{nm}$ is the minimal solution of
$$
c_4\frac{\rho^{\frac{1}{2k}}\sqrt{2\log(n)}}{\sqrt{C_{\beta}cc_1nm}} t^{1-1 /(2 k)}+\frac{2\rho\log(n)}{C_{\beta}cc_1nm}<\frac{t^2}{\rho}.
$$
Without lost of generality assume that $t\ge (nm)^{-\frac{1}{2}}$. Since $\frac{k}{2k-1}>\frac{1}{2}$, we have that $t\ge (nm)^{-\frac{k}{2k-1}}$ and in consequence $t^{\frac{2k-1}{k}}\ge (nm)^{-1}$. Therefore $t^{\frac{2k-1}{2k}}\ge (nm)^{-\frac{1}{2}}$, which can be written as $t^{1-\frac{1}{2k}}\ge (nm)^{-\frac{1}{2}}$. This implies $\frac{t^{1-\frac{1}{2k}}}{\sqrt{nm}}\ge \frac{1}{nm}$. Thus for some constant $c_5>0$ we have that
\begin{equation}
\label{aux-lemma3-equation}
c_5\rho\frac{2\log(n)}{\sqrt{C_{\beta}cc_1nm}} t^{1-1 /(2 k)}<\frac{t^2}{\rho}.
\end{equation}
To end observe the following. If $\rho\ge 1$, from Inequality (\ref{aux-lemma3-equation}),
$$
c_5\rho\frac{2\log(n)}{\sqrt{C_{\beta}cc_1nm}} t^{1-1 /(2 k)}<t^2,$$
which implies
\begin{equation*}
    \frac{C_1\rho^{\frac{2k}{2k+1}}\log(n)^{\frac{2k}{2k+1}}}{(nm)^{\frac{k}{2k+1}}}\le t,
\end{equation*}
for a positive constant $C_1.$
Moreover if $\rho< 1$, from Inequality (\ref{aux-lemma3-equation}) we have that 
$$
c_5\frac{2\log(n)}{\sqrt{C_{\beta}cc_1nm}} t^{1-1 /(2 k)}<\frac{t^2}{\rho}$$
from where
\begin{equation*}
    \frac{C_1\rho^{\frac{2k}{2k+1}}\log(n)^{\frac{2k}{2k+1}}}{(nm)^{\frac{k}{2k+1}}}\le t,
\end{equation*}
for a positive constant $C_1.$
Thus, with probability at least $1-C_2\log(n)e^{-C_3nm  t^2/(\rho^2\log(n))}-\frac{1}{n}$,
$$\Big\vert \frac{1}{\sum_{i=1}^{n}m_i}\sum_{i=1}^{n}\sum_{j=1}^{m_i}(g(x_{i,j})^2-\vert\vert g\vert\vert_{2}^2)\Big\vert\le \frac{1}{2}\vert\vert g\vert\vert_2^2+\frac{1}{2}t^2,\ \text{for all} \ g \in \mathcal{F},$$
when $$
\frac{C_1\rho^{\frac{2k}{2k+1}}\log(n)^{\frac{2k}{2k+1}}}{(nm)^{\frac{k}{2k+1}}}\le t
.$$
Finally from Assumption 1{\bf{f}} it holds that
\begin{align*}
    \vert\vert g\vert\vert_{2}^2-\frac{1}{\sum_{i=1}^{n}m_i}\sum_{i=1}^{n}\sum_{j=1}^{m_i}g(x_{i,j})^2\ge \vert\vert g\vert\vert_{2}^2-\frac{1}{cnm}\sum_{i=1}^{n}\sum_{j=1}^{m_i}g(x_{i,j})^2,
\end{align*}
which implies,
$$
\|g\|_2^2-\frac{1}{c}\|g\|_{nm}^2\leq \frac{1}{2}\|g\|_2^2+\frac{1}{2}t^2 \quad \text { for all } g \in \mathcal{F},
$$
with probability at least $1-C_2\log(n)e^{-C_3nm  t^2/(\rho^2\log(n))}-\frac{1}{n}$. Similarly we conclude 
$$
\frac{1}{C}\|g\|_{nm}^2-\|g\|_2^2\leq \frac{1}{2}\|g\|_2^2+\frac{1}{2}t^2 \quad \text { for all } g \in \mathcal{F},
$$
with probability at least $1-C_2\log(n)e^{-C_3nm  t^2/(\rho^2\log(n))}-\frac{1}{n}$. 
\end{proof}

\begin{lemma}\label{lemma5}
 Let $\left\{x_{i,j}\right\}_{i=1,j=1}^{n,m_i}$ be a collection of random variables in $[0,1]$. Denote
$$
\|f\|_{\sum_{i=1}^nm_i}^2=\frac{1}{\sum_{i=1}^nm_i} \sum_{i=1}^n\sum_{i=1}^{m_i} f^2\left(x_{i,j}\right) .
$$
Then, there exists an absolute constant $A$ such that for $k \geq 1$,
$$
\log N\left(\varepsilon,\|\cdot\|_{\sum_{i=1}^nm_i}, B_{T V}^{(k-1)}(1) \cap B_{\infty}(1)\right) \leq A \varepsilon^{-1 /k} .
$$

\end{lemma}
\begin{proof}
This is a direct consequence of \cite{sadhanala2019additive}, see page 3050 of the mentioned reference.
\end{proof}
\begin{lemma}
    \label{lemma6}
    Let $\{x_{i,j}^*\}_{i \in\{\mathcal{J}_{e,l},1\le j\le m_i}$ i.i.d random variables and $\{\delta_{i}^*\}_{i \in\mathcal{J}_{e,l}}$ i.i.d random functions independent of $\{x_{i,j}^*\}_{i \in\mathcal{J}_{e,l},1\le j\le m_i}$.
 Let $\mathcal{F}$ be a set of functions and for any $f \in \mathcal{F}$ define
$$
Z_{i,j}^*(f):=\left\{f\left(x_{i,j}^*\right)-f^*\left(x_{i,j}^*\right)\right\} \delta_i^*\left(x_{i,j}^*\right)-\left\langle f-f^*, \delta_i^*\right\rangle_{2}
$$
for all $i \in\{\mathcal{J}_{e,l}\}$ and $j \in\{1 \ldots, m_i\}$. Then,  we have that
\begin{align*}
& \mathbb{E}\left( \sup _{f \in \mathcal{F}} \frac{1}{n } \sum_{i\in\mathcal{J}_{e,l}} \frac{1}{m_i}\sum_{j=1}^{m_i} Z_{i,j}^*(f) \mid\left\{\{\delta_i^*\}_{i\in\mathcal{J}_{e,l}}\right\}\right) \\
& \leq 2 \mathbb{E}\left( \sup _{f \in \mathcal{F}} \frac{1}{n} \sum_{i\in\mathcal{J}_{e,l}} \frac{1}{m_i}\sum_{j=1}^{m_i} \xi_{i,j}\left\{f\left(x_{i,j}^*\right)-f^*\left(x_{i,j}\right)\right\} \delta_i^*\left(x_{i,j}^*\right) \mid\left\{\{\delta_i^*\}_{i\in\mathcal{J}_{e,l}}\right\}\right),
\end{align*}
for Rademacher independent random variables $\left\{\xi_{i,j}\right\}$ independent of $$\left(\left\{x_{i,j}^*\right\}_{i\in\mathcal{J}_{e,l},1\le j\le m_i},\{\delta_i^*\}_{i\in\mathcal{J}_{e,l}}\right).$$
\end{lemma}
\begin{proof}
 Let $\left\{x_{i,j}^{\prime}\right\}$ be an independent copy of $\left\{x_{i,j}^*\right\}$. Then for Rademacher independent random variables $\left\{\xi_{i,j}\right\}$ independent of $\left(\left\{x_{i,j}^*\right\},\left\{\delta_i^*\right\}\right)$ we have that
\begin{align*}
& \mathbb{E}\left( \sup _{f \in \mathcal{F}} \frac{1}{n} \sum_{i\in\mathcal{J}_{e,l}} \frac{1}{m_i}\sum_{j=1}^{m_i} Z_{i,j}^*(f) \mid\left\{\delta_i^*\right\}\right) \\
& =\mathbb{E}\Bigg( \sup _{f \in \mathcal{F}} \frac{1}{n } \sum_{i\in\mathcal{J}_{e,l}} \frac{1}{m_i}\sum_{j=1}^{m_i}\Big[\left\{f\left(x_{i,j}^*\right)-f^*\left(x_{i ,j}^*\right)\right\} \delta_i^*\left(x_{i ,j}^*\right) \\
&  -\mathbb{E}\left(\left\{f\left(x_{i, j}^{\prime}\right)-f^*\left(x_{i, j}^{\prime}\right)\right\} \delta_i^*\left(x_{i,j}^{\prime}\right)\right)\Big] \mid\left\{\delta_i^*\right\}\Bigg).
\end{align*}
Using that the supremum of the expected is smaller or equal to the expected value of the supremum, it can be bounded by
\begin{align*}
 &\mathbb{E}\Bigg(\mathbb{E}\Bigg( \sup _{f \in \mathcal{F}} \frac{1}{n } \sum_{i\in\mathcal{J}_{e,l}} \frac{1}{m_i}\sum_{j=1}^{m_i}\left[\left\{f\left(x_{i,j}^*\right)-f^*\left(x_{i, j}^*\right)\right\} \delta_i^*\left(x_{i ,j}^*\right) \right.\\
 &\quad \left. -\left\{f\left(x_{i,j}^{\prime}\right)-f^*\left(x_{i,j}^{\prime}\right)\right\} \delta_i^*\left(x_{i,j}^{\prime}\right)\right] \mid\left\{\delta_i^*\right\},\left\{x_{i,j}^*\right\}\Bigg) \mid\left\{\delta_i^*\right\}\Bigg) \\
& =\mathbb{E}\Bigg(\sup _{f \in \mathcal{F}} \frac{1}{n } \sum_{i\in\mathcal{J}_{e,l}} \frac{1}{m_i}\sum_{j=1}^{m_i}\left[\left\{f\left(x_{i,j}^*\right)-f^*\left(x_{i,j}^*\right)\right\} \delta_i^*\left(x_{i,j}^*\right) \right.\\
&\quad \left.-\left\{f\left(x_{i,j}^{\prime}\right)-f^*\left(x_{i,j}^{\prime}\right)\right\} \delta_i^*\left(x_{i,j}^{\prime}\right)\right] \mid\left\{\delta_i^*\right\}\Bigg) \\
& =\mathbb{E}\Bigg(\sup _{f \in \mathcal{F}} \frac{1}{n } \sum_{i\in\mathcal{J}_{e,l}} \frac{1}{m_i}\sum_{j=1}^{m_i} \xi_{i,j}\left[\left\{f\left(x_{i,j}^*\right)-f^*\left(x_{i,j}^*\right)\right\} \delta_i^*\left(x_{i,j}^*\right) \right.\\
&\quad \left.-\left\{f\left(x_{i,j}^{\prime}\right)-f^*\left(x_{i,j}^{\prime}\right)\right\} \delta_i^*\left(x_{i,j}^{\prime}\right)\right] \mid\left\{\delta_i^*\right\}\Bigg),
\end{align*}
where the last equation is followed by the independence assumption. Therefore 
\begin{align*}
& \mathbb{E}\left(\sup _{f \in \mathcal{F}} \frac{1}{n} \sum_{i\in\mathcal{J}_{e,l}} \frac{1}{m_i}\sum_{j=1}^{m_i} Z_{i,j}^*(f) \mid\left\{\delta_i^*\right\}\right) 
 \\
& \leq \mathbb{E}\left(\sup _{f \in \mathcal{F}} \frac{1}{n } \sum_{i\in\mathcal{J}_{e,l}} \frac{1}{m_i}\sum_{j=1}^{m_i}\xi_{i,j}\left\{f\left(x_{i,j}^*\right)-f^*\left(x_{i,j}^*\right)\right\} \delta_i^*\left(x_{i,j}^*\right) \mid\left\{\delta_i^*\right\}\right)
\\
& +\mathbb{E}\left(\sup _{f \in \mathcal{F}} \frac{1}{n } \sum_{i\in\mathcal{J}_{e,l}} \frac{1}{m_i}\sum_{j=1}^{m_i}-\xi_{i,j}\left\{f\left(x_{i,j}^{\prime}\right)-f^*\left(x_{i,j}^{\prime}\right)\right\} \delta_i^*\left(x_{i,j}^{\prime}\right) \mid\left\{\delta_i^*\right\}\right). 
\end{align*}    
Hence
\begin{align*}
    &\mathbb{E}\left(\sup _{f \in \mathcal{F}} \frac{1}{n} \sum_{i\in\mathcal{J}_{e,l}} \frac{1}{m_i}\sum_{j=1}^{m_i} Z_{i,j}^*(f) \mid\left\{\delta_i\right\}\right) 
    \\
   & \leq 2 \mathbb{E}\left(\sup _{f \in \mathcal{F}} \frac{1}{n } \sum_{i\in\mathcal{J}_{e,l}} \frac{1}{m_i}\sum_{j=1}^{m_i} \xi_{i,j}\left\{f\left(x_{i,j}^*\right)-f^*\left(x_{i,j}^*\right)\right\} \delta_i^*\left(x_{i,j}^*\right) \mid\left\{\delta_i^*\right\}\right),
\end{align*}
concluding the claim.
\end{proof}

\begin{lemma}
\label{lemma7BetaMing}
Suppose Assumption \ref{assume:tv functions} holds. Then 
$$ \big|  \mathbb E(   \langle \delta_0, \delta_i\rangle    )  \big| \le 2\beta^{1/q}(i) \mathbb E(\|\delta_0\|_\lt^{2p})^{1/p},$$
for all $ i\in \mathbb Z^+$ and $p,q>1$ such that $p^{-1} + q^{-1} =1$.

\end{lemma}
\begin{proof}
By Lemma \ref{lemma:coupling} there exists  $\delta_i^*$ such that $ \delta_0 $ and $\delta_i^*$ are independent, that  $\delta_i^*$   are identically distributed to    $\delta_i $, and that 
\begin{equation}
    \label{Lemma7aux1}
\mathbb P( \delta_i^* \not = \delta_i) \le \beta(i).
\end{equation}
Now observe that 
\begin{align*}
   \big| \mathbb E(\langle \delta_0, \delta_i\rangle )\big|  = \big| \mathbb E(\langle \delta_0, \delta_i-\delta_i^*\rangle ) +  \mathbb E(\langle \delta_0  ,  \delta_i^*\rangle )  \big| .
\end{align*}
Moreover by the independence between $\delta_0$ and $\delta_i^*$,
$$\mathbb E( \langle \delta_0 , \delta_i^*\rangle)=0.$$
Therefore, 
\begin{align*}
    &\big| \mathbb E(\langle \delta_0, \delta_i\rangle )\big|\le\big| \mathbb E(\langle \delta_0, \delta_i-\delta_i^*\rangle ) \big| 
    = \big|  \mathbb E \big( \mathbf{1}_{ \{ \delta_i \neq\delta_i^* \} }\langle \delta_0, \delta_i-\delta_i^*\rangle   \big) \big| 
    \le   \mathbb E \big( \mathbf{1}_{ \{ \delta_i \neq\delta_i^* \} }\|  \delta_0\|_{\mathcal{L}_2} \|  \delta_i-\delta_i^*\|_{\mathcal{L}_2}  \big).
\end{align*}
Then, using inequality and the fact that $\delta_i$ and $\delta_i^*$ have the same distribution, we obtain 
\begin{align*}
    \big| \mathbb E(\langle \delta_0, \delta_i\rangle )\big|
    \le & 2 \mathbb E \big( \mathbf{1}_{ \{ \delta_i \neq\delta_i^* \} }\|  \delta_0\|_{\mathcal{L}_2} \|  \delta_i  \|_{\mathcal{L}_2}  \big).
\end{align*}
 Hölder's inequality, followed by Cauchy Schwarz inequality, and that $\delta_0$ and $\delta_i$ have the same distribution, lead to
\begin{align*}
    \big| \mathbb E(\langle \delta_0, \delta_i\rangle )\big|\le 2\mathbb P(\delta_i^* \neq  \delta_i ) ^{1/q} \mathbb E(\|\delta_0\|_{\mathcal{L}_2}^{2p})^{1/p} \le 2\beta^{1/q}(i) \mathbb E(\|\delta_0\|_{\mathcal{L}_2}^{2p})^{1/p}.
\end{align*}
where the last inequality is followed by Inequality (\ref{Lemma7aux1}).
\end{proof}

\begin{lemma}
\label{expected-value-b}
Suppose Assumption \ref{assume:tv functions} holds.
Let  $p, q >1$ be  such that   $\frac{1}{p} + \frac{1}{q} =1 $.
Then assuming $\mathbb  \mathbb{E}(\|\delta_0\|_\lt^{ 2p})<\infty $, it is followed that
$$ \mathbb E \bigg (\big  \| \frac{1}{n} \sum_{i=1}^n \delta_i\big  \|_\lt^2 \bigg) \lesssim    \frac{1}{n} .$$
    
\end{lemma}
\begin{proof} 
Note that 
\begin{align*}
    \mathbb E \bigg (\big  \| \frac{1}{n} \sum_{i=1}^n \delta_i\big  \|_\lt^2 \bigg) =  \frac{1}{n^2 } \sum_{i=1}^n\mathbb E \big\{ \|\delta_i\|_\lt ^2  \big\} + \frac{ 2}{n^2 }  \sum_{i < j }  \mathbb E \big\{    \langle \delta_i, \delta_j \rangle  \big\}.
\end{align*}
and that $\frac{1}{n^2 } \sum_{i=1}^n\mathbb E \big\{ \|\delta_i\|_\lt ^2  \big\} \lesssim \frac{1}{n}.$
In addition, 
\begin{align*} 
&\frac{1}{n^2 }\sum_{i < j }  \mathbb E \big\{    \langle \delta_i, \delta_j \rangle  \big\} \le \frac{1}{n^2 }   \frac{1}{n^2 }\sum_{i =1 }^n \sum_{j=i+1}^\infty  \big|  \mathbb E \big\{    \langle \delta_i, \delta_j \rangle   \big\} \big| 
\\\le   & \frac{1}{n^2 } \sum_{i =1 }^n \sum_{j=i+1}^\infty 2 \beta^{1/q}(j-i) \mathbb E(\|\delta_0\|_\lt^{2p})^{1/p}
 \lesssim   \frac{1}{n^2 } \sum_{i =1 }^n    \mathbb E(\|\delta_0\|_\lt^{2p})^{1/p} \lesssim \frac{1}{n},
\end{align*}
where the first inequality is followed by Lemma \ref{lemma7BetaMing} and the second inequality by the fact that from Assumption \ref{assume:tv functions} we have $ \sum_{l=1}^\infty \beta_{\delta}^{1/q}(l) < \infty $.

\end{proof}

\newpage
\section{\texorpdfstring{$\beta$-}{K}mixing coefficients}\label{BetaMsection}
In the present section, we point out some results of mixing. For more proofs and details, we refer to \cite{doukhan2012mixing}.
Recall that if $X$ is a random variable in $\mathbb{R}^d$, we represent the $\sigma$-algebra generated by $X$ as $\sigma(X)$. Furthermore, for a collection of random variables ${X_t : t \in I} \subseteq \mathbb{R}^d$, we use $\sigma(X_t : t \in I)$ to denote the $\sigma$-algebra generated by the set of random variables ${X_t : t \in I}$. Now a sequence of random variables $\left(X_t, t \in \mathbb{Z}\right)\subset \mathbb{R}^d$ is said to be $\boldsymbol{\beta}$-mixing if
$$
\beta_l=\sup _{t \in \mathbb{Z}} \beta\left(\sigma\left(X_s, s \leq t\right), \sigma\left(X_s, s \geq t+l\right)\right) \underset{l \rightarrow+\infty}{\longrightarrow} 0,
$$
where the $\beta$ coefficients are defined as,
$$ \beta(\sigma\left(X_s, s \leq t\right), \sigma\left(X_s, s \geq t+l\right))=\mathbb{E}\sup _{C \in \sigma\left(X_s, s \geq t+l\right)}|\mathbb{P}(C)-\mathbb{P}(C \mid \sigma\left(X_s, s \leq t\right))|.$$
\begin{theorem} [\cite{doukhan2012mixing} Theorem 1]\label{lemma:coupling}
Let $E$ and $F$ be two polish spaces and $(X, Y)$ some $E\times F $-valued random variables. A random variable $Y^*$ can be defined with the same probability distribution as $Y$, independent of $X$ and such that 
$$\mathbb{P}
(Y\not = Y^*) = \beta(\sigma(X), \sigma(Y)).$$
For some measurable function $f$ on $E\times F \times [0,1]$, and some uniform random variable $\Delta $ on the interval $[0,1]$, $Y^*$ 
takes the form $Y^* = f(X, Y, \Delta)$.
\end{theorem}

For any  $L \in \{1,\ldots, n\} $ and 
$ K=\lfloor n/L \rfloor$,  denote 
$$ \I_s = \begin{cases}
    ((s-1) L, \ s  L] & \text{when } 1 \le s \le K,
    \\
    ( K  L, n] &\text{when } s=K+1  .
\end{cases}$$
Denote for $ 1\le l \le L$,
\begin{align*} \mathcal J_{e,l} =&\{ i \in \mathcal I_s \text{ for }  s  \% 2=0 \text{ and } i    \%   L= l\} \quad \text{and}
\\ \mathcal J_{o, l} = &\{ i \in \mathcal I_s \text{ for }  s  \% 2=1 \text{ and } i    \%   L = l\} .
\end{align*}
   Then 
   \begin{align} \label{eq:temporal spacial effective sample size}\frac{K-1}{2 }\le |\mathcal J_{e,l}|  \le \frac{K+2 }{2} \text{ and }
   \frac{K-1}{2 }\le |\mathcal J_{o,l}|  \le \frac{K+2 }{2},
   \end{align} 
   and that 
   $$\bigcup _{l=1}^{L}\mathcal J_{e,l} \cup  \bigcup_{l=1}^{L}\mathcal J_{o,l} =[1,\ldots, n].$$

\begin{lemma}
\label{Ind-cop}Let $\{x_{i,j}\}_{i=1,j=1}^{n,m_i}$ in $[0,1]^d$ are identically distributed random variables such that $\{\sigma_x(i)\}_{i=0}^n$ are $\beta$-mixing, where $\sigma_x(i)=\sigma(x_{i,j},j\in[m_i])$. Moreover, suppose that for any $i\in[n]$ fixed, $\{x_{i,j}\}_{j=1}^{m_i}$ are independent.
For any $l\in [1,\ldots L]$, there exists a collection of independent identically distributed random variables 
$\{   x^*_{i,j}   \}_{1\le j \le m_i ,i \in \mathcal J_{e,l}} $, with same distribution as $x_{1,1}$, such that for any $i \in \mathcal J_{e,l}$,
$$ \mathbb{P} (x_{i,j}^* \not =  x_{i,j} \text{ for some } j \in [1,\ldots,m_i] ) \le \beta(L).$$
The exact same result is attained for $\mathcal{J}_{o,l}.$
\end{lemma}
\begin{proof}
By Lemma \ref{lemma:coupling}, for any $i\in [1,\ldots,n]$,  there exists $  \{x_{i,j} ^*\}_{j=1}^{m_i} $ such that 
\begin{align} \label{eq:temporal spacial coupling error}\mathbb{P}(x_{i,j}^*\not = x_{i,j} 
\text{ for some } j \in [0, \ldots, m_i]) =\beta(L),
\end{align} 
that $  \{ x_{i,j}\}_{j=1}^{m_i}  $ and  $  \{x_{i,j} ^*\}_{j=1}^{m_i}   $ have the same marginal distribution,
and that $     \{x_{i,j} ^*\}_{j=1}^{m_i}     $ is independent of $     \sigma( \{x_{t,j}  \}_{ t\le i-L ,1\le j \le m_i  } )  .$ In addition, there exists a uniform random variable $\Delta_i$ independent of $\sigma_x$ such that  $  \{x_{i,j} ^*\}_{j=1}^{m_i }  $ 
 is measurable with respect to $\sigma(\{ x_{i,j}\}_{j=1}^{m_i} , \{x_{tj}  \}_{ t\le i-L ,1\le j \le m_i  }, \Delta  ) $.
 \\
 \\
Note that  $\{ x^*_{i,j}\}_{i \in \mathcal{J}_{e,l}}$ are identically distributed because $\{ x^*_{i,j}\}_{i \in \mathcal{J}_{e,l}}$ have the same marginal density. Therefore it suffices to justify the independence. 
\\
\\
 Note that 
$$\mathcal{J}_{e,l} = \{L+l, 3L+l, 5L+l, \dots  \} \cap [1,  n].$$
 Denote 
$$ \mathcal{J}_{e,l, K} = \{L+l, 3L+l, 5L+l, \dots, (2K+1)L+l  \}\cap [1,  n],$$
 and so $ \bigcup _{K \ge 0 } \mathcal{J}_{e,l, K} = \mathcal{J}_{e,l} $.
 By induction, suppose 
 $ \{ x^*_{i,j}\}_{i\in \mathcal J_{e, l, K}, 1\le j \le m_i } $ are jointly independent. Let $ \{ B_{i,j} \}_{i\in \mathcal{J}_{e,l}, 1\le j \le m_i}$  be any collection of Borel sets in $ [0,1]^d$. Then 
 \begin{align*}
     &\mathbb{P}\bigg( x^*_{i, j} \in B_{i,j}  
     , i \in \mathcal J_{e,l,K+1}, 1\le j \le m_i
     \bigg) 
     \\
     =&\mathbb{P} \bigg(  x^*_{i, j} \in B_{i,j}  
     , i \in \mathcal J_{e,l,K}, 1\le j \le m_i  \bigg)  
     \mathbb{P} \bigg(  x^*_{(2K+3)L+l, j} \in B_{(2K+3)L+l,j}   
      , 1\le j \le m  \bigg) , 
 \end{align*}
 where equation follows because $\{ x^*_{(2K+3)L+l, j}\}_{1\le j \le m_i}$ are independent of $     \sigma( \{x_{t,j}  \}_{ t\le (2K+1)L+l ,1\le j \le m_i  } )  .$
 By induction,
 $$\mathbb{P} \bigg(  x^*_{i,j} \in B_{i,j} , 
     i \in \mathcal J_{e,l,K}, 1\le j \le m_i  \bigg)=  \prod_{i \in \mathcal{J}_{e, l, K}, 1\le j \le m_i} \mathbb{P}\bigg(  x^*_{i,j} \in B_{i,j}\bigg).$$
     Since for any $i \in [1, \ldots, n ]$,
       $  \{ x_{i,j}\}_{j=1}^{m_i}  $ and  $  \{x_{i,j} ^*\}_{j=1}^{m_i}   $ have the same marginal distribution,
  \begin{align*}
       &\mathbb{P} \bigg(  x^*_{(2K+3)L+l, j} \in B_{(2K+3)L+l, j} , 
      1\le j \le m_i  \bigg)
     \\
     =&\prod_{1\le j \le m_i }\mathbb{P} \bigg(  x^*_{(2K+3)L+l, j} \in B_{(2K+3)L+l, j}  \bigg) .
  \end{align*}    
  So \begin{align*}
     &\mathbb{P}\bigg( x^*_{i,j} \in B_{i,j}  ,
     i \in \mathcal J_{e,l,K+1}, 1\le j \le m_i
     \bigg)  =  \prod_{i \in \mathcal{J}_{e, l, K+1} , 1\le j \le m_i} \mathbb{P}\bigg(  x^*_{i, j} \in B_{i,j}\bigg).
      \end{align*}
      This implies that  $ \{ x^*_{i,j}\}_{i\in \mathcal J_{e, l, K+1}, 1\le j \le m_i } $ are jointly independent.
      To analyze the case where the indices are considered in $\mathcal{J}_{o,l}$, we follow the same line of arguments, but now we note that
$$\mathcal{J}_{o,l} = \{l, 2L+l, 4L+l, \dots  \} \cap [1,  n],$$
and denote,
$$ \mathcal{J}_{o,l, K} = \{l, 2L+l, 4L+l, \dots, (2K)L+l  \}\cap [1,  n],$$
which implies $ \bigcup _{K \ge 0 } \mathcal{J}_{o,l, K} = \mathcal{J}_{o,l} $. Therefore, the result is followed by replicating the analysis performed above.
\end{proof}

\begin{lemma}
\label{Ind-cop-measerrors}Let $\{\epsilon_{i,j}\}_{i=1,j=1}^{n,m_i}$ in $\mathbb{R}$ are identically distributed random variables such that $\{\sigma_\epsilon(i)\}_{i=0}^n$ are $\beta$-mixing, where $\sigma_\epsilon(i)=\sigma(\epsilon_{i,j},j\in[m_i])$. Moreover, suppose that for any $i\in[n]$ fixed, $\{\epsilon_{i,j}\}_{j=1}^{m_i}$ are independent.
For any $l\in [1,\ldots L]$, there exists a collection of independent identically distributed  random variables 
$\{   \epsilon^*_{i,j}   \}_{1\le j \le m_i ,i \in \mathcal J_{e,l}} $, with same distribution as $\epsilon_{1,1}$, such that for any $i \in \mathcal J_{e,l}$,
$$ \mathbb{P} (\epsilon_{i,j}^* \not =  \epsilon_{i,j} \text{ for some } j \in [1,\ldots,m_i] ) \le \beta(L).$$
The exact same result is attained for $\mathcal{J}_{o,l}.$
\end{lemma}
\begin{proof}
    The arguments used in this proof parallel to those employed in the proof of Lemma  \ref{Ind-cop}. The key distinction lies in our utilization of Borel sets within $\mathbb{R}$ for the present context.
\end{proof}

\begin{lemma}
\label{Ind-cop-func} Let $\{\delta_{i}\}_{i=1}^{n}$ be identically distributed random functions in $[0,1]^d,$ such that $\{\sigma_{\delta}(i)\}_{i=1}^n$ are $\beta$-mixing, where $\sigma_{\delta}(i)=\sigma(\delta_i).$
For any $l\in [0,\ldots L-1]$, there exists a collection $\{   \delta^*_{i}   \}_{i \in \mathcal J_{e,l}} $  of independent identically distributed  random functions in $[0,1]^d,$  with same distribution as $\delta_1,$ 
such that for any $i \in \mathcal J_{e,l}$,
$$ \mathbb{P} (\delta_{i}^* \not =  \delta_{i} ) \le \beta(L).$$
The exact same result is attained for $\mathcal{J}_{o,l}.$
\end{lemma}
\begin{proof}
    
By Lemma \ref{lemma:coupling}, for any $i\in [1,\ldots,n]$,  there exists $  \delta_{i} ^* $ such that 
\begin{align} \label{eq:temporal spacial coupling error-func}\mathbb{P}(\delta_{i}^{*}\not = \delta_{i} 
) =\beta(L),
\end{align} 
that $  \delta_i  $ and  $  \delta_i^*  $ have the same marginal distribution,
and that $     \delta_{i} ^*     $ is independent of $     \sigma( \{\delta_{t}  \}_{ t\le i-L   } )  .$ In addition, there exists a uniform random variable $\Delta_i$ independent of $\sigma_\delta$ such that  $  \delta_{i} ^* $ 
 is measurable with respect to $\sigma(\{ \delta_{i}, \{\delta_{t}  \}_{ t\le i-L }, \Delta  ) $.
 \\
 \\
Note that  $\{ \delta^*_{i}\}_{i \in \mathcal{J}_{e,l}}$ are identically distributed because $\{ \delta^*_{i}\}_{i \in \mathcal{J}_{e,l}}$ have the same marginal density. Therefore it suffices to justify the independence. 
\\
\\
 Note that 
$$\mathcal{J}_{e,l} = \{L+l, 3L+l, 5L+l, \dots  \} \cap [1,  n].$$
 Denote 
$$ \mathcal{J}_{e,l, K} = \{L+l, 3L+l, 5L+l, \dots, (2K+1)L+l  \}\cap [1,  n],$$
 and so $ \bigcup _{K \ge 0 } \mathcal{J}_{e,l, K} = \mathcal{J}_{e,l} $.
 By induction, suppose 
 $ \{ \delta^*_{i}\}_{i\in \mathcal J_{e, l, K}} $ are jointly independent. Let $\{t_i\}_{i\in \mathcal{J}_{e,l}} \subset [0,1]^d$ and $ \{ B_{i} \}_{i\in \mathcal{J}_{e,l}}$  be any collection of Borel sets in $\mathbb{R}$. Then 
 \begin{align*}
     &\mathbb{P}\bigg( \delta^*_{i}(t_i) \in B_i  
     , i \in \mathcal J_{e,l,K+1}
     \bigg) 
     \\
     =&\mathbb{P} \bigg(  \delta^*_{i}(t_i) \in B_{i}  
     , i \in \mathcal J_{e,l,K} \bigg)  
     \mathbb{P} \bigg(  \delta^*_{(2K+3)L+l}(t_{(2K+3)L+l}) \in B_{(2K+3)L+l}   
     \bigg) , 
 \end{align*}
 where equation follows because $\delta^*_{(2K+3)L+l}$ is independent of $     \sigma( \{\delta_{t}  \}_{ t\le 2(K+1)L+l  } )  .$
 By induction,
 $$\mathbb{P} \bigg(  \delta^*_{i}(t_i) \in B_{i} , 
     i \in \mathcal J_{e,l,K}  \bigg)=  \prod_{i \in \mathcal{J}_{e, l, K} } \mathbb{P}\bigg(  \delta^*_{i}(t_i) \in B_{i}\bigg).$$  
  Therefore \begin{align*}
     &\mathbb{P}\bigg( \delta^*_{i}(t_i) \in B_{i}  ,
     i \in \mathcal J_{e,l,K+1}
     \bigg)  =  \prod_{i \in \mathcal{J}_{e, l, K+1} } \mathbb{P}\bigg(  \delta^*_{i}(t_i) \in B_{i}\bigg).
      \end{align*}
      This implies that  $ \{ \delta^*_{i}\}_{i\in \mathcal J_{e, l, K+1}} $ are jointly independent.
      To analyze the case where the indices are considered in $\mathcal{J}_{o,l}$, we follow the same line of arguments, but now we note that
$$\mathcal{J}_{o,l} = \{l, 2L+l, 4L+l, \dots  \} \cap [1,  n],$$
and denote,
$$ \mathcal{J}_{o,l, K} = \{l, 2L+l, 4L+l, \dots, (2K)L+l  \}\cap [1,  n],$$
which implies $ \bigcup _{K \ge 0 } \mathcal{J}_{o,l, K} = \mathcal{J}_{o,l} $. Therefore, the result is followed by replicating the analysis performed above.
\end{proof}

\newpage
\section{Piecewise Lipschitz functions}\label{plsection}
For a set $A \subset [0,1]^d$, we write $B_\varepsilon(A)=\left\{a\right.$ : exists $a^{\prime} \in$ $A$, with $\left.\vert\vert a- a^{\prime}\vert\vert_2 \leq \varepsilon\right\}$. We let $\partial A$ denote the boundary of the set $A$. 
\begin{definition}
\label{Piecewise-def}
For $\varepsilon>0$  let $\Omega_\varepsilon:=[0,1]^d \backslash B_\varepsilon\left(\partial[0,1]^d\right)$. We say that a bounded function $f:[0,1]^d \rightarrow \mathbb{R}$ is piecewise Lipschitz if there exists a set $\mathcal{S} \subset(0,1)^d$ that has the following properties:
 \begin{itemize}
     \item The set $\mathcal{S}$ has Lebesgue measure zero.
     \item For some constants $C_{\mathcal{S}}, \epsilon_0>0$, we have that $Vol\left(B_\epsilon(\mathcal{S}) \cup\left([0,1]^d \backslash \Omega_\epsilon\right)\right) \leq C_{\mathcal{S}} \varepsilon$ for all $0<\varepsilon<$ $\varepsilon_0$.
     \item There exists a positive constant $L_f$ such that if $z$ and $z^{\prime}$ belong to the same connected component of $\Omega_\varepsilon \backslash B_\varepsilon(\mathcal{S})$, then $\left|g(z)-g\left(z^{\prime}\right)\right| \leq L_{f}\left\|z-z^{\prime}\right\|_2$.
 \end{itemize}
\end{definition}

We also require a broader class of functions, which generalizes and relaxes the piecewise Lipschitz condition in Definition~\ref{Piecewise-def}. This class, introduced in \cite{madrid2020adaptive}, allows for functions with weaker regularity while still admitting control over their total variation along local neighborhoods.

Let \( \epsilon > 0 \) and let \( \mathcal{P}_\epsilon \) denote the regular grid partition of the unit cube \( [0,1]^d \) into rectangles of side length \( \epsilon \), that is, induced by the points \( \{0, \epsilon, 2\epsilon, \dots, 1\} \) along each coordinate. Let \( \Omega_{2\epsilon} := [0,1]^d \setminus B_{2\epsilon}(\partial[0,1]^d) \) denote the interior region of the domain away from the boundary. Given a singular set \( \mathcal{S} \subset (0,1)^d \), we define the restricted grid:

\[
\mathcal{P}_{\epsilon, \mathcal{S}} := \left\{ A \cap (\Omega_{2\epsilon} \setminus B_{2\epsilon}(\mathcal{S})) : A \in \mathcal{P}_\epsilon, \; A \cap (\Omega_{2\epsilon} \setminus B_{2\epsilon}(\mathcal{S})) \neq \emptyset \right\},
\]
which restricts attention to grid cells fully contained in the non-singular interior region.

Then, for any bounded measurable function \( g: [0,1]^d \to \mathbb{R} \), we define two complexity functionals over the partition \( \mathcal{P}_{\epsilon, \mathcal{S}} \):
\begin{itemize}
    \item A local Lipschitz average:
    \[
    S_1(g, \mathcal{P}_{\epsilon, \mathcal{S}}) := \sum_{A \in \mathcal{P}_{\epsilon, \mathcal{S}}} \sup_{z_A \in A} \frac{1}{\epsilon} \int_{B_\epsilon(z_A)} |g(z_A) - g(z)| \, dz,
    \]
    which measures how much \( g \) varies over small neighborhoods across regular grid cells.
    
    \item A local smoothness norm:
    \[
    S_2(g, \mathcal{P}_{\epsilon, \mathcal{S}}) := \sum_{A \in \mathcal{P}_{\epsilon, \mathcal{S}}} \sup_{z_A \in A} T(g, z_A) \epsilon^d,
    \]
    where
    \[
    T(g, z_A) := \sup_{z \in B_\epsilon(z_A)} \sum_{l=1}^d \left| \int_{\|z'\|_2 \leq \epsilon} \frac{\partial \psi(z'/\epsilon)}{\partial z_l} \left\{ \frac{g(z_A - z') - g(z - z')}{\|z - z_A\|_2 \epsilon^d} \right\} dz' \right|,
    \]
    and \( \psi \) is a smooth, compactly supported test function used to localize the analysis.
\end{itemize}

If \( g \) is piecewise Lipschitz, then both \( S_1(g, \mathcal{P}_{\epsilon, \mathcal{S}}) \) and \( S_2(g, \mathcal{P}_{\epsilon, \mathcal{S}}) \) are uniformly bounded over small \( \epsilon \); see Appendix~C.1 of \cite{madrid2020adaptive}. Thus, the piecewise Lipschitz class is strictly contained within this more general class.

\begin{definition}
\label{GeneralizedPiecewiseDef}
Let \( \Omega_\varepsilon := [0,1]^d \setminus B_\varepsilon(\partial[0,1]^d) \). A bounded function \( f: [0,1]^d \rightarrow \mathbb{R} \) satisfies the generalized piecewise regularity condition if there exists a set \( \mathcal{S} \subset (0,1)^d \) such that:
\begin{itemize}
    \item The set \( \mathcal{S} \) has Lebesgue measure zero.
    \item For some constants \( C_{\mathcal{S}} > 0 \) and \( \varepsilon_0 > 0 \), we have
    \[
    \text{Vol}\left( B_\varepsilon(\mathcal{S}) \cup \left( [0,1]^d \setminus \Omega_\varepsilon \right) \right) \leq C_{\mathcal{S}} \varepsilon
    \quad \text{for all } 0 < \varepsilon < \varepsilon_0.
    \]
    \item The following complexity quantities are uniformly bounded:
    \[
    \sup_{0 < \varepsilon < \varepsilon_0} \max \left\{ S_1(f, \mathcal{P}_{\varepsilon, \mathcal{S}}), \; S_2(f, \mathcal{P}_{\varepsilon, \mathcal{S}}) \right\} < \infty,
    \]
    where \( S_1 \) and \( S_2 \) are defined as in the previous discussion.
\end{itemize}
\end{definition}

Denote by $\mathcal{F}(L) $, the set of piecewise Lipschitz functions $f:[0,1]^d\rightarrow\mathbb{R}$ such that $L_f\le \mathbb{L}$ for some finite constant $\mathbb{L}.$ A direct consequence of \cite{willett2005faster} is the following. Let $\{x_{i,j}\}_{i=1,j=1}^{n,m_i}\subset [0,1]^d$ denoting the random locations from Assumption \ref{assume:tv functions}, where the spatio-temporal data $\{y_{i,j}\}_{i=1,j=1}^{n,m_i},$ defined by
\begin{equation}
\label{Model-K-appendix}
y_{i,j}=f^*(x_{i,j})+\epsilon_{i,j},
\end{equation}
is observed. Here $\{\epsilon_{i,j}\}_{i=1,j=1}^{n,m_i}$ denotes the measurements error and follow Assumption \ref{assume:tv functions}. Consider $L=\frac{1}{C_{\beta}}\log^d(n)$ with $C_{\beta}>0$ a constant stated in Assumption \ref{assume:tv functions}. Take into account the copies of the design points and measurement errors, denoted as $\{x_{i,j}^*\}_{i=1,j=1}^{n,m_i}$ and $\{\epsilon_{i,j}^*\}_{i=1,j=1}^{n,m_i}$, which were generated in Appendix \ref{BetaMsection}.  During the formation of these copies in Appendix \ref{BetaMsection}, the quantities $ | \mathcal J_{e,l}|$ and $ | \mathcal J_{o,l}|$, which represent the count of even and odd blocks respectively, satisfy $ | \mathcal J_{e,l}|,| \mathcal J_{o,l}|  \asymp n/L$. Thus, there exist positive constants $c_1$ and $c_2$ such that  
 \begin{equation}
     \label{BoundBlocksize-K-ap}
      c_1n/L\le| \mathcal J_{e,l}|,| \mathcal J_{o,l}| \le  c_2n/L.
 \end{equation}
By Lemma \ref{Ind-cop}, we have that $\{x_{i,j}^*\}_{i\in\mathcal{J}_{e,l},j\in[m_i]}$ and $\{\epsilon_{i,j}^*\}_{i\in\mathcal{J}_{o,l},j\in[m_i]}$ are independent. Further, for any $i\in[n]$ it follows that $\mathbb{P}(x_{i,j}\neq x_{i,j}^*\ \text{for some}\ j\in[m_i])\le \beta_{x}(L).$ Moreover, by Assumption $\ref{assume:tv functions}$ and that $L=\frac{1}{C_\beta}\log^d(n)\ge \frac{2}{C_{\beta}}\log(n)$, we note that $\mathbb{P}(x_{i,j}\neq x_{i,j}^*\ \text{for some}\ j\in[m_i])\le \beta_{x}(L)\le \frac{1}{n^2}.$  Therefore, the event $\Omega_x=\bigcap_{i=1}^{n}\{x_{i,j}=x_{i,j}^*\ \forall\ j\in[m_i]\}$ happens with probability at least $\frac{1}{n}$. The same is satisfied for  $\{\epsilon_{i,j}^*\}_{i=1,j=1}^{n,m_i}$ with $\Omega_\epsilon=\bigcap_{i=1}^{n}\{\epsilon_{i,j}=\epsilon_{i,j}^*\ \forall\ j\in[m_i]\}$. We denote by $\Omega$ the event $\Omega_x\cap\Omega_\epsilon.$

Let $\widetilde{f}$ be an estimator for the model described in Equation (\ref{Model-K-appendix}) based on the observations $\{(x_{i,j},y_{i,j})\}_{i=1,j=1}^{n,m_i}$. We let $\theta^*\in\mathbb{R}^{\sum_{i=1}^nm_i}$ given by $\theta_{i,j}^*=f^*(x_{i,j})$ for any $(i,j)\in [n]\times[m_i]$, where the abuse of notation explained in Section \ref{sec:notation} is used. Similarly, we consider $\mu^*,\widetilde{\theta},\widetilde{\mu}\in\mathbb{R}^{\sum_{i=1}^nm_i}$ with each entry given by $\mu_{i,j}^*=f^*(x_{i,j}^*)$, $\widetilde{\theta}_{i,j}=\widetilde{f}(x_{i,j})$, and $\widetilde{\mu}_{i,j}=\widetilde{f}(x_{i,j}^*)$ for any $(i,j)\in [n]\times[m_i]$, respectively.

Let $l\in[L]$. Under the event $\Omega$ we notice that $\widetilde{f}$ is an estimator for the model
$$y_{i,j}=f^*(x_{i,j}^*)+\epsilon_{i,j}^*,\ i\in{\mathcal{J}_{e,l}},\ j\in[m_i].$$
Specifically, we have that under $\Omega$ the truncated vector $\widetilde{\mu}^{\mathcal{J}_{e,l}}\in \mathbb{R}^{\sum_{i\in\mathcal{J}_{e,l}}m_i}$, given by $\widetilde{\mu}^{\mathcal{J}_{e,l}}_{i,j}=\widetilde{\mu}_{i,j}$ for any $(i,j)\in\mathcal{J}_{e,l}\times[m_i]$, is an estimator of the truncated vector $\mu^{*,\mathcal{J}_{e,l}}\in \mathbb{R}^{\sum_{i\in\mathcal{J}_{e,l}}m_i}$ given by $\mu^{*,\mathcal{J}_{e,l}}_{i,j}=\mu^*_{i,j}$ for any $(i,j)\in\mathcal{J}_{e,l}\times[m_i]$.
Using the independence of the copies of the design points and measurement errors, a direct consequence of \cite{willett2005faster} is that 
\[
    \underset{ n \rightarrow \infty}{\lim\sup}\,\underset{f^* \in \mathcal{F}     }{ \sup     }\,\mathbb{P}\left(  \| \widetilde{\mu}^{\mathcal{J}_{e,l}}-\mu^{*,\mathcal{J}_{e,l}} \|_{2}^2 \,\geq \,      \frac{C_{\text{opt,1},\mathcal{J}_{e,l}}}{  \Big(\sum_{i\in\mathcal{J}_{e,l}}m_i\Big)^{1/d} }  \right)\,>0,\,
   \]
   for a positive constant $C_{\text{opt,1},\mathcal{J}_{e,l}}.$ 
   Assume that $Cm\ge m_i\ge cm,$ for a positive constant $c,C$, and $m=\Big(\frac{1}{n}\sum_{i=1}^n\frac{1}{m_i}\Big)^{-1}$. From this fact and Inequality (\ref{BoundBlocksize-K-ap}) we obtain that 
   \begin{equation}
   \label{AppendixK-eq1}
    \underset{ n \rightarrow \infty}{\lim\sup}\,\underset{f^* \in \mathcal{F}     }{ \sup     }\,\mathbb{P}\left( \| \widetilde{\mu}^{\mathcal{J}_{e,l}}-\mu^{*,\mathcal{J}_{e,l}} \|_{2}^2 \,\geq \,       \frac{C_{\text{opt,1},\mathcal{J}_{e,l}}L^{1/d}}{  (Cc_2nm)^{1/d} } \right)\,>C_0>0.\,
   \end{equation}
   for a constant $C_0.$
Similarly, under the event $\Omega$, $\widetilde{f}$ is an estimator for the model
$$y_{i,j}=f^*(x_{i,j}^*)+\epsilon_{i,j}^*,\ i\in{\mathcal{J}_{o,l}},\ j\in[m_i],$$ 
with
\begin{equation}
   \label{AppendixK-eq2}
    \underset{ n \rightarrow \infty}{\lim\sup}\,\underset{f^* \in \mathcal{F}     }{ \sup     }\,\mathbb{P}\left(  \| \widetilde{\mu}^{\mathcal{J}_{o,l}}-\mu^{*,\mathcal{J}_{o,l}} \|_{2}^2 \,\geq \,       \frac{C_{\text{opt,1},\mathcal{J}_{o,l}}L^{1/d}}{  (Cc_2nm)^{1/d} } \right)\,>C_0>0.\,
   \end{equation}
   for a positive constant $C_{\text{opt,1},\mathcal{J}_{o,l}}.$ 

   We are now going to elaborate on the final observations of this discussion. Let $C_{\text{opt,1}}=\frac{1}{(Cc_2)^{1/d}}\min_{l\in[L]}\{{C_{\text{opt,1},\mathcal{J}_{e,l}},C_{\text{opt,1},\mathcal{J}_{o,l}}}\}$. Then,
\begin{align*}
     \underset{ n \rightarrow \infty}{\lim\sup}\,\underset{f^* \in \mathcal{F}     }{ \sup     }\,\mathbb{P}\left(  \| \widetilde{\theta}-\theta^* \|_{2}^2 \,\geq \,       \frac{C_{\text{opt,1}}}{  (nm)^{1/d} } \right)\ge \underset{ n \rightarrow \infty}{\lim\sup}\,\underset{f^* \in \mathcal{F}     }{ \sup     }\,\mathbb{P}\left( \Big\{ \|\widetilde{\theta}-\theta^* \|_{2}^2 \,\geq \,       \frac{C_{\text{opt,1}}}{  (nm)^{1/d} } \Big\}\cap\Omega\right) >0.
\end{align*}
The logic behind the second inequality is as follows. For any $f^*\in\mathcal{F},$
\begin{align}
\label{AppendixK-eq3}
      &\mathbb{P}\left( \Big\{ \| \widetilde{\theta}-\theta^* \|_{2}^2 \,\geq \,       \frac{C_{\text{opt,1}}}{  (nm)^{1/d} } \Big\}\cap\Omega\right)\nonumber
      \\
      \ge &\mathbb{P}\left( \Big\{ \| \widetilde{\theta}-\theta^* \|_{2}^2 \,\geq \,       \frac{L^{1/d}C_{\text{opt,1}}}{  (nm)^{1/d} } \Big\}\cap\Omega\right)\nonumber
      \\
      \ge&\mathbb{P}\left(\bigcup_{l=1}^L \Big\{\Big\{ \| \widetilde{\mu}^{\mathcal{J}_{e,l}}-\mu^{*,\mathcal{J}_{e,l}} \|_{2}^2 \,\geq \,       \frac{L^{1/d}C_{\text{opt,1},\mathcal{J}_{e,l}}}{  (Cc_2nm)^{1/d} } \Big\}\bigcup \Big\{ \| \widetilde{\mu}^{\mathcal{J}_{o,l}}-\mu^{*,\mathcal{J}_{o,l}} \|_{2}^2 \,\geq \,       \frac{L^{1/d}C_{\text{opt,1},\mathcal{J}_{o,l}}}{  (Cc_2nm)^{1/d} } \Big\}\Big\}\cap\Omega\right)\nonumber
      \\
      >&C_0,
\end{align}
where the first inequality is followed by the fact that $L=\frac{1}{C_\beta}\log^d(n)\ge1.$ Moreover the third inequality is achieved from Inequality (\ref{AppendixK-eq1}) and (\ref{AppendixK-eq2}). The derivation of the second inequality is explained below.

Under the event $\Omega$ the relation, $$\| \widetilde{\theta}-\theta^* \|_{2}^2=\|\widetilde{\mu}-\mu^*\|_2^2=\sum_{l=1}^L\Big(\| \widetilde{\mu}^{\mathcal{J}_{e,l}}-\mu^{*,\mathcal{J}_{e,l}} \|_{2}^2 +\| \widetilde{\mu}^{\mathcal{J}_{o,l}}-\mu^{*,\mathcal{J}_{o,l}} \|_{2}^2\Big),$$ is satisfied. Therefore, $$\| \widetilde{\theta}-\theta^*  \|_{2}^2\ge \Big\{\| \widetilde{\mu}^{\mathcal{J}_{e,l}}-\mu^{*,\mathcal{J}_{e,l}} \|_{2}^2,\| \widetilde{\mu}^{\mathcal{J}_{o,l}}-\mu^{*,\mathcal{J}_{o,l}} \|_{2}^2\Big\},$$ for any $l\in[L].$ These specific conditions allow us to conclude that the event $\{\| \widetilde{\theta}-\theta^* \|_{2}^2< \frac{L^{1/d}C_{\text{opt,1}}}{  (nm)^{1/d} } \}$ is contained in the event
$$\bigcap_{l=1}^L \Big\{\Big\{ \| \widetilde{\mu}^{\mathcal{J}_{e,l}}-\mu^{*,\mathcal{J}_{e,l}} \|_{2}^2 \,< \,       \frac{L^{1/d}C_{\text{opt,1},\mathcal{J}_{e,l}}}{  (Cc_2nm)^{1/d} } \Big\}\bigcap \Big\{ \| \widetilde{\mu}^{\mathcal{J}_{o,l}}-\mu^{*,\mathcal{J}_{o,l}} \|_{2}^2 \,< \,       \frac{L^{1/d}C_{\text{opt,1},\mathcal{J}_{o,l}}}{  (Cc_2nm)^{1/d} } \Big\}\Big\},$$
and consequently the second inequality in Inequality (\ref{AppendixK-eq3}) is satisfied. Specifically, suppose that  $\|\widetilde{\theta}-\theta^*  \|_{2}^2 \,< \,       \frac{L^{1/d}C_{\text{opt,1}}}{  (nm)^{1/d} } $. For any $l\in[L]$ by the definition of $C_{\text{opt,1}}$ we have that $\| \widetilde{\theta}-\theta^*  \|_{2}^2 \,< \,       \frac{L^{1/d}C_{\text{opt,1},\mathcal{J}_{e,l}}}{  (Cc_2nm)^{1/d} }$. It follows that $\| \widetilde{\mu}^{\mathcal{J}_{e,l}}-\mu^{*,\mathcal{J}_{e,l}} \|_{2}^2 \,< \,       \frac{L^{1/d}C_{\text{opt,1},\mathcal{J}_{e,l}}}{  (Cc_2nm)^{1/d} }$. Similarly, $\| \widetilde{\mu}^{\mathcal{J}_{o,l}}-\mu^{*,\mathcal{J}_{o,l}}\|_{2}^2 \,< \,       \frac{L^{1/d}C_{\text{opt,1},\mathcal{J}_{o,l}}}{  (Cc_2nm)^{1/d} }$, concluding the contention of the aforementioned events.
   \newpage

{\color{black}{
\subsection{Definition of Bounded Variation in Multiple Dimensions}
\label{app:BV-def}

In this appendix, we recall the notion of bounded variation for functions defined
on multidimensional domains. This definition is used throughout the paper when
referring to functions $f : [0,1]^d \to \mathbb{R}$.

Let $(0,1)^d \subset \mathbb{R}^d$. For a vector-valued function
$g = (g_1,\dots,g_d) : (0,1)^d \to \mathbb{R}^d$ with sufficiently smooth components,
the divergence operator is defined as
\[
\operatorname{div}(g)(x)
:= \sum_{j=1}^d \frac{\partial g_j(x)}{\partial x_j},
\qquad x \in (0,1)^d,
\]
whenever the partial derivatives exist.

We denote by $C_c^1((0,1)^d,\mathbb{R}^d)$ the class of continuously differentiable
vector fields with compact support in $(0,1)^d$. For such a vector field $g$, we
define
\[
\|g\|_\infty
:= \sup_{x \in (0,1)^d} \left( \sum_{j=1}^d g_j(x)^2 \right)^{1/2}.
\]

\begin{definition}[Bounded variation]
A function $f : [0,1]^d \to \mathbb{R}$ is said to have bounded variation if
\[
|f|_{\mathrm{BV}}
:=
\sup \left\{
\int_{(0,1)^d} f(x)\,\operatorname{div}(g)(x)\,dx
\;\middle|\;
g \in C_c^1((0,1)^d,\mathbb{R}^d),\ \|g\|_\infty \le 1
\right\}
< \infty.
\]
\end{definition}

The space of all functions with finite total variation on $[0,1]^d$ is denoted by
$\mathrm{BV}([0,1]^d)$. We refer the reader to \cite{ziemer2012weakly} for a full description of the class of functions of bounded variation
}}
   \newpage

\section{Examples of Functions with Bounded Total Variation along the $K$-NN Graph}

In this appendix, we present examples of functions that exhibit bounded total variation along the $K$-NN graph. To this end, we rely on the following result, which is a restatement of Lemma~\ref{Lemma-nabla}:

\begin{lemma}\label{Lemma-nabla1}
Let Assumption~\ref{assume:tv functions} hold. Consider \( K = \log^{2 + l}(nm) \) for some \( l > 0 \). Then
\begin{equation}
\label{eqn:tv_bound-1-aux}
\| \nabla_{G_K} \theta^* \|_1 = O_{\mathbb{P}}\left((nm)^{1 - 1/d} \log^{2 + l + \frac{2 + l}{d}}(nm)\right).
\end{equation}
\end{lemma}

From this lemma, it follows that any function satisfying Assumption~\ref{assume:tv functions} also satisfies the desired total variation bound along the $K$-NN graph. In particular, the class of piecewise Lipschitz functions described in Appendix~\ref{plsection} satisfies this assumption. Moreover, as shown in Appendix C.1 of \cite{madrid2020adaptive}, these functions are also contained within the generalized function class defined in Definition~\ref{GeneralizedPiecewiseDef} of Appendix~\ref{plsection}. Therefore, it suffices to focus on piecewise Lipschitz functions, which we do throughout this appendix.

\subsection{Step function with a jump}
\begin{example}
Let \( f_1 : [0,1]^2 \rightarrow \mathbb{R} \) be defined by
\( f_1(x_1, x_2) = \begin{cases}
1 & \text{if } x_1 \leq 0.5, \\
0 & \text{if } x_1 > 0.5.
\end{cases} \)
We observe that \( f_1 \in \mathcal{F}(1) \), the class of piecewise Lipschitz functions defined in Definition~\ref{Piecewise-def}. A graph of this function is shown in Figure~\ref{fig1-example}.

To verify this, we take \( \mathcal{S} := \{(x_1, x_2) \in [0,1]^2 : x_1 = 0.5\} \), the vertical line where \( f_1 \) jumps. This set has Lebesgue measure zero. For any \( \varepsilon > 0 \), the tubular neighborhood \( B_\varepsilon(\mathcal{S}) \) is a strip of width \( 2\varepsilon \) and unit height, so \( \operatorname{Vol}(B_\varepsilon(\mathcal{S})) = 2\varepsilon \).

The complement of \( \Omega_\varepsilon \) corresponds to a strip of width \( \varepsilon \) around the boundary of \([0,1]^2\), whose perimeter is 4. Thus, \( \operatorname{Vol}([0,1]^2 \setminus \Omega_\varepsilon) \leq 4 \cdot 2\varepsilon = 8\varepsilon \). Combining both, we have \( \operatorname{Vol}(B_\varepsilon(\mathcal{S}) \cup ([0,1]^2 \setminus \Omega_\varepsilon)) \leq 10\varepsilon \), so the second condition in Definition~\ref{Piecewise-def} holds with \( C_{\mathcal{S}} = 10 \) and \( \varepsilon_0 = 0.25 \), for instance.

Next, we verify local Lipschitz continuity. On each connected component of \( \Omega_\varepsilon \setminus B_\varepsilon(\mathcal{S}) \), the function \( f_1 \) is constant. Therefore, for any \( z, z' \) in the same component, \( |f_1(z) - f_1(z')| = 0 \leq L_f \|z - z'\|_2 \) for any \( L_f > 0 \); we may choose \( L_f = 1 \).

Hence, all conditions are satisfied with \( \mathcal{S} = \{x_1 = 0.5\} \), \( C_{\mathcal{S}} = 10 \), \( \varepsilon_0 = 0.25 \), \( L_f = 1 \), and we conclude that \( f_1 \in \mathcal{F}(1) \).
\end{example}
\begin{figure}[ht]
    \centering
    \includegraphics[width=0.55\textwidth]{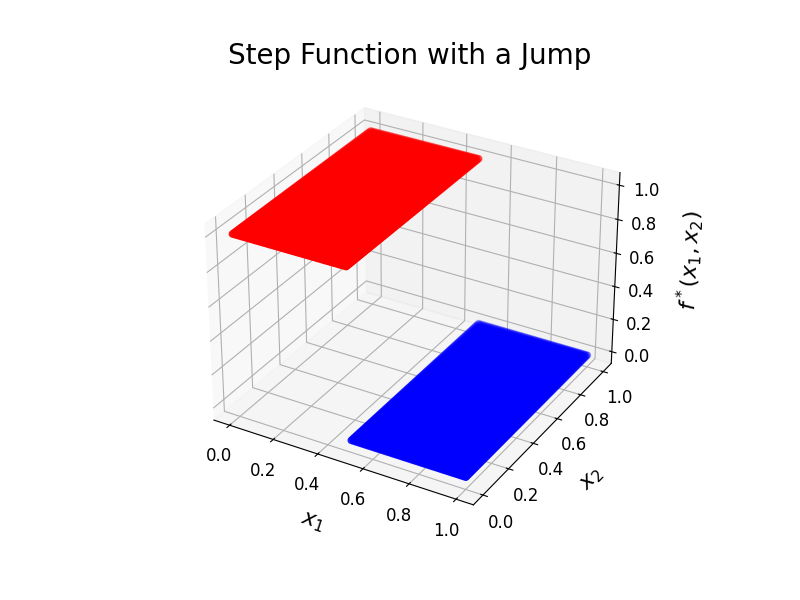}
    \caption{Graph of the function \( f_1(x_1, x_2) = \mathbf{1}\{x_1 \leq 0.5\} \) on the unit square.}
    \label{fig1-example}
\end{figure}

\subsection{Function with a flat disk core and nonlinear exterior}
\begin{example}
Let \( f_2 : [0,1]^2 \rightarrow \mathbb{R} \) be defined by
\[
f_2(x_1, x_2) = \begin{cases}
1 & \text{if } (x_1 - 0.5)^2 + (x_2 - 0.5)^2 \leq r^2, \\
x_1^2 + x_2 & \text{otherwise},
\end{cases}
\]
with \( r = 0.25 \). We claim that \( f_2 \in \mathcal{F}(2) \), the class of piecewise Lipschitz functions defined in Definition~\ref{Piecewise-def}. A graph of this function is shown in Figure~\ref{fig2-example}.

To verify this, define the singular set \( \mathcal{S} := \left\{(x_1, x_2) \in [0,1]^2 : (x_1 - 0.5)^2 + (x_2 - 0.5)^2 = r^2 \right\} \), which is the boundary of a disk where the function transitions from constant to nonlinear. This set has Lebesgue measure zero. The tubular neighborhood \( B_\varepsilon(\mathcal{S}) \) is an annular region of width \( 2\varepsilon \), with volume at most \( 2\pi r \varepsilon \). The boundary strip \( [0,1]^2 \setminus \Omega_\varepsilon \) has volume at most \( 8\varepsilon \). Thus, \( \operatorname{Vol}(B_\varepsilon(\mathcal{S}) \cup ([0,1]^2 \setminus \Omega_\varepsilon)) \leq (2\pi r + 8)\varepsilon \), so the second condition in Definition~\ref{Piecewise-def} holds with \( C_{\mathcal{S}} = 2\pi r + 8 \) and \( \varepsilon_0 = 0.25 \), for instance.

Next, we verify the Lipschitz condition. In the interior of the disk, \( f_2 \equiv 1 \) is constant. Outside the disk, we have
\(
|f_2(z) - f_2(z')| = |x_1^2 - \tilde{x}_1^2 + x_2 - \tilde{x}_2| \leq |x_1^2 - \tilde{x}_1^2| + |x_2 - \tilde{x}_2|.
\)
Using the identity \( |a^2 - b^2| = |a - b||a + b| \), we get \( |x_1^2 - \tilde{x}_1^2| = |x_1 - \tilde{x}_1||x_1 + \tilde{x}_1| \leq 2|x_1 - \tilde{x}_1| \) since \( x_1, \tilde{x}_1 \in [0,1] \). Hence,
\(
|f_2(z) - f_2(z')| \leq 2|x_1 - \tilde{x}_1| + |x_2 - \tilde{x}_2| \leq \sqrt{5} \|z - z'\|_2,
\)
so we may take \( L_f = \sqrt{5} \).

In conclusion, all conditions are satisfied with \( \mathcal{S} = \{ (x_1 - 0.5)^2 + (x_2 - 0.5)^2 = r^2 \} \), \( C_{\mathcal{S}} = 2\pi r + 8 \), \( \varepsilon_0 = 0.25 \), \( L_f = \sqrt{5} \), and we conclude that \( f_2 \in \mathcal{F}(\sqrt{5}) \).
\end{example}
\begin{figure}[ht]
    \centering
    \includegraphics[width=0.55\textwidth]{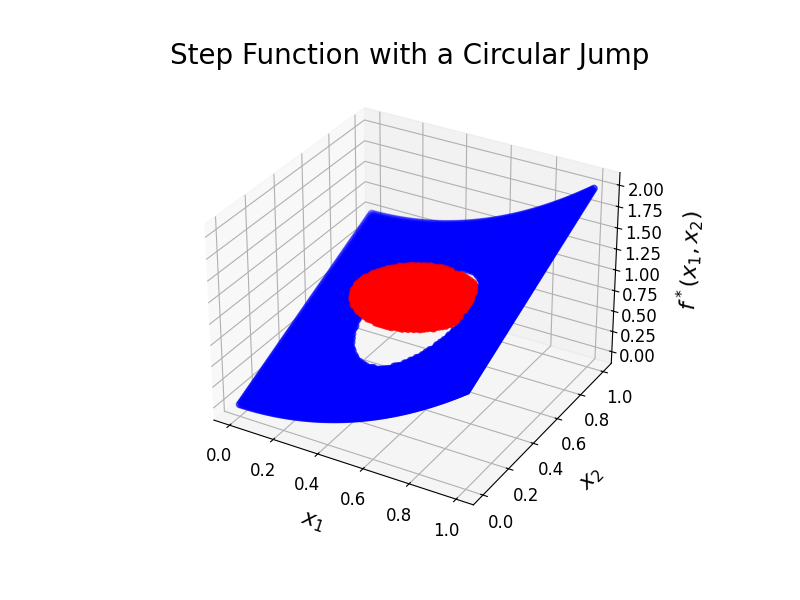}
    \caption{Graph of the function \( f_2(x_1, x_2)\) on the unit square.}
    \label{fig2-example}
\end{figure}

\subsection{Piecewise smooth function with three regions}

\begin{figure}[]
    \centering
    \includegraphics[width=0.55\textwidth]{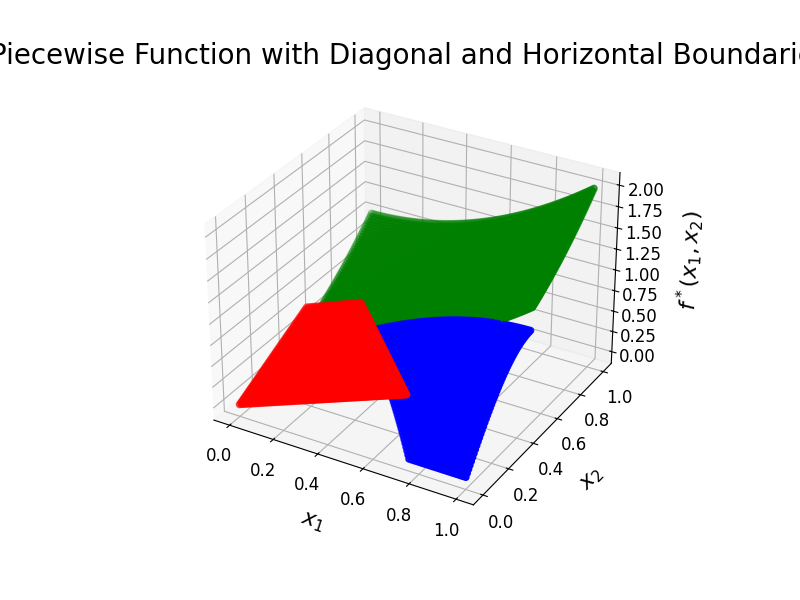}
    \caption{Graph of the function \( f_3(x_1, x_2) \) on the unit square.}
    \label{fig3-example}
\end{figure}
\begin{example}
Let \( f_3 : [0,1]^2 \rightarrow \mathbb{R} \) be defined by
\[
f_3(x_1, x_2) = \begin{cases}
x_1 + x_2 & \text{if } x_1 + x_2 \leq 0.75 \text{ and } x_2 \leq 0.5, \\
\sin(\pi x_1 x_2) & \text{if } x_1 + x_2 > 0.75 \text{ and } x_2 \leq 0.5, \\
x_1^2 + x_2^2 & \text{if } x_2 > 0.5.
\end{cases}
\]
We claim that \( f_3 \in \mathcal{F}(L_f) \), the class of piecewise Lipschitz functions defined in Definition~\ref{Piecewise-def}. A graph of this function is shown in Figure~\ref{fig3-example}.

To verify this, define the singular set
\(
\mathcal{S} := \{ (x_1, x_2) \in [0,1]^2 : x_2 = 0.5 \text{ or } x_1 + x_2 = 0.75 \}.
\)
This set consists of a horizontal line segment and a slanted line segment in \([0,1]^2\), both of Lebesgue measure zero. The tubular neighborhood \( B_\varepsilon(\mathcal{S}) \) includes two strips (a horizontal band and a slanted band) of width \( 2\varepsilon \), and has area at most \( 4\varepsilon \). The boundary strip \( [0,1]^2 \setminus \Omega_\varepsilon \) has area at most \( 8\varepsilon \). Therefore,
\(
\operatorname{Vol}(B_\varepsilon(\mathcal{S}) \cup ([0,1]^2 \setminus \Omega_\varepsilon)) \leq 12\varepsilon,
\)
so the second condition in Definition~\ref{Piecewise-def} holds with \( C_{\mathcal{S}} = 12 \) and \( \varepsilon_0 = 0.25 \), for instance.

Now consider the Lipschitz condition.

In the region \( x_1 + x_2 \leq 0.75 \) and \( x_2 \leq 0.5 \), we have \( f_3(x_1, x_2) = x_1 + x_2 \), a linear function. For any two points \( z = (x_1, x_2) \) and \( z' = (\tilde{x}_1, \tilde{x}_2) \), we have
\(
|f_3(z) - f_3(z')| = |x_1 - \tilde{x}_1 + x_2 - \tilde{x}_2| \leq |x_1 - \tilde{x}_1| + |x_2 - \tilde{x}_2| \leq \sqrt{2} \|z - z'\|_2.
\)

In the region \( x_1 + x_2 > 0.75 \) and \( x_2 \leq 0.5 \), the function is \( f_3(x_1, x_2) = \sin(\pi x_1 x_2) \). For any \( z = (x_1, x_2) \), \( z' = (\tilde{x}_1, \tilde{x}_2) \), we use the inequality
\(
|\sin(\pi a) - \sin(\pi b)| \leq \pi |a - b|
\)
to write
\(
|f_3(z) - f_3(z')| \leq \pi |x_1 x_2 - \tilde{x}_1 \tilde{x}_2|.
\)
We bound the product difference by
\(
|x_1 x_2 - \tilde{x}_1 \tilde{x}_2| = |x_1 x_2 - x_1 \tilde{x}_2 + x_1 \tilde{x}_2 - \tilde{x}_1 \tilde{x}_2| \leq |x_1||x_2 - \tilde{x}_2| + |\tilde{x}_2||x_1 - \tilde{x}_1| \leq |x_2 - \tilde{x}_2| + |x_1 - \tilde{x}_1|.
\)
Hence,
\(
|f_3(z) - f_3(z')| \leq \pi (|x_1 - \tilde{x}_1| + |x_2 - \tilde{x}_2|) \leq \pi \sqrt{2} \|z - z'\|_2.
\)

In the region \( x_2 > 0.5 \), the function is \( f_3(x_1, x_2) = x_1^2 + x_2^2 \). Then
\(
|f_3(z) - f_3(z')| \leq |x_1^2 - \tilde{x}_1^2| + |x_2^2 - \tilde{x}_2^2|.
\)
Using the identity \( |a^2 - b^2| = |a - b||a + b| \) and the fact that all coordinates are in \([0,1]\), we get
\(
|x_1^2 - \tilde{x}_1^2| \leq 2|x_1 - \tilde{x}_1|, \quad |x_2^2 - \tilde{x}_2^2| \leq 2|x_2 - \tilde{x}_2|,
\)
so
\(
|f_3(z) - f_3(z')| \leq 2|x_1 - \tilde{x}_1| + 2|x_2 - \tilde{x}_2| \leq 2\sqrt{2} \|z - z'\|_2.
\)

Thus, the global Lipschitz constant is \( L_f = \max\{\sqrt{2}, \pi \sqrt{2}, 2\sqrt{2}\} \leq \pi \sqrt{2} \). In conclusion, all conditions are satisfied with \( \mathcal{S} = \{x_2 = 0.5\} \cup \{x_1 + x_2 = 0.75\} \), \( C_{\mathcal{S}} = 12 \), \( \varepsilon_0 = 0.25 \), \( L_f = \pi \sqrt{2} \), and we conclude that \( f_3 \in \mathcal{F}(\pi \sqrt{2}) \).
\end{example}

   \newpage
\section{The lattice graph}
\label{Lat-section}
We introduce the Lattice graph, for more details, we refer readers to \cite{madrid2020adaptive}.

Given $N \in \mathbb{N}$, denote the equally spaced lattice graph in $[0,1]^d$ as $G_{\text {lat }}=\left(V_{\text {lat }}, E_{\text {lat }}\right)$, where the distance between any two neighbor lattice points is $1 / N$. Therefore, the number of nodes $|V|_{\text {lat }}=N^d$.
Without loss of generality, we assume that the nodes of the grid correspond to the points
\begin{equation}
\label{Plat-def}
I=P_{\text {lat }}(N)=\left\{\left(\frac{i_1}{N}-\frac{1}{2 N}, \ldots, \frac{i_d}{N}-\frac{1}{2 N}\right): i_1, \ldots, i_d \in\{1, \ldots, N\}\right\}.
\end{equation}
Let $\left\{\mathfrak{c}_r\right\}_{r=1}^{N^d}$ be the corresponding centers and define a collection of cells, $\left\{C(\mathfrak{c}_r)\right\}_{r=1}^{N^d}$, as
$$
\mathcal{I}_r=C(\mathfrak{c}_r)=\left\{z \in[0,1]^d: \mathfrak{c}_r=\underset{x^{\prime} \in P_{\mathrm{lat}}(N)}{\arg \min }\left\|z-x^{\prime}\right\|_{\infty}\right\} .
$$ 
Now, let $\theta \in \mathbb{R}^{\sum_{i=1}^{n}m_i}$ be a collection of signals observed on the sample points $\left\{x_{i,j}\right\}_{i=1, j=1}^{n, m_i}$. Let $\theta_I \in \mathbb{R}^{\sum_{i=1}^{n}m_i}$ be the modified version of $\theta$ according to the grid induced by the lattice graph of size $N^d$. Let $\theta^I \in \mathbb{R}^{N^d}$ be the modified version of $\theta$ that takes values on the centers of the lattice graph. See section E of \cite{madrid2020adaptive} for definitions of $\theta_I$ and $\theta^I$. For completion, the definitions are included as follows.

\begin{definition}
\label{def-2-E}
 Given $\theta \in \mathbb{R}^{\sum_{i=1}^{n}m_i}$, let $\left\{\mathcal{I}_r\right\}_{r=1}^{N^d}$ be the collection of the cells in the lattice graph $G_{l a t}$ and let $\left\{\mathfrak{c}_r\right\}_{r=1}^{N^d}$ be the corresponding centers. Let $\theta_I \in \mathbb{R}^{\sum_{i=1}^{n}m_i}$ be such that for any $1 \leq \alpha \le n, 1\le \beta\le m_\alpha$,
$$
\left(\theta_I\right)_{ \alpha, \beta }=\theta_{i, j},
$$
if $x_{\alpha, \beta}$ and $x_{i,j}$ are in the same cell $\mathcal{I}_r$, and that $(i, j)=\arg \min _{1 \leq i^{\prime} \leq n, 1 \leq j^{\prime} \leq m_{i'}}\left\|x_{i^{\prime} j^{\prime}}-\mathfrak{c}_r\right\|$. Here the abuse of notation explained in Section \ref{sec:notation} is used.
\end{definition}

Denote by $I_r=\{(i,j)\in [n]\times [m_i]: P_{I}(x_{i,j})=\mathfrak{c}_r \}$ for $r=1,...,N^d,$ where $P_{I}(x)$ denotes the point in $P_{\mathrm{lat}}(N)$ such that $x\in C(P_{I}(x))$.

\begin{definition}
\label{def-1-E}Let $\theta \in \mathbb{R}^{\sum_{i=1}^{n}m_i}$ be a collection of signals observed on the sample points $\left\{x_{i,j}\right\}_{i=1, j=1}^{n, m}$. Let $\theta^I \in \mathbb{R}^{N^d}$ be such that, for any $r\in[N^d]$
$$
\theta_r^I=\left\{\begin{array}{lll}
\theta_{i_r, j_r}, & \text { if } & I_{r} \neq \emptyset \quad \text { and } \quad(i_r, j_r)=\arg \min _{1 \leq i^{\prime} \leq n, 1 \leq j^{\prime} \leq m_{i'}}\left\|x_{i^{\prime} ,j^{\prime}}-\mathfrak{c}_r\right\| ; \\
0, & \text { if } & I_{r}=\emptyset .
\end{array}\right.
$$
\end{definition}
This means that $\theta_r^I$ take the value $\theta_{i
_r, j_r}$ if $x_{i_r,j_r}$ is the closest point to the center of $\mathfrak{c}_r$ among all $\left\{x_{i,j}\right\}_{i=1, j=1}^{n, m_i}$. 

\subsection{Auxiliary Lemmas for the proof of Theorem \ref{Cest-d>1} and \ref{Pest-d>1}}

The following series of results are generalizations of those in \cite{madrid2020adaptive} to the spatio-temporal data. 
The initial finding presented below is a widely recognized concentration inequality applicable to binomial random variables. It can be found as Proposition 27 in \cite{JMLR:v15:vonluxburg14a}.
\begin{lemma}
\label{BinCon}
 Let $m$ be a Binomial $(M, q)$ random variable. Then, for all $\delta \in(0,1]$,
$$
\begin{aligned}
& \mathbb{P}(m \leq(1-\delta) M q) \leq \exp \left(-\frac{1}{3} \delta^2 M q\right), \\
& \mathbb{P}(m \geq(1+\delta) M q) \leq \exp \left(-\frac{1}{3} \delta^2 M q\right) .
\end{aligned}
$$
\end{lemma}
We proceed to apply Lemma \ref{BinCon} in order to bound the distance between a given design point and its $K$-nearest neighbors.
\begin{lemma}[{\bf{see Lemma 6 in \cite{madrid2020adaptive}}}]
\label{dist-bet-Knn}
Under Assumption \ref{assume:tv functions}.
     Denote by $R_K(x)$ the distance from $x \in [0,1]^d$ to its $K$ th nearest neighbor in the set $\left\{x_{i,j}\right\}_{i=1,j=1}^{n,m_i}$. Setting
\begin{align*}
& R_{K, \max }=\max _{1 \leq i \leq n,1 \leq j \leq m_i} R_K\left(x_{i,j}\right), \\
& R_{K, \min }=\min _{1 \leq i \leq n,1 \leq j \leq m_i} R_K\left(x_{i,j}\right),
\end{align*}
we have that, 
\begin{align*}
&\mathbb{P}\left\{a\left(\frac{K}{nm}\right)^{1 / d} \leq R_{K, \min } \leq R_{K, \max } \leq \tilde{a}\left(\frac{K}{nm}\right)^{1 / d}\right\} 
\\
\geq& 1-\frac{4}{C_{\beta}}C\log(n)nm\exp(-\frac{\widetilde{C}}{24}\frac{K}{\log(n)})-\frac{4}{C_{\beta}}C\log(n)nm\exp(-\frac{C_{\beta}}{24}\frac{K}{\log(n)})-\frac{2}{n},
\end{align*}
 where $\widetilde{C}=\frac{C_{\beta}c_1c}{c_2C}$, $a=1 /\left(  4Cc_{2,d}c_2v_{\text{max} }\right)^{1 / d}$, and $\tilde{a}=\left(1/(cc_{1,d}c_1v_{\text{min} })\right)^{1 / d}$ for some positive constants $c_2$, and $c_1$, and, $v_{\max}$, $v_{\min}$, $C$, $C_{\beta},$ and $c$  from Assumption \ref{assume:tv functions}. Moreover $c_{1,d}$ and $c_{2,d}$ are positive constants that depends on $d.$
\end{lemma}
\begin{proof}
First note that for any $x \in [0,1]^d$, by Assumptions \ref{assume:tv functions} we have
\begin{equation}
    \label{aux-lemma1-eq0}
\mathbb{P}\left(x_{i,j} \in B_r(x)\right)=\int_{B_r(x)} v(t) d t \leq c_{2,d}r^d v_{\max }:=\mu_{\max }<1,
\end{equation}
for small enough $r$ and a positive constant $c_{2,d}$ that depends on $d.$ 
Let $L=\frac{2}{C_{\beta}}\log(n)$ and $\left\{x_{i,j}^*\right\}_{i=1,j=1}^{n,m_i}$ defined as in Appendix \ref{BetaMsection} and consider the event $\Omega=\bigcap_{i=1}^n\{x_{i,j}=x_{i,j}^*,\ \forall \ 1\le j\le m_i\}$. It is satisfied that
 $ | \mathcal J_{e,l}|,| \mathcal J_{o,l}|  \asymp n/L$, from where, there exist positive constants $c_1$ and $c_2$ such that  
 \begin{equation}
 \label{aux-lemma1-eq9} 
 c_1n/L\le| \mathcal J_{e,l}|,| \mathcal J_{o,l}| \le  c_2n/L. 
 \end{equation}
 Moreover, by Lemma \ref{Ind-cop} and Assumption \ref{assume:tv functions}, we have that $\mathbb{P}(\Omega^c)\le \frac{1}{n}.$ Then, 
\begin{align}
\label{Lem-K-dist-0}
    \mathbb{P}\Big(\left\{R_{K, \min } \leq a(K / (nm))^{1 / d}\right\}\Big)\le \mathbb{P}\Big(\left\{R_{K, \min } \leq a(K / (nm))^{1 / d}\right\}\cap \Omega\Big)+\frac{1}{n}.
\end{align}
To analyze the term $\mathbb{P}\Big(\left\{R_{K, \min } \leq a(K / nm)^{1 / d}\right\}\cap \Omega\Big),$ we observe that
\begin{align}
\label{Lem-K-dist-1-1}
    &\mathbb{P}\Big(\left\{R_{K, \min } \leq a(K / (nm))^{1 / d}\right\}\cap \Omega\Big)\nonumber
    \\
    \le& \mathbb{P}\Big(\left\{\exists (i,j)\in [n]\times [m_i] :R_K\left(x_{i,j}\right) \leq a(K / (nm))^{1 / d}\right\}\bigcap \Omega\Big).
\end{align}
 We also have that $R_K(x) \leq r$ if and only if there are at least $K$ observations among $\left\{x_{i,j}\right\}$ in $B_r(x)$. Therefore, letting 
 \begin{equation}
 \label{radiusLemma1}
 r=a(K / (nm))^{1 / d},
 \end{equation}
 we have that
\begin{align}
\label{aux-lemma1-eq7}
    &\mathbb{P}\Big(\left\{\exists (i,j)\in [n]\times [m_i] :R_K\left(x_{i,j}\right) \leq a(K / (nm))^{1 / d}\right\}\bigcap \Omega\Big)\nonumber
    \\
    \le&
    \sum_{i=1}^n\sum_{j=1}^{m_i}\mathbb{P}\Big(\left\{\sum_{t=1}^n\sum_{s=1}^{m_t}\mathbf{1}_{x_{t,s}\in B_r(x_{i,j})}\ge K\right\}\bigcap \Omega\Big)\nonumber
    \\
    \le&
    Cnm\mathbb{P}\Big(\sum_{t=1}^n\sum_{s=1}^{m_t}\mathbf{1}_{x_{t,s}^*\in B_r(x_{i,j})}\ge K\Big),
\end{align}
where the first inequality is followed by a union bound argument together with the last observation. The second inequality is followed from Assumption \ref{assume:tv functions}{\bf{f}} coupled with the definition of the event $\Omega$ at the beginning of the proof.
Then, using a union bound argument
\begin{align}
\label{aux-lemma1-eq1}
    &\mathbb{P}\Big(\sum_{t=1}^n\sum_{s=1}^{m_t}\mathbf{1}_{x_{t,s}^*\in B_r(x_{i,j})}\ge K\Big)\nonumber
    \\
    \le&
    \sum_{l=1}^{L}\mathbb{P}\Big(\sum_{t \in \mathcal{J}_{e,l}}\sum_{s=1}^{m_t}\mathbf{1}_{x_{t,s}^*\in B_r(x_{i,j})}\ge K/(2L)\Big)+\sum_{l=1}^{L}\mathbb{P}\Big(\sum_{t \in \mathcal{J}_{o,l}}\sum_{s=1}^{m_t}\mathbf{1}_{x_{t,s}^*\in B_r(x_{i,j})}\ge K/(2L)\Big)\nonumber
    \\
    =& I_1+I_2.
\end{align}
We now bound each of the terms $I_1$ and $I_2.$ The analysis employed for the second term follows the same line of arguments as for $I_1$ and is omitted.
Using Lemma \ref{Ind-cop} for any $1\le l\le L$ we have that \begin{equation}
    \label{aux-lemma1-eq2}
\sum_{t \in \mathcal{J}_{e,l}}\sum_{s=1}^{m_t}\mathbf{1}_{x_{t,s}^*\in B_r(x_{i,j})}\sim\operatorname{Binomial}\left(\sum_{t \in \mathcal{J}_{e,l}}m_t,  \mathbb{P}(x_{1,1}^*\in B_r(x_{i,j}))\right).
\end{equation}
Let $\text{Bin}\sim\operatorname{Binomial}\left(\sum_{t \in \mathcal{J}_{e,l}}m_t,\mu_{\max }\right).$ By Inequality (\ref{aux-lemma1-eq0}), 
\begin{align}
\label{aux-lemma1-eq4}
    \mathbb{P}\Big(\sum_{t \in \mathcal{J}_{e,l}}\sum_{s=1}^{m_t}\mathbf{1}_{x_{t,s}^*\in B_r(x_{i,j})}\ge K/(2L)\Big)\le \mathbb{P}\Big(\text{Bin}\ge  K/(2L)\Big).
\end{align}
By Inequality (\ref{aux-lemma1-eq9}), Equation (\ref{radiusLemma1}), Assumption \ref{assume:tv functions}{\bf{f}} and the definition of $\mu_{\max}$ in Inequality (\ref{aux-lemma1-eq0}), observe that
\begin{align}
\label{aux-lemma1-eq3}
    2\mu_{\max }\sum_{t \in \mathcal{J}_{e,l}}m_t \le 2\frac{c_{2,d}v_{\max}}{4Cc_{2,d}c_2v_{\max}}\frac{K}{nm}\frac{Cc_2nm}{L}=\frac{K}{2L}.
\end{align}
From Inequality (\ref{aux-lemma1-eq3}) for any $1\le l\le L,$
\begin{align}
\label{aux-lemma1-eq5}
    \mathbb{P}\Big(\text{Bin}\ge K/(2L)\Big)\le \mathbb{P}\Big(\text{Bin}\ge  2\mu_{\max }\sum_{t \in \mathcal{J}_{e,l}}m_t\Big).
\end{align}
Therefore, by Lemma \ref{BinCon}, Inequality (\ref{aux-lemma1-eq4}) and Inequality (\ref{aux-lemma1-eq5})
\begin{align}
\label{aux-lemma1-eq6}
    \mathbb{P}\Big(\sum_{t \in \mathcal{J}_{e,l}}\sum_{s=1}^{m_t}\mathbf{1}_{x_{t,s}^*\in B_r(x_{i,j})}\ge K/(2L)\Big)\nonumber
    \le&
    \exp(-\frac{1}{3}\mu_{\max }\sum_{t \in \mathcal{J}_{e,l}}m_t)\nonumber
    \\
    \le&\exp(-\frac{c_1c}{12c_2C}\frac{K}{L})=\exp(-\frac{c_1cC_{\beta}}{24c_2C}\frac{K}{\log(n)}).
\end{align}
where the second inequality is followed by Inequality (\ref{aux-lemma1-eq9}), the choice of $r$ in Equation (\ref{radiusLemma1}), Assumption \ref{assume:tv functions}{\bf{f}}, and the choice of $\mu_{\max}$ in Inequality (\ref{aux-lemma1-eq0}). The Equality is attained by the choice of $L$ at the beginning of the proof.
Therefore, using Inequality (\ref{aux-lemma1-eq4}), (\ref{aux-lemma1-eq5}), (\ref{aux-lemma1-eq6}) and the choice of $L$ at the beginning of the proof, the term $I_1$ is upper bounded as follows
\begin{align}
\label{Lem-K-dist-1}
    I_1\le \frac{2}{C_{\beta}}\log(n)\exp(-\frac{c_1c}{24c_2C}\frac{K}{L})=\frac{2}{C_{\beta}}\log(n)\exp(-\frac{c_1cC_{\beta}}{24c_2C}\frac{K}{\log(n)}).
\end{align}
Similarly we obtain that 
\begin{align}
\label{Lem-K-dist-2}
    I_2\le \frac{2}{C_{\beta}}\log(n)\exp(-\frac{C_{\beta}c_1c}{24c_2C}\frac{K}{\log(n)}).
\end{align}
In consequence, by Inequality (\ref{aux-lemma1-eq7}), (\ref{aux-lemma1-eq1}), (\ref{Lem-K-dist-1}) and (\ref{Lem-K-dist-2}) we obtain that 
\begin{align}
\label{Lem-K-dist-3}
    &\mathbb{P}\Big(\left\{\exists (i,j)\in [n]\times [m_i] :R_K\left(x_{i,j}\right) \leq a(K / (nm))^{1 / d}\right\}\bigcap \Omega\Big)\nonumber
    \\
    \le& \frac{4}{C_{\beta}}Cnm\log(n)\exp(-\frac{c_1cC_{\beta}}{24c_2C}\frac{K}{\log(n)}).
\end{align}
From Equation (\ref{Lem-K-dist-0}), (\ref{Lem-K-dist-1-1}) and (\ref{Lem-K-dist-3}) it follows that
\begin{equation}
    \mathbb{P}\Big(\left\{R_{K, \min } \leq a(K / (nm))^{1 / d}\right\}\Big)\le \frac{4}{C_{\beta}}C\log(n)nm\exp(-\frac{c_1cC_{\beta}}{24c_2C}\frac{K}{\log(n)})+\frac{1}{n}.
\end{equation}
Now we analyze $\mathbb{P}\Big(R_{K,\max}\ge \widetilde a \Big(\frac{K}{nm}\Big)^{1/d}\Big)$. To this end, we first observe that,
$$
\mathbb{P}\left\{x_{i,j} \in B_r(x)\right\}=\int_{B_r(x)} v(t) d t  \geq c_{1,d} r^d v_{\min }:=\mu_{\min }>0,
$$
for a positive constant $c_{1,d}$ that depends on $d.$
Following a similar argument as above, we arrive at 
\begin{align*}
    \mathbb{P}\Big(R_{K,\max}\ge \widetilde a \Big(\frac{K}{nm}\Big)^{1/d}\Big)\le \frac{4}{C_{\beta}}C\log(n)nm\exp(-\frac{C_{\beta}}{24}\frac{K}{\log(n)})+\frac{1}{n},
\end{align*}
concluding the proof.
\end{proof}
We provide the following lemma to characterize the minimum and maximum number of observations $\left\{x_{i,j}\right\}$ that fall within each cell $\mathcal{I}_r$, for $r=1,...,N^d$. Remembering that  $I_r=\{(i,j)\in [n]\times [m_i]: P_{I}(x_{i,j})=\mathfrak{c}_r \}$ for $r=1,...,N^d,$ where $P_{I}(x)$ denotes the point in $P_{\mathrm{lat}}(N)$ such that $x\in C(P_{I}(x))$.

 \begin{lemma}[{\bf{see Lemma 7 in \cite{madrid2020adaptive}}}]
     \label{Omega-Oscar-Paper}
Assume that $N$ in the construction of $P_{\text {lat }}(N)$ defined in Equation (\ref{Plat-def}) is chosen as
$$
N=\left\lceil\frac{3 \sqrt{d}  (nm)^{1 / d}}{ aK^{1 / d}}\right\rceil,
$$
where  $a=1 /\left(  4Cc_{2,d}c_2v_{\text{max} }\right)^{1 / d}$, for a positive constants $c_2$, and, $v_{\max}$, and $C$ are from Assumption \ref{assume:tv functions}. Moreover, $c_{2,d}$ is a positive constants that depends on $d.$
Then there exist positive constants $\widetilde{b}_1$ and $\widetilde{b}_2$ depending on $d, v_{\text{min}}$, and $ v_{\text{max} }$ with $ v_{\text{min}}$ from Assumption \ref{assume:tv functions}, such that for $C_{\beta}$ from Assumption 
 \ref{assume:tv functions},
$$ 
\begin{aligned}
& \mathbb{P}\left\{\max _{r=1,...,N^d }|I_r| \geq(1+\delta)  \widetilde{b}_1 K\right\} \leq \frac{4}{C_{\beta}}\log(n)N^d\exp\Big(-\frac{C_{\beta}}{12}\delta^2 \widetilde{b}_2\frac{K}{\log(n)}\Big)+\frac{1}{n}, \\
& \mathbb{P}\left\{\min_{r=1,...,N^d}|I_r| \leq(1-\delta)  \widetilde{b}_2 K\right\} \leq \frac{4}{C_{\beta}}\log(n)N^d\exp\Big(-\frac{C_{\beta}}{12}\delta^2 \widetilde{b}_2\frac{K}{\log(n)}\Big)+\frac{1}{n},
\end{aligned}
$$
for all $\delta \in(0,1)$. Moreover, the symmetric $K$-NN graph has maximum degree $d_{\max }$ satisfying
\begin{align*}
    &\mathbb{P}\left(d_{\max } \geq 3Cc_2c_{2,d} v_{\max }  \tilde{a}^d K\right)
    \\
    \leq& Cnm\left\{\frac{4}{C_{\beta}}\log(n)\exp(-\frac{C_{\beta}}{24}\frac{K}{\log(n)})+\frac{4}{C_{\beta}}C\log(n)nm\exp(-\frac{C_{\beta}}{24}\frac{K}{\log(n)})\right\}+\frac{1}{n},
\end{align*}
where $\tilde{a}=1 /\left(cc_1c_{1,d}v_{\text{min} }\right)^{1 / d}$, and $c_1$ is positive constants. Moreover, $c_{1,d}>0$ is a constant that depends on $d.$
Define the event $\Omega$ as: ``If $x_{i,j} \in C(x_{i,j}^{\prime})$ and $x_{s,t} \in C(x_{s,t}^{\prime})$ for $x_{i,j}^{\prime}, x_{s,t}^{\prime} \in P_{\text{lat}}(N)$ with $\left\|x_{i,j}^{\prime}-x_{s,t}^{\prime}\right\|_2 \leq$ $N^{-1}$, then $x_{i,j}$ and $x_{s,t}$ are connected in the $K$-NN graph". Then,
$$
\mathbb{P}(\Omega) \geq 1-\frac{4}{C_{\beta}}C\log(n)nm\exp(-\frac{\widetilde{C}}{24}\frac{K}{\log(n)})-\frac{1}{n},
$$
where $\widetilde{C}=\frac{C_{\beta}c_1c}{c_2C}$.
 \end{lemma}
\begin{proof}
Through this proof, we make use of the event $\widetilde{\Omega}=\bigcap_{i=1}^n\{x_{i,j}=x_{i,j}^*,\ \forall \ 1\le j\le m_i\}$, where $x_{i,j}^*$ are defined  in Appendix \ref{BetaMsection}. We use 
\begin{equation}
\label{aux-lemma2-Lchoice}
L=\frac{2}{C_{\beta}}\log(n),
\end{equation}
and for $l\in[L]$
\begin{equation}
\label{aux-lemma2-block}
\frac{c_1n}{L}\le |\mathcal{J}_{e,l}|, |\mathcal{J}_{o,l}|\le \frac{c_2n}{L},
\end{equation}
with $c_1$ and $c_2$ positive constants.
\\
First we analyze the event $\Omega$. To this end, let 
$x_{i,j} \in C(x_{i,j}^{\prime})$ and $x_{s,t} \in C(x_{s,t}^{\prime})$ for $x_{i,j}^{\prime}, x_{s,t}^{\prime} \in P_{\text{lat}}(N)$ with $\left\|x_{i,j}^{\prime}-x_{s,t}^{\prime}\right\|_2 \leq$ $N^{-1}$. Then the following holds,
\begin{align}
\label{aux-lemma2-eq1}
\left\| x_{i,j}- x_{s,t}\right\|_2 &\leq  \sqrt{d}\left\|x_{i,j}-x_{s,t}\right\|_{\infty} \nonumber\\
& \leq  \sqrt{d}\left\{\left\|x_{i,j}-x_{i,j}^{\prime}\right\|_{\infty}
+
\left\|x_{i,j}^{\prime}-x_{s,t}^{\prime}\right\|_{\infty}
+
\left\|x_{s,t}-x_{s,t}^{\prime}\right\|_{\infty}
\right\} \nonumber\\
& <3  \sqrt{d} N^{-1}\nonumber \\
& \leq a\left(\frac{K}{nm}\right)^{1 / d},
\end{align}
with $a=1 /\left(  2Cc_2v_{\text{max} }\right)^{1 / d}$ and $v_{\max}$ and $C$ are from Assumption \ref{assume:tv functions}. 
The first inequality is followed by a basic inequality between the $l_2$ and $l_\infty$ norms. The second is due to the triangle inequality, while the third inequality comes after using that $\left\|x_{i,j}^{\prime}-x_{s,t}^{\prime}\right\|_2 \leq$ $N^{-1}$ and $x_{i,j}^{\prime}, x_{s,t}^{\prime} \in P_{\text{lat}}(N)$. Using the definition of $N$ in the statement of the Lemma provides the last inequality.
From Inequality (\ref{aux-lemma2-eq1}) we observe that  $\{a\left(\frac{K}{nm}\right)^{1 / d} \leq R_{K, \min }\}\subset\Omega$. Therefore, using the result from the proof of Lemma \ref{dist-bet-Knn},
$$
\mathbb{P}(\Omega) \geq \mathbb{P}\left\{a\left(\frac{K}{nm}\right)^{1 / d} \leq R_{K, \min }\right\} \geq 1-\frac{4}{C_{\beta}}C\log(n)nm\exp(-\frac{\widetilde{C}}{24}\frac{K}{\log(n)})-\frac{1}{n}.
$$
Next we proceed to derive an upper bound on the counts $\{I_{l}\}_{
l=1}^{N^d}$.  Assume that $x \in P_{\text {lat }}(N)$, and $x^{\prime} \in C(x)$. Since $x^{\prime} \in C(x)$, by the definition of $P_{\text {lat }}(N)$ we have that $\left\|x-x^{\prime}\right\|_{\infty}\le\frac{1}{2N}$. Therefore, using the basic inequality between the $l_2$ and $l_\infty$ norms, 
\begin{align}
\label{aux-lemma2-eq2}
\| x- x^{\prime}\|_2  \le\sqrt{d}\left\|x-x^{\prime}\right\|_{\infty}  \leq \frac{\sqrt{d}}{2 N}  \leq \frac{a}{6}\left(\frac{K}{nm}\right)^{1 / d} =\frac{\left(\tilde{b}_1\right)^{1 / d}}{(c_{2,d}v_{\max }2Cc_2)^{1 / d}}\left(\frac{K}{nm}\right)^{1 / d},
\end{align}
where the third inequality is obtained from the choice of $N$ in the statement of the Lemma, and $\tilde{b}_1$ is an appropriate constant. Here $c_{2,d}>0$ is a constant that depends on $d$. Therefore, from Inequality (\ref{aux-lemma2-eq2})
\begin{equation}
    \label{Omega-Oscar-paper-0}
C(x) \subset B_{\frac{\left(\tilde{b}_1\right)^{1 / d}}{(2c_{2,d}v_{\max }Cc_2)^{1 / d}}\left(\frac{K}{nm}\right)^{1 / d}}(x) .
\end{equation}
On the other hand, let $\tilde{b}_2>0$ a constant such that
$
\frac{\left(\tilde{b}_2\right)^{1 / d}}{(2c_{1,d}cc_1v_{\min })^{1 / d}} \leq \frac{a}{2 } \frac{1}{3 \sqrt{d}}
$ with $c_{1,d}>0$ a constant that depends on $d.$ 
Observe that if the following holds
$$
\left\|x- x^{\prime}\right\|_2 \leq \frac{\left(\tilde{b}_2\right)^{1 / d}}{(2cc_1v_{\min })^{1 / d}}\left(\frac{K}{nm}\right)^{1 / d},
$$
we have that,
$$
\left\|x-x^{\prime}\right\|_{\infty}\leq \left\|x-x^{\prime}\right\|_{2} \leq\frac{\left(\tilde{b}_2\right)^{1 / d}}{(2cc_1v_{\min })^{1 / d}}\left(\frac{K}{nm}\right)^{1 / d}\leq \frac{a}{2 } \frac{1}{3 \sqrt{d}} \left(\frac{K}{nm}\right)^{1 / d}= \frac{1}{2 N} .
$$
The third inequality is followed by the choice of $\widetilde{b}_2$ and the Equality due to the choice of $N$ in the statement of the Lemma.
Therefore,
\begin{equation}
\label{Omega-Oscar-paper-1}
B_{\frac{\left(\tilde{b}_2\right)^{1 / d}}{(2cc_{1,d}c_1v_{\min })^{1 / d}}\left(\frac{K}{nm}\right)^{1 / d}}(x) \subset C(x).
\end{equation}
We are now ready to obtain the bound for $\mathbb{P}\left\{\max _{r=1,...,N^d }|I_r| \geq(1+\delta)  \widetilde{b}_1 K\right\}.$ Notice that
\begin{align}
\label{Omega-Oscar-paper-2}
 \mathbb{P}\left\{\max _{r=1,...,N^d }|I_r| \geq(1+\delta)  \widetilde{b}_1 K\right\}&\le
\mathbb{P}\Big(\left\{\max _{r=1,...,N^d }|I_r| \geq(1+\delta)  \widetilde{b}_1 K\right\} \cap \widetilde{\Omega}\Big)+\frac{1}{n},
\end{align}
with $\widetilde{\Omega}$ defined at the beginning of the proof. Then, by a union bound argument
\begin{align}
\label{Omega-Oscar-paper-31}
  &\mathbb{P}\Big(\left\{\max _{l=1,...,N^d }|I_l| \geq(1+\delta)  \widetilde{b}_1 K\right\} \cap \widetilde{\Omega}\Big)  \nonumber
  \\
& \leq \sum_{r=1}^{N^d} \mathbb{P}\Big(\left\{|I_r| \geq(1+\delta)  \tilde{b}_1 K\right\}\cap \widetilde{\Omega}\Big) \nonumber
\\
=& \sum_{x \in P_{\text{lat}}(N)} \mathbb{P}\left[\Big\{|C(x)| \geq(1+\delta)  \tilde{b}_1 K\Big\}\cap \widetilde{\Omega}\right],
\end{align}
where $|C(x)|=\{(i,j)\in [n]\times[m_i]:P_I(x_{i,j})=x\}$.
 Moreover, for any $x \in P_{\text{lat}}(N)$ using Equation (\ref{Omega-Oscar-paper-0}) we have that
\begin{align*}
    \mathbb{P}\left(x_{1,1} \in C(x)\right)\le  \mathbb{P}\left(x_{1,1} \in B_{\frac{\left(\tilde{b}_1\right)^{1 / d}}{(2c_{2,d}v_{\max }Cc_2)^{1 / d}}\left(\frac{K}{nm}\right)^{1 / d}}(x)\right)\le \frac{\left(\tilde{b}_1\right)}{(2c_{2,d}v_{\max }Cc_2)}\left(\frac{K}{nm}\right)c_{2,d}v_{\max},
\end{align*}
where the last inequality follows the same line of arguments as in Inequality (\ref{aux-lemma1-eq0}) in the proof of Lemma \ref{dist-bet-Knn}. Therefore,
\begin{align*}
    2Cc_2nm\mathbb{P}\left(x_{1,1} \in C(x)\right)  \le 2Cc_2nm\frac{\left(\tilde{b}_1\right)}{(2c_{2,d}v_{\max }Cc_2)}\left(\frac{K}{nm}\right)c_{2,d}v_{\max}=\widetilde{b}_1K.
\end{align*}
The above Inequality and Inequality (\ref{Omega-Oscar-paper-31}) lead to
\begin{align}
\label{Omega-Oscar-paper-3}
&  \sum_{x \in P_{\text{lat}}(N)} \mathbb{P}\left[\Big\{|C(x)| \geq(1+\delta)  \tilde{b}_1 K\Big\}\cap \widetilde{\Omega}\right] \nonumber\\
& \leq \sum_{x \in P_{\text{lat}}(N)} \mathbb{P}\left[\Big\{|C(x)| \geq(1+\delta) 2Cc_2nm \mathbb{P}\left(x_{1,1} \in C(x)\right)\Big\}\cap \widetilde{\Omega}\right],
\end{align}
Then, using the definition of the event $\widetilde{\Omega}$ at the beginning of the proof and a union bound argument, 
\begin{align}
\label{Omega-Oscar-paper-5}
   &\mathbb{P}\left[\Big\{|C(x)| \geq(1+\delta) 2Cc_2nm \mathbb{P}\left(x_{1,1} \in C(x)\right)\Big\}\cap \widetilde{\Omega}\right]\nonumber
   \\
   \le& \sum_{l=1}^L \mathbb{P}\Big(\sum_{i\in\mathcal{J}_{e,l}}\sum_{j=1}^{m_i}\mathbf{1}_{x_{i,j}^*\in C(x)}\ge(1+\delta) 2Cc_2nm \mathbb{P}\left(x_{1,1}^* \in C(x)\right)/(2L) \Big)\nonumber
   \\
   +&\sum_{l=1}^L \mathbb{P}\Big(\sum_{i\in\mathcal{J}_{o,l}}\sum_{j=1}^{m_i}\mathbf{1}_{x_{i,j}^*\in C(x)}\ge (1+\delta)2Cc_2nm \mathbb{P}\left(x_{1,1}^* \in C(x)\right)/(2L) \Big).
\end{align}
Now we analyze each of the terms $\mathbb{P}\Big(\sum_{i\in\mathcal{J}_{e,l}}\sum_{j=1}^{m_i}\mathbf{1}_{x_{i,j}^*\in C(x)}\ge Cc_2nm \mathbb{P}\left(x_{1,1}^* \in C(x)\right)/L \Big).$ By Inequality (\ref{aux-lemma2-block}) and Assumption \ref{assume:tv functions}{\bf{f}}
\begin{align*}
    &\mathbb{P}\Big(\sum_{i\in\mathcal{J}_{e,l}}\sum_{j=1}^{m_i}\mathbf{1}_{x_{i,j}^*\in C(x)}\ge(1+\delta) Cc_2nm \mathbb{P}\left(x_{1,1}^* \in C(x)\right)/L \Big)
    \\
    \le&
    \mathbb{P}\Big(\sum_{i\in\mathcal{J}_{e,l}}\sum_{j=1}^{m_i}\mathbf{1}_{x_{i,j}^*\in C(x)}\ge (1+\delta) \Big(\sum_{i\in \mathcal{J}_{e,l}}m_i\Big)\mathbb{P}\left(x_{1,1}^* \in C(x)\right) \Big).
\end{align*}
By Lemma \ref{BinCon} it follows that
\begin{align}
    \label{aux-lemma2-eq4}&\mathbb{P}\Big(\sum_{i\in\mathcal{J}_{e,l}}\sum_{j=1}^{m_i}\mathbf{1}_{x_{i,j}^*\in C(x)}\ge (1+\delta) \Big(\sum_{i\in \mathcal{J}_{e,l}}m_i\Big)\mathbb{P}\left(x_{1,1}^* \in C(x)\right) \Big)
    \\
    \le& \exp\Big(-\frac{1}{3}\delta^2\Big(\sum_{i\in \mathcal{J}_{e,l}}m_i\Big)\mathbb{P}\left(x_{1,1}^* \in C(x)\right)\Big).
\end{align}
Now Equation (\ref{Omega-Oscar-paper-1}) implies that 
$$\mathbb{P}\left(x_{1,1}^* \in C(x)\right)\ge\mathbb{P}\left(x_{1,1}^*\in B_{\frac{\left(\tilde{b}_2\right)^{1 / d}}{(2cc_{1,d}c_1v_{\min })^{1 / d}}\left(\frac{K}{nm}\right)^{1 / d}}(x)\right)\ge c_{1,d}v_{\min}\frac{\left(\tilde{b}_2\right)}{(2cc_{1,d}c_1v_{\min })}\left(\frac{K}{nm}\right),$$
from where, using Inequality (\ref{aux-lemma2-block}) and Assumption \ref{assume:tv functions}{\bf{f}} we get that
$$\Big(\sum_{i\in \mathcal{J}_{e,l}}m_i\Big)\mathbb{P}\left(x_{1,1}^* \in C(x)\right)\ge cc_1\frac{n}{L}mc_{1,d}v_{\min}\frac{\left(\tilde{b}_2\right)}{(2cc_{1,d}c_1v_{\min })}\left(\frac{K}{nm}\right).$$
Thus, by our choice of $L$ in Equation (\ref{aux-lemma2-Lchoice}) and Inequality (\ref{aux-lemma2-eq4})
\begin{align}
    \label{Omega-Oscar-paper-4}
&\mathbb{P}\Big(\sum_{i\in\mathcal{J}_{e,l}}\sum_{j=1}^{m_i}\mathbf{1}_{x_{i,j}^*\in C(x)}\ge (1+\delta) \Big(\sum_{i\in \mathcal{J}_{e,l}}m_i\Big)\mathbb{P}\left(x_{1,1}^* \in C(x)\right) \Big)
    \\
    &\le\exp\Big(-\frac{1}{3}\delta^2c_1c\frac{n}{L}m\frac{\widetilde{b}_2}{2c_{1,d}v_{\min} cc_1}\frac{K}{nm}c_{1,d}v_{\min}\Big)=\exp\Big(-\frac{C_{\beta}}{12}\delta^2 \widetilde{b}_2\frac{K}{\log(n)}\Big).
\end{align}
A similar analysis is done for $$\mathbb{P}\Big(\sum_{i\in\mathcal{J}_{o,l}}\sum_{j=1}^{m_i}\mathbf{1}_{x_{i,j}^*\in C(x)}\ge (1+\delta)2Cc_2nm \mathbb{P}\left(x_{1,1}^* \in C(x)\right)/(2L) \Big).$$ Thus, by Equation (\ref{Omega-Oscar-paper-2}), (\ref{Omega-Oscar-paper-31}), (\ref{Omega-Oscar-paper-3}), (\ref{Omega-Oscar-paper-5}) and (\ref{Omega-Oscar-paper-4}), we obtain that 
\begin{align*}
   & \mathbb{P}\left\{\max _{r=1,...,N^d }|I_r| \geq(1+\delta)  \widetilde{b}_1 K\right\}
    \\
    \le& \frac{4}{C_{\beta}}\log(n)N^d\exp\Big(-\frac{C_{\beta}}{12}\delta^2 \widetilde{b}_2\frac{K}{\log(n)}\Big)+\frac{1}{n}.
\end{align*}
On the other hand, with a similar argument, we have that
\begin{align*}
\mathbb{P}\left\{\min _{r=1,...,N^d}|I_r| \leq(1-\delta)  \widetilde{b}_2 K\right\} & \leq \sum_{x \in P_{\text{lat}}(N)} \mathbb{P}\Big(\left\{|C(x)| \leq(1-\delta)  \tilde{b}_2 K\right\}\cap \widetilde{\Omega}\Big)+\frac{1}{n} \\
& \leq \sum_{x \in P_{\text{lat}}(N)} \mathbb{P}\left[\left\{|C(x)| \leq(1-\delta) 2cc_1nm \mathbb{P}\left(x_{1,1} \in C(x)\right)\right\}\cap \widetilde{\Omega}\right] +\frac{1}{N}\\
& \le \frac{4}{C_{\beta}}\log(n)N^d\exp\Big(-\frac{C_{\beta}}{12}\delta^2 \widetilde{b}_2\frac{K}{\log(n)}\Big)+\frac{1}{n}.
\end{align*}
Next we proceed to find an upper bound on the maximum degree of the $K$-NN graph. To this end, we proceed as in the proof of Lemma \ref{dist-bet-Knn}. First, define the sets
$$
B_{i,j}(x)=\left\{(i',j') \in[n]\times[m_{i'}], (i',j')\neq (i,j): x_{i',j'}^* \in B_{\tilde{a}(K / nm)^{1 / d}}(x)\right\}
$$
and
$$B_{i,j}=\left\{(i',j') \in[n]\times[m_{i'}],(i',j')\neq (i,j): x_{i',j'}^* \in B_{\tilde{a}(K / nm)^{1 / d}}\left(x_{i,j}\right)\right\},$$
for $(i,j) \in[n]\times[m_i]$, and $x \in [0,1]^d$. Observe that
\begin{align}
\label{maxdegree-1}
    &\mathbb{P}\Big(\vert B_{i,j}(x)\vert>3 c_2Cc_{2,d}v_{\max }  \tilde{a}^d K\Big)\nonumber
    \\
    \le&\mathbb{P}\Big(\sum_{(i',j')\in[n]\times [m_{i'}],(i',j')\neq(i,j)}\mathbf{1}_{\{x_{i',j'}^*\in B_{\tilde{a}(K / nm)^{1 / d}}(x)\}}\ge \frac{3}{2} c_2Cc_{2,d}v_{\max }  \tilde{a}^d K\Big).
\end{align}
To bound the term
$$\mathbb{P}\Big(\sum_{(i',j')\in[n]\times [m_{i'}],(i',j')\neq(i,j)}\mathbf{1}_{\{x_{i',j'}^*\in B_{\tilde{a}(K / nm)^{1 / d}}(x)\}}\ge 3 c_2Cc_{2,d}v_{\max }  \tilde{a}^d K\Big),$$
we use a union bound arguments to get
\begin{align*}
    &\mathbb{P}\Big(\sum_{(i',j')\in[n]\times [m_{i'}],(i',j')\neq(i,j)}\mathbf{1}_{\{x_{i',j'}^*\in B_{\tilde{a}(K / nm)^{1 / d}}(x)\}}\ge 3 c_2Cc_{2,d}v_{\max }  \tilde{a}^d K\Big)
    \\
    \le&
    \sum_{l=1}^L\mathbb{P}\Big(\sum_{(i',j')\in \mathcal{J}_{e,l}\times [m_{i'}],(i',j')\neq(i,j)}\mathbf{1}_{\{x_{i',j'}^*\in B_{\tilde{a}(K / nm)^{1 / d}}(x)\}}\ge 3 c_2Cc_{2,d}v_{\max }  \tilde{a}^d K/(2L)\Big)
    \\
    +&\sum_{l=1}^L\mathbb{P}\Big(\sum_{(i',j')\in \mathcal{J}_{o,l}\times [m_{i'}],(i',j')\neq(i,j)}\mathbf{1}_{\{x_{1,1}^*\in B_{\tilde{a}(K / nm)^{1 / d}}(x)\}}\ge 3 c_2Cc_{2,d}v_{\max }  \tilde{a}^d K/(2L)\Big).
\end{align*}
Further, using Lemma \ref{Ind-cop} for any $l\in[L]$  we have that {\small{$$\sum_{(i',j')\in \mathcal{J}_{e,l}\times [m_{i'}],(i',j')\neq(i,j)}\mathbf{1}_{\{x_{i',j'}^*\in B_{\tilde{a}(K / nm)^{1 / d}}(x)\}}\sim\operatorname{Binomial}\left(\sum_{i' \in \mathcal{J}_{e,l}}m_{i'}-1, \mathbb{P}\Big(\{x_{i',j'}^*\in B_{\tilde{a}(K / nm)^{1 / d}}(x)\}\Big)\right).$$}}
By Equation (\ref{aux-lemma2-Lchoice}), Inequality (\ref{aux-lemma2-block}), Assumption \ref{assume:tv functions}{\bf{f}} and the Inequality (\ref{aux-lemma1-eq0}) in the proof of Lemma \ref{dist-bet-Knn},
\begin{align*}
    \frac{3}{2}\mathbb{P}\Big(\{x_{1,1}^*\in B_{\tilde{a}(K / nm)^{1 / d}}(x)\}\Big)\Big(\sum_{i' \in \mathcal{J}_{e,l}}m_{i'}-1 \Big)\le \frac{3}{2}c_{2,d}v_{\max}\tilde{a}^d\frac{K}{nm}\frac{Cc_2nm}{L}.
\end{align*}
Thus, for any $1\le l\le L,$
\begin{align}
\label{aux-lemma2-eq10}
    &\mathbb{P}\Big(\sum_{(i',j')\in \mathcal{J}_{e,l}\times [m_{i'}],(i',j')\neq(i,j)}\mathbf{1}_{\{x_{i',j'}^*\in B_{\tilde{a}(K / nm)^{1 / d}}(x)\}}\ge 3 c_2Cc_{2,d}v_{\max }  \tilde{a}^d K/(2L)\Big)
    \\
    \le& \mathbb{P}\Big(\sum_{(i',j')\in \mathcal{J}_{e,l}\times [m_{i'}],(i',j')\neq(i,j)}\mathbf{1}_{\{x_{i',j'}^*\in B_{\tilde{a}(K / nm)^{1 / d}}(x)\}}\ge \frac{3}{2}\mathbb{P}\Big(\{x_{1,1}^*\in B_{\tilde{a}(K / nm)^{1 / d}}(x)\}\Big)\Big(\sum_{i' \in \mathcal{J}_{e,l}}m_{i'}-1 \Big)\Big)\nonumber.
\end{align}
Now, from Equation (\ref{aux-lemma2-Lchoice}), Inequality (\ref{aux-lemma2-block}), and Assumption \ref{assume:tv functions}{\bf{f}} we have that
$$
\frac{cc_1}{2} c_{1,d}v_{\min}\tilde{a}^d\frac{K}{L}\le c_{1,d}v_{\min}\tilde{a}^d\frac{K}{nm}\frac{cc_1nm-L}{L}\le\mathbb{P}\Big(\{x_{1,1}^*\in B_{\tilde{a}(K / nm)^{1 / d}}(x)\}\Big)\Big(\sum_{i' \in \mathcal{J}_{e,l}}m_{i'}-1 \Big).$$ 
Thus, using Lemma \ref{BinCon} and the definition of $\widetilde{a}$ in the statement of the Lemma,
\begin{align*}
   & \mathbb{P}\Big(\sum_{(i',j')\in \mathcal{J}_{e,l}\times [m_{i'}],(i',j')\neq(i,j)}\mathbf{1}_{\{x_{i',j'}^*\in B_{\tilde{a}(K / nm)^{1 / d}}(x)\}}\ge \frac{3}{2}\mathbb{P}\Big(\{x_{i',j'}^*\in B_{\tilde{a}(K / nm)^{1 / d}}(x)\}\Big)\Big(\sum_{i' \in \mathcal{J}_{e,l}}m_{i'}-1 \Big)\Big)
    \\
    \le&
    \exp(-\frac{1}{12}\mathbb{P}\Big(\{x_{1,1}^*\in B_{\tilde{a}(K / nm)^{1 / d}}(x)\}\Big)\Big(\sum_{t \in \mathcal{J}_{e,l}}m_t-1\Big))
    \\
    \le&\exp(-\frac{cc_1}{12} c_{1,d}v_{\min}\tilde{a}^d\frac{K}{L})=\exp(-\frac{C_{\beta}}{24}\frac{K}{\log(n)}).
\end{align*}
Analogously, the exact same bound can be achieved when the case of $\mathcal{J}_{o,l}.$
Hence, using the above inequality together with Equation (\ref{aux-lemma2-Lchoice}), Inequality (\ref{maxdegree-1}) and Inequality (\ref{aux-lemma2-eq10}) 
\begin{align*}
    \mathbb{P}\Big(\vert B_{i,j}(x)\vert>3 c_2Cc_{2,d}v_{\max }  \tilde{a}^d K\Big)
    \le&\frac{4}{C_{\beta}}\log(n)\exp(-\frac{C_{\beta}}{24}\frac{K}{\log(n)}).
\end{align*}
Therefore,
\begin{align}
\label{maxdegree-2}
    \mathbb{P}\left\{\left|B_{i,j}\right| \geq 3 c_2c_{2,d}Cv_{\max }  \tilde{a}^d K\right\}=&\int_{[0,1]^d} \mathbb{P}\left\{\left|B_{i,j}(x)\right| \geq 3 c_2Cv_{\max }  \tilde{a}^d K\right\} p(x) \mu(d x) \nonumber
    \\
    \leq& \frac{4}{C_{\beta}}\log(n)\exp(-\frac{C_{\beta}}{24}\frac{K}{\log(n)}).
\end{align}
Finally let $d_{i,j}$ be the degree associated with $x_{i,j}$. Using a union bound argument
$$
\begin{aligned}
\mathbb{P}\left(d_{\max} \leq 3 c_2Cc_{2,d}v_{\max }  \tilde{a}^d K\right)
\le&\mathbb{P}\left(\{d_{\max} \leq 3 c_2Cc_{2,d}v_{\max }  \tilde{a}^d K\}\cap\widetilde\Omega\right)+\frac{1}{n}
\\
\le&Cnm\mathbb{P}\left(\{d_{i,j} \leq 3 c_2Cc_{2,d}v_{\max }  \tilde{a}^d K\}\cap\widetilde\Omega\right)+\frac{1}{n}.
\end{aligned}
$$
Moreover,
$$
\begin{aligned}
&\mathbb{P}\left(\{d_{i,j} \leq 3 c_{2,d}c_2Cv_{\max }  \tilde{a}^d K\}\cap\widetilde\Omega\right) 
\\
\geq & \mathbb{P}\left[\Big\{\left|\left\{(i',j') \in[n]\times[m_{i'}], (i',j')\neq (i,j): \vert\vert x_{i,j}- x_{i',j'}\vert\vert_2 \leq R_{K, \max }\right\}\right| \leq 3 v_{\max } Cc_2c_{2, d} \tilde{a}^d K\Big\}\cap \widetilde\Omega\right] \\
\geq & \mathbb{P}\left[\Big\{\left|\left\{(i',j') \in[n]\times[m_{i'}], (i',j')\neq (i,j): \vert\vert x_{i,j}- x_{i',j'}\vert\vert_2 \leq R_{K, \max }\right\}\right| \leq  3 v_{\max } Cc_2c_{2, d} \tilde{a}^d K\Big\}\right. \\
& \left.\cap \Big\{R_{K, \max } \leq \tilde{a}\left(\frac{K}{nm}\right)^{1 / d}\Big\}\cap\widetilde\Omega\right] \\
\geq & \mathbb{P}\left[\Big\{\left|\left\{(i',j') \in[n]\times[m_{i'}], (i',j')\neq (i,j): \vert\vert x_{i,j}- x_{i',j'}\vert\vert_2 \leq \tilde{a}\left(\frac{K}{nm}\right)^{1 / d}\right\}\right| \leq  3 v_{\max } Cc_2c_{2, d} \tilde{a}^d K\Big\}\right. \\
& \left.\cap\Big\{R_{K, \max } \leq \tilde{a}\left(\frac{K}{nm}\right)^{1 / d}\Big\}\cap\widetilde\Omega\right] \\
\geq & 1-\mathbb{P}\Big(\left\{R_{K, \max }>\tilde{a}\left(\frac{K}{nm}\right)^{1 / d}\right\}\cap \widetilde\Omega\Big)-\mathbb{P}\left\{\left|B_{i,j}\right|>3 v_{\max } Cc_2c_{2, d} \tilde{a}^d K\right\},
\end{aligned}
$$
and the claim follows from Inequality (\ref{maxdegree-2}) and the proof of Lemma \ref{dist-bet-Knn}.
\end{proof}

The Lemma \ref{Omega-Oscar-Paper} demonstrated that the count of observations $x_{i,j}$ that fall within each cell $\mathcal{I}_l$ is proportional to $K$, with high probability. We leverage this understanding to derive an upper limit for the MSE in the following result.

\begin{lemma}[{\bf{see Lemma 8 in \cite{madrid2020adaptive}}}]
\label{lemma8-Oscar}
Consider the notation of Lemma \ref{Omega-Oscar-Paper}.  Assume that the event $\Omega$ from Lemma \ref{Omega-Oscar-Paper} intersected with
\begin{equation}
\label{lemma8-oscar-1}
    \min _{r=1,...,N^d}|\mathcal{I}_r| \geq \frac{1}{2}  \widetilde{b}_2 K,
\end{equation}
holds. Then for all $e \in \mathbb{R}^{\sum_{i=1}m_i}$, it holds that
\begin{equation}
    \label{lemma8-oscar-2}
    \left|e^T\left(\theta-\theta_I\right)\right| \leq 2\|e\|_{\infty}\left\|\nabla_{G_K} \theta\right\|_1, \quad \forall \theta \in \mathbb{R}^{\sum_{i=1}m_i}.
\end{equation}
Moreover,
\begin{equation}
     \label{lemma8-oscar-3}
     \left\|D \theta^I\right\|_1 \leq\left\|\nabla_{G_K} \theta\right\|_1, \quad \forall \theta \in \mathbb{R}^{\sum_{i=1}m_i},
\end{equation}
where $D$ is the incidence matrix of a d-dimensional grid graph $G_{\text {grid }}=\left(V_{\text {grid }}, E_{\text {grid }}\right)$ with $V_{\text {grid }}=\left[N^d\right]$, where $\left(r,r^{\prime}\right) \in E_{\text {grid }}$ if and only if
$$
\left\|P_I(x_{i_r,j_r})-P_I(x_{i_{r^{\prime}},j_{r^{\prime}}})\right\|_2=\left\|\mathfrak{c}_{r}-\mathfrak{c}_{r^{\prime}}\right\|_2=\frac{1}{N},
$$
where $x_{i_r,j_r}\in C(\mathfrak{c}_{r})$ and $x_{i_{r^{\prime}},j_{r^{\prime}}}\in C(\mathfrak{c}_{r^{\prime}})$.
\end{lemma}
\begin{proof}
    We start by introducing the notation $x_{i,j}^{\prime}=P_I\left(x_{i,j}\right)$. To prove (\ref{lemma8-oscar-2}) we proceed in cases.
    \\
{\bf{Case 1}}. If $\left(\theta_I\right)_{1,1}=\theta_{1,1}$. Then clearly $\left|e_{1,1}\right|\left|\theta_{1,1}-\left(\theta_I\right)_{1,1}\right|=0$.
\\
{\bf{Case 2}}. If $\left(\theta_I\right)_{1,1}=\theta_{i,j}$ for $(i,j) \neq (1,1)$. Then
$$
\left\|x_{1,1}^{\prime}-x_{i,j}\right\|_{\infty} \leq\left\|x_{1,1}^{\prime}-x_{1,1}\right\|_{\infty} \leq \frac{1}{2 N} .
$$
Thus, $x_{1,1}^{\prime}=x_{i,j}^{\prime}$, and so $((1, 1),(i,j)) \in E_K$ by the assumption that the event $\Omega$ holds.
Therefore for every $(i,j) \in[n]\times [m_i]$, there exists $(s_{(i,j)},t_{(i,j)}) \in[n]\times[m_i]$ such that $\left(\theta_I\right)_{i,j}=\theta_{{s_{(i,j)},t_{(i,j)}}}$ and either $(i,j)=(s_{(i,j)},t_{(i,j)})$ or $\left((i,j), (s_{(i,j)},t_{(i,j)})\right) \in E_K$. Hence,
\begin{align*}
\left|e^T\left(\theta-\theta_I\right)\right| & \leq \sum_{i=1}^n\sum_{j=1}^{m_i}\left|e_{i,j}\right|\left|\theta_{i,j}-\theta_{{s_{(i,j)},t_{(i,j)}}}\right| \\
& \leq 2\|e\|_{\infty}\left\|\nabla_{G_K} \theta\right\|_1 .
\end{align*}
To verify (\ref{lemma8-oscar-3}), we observe that
$$
\left\|D \theta^I\right\|_1=\sum_{\left(r, r^{\prime}\right) \in E_{\text {grid }}}\left|\theta_{i_r,j_r}-\theta_{i_{r^{\prime}},j_{r^{\prime}}}\right| .
$$
Now, if $\left(r, r^{\prime}\right) \in E_{\text {grid }}$, then $x_{i_r,j_r}$ and $x_{i_{r^{\prime}},j_{r^{\prime}}}$ are in neighboring cells in $I$. This implies that $\left((i_r,j_r), (i_{r^{\prime}},j_{r^{\prime}})\right)$ is an edge in the $K$-NN graph. Thus, every edge in the grid graph $G_{\text {grid }}$ corresponds to an edge in the $K$-NN graph and the mapping is injective. Note that here we have used the fact that (\ref{lemma8-oscar-1}) ensures that every cell has at least one point, provided that $K$ is large enough. The claim is then followed.
\end{proof}

\begin{lemma}[{\bf{see Lemma 9 in \cite{madrid2020adaptive}}}]
\label{lemma9-oscar}
 Consider the notation of Lemma \ref{Omega-Oscar-Paper} and assume that that event $\Omega$ happens. We have that for $e,\theta_1,\theta_2 \in \mathbb{R}^{\sum_{i=1}m_i}$, it holds that
$$
e^T\left(\theta_1-\theta_2\right) \leq 2\|e\|_{\infty}\left(\left\|\nabla_{G_K} \theta_2\right\|_1+\left\|\nabla_{G_K} \theta_2\right\|_1\right)+\varepsilon^T\left({\theta_1}_I-{\theta_2}_I\right).
$$
\end{lemma}
\begin{proof}
We observe that
$$
e^T\left(\theta_1-\theta_2\right)=e^T\left({\theta_1}-{\theta_1}_I\right)+e^T\left({\theta_1}_I-{\theta_2}_I\right)+e^T\left({\theta_2}_I-\theta_2\right),
$$
and the claim is followed by Lemma \ref{lemma8-Oscar}.
\end{proof}
\begin{lemma}[{\bf{see Lemma 10 in \cite{madrid2020adaptive}}}]
\label{lemma11-oscar-1}
Assume that the event $\Omega$ from Lemma \ref{Omega-Oscar-Paper} intersected with
\begin{equation*}
    \min _{r=1,...,N^d}|I_r| \geq \frac{1}{2}  \widetilde{b}_2 K,
\end{equation*}
holds. Let $\theta_1,\theta_2\in \mathbb{R}^{\sum_{i=1}^{n}m_i}$ and $\{\epsilon_{i,j}\}_{i=1,j=1}^{n,m_i}=\epsilon\in \mathbb{R}^{\sum_{i=1}^n m_i}$ sub-Gaussian with parameter $\frac{C^2}{c^2}\varpi_\epsilon^2$. Then, we have that
$$
\epsilon^T\left({\theta_1}_I-{\theta_2}_I\right) \leq \max _{r=1,...,N^d} |I_r|\left(\|\Pi \tilde{\varepsilon}\|_2\left\|\theta_{1}-\theta_{2}\right\|_2+\left\|\left(D^{+}\right)^T \tilde{\epsilon}\right\|_{\infty}\left[\left\|\nabla_{G_K} \theta_{1}\right\|_1+\left\|\nabla_{G_K} \theta_{2_I}\right\|_1\right]\right),
$$
where $\tilde{\epsilon} \in \mathbb{R}^{N^d}$ is a mean zero vector whose coordinates are sub-Gaussian with the same constants $\frac{C^2}{c^2}\varpi_\epsilon^2$. Here, $\Pi$ is the orthogonal projection onto the span of $\mathbf{1}_{N^d} \in \mathbb{R}^{N^d}$, and $D^{+}$is the pseudoinverse of the incidence matrix $D$ from Lemma \ref{lemma8-Oscar}.
\end{lemma}
\begin{proof}
  Using the notation from the proof of Lemma \ref{lemma8-Oscar}, we have that
$$
\epsilon^T\left({\theta_1}_I-{\theta_2}_I\right)=\sum_{r=1}^{N^d} \sum_{(i^{\prime},j^{\prime}) \in I_r} \epsilon_{i^{\prime},j^{\prime}}\left({\theta_1}_{i_r,j_r}-{\theta_2}_{i_r,j_r}\right)=\max _{r=1,...,N^d}|I_r|\tilde{\epsilon}^T\left(\theta_1^I-\theta_2 ^I\right),
$$
where
$$
\tilde{\epsilon}_j=\left\{\max _{r=1,...,N^d}|I_r|\right\}^{-1} \sum_{(i^{\prime},j^{\prime}) \in I_r} \epsilon_{i^{\prime},j^{\prime}}.
$$
Clearly the $\tilde{\epsilon}_1, \ldots, \tilde{\epsilon}_{N^d}$ are also sub-Gaussian with the same constants as the original errors $\epsilon_{1,1}, \ldots, \epsilon_{n,m_n}$.
Moreover let $\Pi$ be the orthogonal projection onto the span of $\mathbf{1}_{N^d} \in \mathbb{R}^{N^d}$. By Hölder's inequality and by the triangle inequality,
\begin{align}
\label{lemma11-oscar-3}
\epsilon^T\left({\theta_1}_I-{\theta_2}_I\right) & \leq\left\{\max _{r=1,...N^d}|I_r|\right\}\left\{\|\Pi \tilde{\epsilon}\|_2\left\|\theta_1^I-\theta_2^{ I}\right\|_2+\left\|\left(D^{+}\right)^T \tilde{\epsilon}\right\|_{\infty}\left\|D\left(\theta_1^I-\theta_2^ I\right)\right\|_1\right\} \nonumber\\
& \leq\left\{\max _{r=1,...,N^d}|I_r|\right\}\left\{\|\Pi \tilde{\epsilon}\|_2\left\|\hat{\theta}^I-\theta^{*, I}\right\|_2+\left\|\left(D^{+}\right)^T \tilde{\epsilon}\right\|_{\infty}\left(\left\|D \theta_1^I\right\|_1+\left\|D \theta_2^ I\right\|_1\right)\right\} .
\end{align}
Next we observe that
\begin{align}
\label{lemma11-oscar-2}
\left\|\theta_1^I-\theta_2^ I\right\|_2=\left\{\sum_{r=1}^{N^d}\left({\theta_1}_{i_r,j_r}-{\theta_2}_{i_r,j_r}\right)^2\right\}^{1 / 2} \leq\left\{\sum_{i=1}^n\sum_{j=1}^{m_i}\left({\theta_1}_{i,j}-{\theta_2}_{i,j}\right)^2\right\}^{1 / 2} .
\end{align}
Therefore, combining (\ref{lemma11-oscar-2}), (\ref{lemma11-oscar-3}) and Lemma \ref{lemma8-Oscar} we arrive at
$$
\epsilon^T\left({\theta_1}_I-{\theta_2}_I\right) \leq \max _{r=1,...,N^d }|I_r|\left\{\|\Pi \tilde{\epsilon}\|_2\left\|\theta_1-\theta_2\right\|_2+\left\|\left(D^{+}\right)^T \tilde{\epsilon}\right\|_{\infty}\left(\left\|\nabla_{G_K} \theta_1\right\|_1+\left\|\nabla_{G_K} \theta_2\right\|_1\right)\right\} .
$$
\end{proof}
\begin{lemma}[{\bf{see Lemma 11 in \cite{madrid2020adaptive}}}]
    \label{lemma11-oscar}
Let $\{\epsilon_{i,j}\}_{i=1,j=1}^{n,m_i}=\epsilon\in \mathbb{R}^{\sum_{i=1}^n m_i}$ sub-Gaussian random variables with parameter $\frac{C^2}{c^2}\varpi_\epsilon^2$ and use the notation $\Delta=\theta_1-\theta_2$, for $\theta_1,\theta_2\in\mathbb{R}^{\sum_{i=1}^n m_i}$. For some positive constants $C_3$ and $C_4$  we have that    \begin{align*}
& \sup _{\|\Delta\|_2 \leq \eta,\left\|\nabla_{G_K} \Delta\right\|_1 \leq 2 V_n} \sum_{i=1}^n \sum_{j=1}^{m_i} \Delta_{i,j} \epsilon_{i,j} \leq C_3\varpi_\epsilon K\log^{1/4}(nm)\eta+ C_4\varpi_\epsilon K \log (n m)  V_n, 
\end{align*}
with a probability of at least 
{\small{$$1-2 \frac{1}{\sqrt{\log(nm)}}-\frac{4}{C_{\beta}}C\log(n)nm\exp(-\frac{\widetilde{C}}{24}\frac{K}{\log(n)})-3\frac{1}{n}-C_2\log(n)\frac{nm}{K}\exp\Big(-\frac{C_{\beta}}{24} \widetilde{b}_2\frac{K}{\log(n)}\Big)-\frac{1}{nm},$$}}
where $C_2>0$ is a constant depending on $d$, $C_{\beta}$ and $v_{\max}$. Moreover, $\widetilde{b}_2>0$ is a constant depending on $d, v_{\text{min} }$, and $ v_{\text{max} }$. Furthermore,
$\widetilde{C}=\frac{C_{\beta}c_1c}{c_2C}$ with $c_1 $ and $c_2$ positive constants. Here  $v_{\max}$, $v_{\min}$, $C$, $C_{\beta}$ and $c$ are from Assumption \ref{assume:tv functions}.
\end{lemma}
\begin{proof} 
 Consider the event $\Omega$ from Lemma \ref{Omega-Oscar-Paper}, and
\begin{equation}
    E_1=\{\min _{r=1,...,N^d}|I_r| \geq \frac{1}{2}  \widetilde{b}_2 K\},
\end{equation}
and
\begin{equation}
    E_2=\{\max _{r=1,...,N^d}|I_r| \leq \frac{3}{2}  \widetilde{b}_1 K\},
\end{equation}
hold. Observe that by Lemma \ref{Omega-Oscar-Paper} these events hold with probability at least,
$$1-\frac{4}{C_{\beta}}C\log(n)nm\exp(-\frac{\widetilde{C}}{24}\frac{K}{\log(n)})-\frac{1}{n},$$
$$1-\frac{4}{C_{\beta}}\log(n)N^d\exp\Big(-\frac{C_{\beta}}{48} \widetilde{b}_2\frac{K}{\log(n)}\Big)-\frac{1}{n},$$
and
$$1-\frac{4}{C_{\beta}}\log(n)N^d\exp\Big(-\frac{C_{\beta}}{48} \widetilde{b}_2\frac{K}{\log(n)}\Big)-\frac{1}{n},$$
respectively, where the notation of Lemma \ref{Omega-Oscar-Paper} is used. Next, by Assumption \ref{assume:tv functions} we have that the event $E_3=\Big\{ \max_{1\le i\le n,1\le j\le m_i}\vert \epsilon_{i,j}\vert\le\frac{C}{c}\varpi_\epsilon \sqrt{2\log(nm)} \Big\}$ happens with probability at least $1-\frac{1}{nm}.$ Further, using the proof of Theorem 2 in \cite{hutter2016optimal} (see page 21 in that paper), if $\widetilde{\epsilon}\in \mathbb{R}^{N^d}$ has entries Sub-Gaussian with parameter $\frac{C^2}{c^2}\varpi_\epsilon^2$,  it yields that the following two inequalities hold simultaneously on an event $E_4$ of probability at least $1-2\frac{1}{\sqrt{\log(nm)}}$,
$$
\left\|\left(D^{\dagger}\right)^{\top} \widetilde{\epsilon}\right\|_{\infty} \leq \frac{C}{c}\varpi_\epsilon \rho \sqrt{2 \log (e N^d \sqrt{\log(nm)} )}, \quad\|\Pi \widetilde\epsilon\|_2 \leq \frac{2C}{c} \varpi_\epsilon \sqrt{2 \log (e \sqrt{\log(nm)})},
$$
where $\rho=\max _{1 \leq v \leq\left|E_{\text {lat }}\right|}\left\|D_{, v}^{\dagger}\right\|_2. $ Additionally, using Lemma \ref{H-R-prop4} and \ref{H-R-prop6} (correspondingly propositions 4 and 6 in \cite{hutter2016optimal}), we have that $\rho\le \widetilde{C(d)} \sqrt{\log (N)},$ where $\widetilde{C(d)}$ is a positive constant depending on $d.$ Hence, 
$$
\left\|\left(D^{\dagger}\right)^{\top} \widetilde{\epsilon}\right\|_{\infty} \leq \frac{C}{c}\varpi_\epsilon C(d) \sqrt{ \log(\frac{nm}{K})\log ( \frac{nm}{K}\sqrt{\log(nm)} )}, \quad\|\Pi \widetilde\epsilon\|_2 \leq  \frac{2C}{c}\varpi_\epsilon \sqrt{2 \log (e \sqrt{\log(nm)})},
$$
with probability at least $1-2\frac{1}{\sqrt{\log(nm)}}$. Here, $C(d)>0$ is a constant depending on $d,$ $C$, and $v_{\max}$. Finally, on the event $\Omega \cap E_1\cap E_2\cap E_3\cap E_4$ the following is satisfied. Let $\theta_1,\theta_2\in \mathbb{R}^{\sum_{i=1}^nm_i}$ such that $\vert\vert\Delta\vert\vert_2\le \eta$ and $\left\|\nabla_{G_K} \Delta\right\|_1 \leq 2 V_n$. First, by Lemma \ref{lemma9-oscar} and $E_3$
\begin{align*}
    \epsilon^T\Delta
    \leq \frac{4C}{c}\varpi_\epsilon \sqrt{2\log(nm)}V_n+\varepsilon^T\left({\theta_1}_I-{\theta_2}_I\right).
\end{align*}
Moreover, by Lemma \ref{lemma11-oscar-1}, $E_2$, and $E_4$
\begin{align*}
    \epsilon^T\Delta
    \leq &\frac{4C}{c}\varpi_\epsilon \sqrt{2\log(nm)}V_n
    + \max _{r=1,...,N^d }|I_r|\left\{\|\Pi \tilde{\epsilon}\|_2\eta+\left\|\left(D^{+}\right)^T \tilde{\epsilon}\right\|_{\infty}\left(4V_n\right)\right\}
    \\
    \le
    &
    \frac{4C}{c}\varpi_\epsilon \sqrt{2\log(nm)}V_n
    \\
    +& \frac{3}{2}  \widetilde{b}_1 K\left\{\frac{C}{c}\varpi_\epsilon C(d)  \log ( \frac{nm}{K}\sqrt{\log(nm)} )\left(4V_n\right)+\frac{2C}{c} \varpi_\epsilon \sqrt{2 \log (e \sqrt{\log(nm)})}\eta\right\}
    \\
    \le&
     C_3\varpi_\epsilon K\log^{1/4}(nm)\eta+ C_4\varpi_\epsilon K \log (n m)  V_n,
\end{align*}
for some positive constants $C_3$ and $C_4$. Here $C_4$ depends on $d$. Thus, 
\begin{align*}
& \sup _{\|\Delta\|_2 \leq \eta,\left\|\nabla_{G_K} \Delta\right\|_1 \leq 2 V_n} \sum_{i=1}^n \sum_{j=1}^{m_i} \Delta_{i,j} \epsilon_{i,j} \le C_3\varpi_\epsilon K\log^{1/4}(nm)\eta+ C_4\varpi_\epsilon K \log (n m)  V_n.
\end{align*}
\end{proof}
The following lemma controls $\left\|\nabla_G \theta^*\right\|_1$. For the grid graph, $G$, considered in \cite{sadhanala2016total}, $\left\|\nabla_G \theta^*\right\|_1 \asymp n^{1-1 / d}$. Subsequently, \cite{madrid2020adaptive} establish that 
$\left\|\nabla_G \theta^*\right\|_1 \asymp n^{1-1 / d}$ when $G =G_K$, provided independence and Lipschitz conditions are met. Notably, our subsequent result broadens the scope of these findings by encompassing the spatio-temporal dependence assumption articulated in Assumption \ref{assume:tv functions}.
 \begin{lemma}\label{Lemma-nabla}
 Let Assumption \ref{assume:tv functions} hold. Consider $K=\log^{2+l}(nm)$ for some $l>0$. It follows that
     \begin{equation}
    \label{eqn:tv_bound-1}
    \| \nabla_{G_K}\theta^*  \|_1 \,=\, O_{\mathbb{P}}\left((nm)^{1-1 / d}\log^{2+l+\frac{2+l}{d}}(nm) \right).
\end{equation} 
 \end{lemma}
 \begin{proof}
Let 
\begin{equation}
\label{varepsilon-lastlemma}
\varepsilon=\Big(\left\lceil\frac{3 \sqrt{d}  (nm)^{1 / d}}{ aK^{1 / d}}\right\rceil\Big)^{-1}=N^{-1},
\end{equation}
with $N$ as in Lemma \ref{Omega-Oscar-Paper}. Here $a=1 /\left(  4Cc_{2,d}c_2v_{\text{max} }\right)^{1 / d}$ for a positive constants $c_2$ and, $v_{\max}$ and $C$ are from Assumption \ref{assume:tv functions}. Moreover $c_{2,d}$ is a positive constant that depends on $d.$ For the rest of the proof we use the notation of Appendix \ref{plsection} and Appendix \ref{additionaNot}. Consider the event
$$
\Lambda_\varepsilon=\left\{x_{1,1} \in B_{4 \varepsilon}(\mathcal{S}) \cup\left[(0,1)^d \backslash \Omega_{4 \varepsilon}\right]\right\},
$$
and note that
\begin{align}
\label{equation1-nabla}
\mathbb{P}\left(\Lambda_\varepsilon\right) & =\int_{B_{4 \varepsilon}(\mathcal{S}) \cup(0,1)^d \backslash \Omega_{4 \epsilon}} v(z) \mu(d z)\nonumber \\
& \leq c_{2,d}v_{\max } Vol\left[\left\{B_{4 \varepsilon}(\mathcal{S}) \cup(0,1)^d \backslash \Omega_{4 \varepsilon}\right\}\right] \nonumber\\
& \leq c_{2,d}v_{\max } C_{\mathcal{S}} (4 \varepsilon),
\end{align}
for a positive constant $c_{2,d}$, where the last inequality follows from the Definition \ref{Piecewise-def}. Denote by,
$$
J=\left\{(i,j) \in[n]\times[m_i]: x_{i,j} \in \Omega_{4 \varepsilon} \backslash B_{4 \varepsilon}(\mathcal{S})\right\}.
$$
Observe that
\begin{align}
&\sum_{((i,j), (i',j')) \in E_K}\left|\theta_{i,j}^*-\theta_{i',j'}^*\right| \nonumber
\\
& =\sum_{((i,j),(i', j')) \in E_K, (i,j), (i',j') \in J}\left|\theta_{i,j}^*-\theta_{i',j'}^*\right|+\sum_{((i,j), (i',j')) \in E_K, (i,j) \notin J \text { or } (i',j') \notin J}\left|\theta_{i,j}^*-\theta_{i',j'}^*\right|\nonumber \\
& \leq \sum_{((i,j), (i',j')) \in E_K (i,j), (i',j') \in J}\left|f^*\left(x_{i,j}\right)-f^*\left(x_{i',j'}\right)\right|+2\left\|f^*\right\|_{\infty} K \tau_d \widetilde{nm},\label{equation0-nabla}
\end{align}
with $\tau_d$ a positive constant that depends on $d$ and 
\begin{equation}
\label{nmtilde}
    \widetilde{nm}=|\{(i,j)\in [n]\times[m_i]\} \backslash J|.
\end{equation}
The inequality happens with high probability as shown in Lemma \ref{Omega-Oscar-Paper}. In fact Lemma \ref{Omega-Oscar-Paper} provides the following inequality $$d_{\max}\le 3Cc_2c_{2,d} v_{\max }  \tilde{a}^d K,$$ 
with probability at least $$1-Cnm\left\{\frac{4}{C_{\beta}}\log(n)\exp(-\frac{C_{\beta}}{24}\frac{K}{\log(n)})+\frac{4}{C_{\beta}}C\log(n)nm\exp(-\frac{C_{\beta}}{24}\frac{K}{\log(n)})\right\}-\frac{1}{n},$$ where $\tilde{a}=1 /\left(cc_1c_{1,d}v_{\text{min} }\right)^{1 / d}$ and $c_1$ is positive constant. Moreover, $c_{1,d}>0$ is a constant that depends on $d,$ and the constants $c$ and $v_{\min}$ are from Assumption \ref{assume:tv functions}. Hence, the subsequent observation validates Inequality (\ref{equation0-nabla}). Let
 $x_{i,j} \notin \Omega_\epsilon \backslash B_\epsilon(\mathcal{S})$. By triangle inequality
$$
\left|f^*\left(x_{i,j}\right)-f^*\left(x_{i',j'}\right)\right| \leq 2\|f^*\|_{\infty} .
$$
Next we bound $\widetilde{nm}$. To this end, we make use of the event $\Omega=\bigcap_{i=1}^n\{x_{i,j}=x_{i,j}^*,\ \forall \ 1\le j\le m_i\}$, where $x_{i,j}^*$ are defined  in Appendix \ref{BetaMsection}. Here we consider \begin{equation}
\label{Lchoice-lastlemma}
    L=\frac{2}{C_{\beta}}\log(n)
\end{equation}
and
\begin{equation}
\label{Block-lastlemma}
\frac{c_1n}{L}\le |\mathcal{J}_{e,l}|, |\mathcal{J}_{o,l}|\le \frac{c_2n}{L}
\end{equation}
with $c_1$ and $c_2$ positive constants and $C_{\beta}$ from Assumption  \ref{assume:tv functions}, for $l=1,...,L.$ Then observe that
\begin{align}
\label{aux-eq1-lastlemma}
    &\mathbb{P}\Big(\widetilde{nm}\ge\frac{3}{2}(\frac{2}{C_{\beta}}\log(n))\Big(\sum_{i=1}^nm_{i}\Big)c_{2,d}v_{\max}C_{\mathcal{S}}4\varepsilon\Big)\nonumber
    \\
=&\mathbb{P}\Big(\sum_{i=1}^n\sum_{j=1}^{m_i}\mathbf{1}_{\{x_{i,j}\in B_{4 \varepsilon}(\mathcal{S}) \cup\left[(0,1)^d \backslash \Omega_{4 \varepsilon}\right]\}}\ge\frac{3}{2}(\frac{2}{C_{\beta}}\log(n))\Big(\sum_{i=1}^nm_{i}\Big)c_{2,d}v_{\max}C_{\mathcal{S}}4\varepsilon\Big)\nonumber
    \\
    \le& \mathbb{P}\Big(\Big\{\sum_{i=1}^n\sum_{j=1}^{m_i}\mathbf{1}_{\{x_{i,j}\in B_{4 \varepsilon}(\mathcal{S}) \cup\left[(0,1)^d \backslash \Omega_{4 \varepsilon}\right]\}}\ge\frac{3}{2}(\frac{2}{C_{\beta}}\log(n))\Big(\sum_{i=1}^nm_{i}\Big)c_{2,d}v_{\max}C_{\mathcal{S}}4\varepsilon\Big\}\cap \Omega \Big)+\frac{1}{n},
\end{align}
where the last inequality is followed by the exponential decay of the $\beta$-mixing coefficients, see Assumption \ref{assume:tv functions}.
Using a union bound argument and the definition of $L$ in Equation (\ref{Lchoice-lastlemma}),
\begin{align}
\label{aux-eq2-lastlemma}
    &\mathbb{P}\Big(\Big\{\sum_{i=1}^n\sum_{j=1}^{m_i}\mathbf{1}_{\{x_{i,j}\in B_{4 \varepsilon}(\mathcal{S}) \cup\left[(0,1)^d \backslash \Omega_{4 \varepsilon}\right]\}}\ge\frac{3}{2}(\frac{2}{C_{\beta}}\log(n))\Big(\sum_{i=1}^nm_{i}\Big)c_{2,d}v_{\max}C_{\mathcal{S}}4\varepsilon\Big\}\cap \Omega \Big)\nonumber
    \\
    \le&\sum_{l=1}^L\mathbb{P}\Big(\Big\{\sum_{i\in \mathcal{J}_{e,l}}\sum_{j=1}^{m_i}\mathbf{1}_{\{x_{i,j}^*\in B_{4 \varepsilon}(\mathcal{S}) \cup\left[(0,1)^d \backslash \Omega_{4 \varepsilon}\right]\}}\ge\frac{3}{2}\Big(\sum_{i\in\mathcal{J}_{e,l}}m_{i}\Big)c_{2,d}v_{\max}C_{\mathcal{S}}4\varepsilon\Big\} \Big)\nonumber
    \\
    +&\sum_{l=1}^L\mathbb{P}\Big(\Big\{\sum_{i\in \mathcal{J}_{o,l}}\sum_{j=1}^{m_i}\mathbf{1}_{\{x_{i,j}^*\in B_{4 \varepsilon}(\mathcal{S}) \cup\left[(0,1)^d \backslash \Omega_{4 \varepsilon}\right]\}}\ge\frac{3}{2}\Big(\sum_{i\in\mathcal{J}_{e,l}}m_{i}\Big)c_{2,d}v_{\max}C_{\mathcal{S}}4\varepsilon\Big\} \Big).
\end{align}
By Inequality (\ref{equation1-nabla}),
\begin{align}
\label{aux-eq3-lastlemma}
    &\mathbb{P}\Big(\Big\{\sum_{i\in \mathcal{J}_{e,l}}\sum_{j=1}^{m_i}\mathbf{1}_{\{x_{i,j}^*\in B_{4 \varepsilon}(\mathcal{S}) \cup\left[(0,1)^d \backslash \Omega_{4 \varepsilon}\right]\}}\ge\frac{3}{2}\Big(\sum_{i\in\mathcal{J}_{e,l}}m_{i}\Big)c_{2,d}v_{\max}C_{\mathcal{S}}4\varepsilon\Big\} \Big)\nonumber
    \\
    \le&
    \mathbb{P}\Big(\Big\{\sum_{i\in \mathcal{J}_{e,l}}\sum_{j=1}^{m_i}\mathbf{1}_{\{x_{i,j}^*\in B_{4 \varepsilon}(\mathcal{S}) \cup\left[(0,1)^d \backslash \Omega_{4 \varepsilon}\right]\}}\ge\frac{3}{2}\Big(\sum_{i\in\mathcal{J}_{e,l}}m_{i}\Big)\mathbb{P}\left(\Lambda_\varepsilon\right)\Big\} \Big).
\end{align}
Observe that
$$\sum_{i\in \mathcal{J}_{e,l}}\sum_{j=1}^{m_i}\mathbf{1}_{\{x_{i,j}^*\in B_{4 \varepsilon}(\mathcal{S}) \cup\left[(0,1)^d \backslash \Omega_{4 \varepsilon}\right]}\sim\operatorname{Binomial}\left(\sum_{i \in \mathcal{J}_{e,l}}m_{i}, \mathbb{P}\left(\Lambda_\varepsilon\right)\right).$$
Therefore if $(nm)_{e,l}^{\prime}\sim\operatorname{Binomial}\left(\sum_{i \in \mathcal{J}_{e,l}}m_{i},c_{2,d}v_{\max}C_{\mathcal{S}}(4\varepsilon) \right)$
by Lemma \ref{BinCon} and Inequality (\ref{equation1-nabla}), it follows that 
\begin{align*}
    &\mathbb{P}\Big(\Big\{\sum_{i\in \mathcal{J}_{e,l}}\sum_{j=1}^{m_i}\mathbf{1}_{\{x_{i,j}^*\in B_{4 \varepsilon}(\mathcal{S}) \cup\left[(0,1)^d \backslash \Omega_{4 \varepsilon}\right]\}}\ge\frac{3}{2}\Big(\sum_{i\in\mathcal{J}_{e,l}}m_{i}\Big)\mathbb{P}\left(\Lambda_\varepsilon\right)\Big\} \Big)
    \\
    &\le \mathbb{P}\Big((nm)_{e,l}^{\prime}\ge\frac{3}{2}\Big(\sum_{i\in\mathcal{J}_{e,l}}m_{i}\Big)c_{2,d}v_{\max}C_{\mathcal{S}}(4\varepsilon)\Big)
    \\
    &\le\exp\left(-\frac{1}{12}\Big(\sum_{i\in\mathcal{J}_{e,l}}m_{i}\Big)c_{2,d}v_{\max}C_{\mathcal{S}}(4\varepsilon)\right)
    \le\exp\left(-\frac{4}{12}\Big(\frac{cc_1nm}{L}\Big)c_{2,d}v_{\max}C_{\mathcal{S}}\Big(\frac{aK^{1 / d}}{3 \sqrt{d}  (nm)^{1 / d} }\Big)\right).
\end{align*}
where the last inequality due to Inequality (\ref{Block-lastlemma}), Assumption \ref{assume:tv functions}{\bf{f}}, and the choice of $\varepsilon$ in Equation (\ref{varepsilon-lastlemma}).
Thus
\begin{align}
\label{aux-eq4-lastlemma}
    &\mathbb{P}\Big(\Big\{\sum_{i\in \mathcal{J}_{e,l}}\sum_{j=1}^{m_i}\mathbf{1}_{\{x_{i,j}^*\in B_{4 \varepsilon}(\mathcal{S}) \cup\left[(0,1)^d \backslash \Omega_{4 \varepsilon}\right]\}}\ge\frac{3}{2}\Big(\sum_{i\in\mathcal{J}_{e,l}}m_{i}\Big)\mathbb{P}\left(\Lambda_\varepsilon\right)\Big\} \Big)
    \nonumber
    \\
    \le&\exp\left(-\frac{C_{\mathcal{S}}cc_1ac_{2,d}v_{\max}}{9\sqrt{d}}\Big(\frac{K^{1 / d}nm^{1-1/d}}{L}\Big)\right).
\end{align}
Similarly
\begin{align}
\label{aux-eq5-lastlemma}
    &\mathbb{P}\Big(\Big\{\sum_{i\in \mathcal{J}_{o,l}}\sum_{j=1}^{m_i}\mathbf{1}_{\{x_{i,j}^*\in B_{4 \varepsilon}(\mathcal{S}) \cup\left[(0,1)^d \backslash \Omega_{4 \varepsilon}\right]\}}\ge\frac{3}{2}\Big(\sum_{i\in\mathcal{J}_{e,l}}m_{i}\Big)\mathbb{P}\left(\Lambda_\varepsilon\right)\Big\} \Big)\nonumber
    \\
    \le&\exp\left(-\frac{C_{\mathcal{S}}cc_1ac_{2,d}v_{\max}}{9\sqrt{d}}\Big(\frac{K^{1 / d}nm^{1-1/d}}{L}\Big)\right).
\end{align}
By our choice of $L$ in Equation (\ref{Lchoice-lastlemma}), Inequalities (\ref{aux-eq1-lastlemma}), (\ref{aux-eq2-lastlemma}), (\ref{aux-eq3-lastlemma}), (\ref{aux-eq4-lastlemma}) and (\ref{aux-eq5-lastlemma})
\begin{align}
\label{equation2-nabla}
    &\mathbb{P}\Big(\widetilde{nm}\ge\frac{3}{2}(\frac{2}{C_{\beta}}\log(n))\Big(\sum_{i=1}^nm_{i}\Big)c_{2,d}v_{\max}C_{\mathcal{S}}(4\varepsilon)\Big)\nonumber
    \\
   \le& \frac{4}{C_\beta}\log(n)\exp\left(-\frac{C_{\beta}C_{\mathcal{S}}cc_1ac_{2,d}v_{\max}}{18\sqrt{d}}\Big(\frac{K^{1 / d}nm^{1-1/d}}{\log(n)}\Big)\right)+\frac{1}{n}.
\end{align}
Using (\ref{equation2-nabla}), Assumption \ref{assume:tv functions}{\bf{b}} and {\bf{f}} and the choice of $\varepsilon$ in Equation (\ref{varepsilon-lastlemma}) we obtain that
\begin{equation}
\label{equation3-nabla}
    2\left\|f^*\right\|_{\infty} K \tau_d \widetilde{nm}\le \frac{8}{C_{\beta}\sqrt{d}}a\tau_{d}c_{2,d}v_{\max}C_{\mathcal{S}}C_{f^*}C\log(n)nm^{1-1/d}K^{1+1/d},
\end{equation}
with high probability.
Therefore it remains to bound $$\sum_{((i,j), (i',j')) \in E_K (i,j), (i',j') \in J}\left|f^*\left(x_{i,j}\right)-f^*\left(x_{i',j'}\right)\right|,$$ the first term in the right hand side of Inequality (\ref{equation0-nabla}). To that end, we notice that if $((i,j), (i',j')) \in E_K$, then we observe that
$$
\left\|x_{i,j}-x_{i',j'}\right\|_2 \leq  R_{K, \max } \leq \frac{3\sqrt{d}\widetilde{a}}{a}\varepsilon,
$$
where the last inequality happens with probability at least $1-\frac{4}{C_{\beta}}C\log(n)nm\exp(-\frac{C_{\beta}}{24}\frac{K}{\log(n)})-\frac{1}{n}$, by Lemma \ref{dist-bet-Knn} and the choice of $\varepsilon$ in Equation (\ref{varepsilon-lastlemma}).
Since $f^*$ is 
piecewise Lipschitz (see Definition \ref{Piecewise-def}) and the choice of $\varepsilon$ in Equation (\ref{varepsilon-lastlemma}), 
\begin{align}
\label{equation4-nabla}
    \sum_{((i,j), (i',j')) \in E_K (i,j), (i',j') \in J}\left|f^*\left(x_{i,j}\right)-f^*\left(x_{i',j'}\right)\right|\le&   \widetilde{a} \tau_{d} L_{f^*}Cnm^{1-1/d}K^{1+1/d}.
\end{align}
Thus, from Inequality (\ref{equation0-nabla}), (\ref{equation3-nabla}) and (\ref{equation4-nabla}), we conclude that
\begin{equation*}
    \sum_{((i,j), (i',j')) \in E_K}\left|\theta_{i,j}^*-\theta_{i',j'}^*\right|=O_{\mathbb{P}}\left((nm)^{1-1 / d}\log^{2+l+\frac{2+l}{d}}(nm) \right).
\end{equation*}
 \end{proof}

\newpage

\section{Fixed design case}

For the data $\{(x_{i,j},y_{i,j})\}_{i=1,j=1}^{n,m} \subset [0,1] \times \mathbb{R}$  as in the main paper, we now consider the model
\begin{equation}
    \label{eqn:new_model}
    y_{i,j}\,=\,f^*(x_{i,j}) \,+\, \epsilon_{i,j} + \delta_i(x_{i,j})\\
\end{equation}
where the covariates $\{x_{i,j}\}_{i=1,j=1}^{n,m}$   are fixed and, following   \cite{cai2011optimal}, we assume that 
\begin{equation}
    \label{eqn:fixed_design}
    x_{1,j} \,=\, x_{2,j}\,=\, \ldots\,=\, x_{n,j},\,\,\,\,\forall j = 1,\ldots, m.
\end{equation}

Before presenting the Trend Filtering estimator for this fixed design case, we recall some  notation on discrete  total variation. For a vector $\theta \in  \mathbb{R}^m,$ define $D^{(0)}(\theta) = \theta, D^{(1)}(\theta) = (\theta_2 - \theta_1,\dots,\theta_m - \theta_{m - 1})^{\top}$ and $D^{(k)}(\theta)$, for $k \geq 2$, is recursively defined as $D^{(k)}(\theta) = D^{(1)}(D^{(k - 1)}(\theta))$, where $D^{(k)}(\theta) \in  \mathbb{R}^{m - k}$.  With this notation, for $k\geq 1$,  the $k$th order total variation of a vector $\theta$  is given as 
\begin{equation}
	\mathrm{TV}^{(k)}(\theta) = m^{k - 1} \|D^{(k)}(\theta)\|_{1}.
\end{equation}
Then, we define the estimator
\begin{equation}
    \label{eqn:fixed_design_est}
    \hat{\theta}\,=\,  \underset{ \theta \in  \mathbb{R}^m  }{\arg \min }\,\left\{ \frac{1}{nm}\sum_{i=1}^n\sum_{j=1}^m  (y_{i,j}  - \theta_{j} )^2  \,+\,\lambda  \mathrm{TV}^{(k)}(\theta)   \right\}
\end{equation}
for a tuning parameter $\lambda>0$, and then set $\hat{f}(x_{i,j} ) = \hat{\theta}_j$ for all $i,j$.

We now present our main result for the fixed design setting.

\begin{theorem}
    \label{thm:10}
    Let $\theta^* \in \mathbb{R}^m$ be given as $\theta^*_j =  f^*(x_{1,j})$ for $j=1,\ldots, m$.
    Suppose that $\{x_{i,j}\}_{i=1,j=1}^{n,m}$   are fixed, (\ref{eqn:fixed_design}) holds,  the $\{\epsilon_{i,j}\}_{i=1,j=1}^{n,m}$ are independent, the $\{\delta_{i}\}_{i=1}^{n}$ are independent, and $ \sigma_\delta$ is independent of  $\sigma_\epsilon$. 
    Then
    \[
       \frac{1}{m}\sum_{j=1}^m (  \theta_j^* -   \hat{\theta}_j )^2 \,=\, O_{\mathbb{P}}\left( \frac{(V^*)^{1/(2k+1)}}{ (nm)^{2k/(2k+1)}    }  \,+\,\frac{1}{n}  \right),
    \]
    where $V^* :=  \mathrm{TV}^{(k)}(\theta^*)$, provided that 
    \begin{equation}
        \label{eqn:815}
        \underset{n \rightarrow \infty}{\lim }\, \frac{\sqrt{\log m }}{ m^{1/(4k+2)} n^{(k+1)/(2k+1)}(V^*)^{1/(2k+1)}  } \,\rightarrow \,0, 
    \end{equation}
    and that $\lambda$ is chosen to satisfy 
    \[
     \lambda \,\asymp\, \frac{ m^{1/(2k+1)} n^{-2k/(2k+1)}(V^*)^{2/(2k+1)} \,+\, mn^{-1}  }{ m V^*  }.
    \]
\end{theorem}


\begin{remark}
    \label{rem_fd}
   We can see from Theorem \ref{thm:10} that the Trend Filtering estimator in the fixed design case attains the rate 
    \begin{equation}
        \label{eqn:rate_fd_est}
        \frac{(V^*)^{1/(2k+1)}}{ (nm)^{2k/(2k+1)}    }  \,+\,\frac{1}{n},
    \end{equation}
   in terms of the mean squared error.  This, in the canonical scaling $V^* = O(1)$, matches up to logarithm factors,  the  upper bounds obtained in the random design case as stated in Theorems \ref{thm:main tv} and \ref{Penalized}.  These upper bound is different to that in Theorem 2.2 from \cite{cai2011optimal} which considered the $\ell_2$ rather than the mean squared error as in Theorem \ref{thm:10} and also it considered a different function class. 

   Finally, while Theorem \ref{thm:10} requires the data to be independent, this can be relaxed to allow temporal dependence using similar arguments to those in  Theorems \ref{thm:main tv} and \ref{Penalized}.

\end{remark}
\begin{proof}
Let us write   $\delta_{i,j} := \delta_i(x_{i,j})$. We see that the basic inequality,  for any $\theta \in \Lambda \,:=\, \{ s \hat{\theta}+(1-s)\theta^* \,:\, s\in [0,1]  \}$,
\[
 \frac{1}{nm}\sum_{i=1}^n\sum_{j=1}^m  (y_{i,j}  - \theta_{j} )^2  \,+\,\lambda  \mathrm{TV}^{(k)}(\theta) \,\leq\,   \frac{1}{nm}\sum_{i=1}^n\sum_{j=1}^m  (y_{i,j}  - \theta_{j}^* )^2  \,+\,\lambda  \mathrm{TV}^{(k)}(\theta^*)
\]
 which implies
 \[
  \begin{array}{lll}
    \displaystyle    \frac{1}{m} \sum_{j=1}^m (\theta_j -  \theta^*_j)^2  & \leq    & \displaystyle  \frac{2}{nm}\sum_{i=1}^n\sum_{j=1}^m (\theta_{j}  - \theta_j^*   )( \epsilon_{i,j} + \delta_{i,j}  )  \,+\,  \lambda (  \mathrm{TV}^{(k)}(\theta^*) -   \mathrm{TV}^{(k)}(\theta)   )\\
       &= & \displaystyle  \frac{2}{m}\sum_{j=1}^m ( \theta_j - \theta_j^* )\left( \frac{1}{n} \sum_{j=1}^n \epsilon_{i,j}  \right) \,+\, \frac{2}{m}\sum_{j=1}^m ( \theta_j - \theta_j^* )\left( \frac{1}{n} \sum_{j=1}^n \delta_{i,j}  \right)  \,+\,\\
       &&\displaystyle \lambda (  \mathrm{TV}^{(k)}(\theta^*) -   \mathrm{TV}^{(k)}(\theta)   ).
  \end{array}
 \]
 Hence, letting 
 \[
   \bar{\delta}_{j} \,=\, \frac{1}{n} \sum_{i=1}^n \delta_{i,j}
   \]
 for all $j= 1,\ldots,m$, we obtain that 
  \[
  \begin{array}{lll}
    \displaystyle    \frac{1}{m} \sum_{j=1}^m (\theta_j -  \theta^*_j)^2  & \leq    & \displaystyle  \frac{2}{nm}\sum_{i=1}^n\sum_{j=1}^m (\theta_{j}  - \theta_j^*   )( \epsilon_{i,j} + \delta_{i,j}  )  \,+\,  \lambda (  \mathrm{TV}^{(k)}(\theta^*) -   \mathrm{TV}^{(k)}(\theta)   )\\
       &\leq  & \displaystyle  \frac{2}{m}\sum_{j=1}^m ( \theta_j - \theta_j^* )\left( \frac{1}{n} \sum_{j=1}^n \epsilon_{i,j}  \right) \,+\,  \frac{2}{m} \| \theta-\theta^*\| \cdot\| \bar{\delta}\|   \,+\,\\
       &&\displaystyle \lambda (  \mathrm{TV}^{(k)}(\theta^*) -   \mathrm{TV}^{(k)}(\theta)   )\\
          &\leq  & \displaystyle  \frac{2}{m}\sum_{j=1}^m ( \theta_j - \theta_j^* )\left( \frac{1}{n} \sum_{j=1}^n \epsilon_{i,j}  \right) \,+\,  \frac{1}{2m} \| \theta-\theta^*\|^2\,+\, \frac{2}{m}  \| \bar{\delta}\|^2   \,+\,\\
       &&\displaystyle \lambda (  \mathrm{TV}^{(k)}(\theta^*) -   \mathrm{TV}^{(k)}(\theta)   ).
  \end{array}
 \]

 This implies that 
 \begin{equation}
     \label{eqn:main_fd}
        \begin{array}{lll}
    \displaystyle    \frac{1}{2m} \sum_{j=1}^m (\theta_j -  \theta^*_j)^2  & \leq    &
    \displaystyle  \frac{2}{m}\sum_{j=1}^m ( \theta_j - \theta_j^* )\left( \frac{1}{n} \sum_{j=1}^n \epsilon_{i,j}  \right) \,+\,  \frac{2}{m}  \| \bar{\delta}\|^2   \,+\,\\
       &&\displaystyle \lambda (  \mathrm{TV}^{(k)}(\theta^*) -   \mathrm{TV}^{(k)}(\theta)   ).
  \end{array}
 \end{equation}

Therefore, from (\ref{eqn:main_fd}),  for $\theta \in \Lambda$, 
\begin{equation}
    \label{eqn:800}
       \begin{array}{lll}
    \displaystyle  \mathrm{TV}^{(k)}(\theta- \theta^*)  &\leq &\displaystyle \mathrm{TV}^{(k)}(\theta )\,+\,  \mathrm{TV}^{(r)}(\theta^*) \\
        &\leq  &\displaystyle  \frac{2}{m \lambda}\sum_{j=1}^m ( \theta_j - \theta_j^* )\left( \frac{1}{n} \sum_{j=1}^n \epsilon_{i,j}  \right) \,+\,  \frac{2}{m\lambda }  \| \bar{\delta}\|^2  \,+\,  2\mathrm{TV}^{(k)}(\theta^*).
   \end{array}
\end{equation}

Next, let $\theta \in \Lambda$ and suppose that $\| \theta - \theta^*\| \leq \eta $ for some $\eta>0$, and $\mathrm{TV}^{(k)}(\theta ) \geq 5 \mathrm{TV}^{(k)}(\theta^* )$. Then 
 \[
 \mathrm{TV}^{(k)}(\theta - \theta^* )   \,\geq \, \mathrm{TV}^{(k)}(\theta  ) - \mathrm{TV}^{(k)}(\theta^* )  \geq 4 \mathrm{TV}^{(k)}(\theta^* ). 
 \]
Hence, we can set, 
\[
s   =  \frac{4  \mathrm{TV}^{(k)}(\theta^* )}{ \mathrm{TV}^{(k)}(\theta - \theta^* )  } \in [0,1],
\]
and 
\[
   \theta^{\prime }\,=\, s\theta + (1-s) \theta^* \in \Lambda. 
\]
Then
\[
\|\theta^{\prime} - \theta^*\|^2 \,\leq\, \|\theta- \theta^*\|^2 \,\leq\, \eta^2,
\]
and
\[
 \mathrm{TV}^{(k)}(\theta^* -  \theta^{\prime }  ) \,=\, s\mathrm{TV}^{(k)}(\theta^* -  \theta  )\,=\, 4  \mathrm{TV}^{(k)}(\theta^* ).
\]
Hence, from (\ref{eqn:800}), 
 \[
 2  \mathrm{TV}^{(k)}(\theta^* ) \,\leq \,  \frac{2}{m \lambda}\sum_{j=1}^m ( \theta_j^{\prime } - \theta_j^* )\tilde{\epsilon}_j \,+\,  \frac{2}{m\lambda }  \| \bar{\delta}\|^2,
 \]
 where 
 \[
\tilde{\epsilon}_j \,:=\,  \frac{1}{n} \sum_{i=1}^n \epsilon_{i,j},
 \]
 for $j=1,\ldots,m$.
 
 Then, letting 
 \begin{equation}
     \label{eqn:fd_lambda}
      \lambda \,=\, \frac{\eta^2}{4 m \mathrm{TV}^{(k)}(\theta^*)   },
 \end{equation}
 we obtain 
 \[
  \frac{2}{m} (\theta^{\prime}-\theta^*)^{\top}\tilde{\epsilon}  +   \frac{2}{m }  \| \bar{\delta}\|^2 \,\geq\, \frac{\eta^2}{2m}.
 \]

Therefore, 
\begin{equation}
    \label{eqn:810}
    \Omega_1\,:=\, \left\{ \underset{\theta\,:\,  \| \theta -\theta^*   \|\leq \eta }{ \sup}\,  \,   \mathrm{TV}^{(k)}(\theta)  \geq\,  5\mathrm{TV}^{(k)}(\theta^*)   \right\} \,\subset \, \Omega_2,
\end{equation}
where
\[
\Omega_2\,:=\, \left\{ \underset{\theta\,:\,  \| \theta -\theta^*   \|\leq \eta,\, \,\mathrm{TV}^{(k)}(\theta - \theta^*)\,\leq\,  4\mathrm{TV}^{(k)}(\theta^*)     }{ \sup} 
 \frac{2}{m} (\theta-\theta^*)^{\top}\tilde{\epsilon}  +   \frac{2}{m }  \| \bar{\delta}\|^2 \,\geq\, \frac{\eta^2}{2m} \right\}.
\]
Next, if $\| \theta^* - \hat{\theta}\|>\eta$,  then there exists $\theta \in \Lambda$ such that $\| \theta - \theta^*\| =\eta$. Therefore, from (\ref{eqn:main_fd}),
\[
  \frac{\eta^2}{2m}\,\leq\,   \frac{2}{m} (\theta-\theta^*)^{\top}\tilde{\epsilon}\,  + \,  \frac{2}{m }  \| \bar{\delta}\|^2 \,+\, \lambda  \mathrm{TV}^{(k)}(\theta^*),
\]
which by our choice of $\lambda$,
\[
  \frac{\eta^2}{4m}\,\leq\,   \frac{2}{m} (\theta-\theta^*)^{\top}\tilde{\epsilon}\,  + \,  \frac{2}{m }  \| \bar{\delta}\|^2.
\]
Therefore, setting $\eta^2 = \eta_1^2 + \eta_2^2$ for some $\eta_1>0$ and $\eta_2>0$ to be chosen later, we obtain that 
\begin{equation}
    \label{eqn:811}
    \begin{array}{lll}
      \displaystyle   \mathbb{P}( \|\hat{\theta} -\theta^*\| >  \eta  )   &\leq&  \displaystyle \mathbb{P}( \{\|\hat{\theta} -\theta^*\| >  \eta \} \cap \Omega_1^c ) \,+\,\mathbb{P}(\Omega_1)  \\
         & \leq&    \displaystyle \mathbb{P}( \{\|\hat{\theta} -\theta^*\| >  \eta \} \cap \Omega_1^c ) \,+\,\mathbb{P}(\Omega_2)\\
          & \leq&   \displaystyle \mathbb{P}\bigg( \underset{\theta  \in \Lambda\,:\,\|\theta-\theta^*\|\leq \eta,\,\,\,  \mathrm{TV}^{(k)}(\theta)  \leq\,  5\mathrm{TV}^{(k)}(\theta^*)     }{\sup} \big\{ \frac{2}{m} (\theta-\theta^*)^{\top}\tilde{\epsilon}\,  + \,  \frac{2}{m }  \| \bar{\delta}\|^2\big\} \geq \frac{\eta^2}{4m}   \bigg)\,+\,\\
          & & \displaystyle \mathbb{P}(\Omega_2) \\
             & \leq&   \displaystyle 2\mathbb{P}\bigg( \underset{\theta  \in \Lambda\,:\,\|\theta-\theta^*\|\leq \eta,\,\,\,  \mathrm{TV}^{(k)}(\theta)  \leq\,  5\mathrm{TV}^{(k)}(\theta^*)     }{\sup} \big\{ \frac{2}{m} (\theta-\theta^*)^{\top}\tilde{\epsilon}\,  + \,  \frac{2}{m }  \| \bar{\delta}\|^2\big\} \geq \frac{\eta^2}{4m}   \bigg)\\
                         & \leq&   \displaystyle 2\mathbb{P}\bigg( \underset{\theta  \in \Lambda\,:\,\|\theta-\theta^*\|\leq \eta,\,\,\,  \mathrm{TV}^{(k)}(\theta)  \leq\,  5\mathrm{TV}^{(k)}(\theta^*)     }{\sup}  \frac{2}{m} (\theta-\theta^*)^{\top}\tilde{\epsilon}\, \geq \frac{\eta_1^2}{4m}   \bigg)\,+\,  \\
                    &&   \displaystyle 2\mathbb{P}\bigg(   \frac{2}{m }  \| \bar{\delta}\|^2\geq \frac{\eta_2^2}{4m}   \bigg)\\     
    \end{array}
\end{equation}
where the last inequality follows from union bound. 

As a result, from the previous inequality and Markov's inequality, 
\begin{equation}
    \label{eqn:812}
        \begin{array}{lll}
         \mathbb{P}( \|\hat{\theta} -\theta^*\| >  \eta  ) &\leq& \displaystyle \frac{16}{\eta_1^2} \mathbb{E}\left(  \underset{\theta  \in \Lambda\,:\,\|\theta-\theta^*\|\leq \eta,\,\,\,  \mathrm{TV}^{(r)}(\theta)  \leq\,  5\mathrm{TV}^{(k)}(\theta^*)     }{\sup}   (\theta-\theta^*)^{\top}\tilde{\epsilon} \right)\,+\,\\
         &&\displaystyle \frac{16}{\eta_2^2  }\mathbb{E}\left(  \| \bar{\delta}\|^2  \right)\\
          & =:&\displaystyle A_1+A_2,
     \end{array}
\end{equation}
and we proceed to bound $A_1$ and $A_2$. 

To bound $A_1$,  we observe that the entries of $\sqrt{n} \tilde{\epsilon} $ are independent and satisfy that $\sqrt{n} \tilde{\epsilon}_j$ is sub-Gaussian with parameter uniformly bounded in $j$. As a result, 
\begin{equation}
    \label{eqn:813}
   \begin{array}{lll}
         A_1 &\leq& \displaystyle\frac{16}{\eta_1^2  \sqrt{n} } \mathbb{E}\left(  \underset{\theta  \in \mathbb{R}^n\,:\,\|\theta-\theta^*\|\leq \eta_1,\,\,\,  \mathrm{TV}^{(k)}(\theta)  \leq\,  5\mathrm{TV}^{(k)}(\theta^*)     }{\sup}  (\theta-\theta^*)^{\top}  (\sqrt{n} \tilde{\epsilon}) \right)\\
    &\leq & \displaystyle  \frac{C}{\eta_1^2  \sqrt{n} } \left[   \eta_1\left(  \frac{\sqrt{m  }\cdot\mathrm{TV}^{(k)}(\theta^*)     }{\eta_1} \right)^{1/(2k)} \,+\,\eta_1 \sqrt{\log (me)}   \right]
   \end{array}
\end{equation}
where $C_1 >0$  is a positive constant and the second inequality holds from Lemma B.1 in \cite{guntuboyina2020adaptive}.

Hence, given $\varepsilon>0$, there exists $C_{1,\varepsilon}>0$ depending on $\varepsilon$ such that 
\[
\eta_1 \,=\, C_{1,\varepsilon}  m^{1/(4k+2)} n^{-k/(2k+1)}(V^*)^{1/(2k+1)}, 
\]
satisfies 
\begin{equation}
    \label{eqn:814}
    \frac{C}{\eta_1^2  \sqrt{n} } \left[   \eta_1\left(  \frac{\sqrt{m  }\cdot\mathrm{TV}^{(k)}(\theta^*)     }{\eta_1} \right)^{1/(2k)} \,+\,\eta_1 \sqrt{\log (me)}   \right] \,\leq \, \frac{\varepsilon}{2}
\end{equation}
provided that 
\[
\frac{\sqrt{\log m }}{\sqrt{n} \eta_1  } \,=\,\frac{\sqrt{\log m }}{\sqrt{n} \cdot  C_{1,\varepsilon}  m^{1/(4k+2)} n^{-k/(2k+1)}(V^*)^{1/(2k+1)}  } \,\rightarrow \,0, 
\]
which holds by our assumption on (\ref{eqn:815}).

Furthermore, 
\[
\begin{array}{lll}
  \displaystyle  \mathbb{E}\left(  \| \bar{\delta}\|^2  \right) &= &  \displaystyle \mathbb{E}\bigg( \sum_{j=1}^m \bigg(\frac{1}{n} \sum_{i=1}^n \delta_{i,j}  \bigg)^2  \bigg)\\
   & = & \displaystyle \mathbb{E}\bigg(  \frac{1}{n^2} \sum_{j=1}^m \sum_{i=1}^n \delta_{i,j}^2  \,+\, \frac{2}{n^2} \sum_{j=1}^m \sum_{i \neq i^{\prime}  }  \delta_{i,j} \delta_{ i^{\prime}, j }  \bigg)\\
      & = & \displaystyle \mathbb{E}\bigg(  \frac{1}{n^2} \sum_{j=1}^m \sum_{i=1}^n \delta_{i,j}^2  \bigg)\\
         & = &\displaystyle \frac{1}{n^2} \sum_{j=1}^m \displaystyle \mathbb{E}\bigg(  \sum_{i=1}^n \delta_{i,j}^2  \bigg)\\
               & = & \displaystyle O\bigg( \frac{m}{n}\bigg).
\end{array}
\]
Hence, there exists a  constant $C_{2,\varepsilon}>0$  such that $\eta_2$ given as 
\[
\eta_2 \,=\, \frac{ C_{2,\varepsilon}  \sqrt{m} }{ \sqrt{n} }
\]
satisfies 
\begin{equation}
    \label{eqn:816}
    \frac{16}{\eta_2^2  }\mathbb{E}\left(  \| \bar{\delta}\|^2  \right) \,\leq\, \frac{\varepsilon}{2}.
\end{equation}
Therefore, (\ref{eqn:812})--(\ref{eqn:816})  lead to 
 \[
     \mathbb{P}( \|\hat{\theta} -\theta^*\| >  \eta  )\,\leq\, \varepsilon
 \]
 and the claim follows. 
\end{proof}

\newpage

\section{Canonical scaling}

Now we elaborate on the convergence rates obtained in Theorems \ref{thm:main tv} and \ref{Penalized}. Specifically, recall that these show that, ignoring logarithmic factors, the {\color{black}{Locally Adaptive Regression Splines}} estimator attains the rate 
\[ 
\frac{ (J_k(f^*) +1)^2 }{(mn)^{\frac{ 2 k   }{2 k +1 }} }   +  \frac{1}{  n }
\]
in terms of the $\ell_2$ squared error. Furthermore, according to (\ref{eqn:lower2}), in the function class 
\[
    \{  f^*\,:\,  J_k(f^*) \leq  V \}
\]
for $V \asymp 1$, the optimal rate is 
\[
 \frac{1}{n }  \,+\, \frac{1}{  (nm)^{\frac{2k}{2k+1}} }.  
\]
Thus, the {\color{black}{Locally Adaptive Regression Splines}} estimator is minimax optimal when $J_k(f^*) \asymp 1$. A question that we leave open is whether or not the {\color{black}{Locally Adaptive Regression Splines}} estimator is optimal for general $J_k(f^*)$ in the model with temporal and spatial dependence that we study.  However, we now elaborate more on why the setting $J_k(f^*)\asymp 1 $ is natural, and it is often referred as canonical scaling. Towards that end, consider that the regression function $f^*$ is fixed, thus independent of $n$ and $m$ in Assumption \ref{assume:tv functions}. Then, we recall that 
\[
J_k(f^*)\,=\,\mathrm{TV}(f^{ (k-1) })
\] 
which does not depend on $n$ and $m$, and hence we can think of $J_k(f^*)$ as being a constant or  $J_k(f^*) \,\asymp\, 1$. This is known as the canonical setting, see the discussion on Page 23 of  
\cite{padilla2024variance} and also \cite{sadhanala2016total,sadhanala2017higher}, and \cite{madrid2022risk}.

\newpage 
\section{Additional Technical Results for proof of Theorem \ref{Cest-d>1} and Theorem \ref{Pest-d>1}}

\begin{lemma}\label{lemma7}
   Suppose Assumption \ref{assume:tv functions} holds with $\mathbb{E}\left(\left\|\delta_i\right\|_{2}\right)<C_\delta$. Then for any $\kappa>0$
$$
\mathbb{P}\left(\left\|\frac{1}{n} \sum_{i=1}^n \delta_i\right\|_{2} \geq  \sqrt{\frac{C_\delta}{n \kappa}}\right) \leq \kappa .
$$ 
\end{lemma}

\begin{proof}
    Note that by Lemma \ref{expected-value-b}
$$ \mathbb E \bigg (\big  \| \frac{1}{n} \sum_{i=1}^n \delta_i\big  \|_\lt^2 \bigg) \lesssim    \frac{1}{n} .$$
    
$$
\mathbb{E}\left(\left\|\frac{1}{n} \sum_{i=1}^n \delta_i\right\|_{2}^2\right)=\frac{1}{n^2} \sum_{i, j=1}^n \mathbb{E}\left(\left\langle\delta_i, \delta_j\right\rangle_{2}\right)=\frac{1}{n^2} \sum_i^n \mathbb{E}\left(\left\langle\delta_i, \delta_i\right\rangle_{2}\right)\le\frac{C_\delta}{n} .
$$
The desired result follows by applying the Markov's inequality to $\mathbb{E}\left(\left\|\frac{1}{n} \sum_{i=1}^n \delta_i\right\|_{2}^2\right)$.
\end{proof}
\begin{lemma}
\label{lemma8}
For any signal $\theta \in \mathbb{R}^{\sum_{i=1}^n m_i}$ on the grid $\left\{x_{i,j}\right\}_{i=1, j=1}^{n, m_i}$, it holds that with probability 1,
$$
\left\|\theta^I\right\|_2 \leq\|\theta\|_2.
$$
\end{lemma}
\begin{proof}
    It follows directly from the definition.
\end{proof}

 \begin{lemma}
     \label{H-R-prop6}
     Let $d\ge2.$ For the incidence matrix of the regular grid on $N^d$ nodes in dimensions, the inverse scaling factor $\rho \leq$ $C(d)$, for some $C(d)>0$.
 \end{lemma}
 \begin{proof}
     This is Proposition 6 in \cite{hutter2016optimal}.
 \end{proof}

 \begin{lemma}
 \label{H-R-prop4}
      The incidence matrix $D_2$ of the $2 D$ grid on $N$ vertices has inverse scaling factor $\rho \lesssim \sqrt{\log (N)}$.
 \end{lemma}
 \begin{proof}
     This is Proposition 4 in \cite{hutter2016optimal}.
 \end{proof}

\vskip 0.2in
\bibliography{references}

\end{document}